\documentclass[acmsmall,screen,nonacm]{acmart}
\AtBeginDocument{%
  }

\newtoggle{supplementary}
\toggletrue{supplementary}

\setcopyright{none}
\acmYear{2027}
\acmDOI{}
\acmJournal{PACMPL}
\acmVolume{11}
\acmNumber{POPL}
\acmArticle{}
\acmMonth{1}

\usepackage{mathpartir} 
\usepackage{mdframed}
\usepackage{listings}
\usepackage{enumerate}  
\usepackage{mathtools}  
\usepackage{scalerel} 

\usepackage{macros}
\usepackage{preamble}
\PassOptionsToPackage{normalem}{ulem}
\usepackage{revisions}
\usepackage{camera-revisions}
\ShowRevisionsfalse
\ShowCameraRevisionsfalse

\CameraRevisionNewReferences{ayhan-intro-i,ayhan-intro-ii}
\CameraRevisionDeletedReferences{ayhan-intro}
\AtBeginEnvironment{thebibliography}{\CameraRevisionBibliographySetup}

\begin{document}
\newcounter{complrules}
\title[A Logical Interpretation of Two-Sided Type Systems]{\revreplace{Logical Foundations of Two-Sided Type Theory}{A Logical Interpretation of Two-Sided Type Systems}}

\author{Celia Mengyue Li}
\orcid{0000-0002-8023-3052}
\email{celia.li@bristol.ac.uk}
\affiliation{
  \institution{University of Bristol}
  \city{Bristol}
  \country{UK}
}

\author{Steven Ramsay}
\orcid{0000-0002-0825-8386}
\email{steven.ramsay@bristol.ac.uk}
\affiliation{%
  \institution{University of Bristol}
  \city{Bristol}
  \country{UK}
}


\begin{abstract}
  Two-sided type systems, introduced in POPL'24, are an extension of the traditional notion of type system that allows for stating and deriving typing judgements in which (a) assumptions can be made about the types of arbitrary terms and not only variables, and (b) conclusions can be made about any number of type assignments, and not exactly one.  In this work, we investigate \revreplace{the logical foundations of two-sided type systems}{a logical interpretation of a core of two-sided typing rules} in the sense of the propositions-as-types paradigm. \revadd{A key step is to strengthen the meaning of a distinguished left-side typing $M:A$ from ``$M$ is not a proof of $A$'' to ``$M$ constructs a refutation of $A$''.} \revreplace{We introduce new two-sided type systems \tintty{} and \tintcty{} that correspond with Wansing's bilateral logic \tint{} and its extension with Nelson's strong negation respectively}{Building on the core two-sided rules for pure \(\lambda\)-terms, we introduce new two-sided type systems \tintty{} and \tintcty{}, and establish a correspondence of derivability/type inhabitation with Wansing's bilateral logic \tint{} and its extension with Nelson's strong negation, respectively.}  Going beyond the propositional case, we introduce \holtty{} as a \revadd{two-sided }extension of \revreplace{Guevers'}{Geuvers'} \holty{}, and we show its expressive adequacy, its consistency and that it satisfies both the existence property and its dual.
\end{abstract}




\maketitle

\section{Introduction}

Two-sided type systems are a generalisation of the usual conception of type system, designed to allow the refutation as well as the verification of type assignments \cite{ramsay-walpole-popl24,li-etal-popl26}.  A two-sided typing judgement generally has shape
$
  M_1:A_1,\,\ldots,\,M_k:A_k \types N_1:B_1,\,\ldots,\,N_m:B_m
$
in which \emph{all} the $M_i$ and $N_j$ are arbitrary terms.  So, in a two-sided judgement: (i) one may make assumptions about the types of a compound expression $M_i:A_i$, not necessarily a variable, and (ii) one may conclude any number of typings $N_j:B_j$, not only a singleton.  \revreplace{The meaning of such a judgement}{Under the program interpretation of the earlier systems, the meaning of such a judgement} smoothly generalises that of a traditional judgement, namely: if \emph{all} of the $M_i$ each evaluates to a value of type $A_i$ then \emph{some} $N_j$ will evaluate to a value of type $B_j$ (or diverge).

\begingroup
\setlength{\emergencystretch}{1em}
\revadd{In this paper, we investigate a logical interpretation of a \emph{core} of these rules through the propositions-as-types paradigm.  The key step is to strengthen the meaning of a distinguished left-side typing $M:A$: instead of just saying that this particular term does not evaluate to a value of type $A$ -- or, under a proposed logical reading, is not a proof of $A$, it states that $M$ \emph{constructs a refutation of $A$}.  We develop new two-sided type systems with this reading and establish a correspondence with constructive logic at the level of derivability and type inhabitation.}
\par
\endgroup

\revadd{\paragraph{Two-Sided Typing for Programs}}
\revreplace{This generalisation opens up interesting possibilities for the kinds of properties that can be reasoned about using type assignment.  The following example is taken from \cite{ramsay-walpole-popl24}. Consider a typical definition of addition over natural number numerals in a simple PCF-like functional language:}{To recall the original motivation, consider the following example from \citet{ramsay-walpole-popl24}, defining addition in a PCF-like functional language:}
\[
  \abbv{add} \coloneqq \fixtm{f(x,y)}{\ifztm{y}{x}{\succtm\,(f\,(x,\predtm\,y))}}
\]
In the two-sided system of \cite{ramsay-walpole-popl24}, as usual one can show that, whenever the inputs $x$ and $y$ are numerals, then their addition will also evaluate to a numeral:
\[
  x:\natty,\,y:\natty \types \abbv{add}\,(x,y) : \natty
\]
However, one may also use the freedom of the two-sided judgement to show that, whenever the addition of $x$ and $y$ evaluates to a numeral, then $x$ and $y$ must have been numerals:
\[
  \abbv{add}\,(x,y) : \natty \types x : \natty \qquad \abbv{add}\,(x,y) : \natty \types y : \natty
\]

When discussing informally the behaviour of this program, it is quite natural to state that $\abbv{add}\,(x,(\abs{y}{y}))$ will \emph{not} evaluate to a numeral, but this is not something that can typically be made precise in a traditional type system.  However, using the freedom to conclude any number of type assignments allows for this statement to be formalised in the two-sided system:
\[
  \abbv{add}\,(x,(\abs{y}{y})) : \natty \types
\]
That is, for any $x$, if $\abbv{add}\,(x,(\abs{y}{y}))$ evaluates to a numeral then absurdity follows.  Thus, a typing derivation rooted at this judgement constitutes a certificate of the defectiveness of the program $\abbv{add}\,(x,(\abs{y}{y}))$.  The subsequent work of \citet{li-etal-popl26} develops a more feature-full system that certifies the defectiveness of textbook examples of Erlang-like programs. 
\begin{figure}
    \begin{minipage}[c]{.65\textwidth}
      \begin{mathpar}
        \infer[Id]{ }
        {
          \Gamma,\,x:A \types x:A,\,\Delta
        }
        \\
        \infer[App]{
          \Gamma \types M : B \to A,\,\Delta
          \and
          \Gamma \types N : \revreplace{A}{B},\,\Delta
        }
        {
          \Gamma \types M\,N : A\revadd{,\,\Delta}
        }
        \\
        \infer[CoApp]{
          \Gamma,\,M:B \cotofun A \types \Delta
          \and
          \Gamma,\,N:B \types \Delta
        }
        {
          \Gamma,\,M\,N:A \types \Delta
        }
        \\
        \infer[Abs]{
          \Gamma,\,x:B \types M : A,\,\Delta
        }
        {
          \Gamma \types \abs{x}{M} : B \to A,\,\Delta
        }
        \and
        \infer[CoAbs]{
          \Gamma,\,M:A \types x:B,\,\Delta
        }
        {
          \Gamma,\,\abs{x}{M} : B \cotofun A \types \Delta
        }
      \end{mathpar}
    \end{minipage}
    \hfill
    \raisebox{-0.5\totalheight}{\rule{0.5pt}{5.5cm}}
    \hfill
    \begin{minipage}[c]{.3\textwidth}
      \begin{mathpar}
        \infer[CompL]{
          \Gamma \types M:A,\, \Delta
        }
        {
          \Gamma,\,M:A^c \types \Delta
        }
        \\
        \infer[CompL']{
          \Gamma,\,M:A^c \types \Delta
        }
        {
          \Gamma \types M:A,\,\Delta
        }
        \\
        \infer[CompR]{
          \Gamma,\,M:A \types \Delta
        }
        {
          \Gamma \types M:A^c,\,\Delta
        }
        \\
        \infer[CompR']{
          \Gamma \types M:A^c,\, \Delta
        }
        {
          \Gamma,\,M:A \types \Delta
        }
      \end{mathpar}
    \end{minipage}
  \caption{Rules of the two-sided systems for pure $\lambda$-terms from \cite{ramsay-walpole-popl24,li-etal-popl26}\protect\footnotemark.}\label{fig:prev-typing-rules}
  \Description{The left panel shows identity, application, abstraction, and the dual rules for application and abstraction. The right panel shows four complementation rules, two for each side of the judgement.}
\end{figure}

\footnotetext{In the rules \rlnm{Abs} and \rlnm{CoAbs}, $x$ must not occur free in $\Gamma$ or $\Delta$.}

\paragraph{Rules of Two-Sided Systems}
We represent on the left-hand side of Figure~\ref{fig:prev-typing-rules} the simple set of core rules for pure $\lambda$-terms that was identified in \cite{ramsay-walpole-popl24} and refined in \cite{li-etal-popl26}.  \revdel{On the right-hand side we enumerate the complementation rules that were the principal subject of \cite{li-etal-popl26}\textsuperscript{\ref{footnote:compl}}.}  The idea is that the rules \rlnm{Id}, \rlnm{App} and \rlnm{Abs} describe how to prove a typing of a specific term shape on the right-hand side of the judgement, as usual, and the rules \rlnm{Id}, \rlnm{CoApp} and \rlnm{CoAbs} describe how to refute a typing of a specific term shape that occurs on the left-hand side of the judgement.

A key component is the type $B \cotofun A$.  To understand this type, it is instructive to first introduce the type $B \revreplace{\coto}{\nectofun}A$\revadd{\footnote{This type is denoted by $B \coto A$ in \cite{ramsay-walpole-popl24,li-etal-popl26}. We instead use $\nectofun$, reserving $\cotol$ for logical coimplication, for which it is a slight variant of the standard notation (see also footnote~\ref{footnote:coimp-notation})}.}, which is the type of those functions $f$ that satisfy the following property: for all inputs $x$, if $f\,x$ evaluates to a value of type $A$, then $x$ was necessarily a value of type $B$.  Then the type $B \cotofun A$ is the complement of this type, i.e. it describes those terms that do \emph{not} behave according to $B \revreplace{\coto}{\nectofun} A$.  Hence, to \emph{refute} that some term $\abs{x}{M}$ behaves in accordance with $B \cotofun A$ is to \emph{prove} that it behaves in accordance with $B \revreplace{\coto}{\nectofun} A$, that is: to demonstrate that whenever it returns an $A$, then its argument must have been a $B$.  This is exactly the content of the \rlnm{CoAbs} rule.  Conversely, when \emph{refuting} that an arbitrary application $M\,N$ evaluates to a value of type $A$, it suffices to show that $M$ is a function that, to return an $A$ necessitates an argument that is a $B$, and then to \emph{refute} that $N$ is a $B$.  This is exactly the content of the \rlnm{CoApp} rule.

\revadd{The right-hand side of Figure~\ref{fig:prev-typing-rules} gives the complementation rules of \cite{li-etal-popl26}\footnote{\label{footnote:compl}\textcolor{black}{Actually, only \rlnm{CompL} and \rlnm{CompR} are given explicitly in \cite{li-etal-popl26}.  However, rules \rlnm{CompL'} and \rlnm{CompR'} are admissible via the subtyping equivalence $(A^c)^c \equiv A$.  We add them here explicitly since we will not consider subtyping in this work.}}. They move the \emph{same term} across the turnstile while introducing or removing a complement. Rules \rlnm{CompL} and \rlnm{CompL'} express that establishing $M:A$ is equivalent to refuting $M:A^c$; rules \rlnm{CompR} and \rlnm{CompR'} express that refuting $M:A$ is equivalent to establishing $M:A^c$. In the reading of \cite{li-etal-popl26}, these rules are roughly about asserting whether the particular term belongs to a type or its complement.}

\revreplace{\paragraph{Propositions-as-Types}}{\subsection{Propositions-as-Types}}
The natural symmetry of the rules is appealing. Moreover, when extended with appropriate rules for typing constants, they can be used to prove and refute interesting properties of program expressions. This leads to a natural question: do two-sided type systems admit a logical interpretation in the \emph{propositions-as-types} sense?

\revadd{\paragraph{From Type Exclusion to Constructive Refutation.}}
\revadd{
  In the systems of \cite{ramsay-walpole-popl24,li-etal-popl26}, a judgement $M:A \types\,$ expresses that $M$ does not evaluate to a value of type $A$. From a logical perspective, one might therefore read this as saying that $M$ is not a proof of $A$. However, this reading says nothing against
  the proposition $A$. For example, the earlier systems can derive
  $(x,y):a\to a \types\,$, since a pair cannot evaluate to a function. Yet the proposition $a\to a$ is provable: the identity function $I\coloneqq\abs{z}{z}$ constitutes a proof. Thus, excluding one expression from a type does not give any evidence to refute the corresponding proposition. A propositions-as-types interpretation asks more of the negative judgement: $M$ should itself
  describe how $A$ is constructively refuted.  We instead require $M:A \types\,$ to construct evidence for the falsity of $A$.  For example, for coimplication, written $A \coto B$, we will see that adequate evidence is a function mapping refutations of $A$ to refutations of $B$.  Consider:
  \[
    (i)\quad (x,y):a \coto a \types
    \qquad\qquad
    (ii)\quad I:a \coto a \types
  \]
  Under a mere `not a proof' reading of the logical formula, both judgements would be justified: neither a pair nor an identity function proves the invalid formula $a \coto a$.  In our system only (ii) is derivable, since $I$ maps any refutation of $a$ to itself.
}

\revadd{\paragraph{Bilateral Natural Deduction.}}
\revadd{
  To see how such an interpretation might arise, consider the \rlnm{CoApp} rule. Its structure is that of an elimination rule: the compound type occurs in a premise, and the term in the conclusion is an application. Although the earlier two-sided systems were introduced as sequent calculi for typing formulas, this suggests a natural-deduction reading of their logical content. More precisely, we seek a bilateral reading in which application eliminates an implication in a proof and a coimplication in a refutation.

  Wansing's \emph{bilateral} propositional logic \tint{} \cite{wansing2016falsification} provides such a framework. Its proof system consists of mutually inductive notions of proof and \emph{dual proof} (refutation), so that proofs may contain dual proofs as proper parts and conversely. The logical counterpart of \rlnm{CoApp} is the rule for eliminating a \emph{coimplication} from a dual proof. On this reading, the term $M$ in the first premise represents a construction that transforms any refutation of $B$ into a refutation of $A$, while the term $N$ in the second premise supplies a refutation of $B$. The application $M\,N$ in the conclusion therefore supplies a refutation of $A$. Thus, just as implication internalises deduction for proofs, coimplication internalises deduction for refutations.
}

\revadd{\paragraph{Single-Subject Judgements.}}
\revadd{
  To express this bilateral reading at the level of typing judgements, we distinguish the evidence being constructed from the evidence assumed in the context. In \tint{}, a derivation constructs a proof or a refutation of a single formula relative to assumptions and co-assumptions. We reflect this structure through \emph{single-subject judgements}, of shape $\Gamma \types M:A,\Delta$ or $\Gamma,M:A \types \Delta$, where $\Gamma$ and $\Delta$ are disjoint variable environments. The variables in these environments name assumed proofs and refutations, respectively, while in the first case, the distinguished term $M$ represents a proof of $A$ and in the second case, it represents a refutation of $A$. Thus, the judgement identifies both the evidence being constructed and the hypotheses on which it depends. We introduce these judgements formally in Section~\ref{sec:two-lambda-int}.
}

\revadd{\paragraph{Strong Negation.}}
\revadd{
  It remains to account for the complementation rules \rlnm{CompL}, \rlnm{CompL'}, \rlnm{CompR}, and \rlnm{CompR'}. Together, these rules allow a form of double-negation elimination, since one may use them to derive $\types \abs{x}{x} : (A^c)^c \to A$ for any type $A$. Double-negation elimination is notoriously \emph{non-constructive}, but the situation here is more nuanced. Although introducing or eliminating a complement \emph{typing formula} $M: A^c$ behaves like a negation in switching between the sides of the sequent, if we consider only the \emph{propositional formula} $A^c$, it does not behave like the logical negation of $A$, for example, there is no sense in which its inhabitants amount to functions from $A$ to the empty type. To incorporate the complementation rules, we look to David Nelson's seminal 1949 work on Constructible Falsity \cite{nelson1949constructible}, which introduced the \emph{strong negation}, $\sneg{A}$, as a means to repair the asymmetry between verification and falsification in intuitionistic logic. A constructive proof of the strong negation of $A$ amounts to a constructive refutation of $A$, and a constructive refutation of the strong negation of $A$ amounts to a constructive proof of $A$. Thus, the provability of $\sneg{(\sneg{A})} \to A$ does not imply collapse.
}

\revadd{\paragraph{Beyond Propositional Logic.}}
\revadd{The computational value of constructive refutation becomes especially clear when we add quantification. We introduce }\revdel{Moreover, we show that the underlying ideas extend beyond the propositional case by introducing} a two-sided extension of Geuvers' \holty{} \cite{geuvers-types06}, which is a type system that corresponds to a version of Church's Simple Type Theory \cite{church-jsl40}.  We show that, as one might expect of a constructive logic, for any appropriate context $\Gamma$, our system enjoys the existence property:
\begin{description}
  \item[\rlnm{Existence}] If $\Gamma \types M : \SigmaType{a}{K}{A}$ then there are $B$, $N$ such that $\Gamma \types B : K$ and $\Gamma \types N : A[B/a]$.
\end{description}
Moreover, its bilateral nature, mediated by the strong negation, ensures it also satisfies the dual property.  Namely, \revreplace{in any appropriate context $\Gamma$}{in any such kinding environment $\Gamma$}, a refutation of $\PiType{a}{K}{A}$ yields directly a construction $B$ and a refutation of $A[B/a]$.
\begin{description}
  \item[\rlnm{Dual Existence}] If $\Gamma;\, M:\PiType{a}{K}{A} \types\,$ then there are $B$, $N$ such that $\Gamma \types B : K$ and $\Gamma;\,N : A[B/a] \types $\revadd{.}
\end{description}
In many cases of interest in computer science, this construction associated with the dual existence property is of inherent interest.  For example, suppose we define reachability on a rooted, directed graph with edge relation $E$ and root $I$ by the usual second order formula:
\[
  \abbv{Reachable}_{I,E} \coloneqq \tyabs{y}{\mathsf{node}}{\PiType{z}{\mathsf{node} \to \starsort}{z\,I \wedge (\PiType{uv}{\mathsf{node}}{z\,u \wedge E\,u\,v \to z\,v}) \to z\,y}}
\]
Then a refutation of this property\revadd{, equivalently a proof of} $\sneg{(\abbv{Reachable}_{I,E}\,J)}$ for some target node $J$ yields, via the Dual Existence property, the construction of a set of nodes that both contains the root and is closed under the edge relation, and which nevertheless excludes $J$: i.e. an inductive invariant.

\revadd{\subsection{Contributions}}
\revdel{\paragraph{Contributions.}} In this paper we establish an understanding of \revadd{the core rules of }two-sided type systems in terms of proof and refutation in systems of constructive logic.
\begin{enumerate}[I.]
  \item We identify the dual notions of constructive proof and refutation (dual proof) in Wansing's bilateral logic \tint{} \cite{wansing2016falsification}, and Nelson's strong negation \cite{nelson1949constructible}, \revreplace{as the logical counterparts of the two-sided typing judgements}{as interpretations of the core two-sided rules} and the complementation operator, respectively.
        \begin{enumerate}[(a)]
          \item We show that derivability in our system \revreplace{\tintty{}}{\tintcty{}} and provability in \revadd{an extension of }Wansing's system \tint{} \revadd{with strong negation }correspond precisely (Theorem~\ref{thm:wansing-equivalence}).
          \item We show that derivability in our system \revreplace{\tintcty{}}{\tintty{}} and provability in \revdel{an extension of }Wansing's system \tint{}\revdel{ with strong negation} correspond precisely (Corollary~\ref{thm:wansing-equivalence-sneg}).
        \end{enumerate}
  \item We introduce a new two-sided system \holtty{} as an extension of Geuvers' \holty{}, which is itself a system corresponding to Church's Simple Type Theory.
        \begin{enumerate}[(a)]
          \item Due to the bilateral nature of types in our system, the usual encoding of logical operators using higher-order quantification is insufficient.  We give a new encoding that accounts for the dual nature of each operator, and show that \holtty{} is therefore, in a certain sense, expressively adequate (Theorem~\ref{thm:hol2-expr-adequacy}).
          \item By introducing a mapping from \holtty{} to \holty{}, we show that our system is strongly normalising (Theorem~\ref{thm:hol2-sn}), consistent (Theorem~\ref{thm:hol2-consistency}), and satisfies both the existence and dual existence properties (Theorem~\ref{thm:hol2-existence}).
        \end{enumerate}
\end{enumerate}

\paragraph{Notation.}
In the rest of this paper, since we only construct two-sided type systems for the proof theory of logics, we will leave behind the notation for types used in \cite{ramsay-walpole-popl24,li-etal-popl26} and seen in Figure~\ref{fig:prev-typing-rules}.  For simplicity, we will instead adopt notation that is closer to that for formulas.  So, we will use $\sneg{A}$ in preference to $A^c$ for \camrevadd{the }type corresponding to the strong negation of $A$, and we will use $A \cotot B$ in preference to $A \cotofun B$ for the type corresponding to $A$ coimplies $B$ (see also footnote~\ref{footnote:coimp-notation}).

\paragraph{Structure} The rest of the paper is structured as follows.  In Section~\ref{sec:background}, we introduce the ideas of strong negation and bilateral logics.  In Section~\ref{sec:two-lambda-int} we define our two-sided type systems \tintty{} and \tintcty{} and show their correspondence with \tint{} and \tintc{}.  In Section~\ref{sec:2hol} we define our two-sided system \holtty{} and, in Section~\ref{sec:hol2-power} we demonstrate its expressive power.  In Section~\ref{sec:hol2-consistency} we prove consistency and in Section~\ref{sec:hol2-proofs} we prove the existence \camrevadd{and dual existence properties} and give some examples of \camrevdel{its} their use in practice.  Finally, in Section~\ref{sec:concl-related} we conclude and discuss related work.

\section{Constructible Refutation, Strong Negation and Bilateralism}\label{sec:background}
Intuitionistic logic is centred on proof. A proposition is explained by
the constructions that count as its proofs, and the logical connectives
are read as instructions specifying the form of evidence that must be
supplied. For example, to prove $A \wedge B$ is to give both a proof of $A$
and a proof of $B$; to prove $A \vee B$ is to choose one of the two
disjuncts and give a proof of it; and to prove $A \to B$ is to give a
construction which transforms any proof of $A$ into a proof of $B$.

One important consequence of this constructive meaning of proof is that
intuitionistic logic satisfies the \emph{disjunction property}.
\bgroup
\settowidth{\titleindent}{\textit{(Constructible Refutation)}}
\begin{myprop}
  \item[\textit{(Disjunction Property)}]
  if we have a proof of $A \vee B$, then either we have\\ a proof of $A$,
  or we \revdel{we }have a proof of $B$
\end{myprop}
\egroup
In other words, a proof of a disjunction is not only evidence
that one of the disjuncts holds, but contains enough computational
information to recover which disjunct has been established.

If refutation is to be understood constructively in the same sense, then
one might expect an analogous principle. That said, a refutation of a conjunction
should not only rule out the conjunction as a whole; it should identify
where the conjunction fails. This suggests the following property\footnote{This property is known under the name \emph{\revreplace{Construbtible Falsity}{Constructible Falsity}}
  in some literature (e.g. \cite{cantor2024dual}). Here, we choose the name
  \emph{Constructible Refutation}, in order to distinguish the property from
  Nelson's notion of constructible falsity in \cite{nelson1949constructible}
  and the subsequent Nelson family logics.}:
\bgroup
\settowidth{\titleindent}{\textit{(Constructible Refutation)}}
\begin{myprop}
  \item[\textit{(Constructible Refutation)}]
  if we have a refutation of $A \wedge B$, then either we have\\ a refutation of $A$,
  or we \revdel{we }have a refutation of $B$
\end{myprop}
\egroup

However, this property is not validated by ordinary intuitionistic logic.
The reason is that intuitionistic logic has no primitive treatment of
refutation. Rather, refuting $A$ is handled only indirectly, as proving its
intuitionistic negation: $\neg A \;:=\; A \to \bot$. The proposed principle
of constructible refutation is then read intuitionistically as the following:
\[
  \text{if } \vdash \neg(A \wedge B),
  \text{ then either } \vdash \neg A
  \text{ or } \vdash \neg B
\]
This statement is not valid, however. To see this, let $A$ be a propositional
variable and take $B$ to be $\neg A$. Intuitionistic logic proves $\vdash \neg(A \wedge \neg A)$,
since a proof of $A \wedge \neg A$ immediately yields a contradiction, but it proves neither $\neg A$ nor $\neg\neg A$. Thus we have \revreplace{a}{an} intuitionistic refutation of $A \wedge \neg A$, despite not having \revreplace{a}{an} intuitionistic refutation of either conjunct.

This failure isolates the difference between the construction of an intuitionistic refutation and
a genuinely constructive account of refutation. To make the distinction even clearer,
we adapt the discussion from \citet{prawitz2007pragmatist} (also highlighted in \citet[Section~1.3]{wansing-2017}), that separates \textit{the standard way of refuting a compound sentence}, from \textit{the obvious rules for the refutation of a sentence}. He describes the standard method as follows (terminology slightly adjusted to align with ours):
\begin{quote}
  The standard way of refuting a compound sentence $A$ is of course to make a series of inferences from the assumption that $A$ holds that finally ends in the assertion (under the assumption $A$) of a contradiction \ldots.
\end{quote}
We see this in the intuitionistic notion of refutation, where $\neg(A \wedge B)$ expresses that there cannot be any proof of $A \wedge B$, but it does not by itself witness an actual failure that lies in $A$ or lies in $B$. Prawitz contrasts this with what he calls \textit{the obvious rules for the refutation of a sentence} \citep[p.~475]{prawitz2007pragmatist}, adapted here to our terminology:
\begin{quote}
  There are obvious rules for the refutation of a sentence such as inferring a refutation of $(A \wedge B)$ from a refutation of either $A$ or $B$, a refutation of $A \vee B$ from refutations of both $A$ and $B$, a refutation of $A \to B$ from a proof of $A$ and a refutation of $B$ and so on.
\end{quote}
These rules characterize a \camrevdel{genuinely}\revadd{\camrevreplace{, primitive}{primitive}} constructive account of refutation. Under this approach,
just as proofs are built connective by connective, refutations are also
built by rules determined by the structure of the formula being refuted.
This way, refutation is a primitive kind of evidence in its own right just like proof\camrevreplace{, in which sense they enjoy}{; the two therefore enjoy} \revreplace{a}{\camrevdel{an}} \camrevreplace{equivalent status of proof}{equal status} in the system.

Conceptually, the idea of a more symmetric perspective on proof and refutation is quite natural. Proof is evidence for truth, while refutation is evidence for falsity; together, they form two sides of the same balance.
\begin{itemize}[leftmargin=.2in]
  \item At the atomic level, both proof and refutation often arise from the same kind of basic act.
        For example, to verify or refute that a lemon is red, we just look at it. Here, both proof and refutation \revreplace{consists}{consist} of \revreplace{an}{a} single visual experiment (with some meta-analysis of the result). In such cases, neither construction proceeds by assuming the opposite claim and deriving absurdity.
  \item Moreover, it is natural that evidence of different kinds can combine to form more complex evidence. This is already visible in Prawitz's obvious rules: to refute $A \to B$, one supplies positive evidence for $A$ together with negative evidence against $B$. Thus, the construction of refutations cannot be described in isolation from the construction of proofs. A constructive account of refutation therefore calls for a framework in which proof and refutation coexist.
\end{itemize}

\paragraph{Towards Constructivity for Both Proof and Refutation}
In what follows, we consider two systems that give primitive constructive content to both proof and refutation. The first is Nelson’s family of logics for constructible falsity, which begins with a primitive connective of \emph{strong negation}. This connective functions as an object-language operation relating evidence for truth to evidence for falsity. The second is Wansing’s system \tint{}, which has its roots in a more proof-theoretic standpoint. It treats proof and refutation as primitive, mutually related forms of bilateral judgement, thereby providing a different route to a constructive theory of refutation.

\subsection{Nelson's Constructible Falsity and Strong Negation}\label{sec:constructive-falsity}


Nelson extends Kleene's realizability \cite{kleene1945interpretation} from truth to both truth and falsity. He observes
that, from the viewpoint of constructibility, there are two distinct methods to
falsify a formula. Consider, for example, a universal statement $\forall x\,Q(x)$.
One way to falsify it is to give an effective counterexample: a particular witness
$n$ together with evidence that $Q(n)$ is false. Another way is to demonstrate that
$\forall x.Q(x)$ implies an absurdity. Nelson calls the former notion \emph{constructible falsity}.\footnote{This notion, which we will refer to as Nelson-style constructible falsity, is also more commonly called \emph{strong negation} directly nowadays. Although it is
  motivated by a requirement analogous to Constructible Refutation,
  that property alone does not uniquely determine Nelson's connective; other
  notions of negation may also \revreplace{satisfies}{satisfy} Constructible Refutation. } To formalize
this idea, he introduces two realizability relations:
\begin{align*}
  a \text{ P-realizes } A
   & \quad\text{means}\quad
  a \text{ encodes constructive evidence for the truth of } A, \\
  a \text{ N-realizes } A
   & \quad\text{means}\quad
  a \text{ encodes constructive evidence for the falsity of } A .
\end{align*}
The clauses for $P$-realizability and $N$-realizability are then defined
simultaneously, by induction on the structure of formulas. In particular,
Nelson introduces a primitive strong negation connective, which we will write $\sneg{}$, to
internalize constructible falsity: the constructible truth of $\sneg{A}$ is identified with the
constructible falsity of $A$. It is interpreted by a switch between $P/N$-realizability:
\[
  \begin{array}{l@{\qquad}l@{\quad}c@{\quad}l}
    (\sneg{P}) & a \text{ P-realizes } \sneg{A}
               & \text{iff}                     & a \text{ N-realizes } A,  \\
    (\sneg{N}) & a \text{ N-realizes } \sneg{A}
               & \text{iff}                     & a \text{ P-realizes } A .
  \end{array}
\]
This switchability is the key for $\sneg{}$ being constructible. To P-realize $\sneg{(\forall x\,Q(x))}$ is, by the clause $(\sneg{P})$, to N-realize $\forall x\,Q(x)$. But Nelson's clause for N-realizing $\forall x\,Q(x)$ requires exactly a witness $n$ together with an N-realizer of $Q(n)$.
Using $(\sneg{P})$ again, such an N-realizer of $Q(n)$ is precisely a P-realizer
of $\sneg{Q(n)}$. Thus, the evidence of $\sneg{(\forall x\,Q(x))}$ contains a witness $n$ together with evidence of $\sneg{Q(n)}$. In this sense, the falsification internalized by $\sneg{}$ is genuinely constructible.

Nelson showed that his logic is conservative over an intuitionistic system of arithmetic, and he set out a Hilbert-style calculus, in which the behaviour of the strong negation is axiomatised through a number of equivalences:

\noindent
\begin{center}
  \begin{minipage}{.45\textwidth}
    \[
      \begin{array}{ccrcl}
        (a) &  & \sneg{(A \wedge B)} & \!\!\equiv\!\! & \sneg{A} \vee \sneg{B}   \\
        (b) &  & \sneg{(A \vee B)}   & \!\!\equiv\!\! & \sneg{A} \wedge \sneg{B} \\
        (c) &  & \sneg{(A \to B)}    & \!\!\equiv\!\! & A \wedge \sneg{B}
      \end{array}
    \]
  \end{minipage}
  \begin{minipage}{.45\textwidth}
    \[
      \begin{array}{ccrcl}
        (d) &  & \sneg{(\sneg{A})}      & \!\!\equiv\!\! & A                    \\
        (e) &  & \sneg{(\exists a.\,A)} & \!\!\equiv\!\! & \forall a.\,\sneg{A} \\
        (f) &  & \sneg{(\forall a.\,A)} & \!\!\equiv\!\! & \exists a.\,\sneg{A}
      \end{array}
    \]
  \end{minipage}
\end{center}





\subsection{Wansing's System \tint{}}

In \cite{wansing2016falsification}, Wansing introduces the system
\tint{}, a bilateral\footnote{Here \emph{bilateral} can be understood as: giving equal import to proof and refutation.  The term originates with Rumfitt \cite{rumfitt-yes-no}.} natural-deduction system designed to combine both
constructive proof and refutation, where refutation is understood as the dual of proof. The
guiding idea is as follows: we start with the ordinary
natural-deduction rules for intuitionistic proof; then read these
rules from the opposite, falsificationist side; and finally add the
mixed rules needed to let proof and refutation interact in a single
calculus.

\paragraph{Two Primitive Kinds of Derivation} The resulting system has two
primitive kinds of derivation. A derivation of the first kind, called
\camrevadd{a }proof, gives constructive support for truth; a derivation
of the second kind, called \camrevadd{a }dual proof, gives constructive
support for falsity.
We write:
\begin{itemize}
  \item $\wanp{\Pi,\,A,\,(L;R)}$\camrevdel{,} to mean that \(\Pi\)
        is a proof of \(A\) relative to the bilateral basis $(L;R)$;
  \item $\wandp{\Pi,\,A,\,(L;R)}$\camrevdel{,} to mean that \(\Pi\)
        is a dual proof, or refutation, of \(A\) relative to
        $(L;R)$\camrevadd{.}
\end{itemize}
Here \(L\) is a finite set of assumptions, whose formulas are
\camrevreplace{treated as}{assumed to be} true, and \(R\) is a finite
set of co-assumptions, whose formulas are
\camrevreplace{treated as}{assumed to be} false.
\camrevreplace{%
  In particular, $A$ is considered as a proof of $A$ from the pair
  $(\revreplace{A}{\{A\}};\varnothing)$ and
  $\overline{A}$\revadd{\footnote{The overline marks a refutation leaf;
      it is not an object-language connective.}}
  as a dual proof of $A$ from the pair
  $(\varnothing;\revreplace{A}{\{A\}})$.
  Moreover, proofs may contain dual proofs as proper parts and dual
  proofs may contain proofs as proper parts.
}{%
  Both kinds of derivation are represented as trees. A single inference
  line at the root marks a proof, and a double inference line marks a
  dual proof. This convention also applies to leaves, which have no premises above the line.
  In particular, $\wanp{\overline{A},,A,,({A};\varnothing)}$ expresses that $\overline{A}$ is the proof obtained by assuming $A$, while $\wandp{\dbar{A},,A,,(\varnothing;{A})}$ expresses that $\dbar{A}$ is the dual proof obtained by co-assuming $A$. More complex proofs and dual proofs are defined by mutual induction: a proof may contain dual proofs as subderivations, and conversely, a dual proof may contain proofs as subderivations.
}

For convenience, a full definition of \tint{}, which encompasses all
the usual propositional connectives, is given in
\iftoggle{supplementary}{%
  Appendix~\ref{sec:apx-2int}.
}{%
  \camrevadd{the long version of the paper~\cite{li2026logicalfoundationstwosidedtype}.}
}%
\camrevadd{F}or each propositional connective there are deduction rules
for introducing it into a proof and deduction rules for introducing it
into a dual proof. Similarly, there are deduction rules for eliminating
it from a proof, and deduction rules for eliminating it from a dual proof.
We mention here four that foreshadow the correspondence with two-sided
type systems\revdel{:}, \revadd{in which $A \cotol B$ (read as $A$ coimplies
  $B$) denotes the logical dual of $A \to B$\footnote{%
    \label{footnote:coimp-notation}%
    \citeauthor{wansing2016falsification} denoted this as $B \wcoto A$,
    that is, $A \cotol B$ with both the symbol and the order of the arguments
    reversed. This difference is purely notational%
    \camrevreplace{, however}{; however,} we use the latter throughout the
    paper to align better with the notation of \cite{li-etal-popl26}.}:}
\begin{mathpar}
  \prftree{
    \prftree{
      \prftree[noline]{
        ([A], L; R)
      }
      {
        \vdots{}
      }
    }
    {
      B
    }
  }
  {
    A \to B
  }
  \and
  \prftree{
    \prftree{
      \prftree[noline]{
        (L; R)
      }
      {
        \vdots{}
      }
    }
    {
      A
    }
    \and
    \prftree{
      \prftree[noline]{
        (L; R)
      }
      {
        \vdots{}
      }
    }
    {
      A \to B
    }
  }
  {
    B
  }
  \and
  \prftree[d]{
    \prftree[d]{
      \prftree[noline]{
        (L;R,[\![A]\!])
      }
      {
        \vdots{}
      }
    }
    {
      B
    }
  }
  {
    A \cotol B
  }
  \and
  \prftree[d]{
    \prftree[d]{
      \prftree[noline]{
        (L';R')
      }
      {
        \vdots{}
      }
    }
    {
      A \cotol B
    }
  }
  {
    \prftree[d]{
      \prftree[noline]{
        (L;R)
      }
      {
        \vdots{}
      }
    }
    {
      A
    }
  }
  {
    B
  }
\end{mathpar}

\camrevdel{In this tree format, an inference ending with a formula
  under a single line represents a proof, and an inference ending with
  a formula under a double line represents a dual proof.}
\camrevreplace{Assumptions}{Here, assumptions} are discharged by writing
them in single square brackets, and co-assumptions are discharged by
writing them in double square brackets.
The rules are, from left-to-right: implication introduction (proof), implication elimination
(proof), coimplication introduction (dual proof), and coimplication
elimination (dual proof).\camrevdel{ Here}

The two proof rules are quite standard.  The first of the dual proof rules says that to construct a dual proof of $A \cotol B$, one can construct a dual proof of $B$ in which one can discharge the co-assumption $A$\revadd{, that is, an assumed refutation of $A$}.  The second of the dual proof rules says that one can construct a dual proof of $B$ from a dual proof of $A \cotol B$ and a dual proof of $A$\revadd{ supplied as input}.  Thus, just as implication internalises the notion of deduction for proofs, coimplication internalises the notion of deduction for dual proofs.

We draw the reader's attention to the structural similarity of the above proof and dual proof rules with the typing rules on the left-hand side of Figure~\ref{fig:prev-typing-rules}.  Aside from the absence of terms, the Gentzen-style natural deduction format, and the alternative notation for coimplication, the rules above and the rules \rlnm{Abs}, \rlnm{App}, \rlnm{CoAbs}, and \rlnm{CoApp} clearly correspond.
Extrapolating from this observation, in the following section we will develop a two-sided type system that corresponds, \revreplace{in a precise sense}{at the level of logical derivability and type inhabitation for single-subject judgements}, to \tint{}.

\section{A Language with \camrevadd{a }Two-Sided Type System}
\label{sec:two-lambda-int}


\begin{figure}
    \input{2intlam-left.tex}
    \caption{$\tintty{}$: Refutation Rules.}\label{fig:2intlam-left}
    \Description{Twelve inference rules for left-hand typings cover truth and falsity, disjunction, conjunction, implication, and coimplication. Each places its premises above its conclusion.}
\end{figure}

We now introduce \tintty{}, a term language equipped with a propositional
two-sided type system corresponding to Wansing's logic \tint{}, and \tintcty{}, its extension with Nelson-style strong negation.
The design of the language follows the central idea of \tint{}: \camrevdel{The design of the language follows the central idea of \tint{}: }proofs and refutations are treated uniformly as forms of evidence.
\revadd{A term can serve as either kind of evidence, according to its typing judgement. We state
    the typing rules over general two-sided contexts, while our logical
    \camrevreplace{logical}{correspondence} concerns single-subject judgements, which distinguish
    one evidence term from its assumptions and co-assumptions.}
\begin{definition}[Types and Terms of \tintty{}]
    \label{def:types-and-terms-2intty}
    Assume a denumerable set of term variables \(x,y,z\)\revadd{ and a denumerable set of propositional atoms \(a,b,c\)}. The \emph{types} of the system, typically $A, B,C$\revreplace{;}{,} and terms of the system, typically \(M,N,P\), are generated by the following grammars (usual conventions for bound variables are assumed):
    \[
        \begin{array}{rrcl}
            \rlnm{Types}
             & A,B,C
             & \Coloneqq
             & \revadd{a \mid }\falsefm
            \mid \truefm
            \mid A \wedge B
            \mid A \vee B
            \mid A \to B
            \mid B \cotot A
            \\[0.4em]

            \rlnm{Terms}
             & M,N,P
             & \Coloneqq
             & x
            \mid \unittm
            \mid \abort(M)
            \mid (M,N)
            \mid \pi_1(M)
            \mid \pi_2(M)
            \\[0.2em]
             &                          & \mid
             & \inj_1(M)
            \mid \inj_2(M)
            \mid \casetm{M}{N}{P}
            \\[0.2em]
             &                          & \mid
             & \abs{x}{M}
            \mid M\,N
            \mid M\cdot{}N
            \mid \iota_1(M)
            \mid \iota_2(M)
        \end{array}
    \]
\end{definition}

The types of \tintty{} are the formulae of \tint{}.  The types $\falsefm$ and $\truefm$ are types for \revreplace{truth and falsity}{falsity and truth}. $A \wedge B$, $A \vee B$ and $A \to B$ are the standard \revreplace{sum, product}{product, sum} and arrow types that \revreplace{relates}{relate} to the connectives $\wedge$, $\vee$ and $\to$ in intuitionistic logic. In addition, we have the type $A \cotot B$, pronounced $A$ `coto' $B$. This corresponds to the coimplication connective $\cotol$ in \tint{}. As the name suggests, this is supposed to be the dual of the implication connective $\to$.

The terms are the evidence objects manipulated by the type system.  A variable
\(x\) names an available piece of evidence.  The constant \(\unittm\) is
nullary evidence, while \(\abort(M)\) is the eliminator associated with
impossible evidence.
The construction \(M\cdot{}N\)
packages a pair of opposite evidence $M$ and $N$. Correspondingly,
\(\iota_1\) and \(\iota_2\) are eliminators for mixed evidence $M\cdot N$, where
\(\iota_1\) extracts the first component $M$, and \(\iota_2\) extracts the second
component $N$.
The other forms are quite standard.

\begin{definition}[Type contexts of \tintty{}]
    \label{def:type-contexts}
    \revadd{A \emph{typing formula} is a pair $M:A$ of a term $M$ and a type $A$.
    }  A \emph{type context} is a finite set of typing formulae.  Type contexts
    are usually denoted by \(\Gamma,\,\Delta,\,\Sigma,\,\Pi\), and written comma-separated: \(\Gamma = M_1:A_1,\dots,M_n:A_n\).
    Since contexts are finite sets, the order of the displayed formulae is irrelevant,
    and repeated occurrences of the same typing formula are ignored.  Thus,
    whenever a context is displayed as above, we assume that the typing formulae
    \(M_1:A_1,\dots,M_n:A_n\) are pairwise distinct.
\end{definition}

\begin{definition}[Type assignment for \tintty{}]
    \label{def:type-contexts-judgements}
    A \emph{typing judgement} is a pair of type contexts, written $\Gamma \types \Delta$.
    A typing judgement is said to be \emph{provable}, or \emph{derivable} in $\tintty$, if
    it can be obtained from the rules in \cref{fig:2intlam-left,fig:2intlam-right} and the identity rule:
    \[
        \infer{ }{\Gamma,\,x:A \types x:A,\,\Delta}\;\textsc{Id}
    \]
    Excepting this one, the conclusion of each rule distinguishes a single
    typing, which is called the \emph{principal formula} of the rule.
\end{definition}

\revadd{For the logical interpretation \camrevreplace{for}{of} the typing rules, we restrict our attention to \emph{single-subject judgements}, that distinguish one typing containing the proof or refutation being constructed.}
\revadd{
    \begin{definition}[Variable environments]
        \label{def:typing-formula-variable-typing}
        \label{def:variable-environments}
        A typing formula \(M:A\) is called a \emph{variable typing} just if \(M\) is a variable.
        A \emph{variable environment} is a type context consisting only of variable
        typings, with at most one type assigned to each variable.  Equivalently, it
        is a context of the form $x_1:A_1,\dots,x_n:A_n$ where \(x_1,\dots,x_n\) are pairwise distinct.
        Thus a variable environment may be viewed as a finite partial map from
        variables to types.  We write \(\fv(\Gamma)\) for its domain.
    \end{definition}

    \begin{definition}[Disjoint variable environments]
        \label{def:disjoint-variable-environments}
        Two variable environments \(\Gamma\) and \(\Delta\) are called
        \emph{disjoint} if their domains are disjoint, i.e. $\fv(\Gamma)\cap\fv(\Delta)=\varnothing$.
    \end{definition}

    \begin{definition}[Single-subject judgement]
        \label{def:single-subject-presentations}
        A \emph{single-subject judgement} is a typing judgement with a chosen presentation of the form $\Gamma \types M:A, \Delta$ or $\Gamma, M:A \types \Delta$ with $\Gamma$ and $\Delta$ being disjoint variable environments. The
        variable typings in \(\Gamma\) are called \emph{assumptions}, and the variable
        typings in \(\Delta\) are called \emph{co-assumptions}.
    \end{definition}
}
When $\Gamma \types M : A,\,\Delta$ and $\Gamma$ and $\Delta$ are disjoint variable environments, we can view \revreplace{M}{\(M\)} as a proof of \revreplace{A}{\(A\)}.  When $\Gamma,\, M:A \types \Delta$ and $\Gamma$ and $\Delta$ are disjoint variable environments, we can view \revreplace{M}{\(M\)} as a refutation of \revreplace{A}{\(A\)}.  \revreplace{Note, in the case when $M$ above is a variable, it}{For example, the subject of the identity judgement $x:A \types x:A$} can be understood either as a proof of $A$ under the assumption $A$ or as a refutation of $A$ under the co-assumption $A$\revadd{, according to whether the right- or left-hand typing is distinguished.}

\begin{figure}
    \input{2intlam-right.tex}
    \caption{$\tintty{}$: Proof Rules.}\label{fig:2intlam-right}
    \Description{Twelve inference rules for right-hand typings cover truth and falsity, conjunction, disjunction, implication, and coimplication. Each places its premises above its conclusion.}
\end{figure}

\begin{figure}
    \input{negation-rules.tex}
    \caption{$\tintcty{}$: Negation Rules.}\label{fig:negation-rules}
    \Description{The first pair of rules introduces and eliminates strong negation with a conclusion on the left. The second pair does so with a conclusion on the right. Each rule moves the distinguished typing to the opposite side.}
\end{figure}



\revadd{We now give the logical interpretations for the typing rules}. The rules in \cref{fig:2intlam-left} are the rules used for constructing the refutation of a formula.  Rule \rlnm{0L} says that falsity $0$ requires a token refutation $\abra{}$.  Rule \rlnm{1L-Elim} says that $\abort(M)$ can be used as a refutation of any formula $A$ whenever one has $M$, a refutation of truth.  Rule \rlnm{$\vee$L} explains that a refutation of $A \vee B$ consists of a pair, consisting of a refutation of $A$ and a refutation of $B$.  Rule \rlnm{PiL} explains how to recover those refutations separately.  The rules for conjunction are dual.  Rule \rlnm{$\to$L} explains that a refutation of $A \to B$ consists of a proof of $A$ and a refutation of $B$, which are put together in a package $M \cdot{} N$, and \rlnm{Iota2L} explains how one may obtain the refutation $N$ in the second component using $\iota_2$ (the complementary rule \rlnm{Iota1R} explains how to obtain the proof component, but this is in Figure~\ref{fig:2intlam-right} because it constructs a proof).  Finally, the rule \rlnm{$\cotot$L} constructs a refutation of a coimplication.  A \emph{refutation} of $A \cotol B$ expresses `if $A$ is refutable then so is $B$' and so the evidence consists of a function that assigns, to every refutation $y$ of $A$, a refutation $M$ of $B$.

The rules in \cref{fig:2intlam-right} are largely the usual rules for proving type assignments, excepting that now there is an additional context $\Delta$ on the right-hand side \revreplace{which accounts for the fact that there may be further typing formulas on the right as well as on the left}{which, under the single-subject reading, supplies co-assumed refutations alongside the assumed proofs in $\Gamma$}.  The rule \rlnm{$\cotot$R} assigns $M \cdot{} N$ the type $A \cotot B$: given the meaning of a refutation of $A \cotol B$ above, a \emph{proof} should demonstrate that there is a refutation of $A$ and yet a proof of $B$.  These opposing forms of evidence are packaged up in the dot construction, and \rlnm{Iota2R} explains how one can extract the proof using $\iota_2$.




\begin{definition}[\tintcty{}]
    The system \tintcty{} extends \tintty{} with strong negation.
    The terms of \tintcty{} are the same as those of \tintty{}, while its types are obtained by adding the type former $\sneg A$.  Type contexts and typing judgements are the same as those of \tintty{}, excepting that types occurring may now contain the strong negation.
    A typing judgement is said to be \emph{provable}, or \emph{derivable} in $\tintcty$, if
    it can be obtained from the rules of \tintty{} extended with the rules in \cref{fig:negation-rules}.
\end{definition}

The rules in \cref{fig:negation-rules} describe the strong negation.
As in \cite{wansing93}, there is no term former for \(\sneg\).  Instead, the same term is re-used as the
evidence is transported across the turnstile.  The rule \rlnm{$\sneg_I$ R}, for
example, says that a refutation \(M\) of \(A\) also constitutes a proof of \(\sneg A\), while
\rlnm{$\sneg_I$ L} says that a proof \revdel{of }\(M\) of \(A\) also constitutes a refutation of \(\sneg A\).
The elimination rules express the converse movements.  In this sense, strong
negation internalises the switch between proof and refutation.

\begin{example}
    All of the equivalences (a)--(d) of Section~\ref{sec:constructive-falsity} are provable.  For example, the forward direction of (c) can be derived as follows.  Let $\Gamma$ abbreviate $x:\sneg{(A \to B)}$ in:
    \begin{mathpar}
        \infer*[left=$\to$R]{
            \infer*[left=$\wedge$R]{
                \infer*[left=Iota1R]{
                    \revreplace{\infer*[left=Id]{ }{\Gamma \types x : \sneg{(A \to B)}}}{\infer*[left=$\sneg{}_E$L]{
                            \infer*[left=Id]{ }{\Gamma \types x : \sneg{(A \to B)}}
                        }{\Gamma,\,x:A\to B \types}}
                }
                {
                    \Gamma \types \iota_1(x) : A
                }
                \and
                \infer*[left=$\sneg{}_I$R]{
                    \infer*[left=Iota2L]{
                        \revreplace{\infer*[left=Id]{ }{\Gamma \types x : \sneg{(A \to B)}}}{\infer*[left=$\sneg{}_E$L]{
                                \infer*[left=Id]{ }{\Gamma \types x : \sneg{(A \to B)}}
                            }{\Gamma,\,x:A\to B \types}}
                    }
                    {
                        \Gamma,\,\iota_2(x) : B \types
                    }
                }
                {
                    \Gamma \types \iota_2(x) : \sneg{B}
                }
            }
            {
                \Gamma \types (\iota_1(x),\,\iota_2(x)) : A \wedge \sneg{B}
            }
        }
        {
            \types \abs{x}{(\iota_1(x),\,\iota_2(x))} : \sneg{(A \to B)} \to A \wedge \sneg{B}
        }
    \end{mathpar}
\end{example}

\begin{example}
    As explained in \cite{li-etal-popl26}, the coimplication type allows functions to be curried on the left in an analogous way to implication on the right.
    Let us abbreviate $\Delta \coloneqq x\ty:a,f\ty:a \cotot a$ in:
    \begin{mathpar}
        \infer*[left=$\cotot$L]{
            \infer*[Left=$\cotot$L]{
                \infer*[Left=AppL]{
                    \infer*[Left=Id]{ }{
                        f : a \cotot a \types \Delta
                    }
                    \and
                    \infer*[left=AppL]{
                        \infer*[Left=Id]{ }{f:a \cotot a \types \Delta}
                        \and
                        \infer*[Left=Id]{ }{x:a \types \Delta}
                    }
                    {
                        f\,x : a \types \Delta
                    }
                }
                {
                    f\,(f\,x) : a \types \Delta
                }
            }
            {
                \abs{x}{f\,(f\,x)} : a \cotot a \types f : a \cotot a
            }
        }
        {
            \abs{fx}{f\,(f\,x)} : (a \cotot a) \cotot a \cotot a \types
        }
    \end{mathpar}
\end{example}

\begin{example}
    Define the co-negation $-A \coloneqq A \cotot 1$, which is dual to intuitionistic negation.  Then we can prove the dual de Morgan law $\cneg{(A \vee B)} \to (\cneg{A}\wedge\cneg{B})$. For brevity, let $\Gamma \coloneqq x:\cneg{(A\vee B)}$ in:

\begin{mathpar}
  \infer*[Left=$\to$R]{
    \infer*[Left=$\wedge$R]{
      \infer*[Left=$\cotot$R]{
        \infer*[Left=PiL$_1$]{
          \infer*[Left=Iota1L]{
            \infer*[Left=Id]{ }{
              \Gamma
              \types
              x:\cneg{(A \vee B)}
            }
          }{
            \Gamma,\,
            \iota_1(x):A\vee B
            \types
          }
        }{
          \Gamma,\,
          \pi_1(\iota_1(x)):A
          \types
        }
        \and
        \infer*[Left=$\truefm$R]{ }{
          \Gamma
          \types
          \unittm:\truefm
        }
      }{
        \Gamma
        \types
        \pi_1(\iota_1(x))\cdot{\unittm}:\cneg{A}
      }
      \and
      \infer*[Left=$\cotot$R]{
        \infer*[Left=PiL$_2$]{
          \infer*[Left=Iota1L]{
            \infer*[Left=Id]{ }{
              \Gamma
              \types
              x:\cneg{(A \vee B)}
            }
          }{
            \Gamma,\,
            \iota_1(x):A\vee B
            \types
          }
        }{
          \Gamma,\,
          \pi_2(\iota_1(x)):B
          \types
        }
        \and
        \infer*[Left=$\truefm$R]{ }{
          \Gamma
          \types
          \unittm:\truefm
        }
      }{
        \Gamma
        \types
        \pi_2(\iota_1(x))\cdot{\unittm}:\cneg{B}
      }
    }{
      \Gamma
      \types
      \left(
      \pi_1(\iota_1(x))\cdot{\unittm},
      \pi_2(\iota_1(x))\cdot{\unittm}
      \right)
      :
      \cneg{A}\wedge\cneg{B}
    }
  }
  {
    \types
    \abs{x}
    {\left(
      \pi_1(\iota_1(x)):\cneg{A}\cdot{\unittm},\,
      \pi_2(\iota_1(x)):\cneg{B}\cdot{\unittm}
      \right)
      :
      \cneg{(A \vee B)} \to (\cneg{A}\wedge\cneg{B})}
  }
\end{mathpar}
\end{example}

\revadd{
    \paragraph{Comparison with Earlier Two-Sided Systems.}
    The common core of Figure~\ref{fig:prev-typing-rules} supports both type-exclusion reading of \cite{ramsay-walpole-popl24,li-etal-popl26} and the
    constructive reading developed here, but their extensions differ. As is
    usual for propositions-as-types, we omit unrestricted recursion. \camrevdel{As is
        usual for propositions-as-types, we omit unrestricted recursion. }In the
    other direction, the constructive reading justifies rules that are unsound
    under the previous type-exclusion meaning. For example, our left projection rule admits the
    instance
    \[
        \infer*[left=PiL]{
            M:A\vee B \types
        }{
            \pi_1(M):A \types
        }.
    \]
    A constructive refutation of a disjunction contains refutations of both
    disjuncts, so projection extracts a refutation of $A$. Under a non-membership
    reading, however, take $M=(P,Q)$, where $P$ and $Q$ are closed values and $P$
    has type $A$. Interpreting $A\vee B$ as the tagged sum type of our language,
    the pair is not a value of this type, so the premise holds. Yet $\pi_1(M)$
    evaluates to $P$, so the conclusion fails. Thus, the interpretation of a left
    subject as constructive evidence, rather than merely as a term excluded from
    a type, determines which extensions of the common core are sound. This
    change of interpretation is distinct from the restriction to single-subject
    judgements used in the correspondence below.
}

\subsection{Correspondence with Wansing's \tint{}}
\revadd{We now establish the correspondence between derivable single-subject
    judgements, as defined in
    Definition~\ref{def:single-subject-presentations}, and proofs and dual
    proofs in Wansing's logic. We first work with the systems extended by
    strong negation, \tintcty{} and \tintc{}, and then recover the
    correspondence for their fragments without it.}

\revdel{
    To obtain a correspondence with \tint{}, it seems sensible to restrict our attention to judgements with only a single occurrence of a non-variable term.  This term can then be understood as \emph{the} proof term, or \emph{the} refutation term, depending on which side of the turnstile it occurs.
}

\revdelblock{
    \theoremstyle{acmdefinition}
    \newtheorem*{deleteddefinition}{\revdel{Definition}}
    \begin{deleteddefinition}[\revdel{Variable environments}]
        \revdel{A typing formula \(M:A\) is called a \emph{variable typing} just if \(M\) is a variable.
            A \emph{variable environment} is a type context consisting only of variable
            typings, with at most one type assigned to each variable.  Equivalently, it
            is a context of the form $x_1:A_1,\dots,x_n:A_n$ where \(x_1,\dots,x_n\) are pairwise distinct.
            Thus a variable environment may be viewed as a finite partial map from
            variables to types.  We write \(\fv(\Gamma)\) for its domain.}
    \end{deleteddefinition}
    \begin{deleteddefinition}[\revdel{Disjoint variable environments}]
        \revdel{Two variable environments \(\Gamma\) and \(\Delta\) are called
            \emph{disjoint} if their domains are disjoint, i.e. $\fv(\Gamma)\cap\fv(\Delta)=\varnothing$.}
    \end{deleteddefinition}
    \begin{deleteddefinition}[\revdel{Single-subject judgement}]
        \revdel{A \emph{single-subject judgement} is a typing judgement that can be written: $\Gamma \types M:A, \Delta$ or $\Gamma, M:A \types \Delta$ with $\Gamma$ and $\Delta$ being disjoint variable contexts. The
            variable typings in \(\Gamma\) are called \emph{assumptions}, and the variable
            typings in \(\Delta\) are called \emph{co-assumptions}.}
    \end{deleteddefinition}
}
\revdel{
    When $\Gamma \types M : A,\,\Delta$ and $\Gamma$ and $\Delta$ are disjoint variable environments, we can view M as a proof of A.  When $\Gamma,\, M:A \types \Delta$ and $\Gamma$ and $\Delta$ are disjoint variable environments, we can view M as a refutation of A.  Note, in the case when $M$ above is a variable, it can be understood either as a proof of $A$ under the assumption $A$ or as a refutation of $A$ under the co-assumption $A$.
}

\begin{definition}[Erasure]\label{def:proof-term-erasure}
    We write $\erasure{M:A} = A$ for the erasure of a \revreplace{proof term,}{typing formula} (including when $M$ is a variable). We extend this notation to environments by erasing \revdel{proof} terms pointwise and identifying repeated formulas, so that $\erasure{\Gamma}$ and $\erasure{\Delta}$ are finite sets of formulas.
\end{definition}

\begin{lemma}[From $\tintcty{}$ to $\tintc{}$]
    \label{lem:typing-to-2int}
    The following hold simultaneously:
    \begin{enumerate}[(i)]
        \item If $\Gamma \types \mf{M:A},\,\Delta$ in \tintcty{} and $\Gamma$ and $\Delta$ are disjoint variable \revreplace{contexts}{environments}, then there exists some $\Pi$, such that
              $\wanp{\Pi,\,A,\,(\erasure{\Gamma}\revreplace{,}{;}\,\erasure{\Delta})}$ in \tintc{}.
        \item If $\Gamma,\,\mf{M:A} \types \Delta$ in \tintcty{} and $\Gamma$ and $\Delta$ are disjoint variable \revreplace{contexts}{environments}, then there exists some $\Pi$, such that
              $\wandp{\Pi,\,A,\,(\erasure{\Gamma}\revreplace{,}{;}\,\erasure{\Delta})}$ in \tintc{}.
    \end{enumerate}
\end{lemma}

In the opposite direction, as is typical, it is convenient to choose some canonical naming of assumptions (and co-assumptions).

\begin{definition}\label{def:logic-single-subject}
    Fix once and for all, for each type $A$, two variables $x_A$ and $y_A$, in such a way that the families $\{x_A \mid A \text{ is a type}\}$ and $\{y_A \mid A \text{ is a type}\}$ are jointly pairwise distinct.
    Now let $L = \{A_1,\ldots,A_m\}$ and $R = \{C_1,\ldots,C_n\}$ be finite sets of formulas. We define:
    \begin{itemize}
        \item $L \types' B;\,R$ just if there is a term $M$ and $x_{A_1} \ty: A_1,\ldots,x_{A_m}\ty:A_m \types \mf{M : B},\,y_{C_1} \ty: C_1,\ldots,y_{C_n}\ty: C_n$
        \item $L;\,B \types' R$ just if there is a term $M$ and $x_{A_1} \ty: A_1,\ldots,x_{A_m}\ty:A_m,\, \mf{M : B} \types y_{C_1} \ty: C_1,\ldots,y_{C_n}\ty: C_n$
    \end{itemize}
    If $L=\varnothing$ or $R=\varnothing$, the corresponding family is empty.
\end{definition}

Then, from a proof or a dual proof in \tintc{}, we can obtain a proof term or refutation\revadd{ term} respectively, with (co-)assumptions named according to the above scheme.

\begin{lemma}[From $\tintc{}$ to $\tintcty{}$]
    \label{lem:2int-to-typing}
    The following hold simultaneously:
    \begin{enumerate}[(i)]
        \item If $\wanp{\Pi,\,A,\,(L\revreplace{,}{;}R)}$ in \tintc{}, then $L \vdash' A;R$ in \tintcty{}.
        \item If $\wandp{\Pi,\,A\revadd{,}\,(L\revreplace{,}{;}R)}$ in \tintc{}, then $L;A \vdash' R$ in \tintcty{}.
    \end{enumerate}
\end{lemma}

\noindent
Putting the foregoing results together, we obtain correspondence for the larger system.

\begin{theorem}[Correspondence between $\tintc{}$ and $\tintcty$]
    \label{thm:wansing-equivalence}
    Both of the following\revadd{ hold for finite sets of formulas $L,R$ and a formula $A$}:
    \begin{enumerate}[(i)]
        \item \revadd{A derivation $\Pi$ satisfying }$\wanp{\Pi,\,A,\,\revreplace{(L,R)}{(L;R)}}$\revadd{ exists} in \tintc{} iff $L \vdash' A;R$ in \tintcty{}\revadd{.}
        \item \revadd{A derivation $\Pi$ satisfying }$\wandp{\Pi,\,A\revadd{,}\,\revreplace{(L,R)}{(L;R)}}$\revadd{ exists} in \tintc{} iff $L;A \vdash' R$ in \tintcty{}.
    \end{enumerate}
\end{theorem}

\noindent
By inspection of the argument, one can immediately further conclude that the fragments without strong negation are in similar correspondence.

\begin{corollary}[Correspondence between $\tint{}$ and $\tintty{}$]
    \label{thm:wansing-equivalence-sneg}
    Both of the following\revadd{ hold for finite sets of formulas $L,R$ and a formula $A$ in the language without strong negation}:
    \begin{enumerate}[(i)]
        \item \revadd{A derivation $\Pi$ satisfying }$\wanp{\Pi,\,A,\,\revreplace{(L,R)}{(L;R)}}$\revadd{ exists} in \tint{} iff $L \vdash' A;R$ in \tintty{}.
        \item \revadd{A derivation $\Pi$ satisfying }$\wandp{\Pi,\,A\revadd{,}\,\revreplace{(L,R)}{(L;R)}}$\revadd{ exists} in \tint{} iff $L;A \vdash' R$ in \tintty{}.
    \end{enumerate}
\end{corollary}

\section{\holtty{}}\label{sec:2hol}

To extend these ideas beyond the propositional case, we here introduce \holtty{}, a 2-sided extension of Geuvers' \holty{} \cite{geuvers-types06}.  Geuvers' motivation for \holty{} was to construct a type system with a straightforward correspondence with Church's Simple Theory of Types \cite{church-jsl40}.  It can be presented as a pure type system (PTS) \cite{barendregt-lambda-types}:
\[\def\arraystretch{1.2}
  \holty{}
  \quad
  \begin{array}{|rl|}\hline
    \mathcal{S} & \starsort,\,\boxsort,\,\trisort                                    \\
    \mathcal{A} & \starsort : \boxsort,\, \boxsort : \trisort                        \\
    \mathcal{R} & (\starsort,\starsort),\,(\boxsort,\starsort),\,(\boxsort,\boxsort) \\\hline
  \end{array}
\]
\revadd{Here $\mathcal{S}$ gives the sorts, $\mathcal{A}$ the axioms, and $\mathcal{R}$ the product-formation rules. }From this one can deduce that \holty{} is a mild extension of System F$\omega$ in which one may additionally introduce new kind variables. \revadd{The axiom $\boxsort:\trisort$ permits these kind variables.} The correspondence with higher-order (predicate) logic is achieved by representing all \revadd{logical }terms (including formulas, which are a subclass of terms in \cite{church-jsl40}) as \holty{} types and representing the types that occur within those terms (sometimes called domain sorts) as \holty{} kinds.  For example, a universal quantification in the logic $\PiType{a}{K}{A}$ corresponds to a type, the domain sort $K$ that occurs within the formula corresponds to a kind and a term of sort $K$ that might be used to instantiate the formula $A$ corresponds to a type.  Consequently, there is no need for types that depend on terms; terms of the type system are used exclusively to represent proofs.  For a detailed account, see \cite{geuvers-phd,geuvers-intro-cc}.

In \holtty{}, we will eschew the compact PTS presentation in favour of something more explicit.  This is because our intention is to extend the proof rules quite substantially, whilst leaving the kinding rules largely untouched.

\begin{definition}[Syntax]\label{def:hol2-syntax}
  \camrevreplace{Let us assume}{Assume} \revadd{disjoint }countably infinite supplies of kind variables $k$; type variables $a$, $b$, $c$; and term variables $x$, $y$, $z$.  The \emph{kinds}, typically $K$; \emph{types}, typically $A, B$; \emph{terms}, typically $M$, $N$; and environments, typically $\Gamma,\Delta$; of \holtty{} are described by the following grammars:
  \[
    \begin{array}{rrcl}
      \rlnm{Kinds} & K             & \Coloneqq & k \mid \starsort \mid K_1 \to K_2                                                                                                  \\
      \rlnm{Types} & A,B           & \Coloneqq & a \mid A\,B \mid \tyabs{a}{K}{A} \mid A \to B \mid \PiType{a}{K}{A} \mid \sneg{A}                                                  \\
      \rlnm{Terms} & M,N           & \Coloneqq & x \mid M\,N \mid M\,[A] \mid \tyabs{x}{A}{M} \mid \Tyabs{a}{K}{M} \mid \cdot \mid \iota_1 \mid \iota_2 \mid \pack_K \mid \unpack_K \\
      \rlnm{Envs}  & \Gamma,\Delta & \Coloneqq & \varnothing \mid \Gamma,a\mathord{:}K \mid \Gamma,x\mathord{:}A
    \end{array}
  \]
  We consider the types $\tyabs{a}{K}{A}$, $\PiType{a}{K}{A}$; and the terms $\tyabs{x}{A}{M}$ and $\Tyabs{a}{K}{M}$; as variable binders, and we treat them accordingly. Environments are (possibly empty) sequences of variable kind- and type-assignments.  We will occasionally enclose such a sequence in square brackets to increase readability.  We write $\Delta^R$ for the reversal of the sequence $\Delta$, so $[x\ty:a,y\ty:b,z\ty:a \to b]^R$ is the environment $[z\ty:a\to b,y\ty:b,x\ty:a]$.  We say that an environment $\Gamma$ is a \emph{kinding environment} just if it does not contain any type assignment, that is any statement of shape $x\ty:A$.

  \begingroup
  \setlength{\emergencystretch}{1em}
  We will sometimes refer to the introduction term forms and types by names.  We write $\cdot{}\,[A]\,[B]\,M\,N$ rather as $M \cdot_{A,B} N$ and refer to a term of this shape as a \emph{dot construction}.  We refer to a term of shape $\pack_K\,[\tyabs{a}{K}{A}]\,[B]\,M$ as a \emph{$K$-pack}.  A term of shape $\tyabs{x}{A}{M}$ is called a \emph{(term) abstraction} and $\Tyabs{a}{K}{M}$ a \emph{type abstraction}.  A type of shape $A \to B$ is called an \emph{arrow}, $\PiType{a}{K}{A}$ a \emph{universal over $K$} and $\sneg{A}$ a \emph{strong negation}.
  \par
  \endgroup
\end{definition}

\begin{definition}[Type Conversion]
  Conversion $A \conv B$ between two types $A$ and $B$, is defined as the reflexive, symmetric, transitive closure of the following inductively defined relation:
  \begin{mathpar}
    \infer[CRedex]{ }{
      (\tyabs{a}{K}{A})\,B \conv A[B/a]
    }
    \and
    \infer[CAbs]{
      A \conv B
    }
    {
      \tyabs{a}{K}{A} \conv \tyabs{a}{K}{B}
    }
    \and
    \infer[CApp]{
      A \conv B
      \and
      C \conv D
    }
    {
      A\,C \conv B\,D
    }
    \\
    \infer[CSneg]{
      A \conv B
    }
    {
      \sneg{A} \conv \sneg{B}
    }
    \and
    \infer[CForall]{
      A \conv B
    }
    {
      \PiType{a}{K}{A} \conv \PiType{a}{K}{B}
    }
    \and
    \infer[CArrow]{
      A \conv B
      \and
      C \conv D
    }
    {
      A \to C \conv B \to D
    }
  \end{mathpar}
  Note: since strong negation does not play an active role in conversion, one can immediately conclude that \holtty{} convertibility has all the good properties of (System F$\omega$) $\beta$-convertibility, by regarding the construction $\sneg{A}$ as the application of a free variable $\sneg{}$ to a type $A$.
\end{definition}

\begin{definition}[Kinding]
  Well-formedness of an environment $\Gamma$, written $\types \Gamma$, and \camrevreplace{the assignment of a kind $K$ to}{well-kindedness of} a type $A$ in the environment $\Gamma$, written $\Gamma \types A : K$\camrevadd{ for some kind $K$}, are \camrevreplace{given simultaneously by the following inductive definition}{defined simultaneously by the following rules}, where $\subj(\Gamma)$ refers to the subjects of the \revadd{kind and }type assignments in $\Gamma$.
  \begin{mathpar}
    \infer[Empty]{ }{
      \types \varnothing
    }
    \and
    \infer[Kinding]{
      \types \Gamma
    }
    {
      \types \Gamma,a\ty:K
    }\;a \notin \subj(\Gamma)
    \and
    \infer[Typing]{
      \Gamma \types A : \starsort
    }
    {
      \types \Gamma,x\ty:A
    }\;x \notin \subj(\Gamma)
    \\
    \infer[Var]{
      \types \Gamma
    }
    {
      \Gamma \types a : K
    }\;a\mathord{:}K \in \Gamma
    \and
    \infer[Abs]{
      \Gamma,\,a\ty:K_1 \types B : K_2
    }
    {
      \Gamma \types \tyabs{a}{K_1}{B} : K_1 \to K_2
    }
    \and
    \infer[App]{
      \Gamma \types A : K_1 \to K_2
      \and
      \Gamma \types B : K_1
    }
    {
      \Gamma \types A\,B : K_2
    }
    \and
    \infer[Arr]{
      \Gamma \types A : \starsort
      \and
      \Gamma \types B : \starsort
    }
    {
      \Gamma \types A \to B : \starsort
    }
    \and
    \infer[All]{
      \Gamma,\,a\ty:K \types A : \starsort
    }
    {
      \Gamma \types \PiType{a}{K}{A} : \starsort
    }
    \and
    \infer[SNeg]{
      \Gamma \types A : \starsort
    }
    {
      \Gamma \types \sneg{A} : \starsort
    }
  \end{mathpar}
\end{definition}

\newcommand{\seqsep}{\triangleright}

In our propositional system of Section~\ref{sec:two-lambda-int}, a key restriction on the shape of judgement was that there is at most one non-variable typing, which may occur either on the left or right.  In \holtty{}, it will be convenient for the metatheory to make this more explicit by distinguishing the typing formula that contains that single proof (or refutation) term.  Thus, we will have two judgements $\Gamma \types M : A;\,\Delta$ and $\Gamma;\,N:B \types \Delta$, corresponding to the distinguished typing formula being a proof $M:A$ or a refutation $N:B$, respectively.  We will call this distinguished typing formula the \emph{principal formula} of the judgement.

\begin{definition}[\revreplace{Proof}{Evidence} Terms]
  The assignment of a type $A$ to a (proof) term $M$ under assumptions $\Gamma$ and co-assumptions $\Delta$, is written $\Gamma \types M : A;\,\Delta$.  The assignment of a type $A$ to a (refutation) term $M$ under assumptions $\Gamma$ and co-assumptions $\Delta$, is written $\Gamma;\,M:A \types \Delta$.  In both cases, the environment $\Delta$ is required to only contain typing statements, of shape $x:B$ (i.e. no kinding statements $a:K$).  The theory of the two judgements is defined simultaneously by induction.
  \begin{mathpar}
    %
    %
    %
    %
    \infer[Id-L]
    {
      \types \Gamma,\Delta^R
    }
    {
      \Gamma;\, x:A \types \Delta
    }\;x\mathord{:}A \in \Delta
    \and
    \infer[Id-R]
    {
      \types \Gamma,\Delta^R
    }
    {
      \Gamma \types x:A;\,\Delta
    }\;x\mathord{:}A \in \Gamma
    \and
    %
    %
    \infer[Conv]{
      \Gamma \types M : B;\,\Delta
      \and
      \Gamma \types A : \starsort
    }
    {
      \Gamma \types M : A;\,\Delta
    }\;A =_\beta B
    \\
    %
    %
    %
    \infer[$\forall$-Intro-L]{
      \Gamma \types A : K
      \and
      \Gamma;\, N : B[A/a] \types \Delta
    }
    {
      \Gamma;\,\pack_K\,[\tyabs{a}{K}{B}]\,[A]\,N : \PiType{a}{K}{B} \types \Delta
    }
    \and
    \infer[$\forall$-Intro-R]{
      \Gamma,a : K \vdash M : B;\,\Delta
    }
    {
      \Gamma \vdash \Tyabs{a}{K}{M} : \PiType{a}{K}{B};\,\Delta
    }
    \and
    %
    %
    \infer[$\forall$-Elim-L]{
      \Gamma;\,M : \PiType{a}{K}{B} \types \Delta
      \and
      \Gamma,\,a\ty:K;\,N : C \types x\ty:B,\,\Delta
    }
    {
      \Gamma;\,\unpack_K\,[\tyabs{a}{K}{B}]\,M\,[C]\,(\Tyabs{a}{K}{\tyabs{x}{\sneg{B}}{N}}) : C \types \Delta
    }
    \and
    \infer[$\forall$-Elim-R]{
      \Gamma \types M : \PiType{a}{K}{B};\,\Delta
      \and
      \Gamma \types A : K
    }
    {
      \Gamma \types M\,[A] : B[A/a];\,\Delta
    }
    \\
    \infer[$\to$-Intro-L]{
      \Gamma \types M:A;\, \Delta
      \and
      \Gamma;\,N:B \types \Delta
    }
    {
      \Gamma;\, M \cdot_{A,B} N : A \to B \types \Delta
    }
    \and
    \infer[$\to$-Elim-L$_1$]{
      \Gamma;\,M:A \to B \types \Delta
    }
    {
      \Gamma \types \iota_1\,[A]\,[B]\,M : A;\,\Delta
    }
    \and
    \infer[$\to$-Elim-L$_2$]{
      \Gamma;\,M:A \to B \types \Delta
    }
    {
      \Gamma;\,\iota_2\,[A]\,[B]\,M : B \types \Delta
    }
    \and
    \infer[$\to$-Intro-R]{
      \Gamma,\,x\ty:A \types M : B;\,\Delta
    }
    {
      \Gamma \types \tyabs{x}{A}{M} : A \to B;\,\Delta
    }
    \and
    \infer[$\to$-Elim-R]{
      \Gamma \types M : A \to B;\,\Delta
      \and
      \Gamma \types N : A;\,\Delta
    }
    {
      \Gamma \types M\,N : B;\,\Delta
    }
    \\
    \infer[$\sneg{}$-Intro-L]{
      \Gamma \types M:A;\, \Delta
    }
    {
      \Gamma;\, M : \sneg{A} \types \Delta
    }
    \and
    \infer[$\sneg{}$-Intro-R]{
      \Gamma;\,M:A \types \Delta
    }
    {
      \Gamma \types M : \sneg{A};\,\Delta
    }
    \and
    \infer[$\sneg{}$-Elim-L]{
      \Gamma;\,M:\sneg{A} \types \Delta
    }
    {
      \Gamma \types M : A;\,\Delta
    }
    \and
    \infer[$\sneg{}$-Elim-R]{
      \Gamma \types M:\sneg{A};\, \Delta
    }
    {
      \Gamma;\, M : A \types \Delta
    }
  \end{mathpar}
  In \rlnm{$\forall$-Intro-L}, \rlnm{$\forall$-Intro-R}, \rlnm{$\forall$-Elim-L} and \rlnm{$\to$-Intro-R} we require that the variables $a$ and $x$ are fresh for their context, in the sense that neither \camrevreplace{occur}{occurs} free in $\Gamma$ or $\Delta$. \revadd{In \rlnm{$\forall$-Elim-L}, we additionally require $a\notin\fv(C)$, so that the locally unpacked witness does not escape into the result formula.}  In what follows we will typically omit the type applications on $\iota_1$ and $\iota_2$, since they are generally clear from the context.
\end{definition}

The format of the two judgements unfortunately necessitates a pair of identity rules \rlnm{Id-L} and \rlnm{Id-R}, but one simply chooses the rule according to the position of the proof (or refutation) term.  The rules \rlnm{Conv}, \rlnm{$\forall$-Intro-R}, \rlnm{$\forall$-Elim-R}, \revreplace{\rlnm{$\to$-intro-R}}{\rlnm{$\to$-Intro-R}} and \rlnm{$\to$-Elim-R} are quite standard, except for the addition of the environment $\Delta$ of co-assumptions.  Rules \rlnm{$\to$-Intro-L}, \rlnm{$\to$-Elim-L$_1$} and \rlnm{$\to$-Elim-L$_2$} are essentially as in Section~\ref{sec:two-lambda-int}, except that we now decorate their \revreplace{proofs}{evidence terms} with type applications, which makes our later translations to \holty{} more straightforward. Two new concepts concern universal quantification.
\begin{itemize}
  \item Rule \rlnm{$\forall$-Intro-L} describes how to refute a universal quantification $\PiType{a}{K}{B}$ over kind $K$: this can be done by finding a witness $A$ of kind $K$ and constructing a refutation $N$ of $B[A/a]$.\revadd{ The $K$-pack records the family $\tyabs{a}{K}{B}$, the witness $A$, and the refutation evidence $N$.}
  \item Rule \rlnm{$\forall$-Elim-L} describes how one may use a given refutation of a universal quantification $\PiType{a}{K}{B}$ over kind $K$ in order to \revreplace{establish some arbitrary proposition $C$}{construct a refutation of \camrevadd{some arbitrary proposition }$C$}\camrevadd{, without knowing a exact witness}.  To do so, one must \revreplace{show that $C$ already follows from any choice of witness $a\ty:K$ and refutation $x\ty:B$}{construct a refutation of $C$ under an arbitrary (fresh) witness $a\ty:K$ and \camrevreplace{a co-assumed}{its accompanying (co-assumed)} refutation $x\ty:B$}.
\end{itemize}
Finally, rules \rlnm{$\sneg{}$-Intro-L}, \rlnm{$\sneg{}$-Intro-R}, \rlnm{$\sneg{}$-Elim-L}, \rlnm{$\sneg{}$-Elim-R} allow for switching from proofs to refutations or vice versa.  Unlike the system of Section~\ref{sec:two-lambda-int}, where there was no syntactically distinguished typing, these rules can only affect the principal typing formula.  However, it can be shown that the rules of Figure~\ref{fig:ctx-switching}, which allow for some movement between assumptions and co-assumptions, are admissible in the system.

\begin{figure}
    \begin{mathpar}
      \infer[Ctx$_1$]{
        \Gamma;\,M:B \types x\ty:A,\,\Delta
      }
      {
        \Gamma,\,x\ty:\sneg{A};\,M:B \types \Delta
      }
      \and
      \infer[Ctx$_2$]{
        \Gamma \types M:B;\,x\ty:A,\,\Delta
      }
      {
        \Gamma,\,x\ty:\sneg{A} \types M:B;\, \Delta
      }
      \and
      \infer[Ctx$_3$]{
        \Gamma,\,x\ty:A \types M:B;\,\Delta
      }
      {
        \Gamma \types M:B;\,x\ty:\sneg{A},\,\Delta
      }
      \and
      \infer[Ctx$_4$]{
        \Gamma,\,x\ty:A;\,M:B \types \Delta
      }
      {
        \Gamma;\,M:B \types x\ty:\sneg{A},\,\Delta
      }
    \end{mathpar}
  \caption{Admissible context switching rules.}\label{fig:ctx-switching}
  \Description{Four rules, Ctx1 through Ctx4, move a variable typing between contexts and negate its type while keeping the distinguished typing on the same side of the judgement.}
\end{figure}

\section{\holtty{} Expressive Power}\label{sec:hol2-power}

As an extension of System F$_\omega$, \holty{} can express the logical operations - conjunction, disjunction, existential quantification and so on - via impredicative encodings.  In this section, we shall show that it is also possible to express these operations in \holtty{}.  However, here there is an added layer of complexity, because each type represents \emph{both} its proofs (as usual) \emph{and} its refutations.  That is, a single type, say $A \wedge B$, must be able to support the appropriate introductions and eliminations on the right and on the left\revdel{, and the appropriate eliminations on the right and left}.  In some cases the same encoding works, and in many other cases we can obtain the desired type by taking the dual of an existing type using strong negation, but for conjunction $A \wedge B$ we have to invent \revdel{to} a considerably more complex encoding, to which is devoted a large part of the section.

\paragraph{Existentials and Coimplication.}
We start with the definition of existential quantification and coimplication, since these are, in a sense, already contained in the system as a refutational part of universal quantification and implication respectively.
Both are defined by dualisation:
\[
  \SigmaType{a}{K}{B} \coloneqq \sneg{(\PiType{a}{K}{\sneg{B}})}
  \qquad
  A \cotot B \coloneqq \sneg{(\sneg{A} \to \sneg{B})}
\]
Then, in each case, all the introduction and elimination rules are admissible.  For example, for existentials, we can derive both of the introduction rules (left and right):
\begin{mathpar}
  \infer*{
    \infer*{
      \infer*{
        \Gamma,\,a\ty:K;\,M:B \types \Delta
      }
      {
        \Gamma,\,a\ty:K \types M : \sneg{B};\,\Delta
      }
    }
    {
      \Gamma \types \Tyabs{a}{K}{M} : \PiType{a}{K}{\sneg{B}};\,\Delta
    }
  }
  {
    \Gamma;\,\Tyabs{a}{K}{M} : \SigmaType{a}{K}{B} \types \Delta
  }
  \and
  \infer*{
    \infer*{
      \Gamma \types A : K
      \and
      \infer*{
        \Gamma \types N : B[A/a];\,\Delta
      }
      {
        \Gamma;\, N : \sneg{B}[A/a] \types \Delta
      }
    }
    {
      \Gamma;\,\pack_K\,[\tyabs{a}{K}{\sneg{B}}]\,[A]\,N : \PiType{\revreplace{x}{a}}{K}{\sneg{B}} \types \Delta
    }
  }
  {
    \Gamma \types \pack_K\,[\tyabs{a}{K}{\sneg{B}}]\,[A]\,N : \SigmaType{\revreplace{x}{a}}{K}{B};\,\Delta
  }
\end{mathpar}
Due to the admissibility of the context switching rules (Figure~\ref{fig:ctx-switching}), both of the elimination rules are admissible.  For space reasons, we suppress the type arguments $[\tyabs{a}{K}{\sneg{B}}]$ and $[\sneg{C}]$ of $\unpack_K$.
\begin{mathpar}
  \infer*{
    \infer*{
      \infer*{
        \Gamma;\,M:\SigmaType{a}{K}{B} \types \Delta
      }
      {
        \Gamma \types M : \PiType{a}{K}{\sneg{B}};\,\Delta
      }
      \and\!\!\!\!
      \Gamma \types A : K
    }
    {
      \Gamma \types M\,[A] : \sneg{B[A/a]};\,\Delta
    }
  }
  {
    \Gamma;\,M\,[A]:B[A/a] \types \Delta
  }
  \and\!\!
  \infer*{
    \infer*{
      \infer*{
        \Gamma \types M : \SigmaType{a}{K}{B};\,\Delta
      }
      {
        \Gamma;\,M:\PiType{a}{K}{\sneg{B}} \types \Delta
      }
      \and\!\!\!\!
      \infer*{
        \infer*{
          \Gamma,a\ty:K,x\ty:B \types N : C;\,\Delta
        }
        {
          \Gamma,a\ty:K \types N : C;\,x\ty:\sneg{B},\,\Delta
        }
      }
      {
        \Gamma,a\ty:K;\,N : \sneg{C} \types x\ty:\sneg{B},\,\Delta
      }
    }
    {
      \Gamma;\,\unpack_K\,M\,(\Tyabs{a}{K}{\tyabs{x}{\sneg{(\sneg{B})}}{N}}) : \sneg{C} \types \Delta
    }
  }
  {
    \Gamma \types \unpack_K\,\,M\,(\Tyabs{a}{K}{\tyabs{x}{\sneg{(\sneg{B})}}{N}}) : C;\, \Delta
  }
\end{mathpar}



\newcommand{\labort}{\mathsf{\ell{}abort}}
\newcommand{\rabort}{\mathsf{rabort}}

\paragraph{Truth, Absurdity and Negations}
We can define absurdity using the standard impredicative encoding: its elimination rule\revdel{,} $\rabort$ on the right follows as usual, and it can be introduced on the left by providing a refutation of $0 \cotot 0$ (just as $A \to A$ is provable for any $A$, so $A \cotot A$ is refutable for any $A$).  Then truth can be defined by dualisation:
\begin{align*}
  \falsefm & \coloneqq \PiType{a}{\starsort}{a}                                                                                & \truefm & \coloneqq \sneg{0}                                                                \\
  \rabort  & \coloneqq \Tyabs{c}{\starsort}{\tyabs{x}{\falsefm}{x\,[c]}}                                                       & \labort & \coloneqq \Tyabs{c}{\starsort}{\tyabs{x}{\sneg{\truefm}}{\rabort\,[\sneg{c}]\,x}} \\
  \abra{}  & \coloneqq \pack_\starsort\,[\tyabs{a}{\starsort}{a}]\,[\falsefm \cotot \falsefm]\,(\tyabs{x}{\sneg{\falsefm}}{x})
\end{align*}
It follows immediately that $\types \rabort : \PiType{c}{\starsort}{\falsefm \to c}$ and that $\labort : \SigmaType{c}{\starsort}{\truefm \cotot c} \types\,$.  The same term can be used to introduce both absurdity on the left and truth on the right, i.e. $\types \abra{} : \truefm$ and $\abra{} : \falsefm \types$.  Having defined truth and absurdity we can use the standard definitions of intuitionistic and co-negation, respectively:
\[
  \neg A \coloneqq A \to 0 \qquad\qquad \mathord{-}A \coloneqq A \cotot 1
\]

\newcommand{\refwedge}{\mathrel{\vert}}
\newcommand{\refcase}{\mathsf{case}}
\newcommand{\linj}{\mathsf{\ell{}in}}
\newcommand{\rinj}{\mathsf{rin}}
\newcommand{\lcase}{\mathsf{\ell{}case}}
\newcommand{\rcase}{\mathsf{rcase}}
\newcommand{\lpi}{\mathsf{\ell}\pi}
\newcommand{\rpi}{\mathsf{r}\pi}
\newcommand{\lpair}{\mathsf{\ell{}pair}}
\newcommand{\rpair}{\mathsf{rpair}}

\paragraph{Conjunction and Disjunction}
However, conjunction is more tricky to define.  The usual encoding has the conjunction of $A$ and $B$ as $\PiType{c}{\starsort}{(A \to B \to c) \to c}$ and then we can construct terms that witness the \emph{right} introduction and elimination rules as usual.  However, this encoding does not capture how $A \wedge B$ behaves refutationally, i.e. on the \emph{left}.  For example, if we have a refutation $M$ of $A \wedge B$, then we should expect to be able to eliminate it on the left, by covering both cases (when the refutation of $A \wedge B$ is actually a refutation of $A$, or actually a refutation of $B$, respectively).
\begin{mathpar}
  \infer{
    \Gamma;\,M:\PiType{c}{\starsort}{(A \to B \to c) \to c} \types \Delta
    \and
    \Gamma;\,N:C \types x:A,\,\Delta
    \and
    \Gamma;\,P:C \types \revreplace{x}{y}:B,\,\Delta
  }
  {
    \Gamma;\, \revreplace{f(M,N,P)}{f(M;\,x.N;\,y.P)} : C \types \Delta
  }
\end{mathpar}
\revadd{Here $f$ denotes a construction being sought, with $x$ and $y$ bound in their respective branches.} However, it is not clear how one could define \revreplace{$f(M,N,P)$}{$f(M;\,x.N;\,y.P)$}; \revreplace{to obtain a term of type $C$}{to construct a refutation of $C$}, the second and third premises require refutations of $A$ and $B$ respectively.\revadd{  These are the co-assumptions $x:A$ and $y:B$; each branch constructs a refutation of the same result formula $C$.}  However, \revreplace{the first premise yields only}{unpacking the refutation in the first premise yields, inside its branch,} a proof of $A \to B \to c$ and a refutation of $c$ \revreplace{for some type $c$}{for a fresh type witness $c:\starsort$}.  \revreplace{This is not enough to conclude that either $A$ or $B$ has a refutation, it may only be, for example, that one of $A$ or $B$ is not provable.}{This evidence does not directly supply the co-assumed refutation of $A$ or $B$ required by either branch. The difficulty is thus a mismatch between the evidence supplied by the encoding and the evidence needed for refutation elimination.}

We proceed as follows.  First note that we can isolate the refutational content of $A \wedge B$ by \revdel{a }dualising the usual construction for $A \vee B$.  Let us call this type $A \refwedge B$.
\begin{align*}
  A \refwedge B & \coloneqq \SigmaType{b}{\starsort}{(A \cotot b) \cotot (B \cotot b) \cotot b}                                                                                                               \\
  \inj_1        & \coloneqq \Tyabs{a}{\starsort}{\Tyabs{b}{\starsort}{\tyabs{z}{\sneg{a}}{\Tyabs{c}{\starsort}{\tyabs{x}{\sneg{(a \cotot c)}}{\tyabs{y}{\sneg{(b \cotot c)}}{x\,z}}}}}}                       \\
  \inj_2        & \coloneqq \Tyabs{a}{\starsort}{\Tyabs{b}{\starsort}{\tyabs{z}{\sneg{b}}{\Tyabs{c}{\starsort}{\tyabs{x}{\sneg{(a \cotot c)}}{\tyabs{y}{\sneg{(b \cotot c)}}{y\,z}}}}}}                       \\
  \refcase      & \coloneqq \Tyabs{a}{\starsort}{\Tyabs{b}{\starsort}{\tyabs{x}{\sneg{(a \refwedge b)}}{\Tyabs{c}{\starsort}{\tyabs{f}{\sneg{(a \cotot c)}}{\tyabs{g}{\sneg{(b \cotot c)}}{x\,[c]\,f\,g}}}}}}
\end{align*}
Then, as expected, we can establish \revadd{that }the preceding are all \emph{refutations} of the appropriate types:
\begin{align*}
  \inj_1   & : \SigmaType{ab}{\starsort}{a \cotot (a \refwedge b)} \types                                                                   \\
  \inj_2   & : \SigmaType{ab}{\starsort}{b \cotot (a \refwedge b)} \types                                                                   \\
  \refcase & : \SigmaType{ab}{\starsort}{(a \refwedge b) \cotot \SigmaType{c}{\starsort}{(a \cotot c) \cotot (b \cotot c) \cotot c}} \types
\end{align*}

\newcommand{\andty}[3]{#1 \mathrel{\&}_{#3} \!#2}

Then, to construct a type that faithfully encodes both the proof and refutational content of conjunction, we carefully combine both encodings.
So, define the following family of abbreviations, indexed over types $A$, $B$ and $C$:
\[
  \andty{A}{B}{C} \coloneqq (A \to B \to C) \to \revadd{(}\sneg{(C \cotot (A \refwedge B)) \to C}\revadd{)}
\]
We use these abbreviations in the following:
\begin{align*}
  A \wedge B & \coloneqq \PiType{c}{\starsort}{\andty{A}{B}{c}}                                                                                                                                                                                                    \\
  \linj_1    & \coloneqq \Tyabs{ab}{\starsort}{\tyabs{z}{\sneg{a}}{\revreplace{\pack_K}{\pack_{\starsort}}\,[\tyabs{c}{\starsort}{\andty{a}{b}{c}}]\,[a]\,((\tyabs{x}{a}{\tyabs{y}{b}{x}}) \cdot{} ((\inj_1\,[a]\,[b]) \cdot{} z))}}                               \\
  \linj_2    & \coloneqq \Tyabs{ab}{\starsort}{\tyabs{z}{\sneg{b}}{\revreplace{\pack_K}{\pack_{\starsort}}\,[\tyabs{c}{\starsort}{\andty{a}{b}{c}}]\,[b]\,((\tyabs{x}{a}{\tyabs{y}{b}{y}}) \cdot{} ((\inj_2\,[a]\,[b]) \cdot{} z))}}                               \\
  \lcase     & \coloneqq \Tyabs{ab}{\starsort}{\tyabs{z}{\sneg{(a \wedge b)}}{\Tyabs{d}{\starsort}{\tyabs{f}{\sneg{(a \cotot d)}}{\tyabs{g}{\sneg{(b \cotot d)}}{}{}}}}}                                                                                           \\
             & \qquad \revreplace{\unpack_K}{\unpack_{\starsort}}\,[\tyabs{c}{\starsort}{\andty{a}{b}{c}}]\,z\,[d]\,(\Tyabs{c}{\starsort}{\tyabs{x}{\sneg{(\andty{a}{b}{c})}}{\refcase\,[a]\,[b]\,((\iota_1\,(\iota_2\,x))\,(\iota_2\,(\iota_2\,x)))\,[d]\,f\,g}})
  \\
  \rpair     & \coloneqq \Tyabs{ab}{\starsort}{\tyabs{x}{a}{\tyabs{y}{b}{\Tyabs{c}{\starsort}{\tyabs{f}{a \to b \to c}{\tyabs{g}{\sneg{(c \cotot (a \refwedge b))}}{f\,x\,y}}}}}}
  \\
  \rpi_1     & \coloneqq \Tyabs{ab}{\starsort}{\tyabs{z}{(a \wedge b)}{z\,[a]\,(\tyabs{x}{a}{\tyabs{y}{b}{x}})\,(\inj_1\,[a]\,[b])}}                                                                                                                               \\
  \rpi_2     & \coloneqq \Tyabs{ab}{\starsort}{\tyabs{z}{(a \wedge b)}{z\,[b]\,(\tyabs{x}{a}{\tyabs{y}{b}{y}})\,(\inj_2\,[a]\,[b])}}
\end{align*}
The idea is that the type $A \wedge B$ combines both how it can be eliminated, $A \to B \to c$, when it occurs as a proof \emph{and} how it can be eliminated, $c \cotot (A \refwedge B)$, when it occurs as a refutation.  In the latter case, this is \revreplace{a}{an} description of how to convert a refutation of $\andty{A}{B}{c}$ into a refutation of $A \refwedge B$, for which we already have an eliminator.

To see that these constructions have the appropriate types, we proceed as follows.  For the injections, note that, for $i \in \{1,2\}$:
\begin{mathpar}
  \infer*{
    \infer*{
      a_1,a_2\ty:\starsort;\,\inj_i\,[a_1]\,[a_2] : a_i \cotot (a_1 \refwedge a_2) \types z\ty:a_i,\,\Delta
    }
    {
      a_1,a_2\ty:\starsort \types \inj_i\,[a_1]\,[a_2] : \sneg{(a_i \cotot (a_1 \refwedge a_2))};\,z\ty:a_i,\,\Delta
    }
    \and
    \infer*{ }
    {
      a_1,a_2\ty:\starsort;\,z:a_i \types z\ty:a_i,\,\Delta
    }
  }
  {
    a_1,a_2\ty:\starsort;\,(\inj_i\,[a_1]\,[a_2]) \cdot{} z : \sneg{(a_i \cotot (a_1 \refwedge a_2))} \to a_i \types z\ty:a_i,\,\Delta
  }
\end{mathpar}
Since also $a_1,a_2\ty:\starsort \types \tyabs{x_1}{a_1}{\tyabs{x_2}{a_2}{x_i}} : a_1 \to a_2 \to a_i\revreplace{,}{;}\,z\ty:a_i\revadd{,\,\Delta}$, it follows that both:
\begin{align*}
  \linj_i & : \SigmaType{a_1}{\starsort}{\SigmaType{a_2}{\starsort}{a_i \cotot \revadd{(}a_1 \wedge a_2\revadd{)}}} \types
\end{align*}
Next, we verify the expected type of the case analysis construction.  Let $\revreplace{\Gamma}{\Gamma_0} \coloneqq [a\ty:\starsort,\,b\ty:\starsort,\,\revreplace{c}{d}\ty:\starsort]$ and $\revreplace{\Delta}{\Delta_0} \coloneqq [\revdel{x\ty:\andty{a}{b}{c},\,}g\ty:b \cotot d,\,f\ty:a \cotot d,\,z:(a \wedge b)]$ \revreplace{in the following derivations:}{be the outer environments.}\revadd{  Inside the unpacking branch, choose fresh $c$ and $x$ and put $\Gamma=\Gamma_0,c\ty:\starsort$ and $\Delta=x\ty:\andty{a}{b}{c},\,\Delta_0$. The local co-assumption $x$ supplies a refutation of $\andty{a}{b}{c}$ in the following derivation:}
\begin{mathpar}
  \infer*{
    \infer*{
      \infer*{
        \infer*{
          \infer*{ }
          {
            \Gamma;\,x:\andty{a}{b}{c} \types \Delta
          }
        }
        {
          \Gamma;\,\iota_2\,x : \sneg{(c \cotot (a \refwedge b))} \to c \types \Delta
        }
      }
      {
        \Gamma \types \iota_1\,(\iota_2\,x) : \sneg{(c \cotot (a \refwedge b))};\,\Delta
      }
    }
    {
      \Gamma;\,\iota_1\,(\iota_2\,x) : c \cotot (a \refwedge b) \types \Delta
    }
    \and
    \infer*{
      \infer*{
        \infer*{ }
        {
          \Gamma;\,x:\andty{a}{b}{c} \types \Delta
        }
      }
      {
        \Gamma;\,\iota_2\,x : \sneg{(c \cotot (a \refwedge b))} \to c \types \Delta
      }
    }
    {
      \Gamma;\,(\iota_2\,(\iota_2\,x)) : c \types \Delta
    }
  }
  {
    \Gamma;\,(\iota_1\,(\iota_2\,x))\,(\iota_2\,(\iota_2\,x)) : a \refwedge b \types \Delta
  }
\end{mathpar}
Therefore, also: $\Gamma;\,\refcase\,[a]\,[b]\,((\iota_1\,(\iota_2\,x))\,(\iota_2\,(\iota_2\,x)))\,[d]\,f\,g : d \types \Delta$\revreplace{, and so:}{. Universal-refutation elimination then discharges the local witness $c$ and co-assumption $x$:}
\[
  \revreplace{\Gamma}{\Gamma_0};\,\revreplace{\unpack_K}{\unpack_{\starsort}}\,[\tyabs{c}{\revreplace{K}{\starsort}}{(\andty{a}{b}{c})}]\,z\,[d]\,(\Tyabs{c}{\starsort}{\tyabs{x}{\sneg{(\andty{a}{b}{c})}}{\refcase\,[a]\,[b]\,((\iota_1\,(\iota_2\,x))\,(\iota_2\,(\iota_2\,x)))\,[d]\,f\,g}}) : d \types \revreplace{\Delta}{\Delta_0}
\]
Hence, it follows that we have the refutation:
\[
  \lcase : \SigmaType{ab}{\starsort}{(a \wedge b) \cotot \SigmaType{d}{\starsort}{(a \cotot d) \cotot (b \cotot d) \cotot d}} \types
\]
Since $\inj_i : \SigmaType{a_1}{\starsort}{\SigmaType{a_2}{\starsort}{a_i \cotot (a_1 \refwedge a_2)}} \types\,$, it is straightforward to see that we have proofs:
\begin{align*}
  \types \rpi_1 & : \PiType{ab}{\starsort}{(a \wedge b) \to a}       \\
  \types \rpi_2 & : \PiType{ab}{\starsort}{(a \wedge b) \to b}       \\
  \types \rpair & : \PiType{ab}{\starsort}{a \to b \to (a \wedge b)}
\end{align*}
Then, we can define disjunction by taking the dual of $A \wedge B$ using the strong negation:
\begin{align*}
  A \vee B & \coloneqq \sneg{(\sneg{A} \wedge \sneg{B})}
           & \lpair                                                                                                                                                                                           & \coloneqq \Tyabs{ab}{\starsort}{\tyabs{x}{\revreplace{\sneg{(\sneg{a})}}{\sneg{a}}}{\tyabs{y}{\revreplace{\sneg{(\sneg{b})}}{\sneg{b}}}{\rpair\,[\sneg{a}]\,[\sneg{b}]\,x\,y}}} \\
  \rinj_1  & \coloneqq \Tyabs{ab}{\starsort}{\tyabs{x}{a}{\linj_1\,[\sneg{a}]\,[\sneg{b}]\,x}}\qquad
           & \lpi_1                                                                                                                                                                                           & \coloneqq \Tyabs{ab}{\starsort}{\tyabs{x}{\revreplace{a \vee b}{\sneg{(a \vee b)}}}{\rpi_1\,[\sneg{a}]\,[\sneg{b}]\,x}}                                                         \\
  \rinj_2  & \coloneqq \Tyabs{ab}{\starsort}{\tyabs{x}{b}{\linj_2\,[\sneg{a}]\,[\sneg{b}]\,x}}
           & \lpi_2                                                                                                                                                                                           & \coloneqq \Tyabs{ab}{\starsort}{\tyabs{x}{\revreplace{a \vee b}{\sneg{(a \vee b)}}}{\revreplace{\rpi_1}{\rpi_2}\,[\sneg{a}]\,[\sneg{b}]\,x}}                                    \\
  \rcase   & \coloneqq \mathrlap{\Tyabs{ab}{\starsort}{\tyabs{z}{(a \vee b)}{\Tyabs{d}{\starsort}{\tyabs{f}{a \to d}{\tyabs{g}{b \to d}{}}}}}}                                                                                                                                                                                                                                                  \\
           & \qquad\quad \mathrlap{\lcase\,[\sneg{a}]\,[\sneg{b}]\,z\,[\sneg{d}]\,(\tyabs{x}{\revreplace{a}{\sneg{(\sneg{a})}}}{f\,x})\,(\tyabs{y}{\revreplace{b}{\sneg{(\sneg{b})}}}{g\,\revreplace{x}{y}})}
\end{align*}
We therefore have the following kind of expressive adequacy.
\begin{theorem}[Expressive Adequacy]\label{thm:hol2-expr-adequacy}
  \holtty{} can adequately express the logical constructions $0$ and $1$, $\PiType{a}{K}{A}$ and $\SigmaType{a}{K}{A}$, $A \to B$ and $A \cotol B$, $A \wedge B$ and $A \vee B$, $\neg A$ and $\mathord{-} A$.
\end{theorem}

\section{\holtty{} Consistency}\label{sec:hol2-consistency}

We establish certain key properties of the system by a translation into Geuvers' \holty{}.  We give a full definition in
\iftoggle{supplementary}{%
  Appendix~\ref{sec:apx-hol2-consistency},
}%
{%
  the appendices of the supplementary material,
}%
but for short one can understand \holty{} as a reduction of \holtty{} by the following operations:
\begin{itemize}
  \item Removing $\sneg{}$ from the syntax of types, $\cdot{}$, $\iota_1$, $\iota_2$, $\pack_K$ and $\unpack_K$ from the syntax of terms.
  \item Removing the kinding rule \rlnm{SNeg}.
  \item Removing the rules \rlnm{Id-L}, \rlnm{$\forall$-Intro-L}, \rlnm{$\forall$-Elim-L}, \rlnm{$\to$-Intro-L}, \rlnm{$\to$-Elim-L$_1$}, \rlnm{$\to$-Elim-L$_2$}, \rlnm{$\sneg$-Intro-L}, \rlnm{$\sneg$-Intro-R}, \rlnm{$\sneg$-Elim-L}, \rlnm{$\sneg$-Elim-R}.
  \item Removing the co-assumption environment $\Delta$ from all the remaining rules.
\end{itemize}

\newcommand{\mkPair}{\mathsf{mkPair}}

We assume the standard encoding of conjunction $\mathord{\wedge} : \starsort \to \starsort \to \starsort$ in \holty{}, along with introduction form $\mkPair : \PiType{ab}{\starsort}{a \to b \to (a \wedge b)}$, and eliminators $\pi_1 : \PiType{ab}{\starsort}{(a \wedge b) \to a}$ and $\pi_2 : \PiType{ab}{\starsort}{(a \wedge b) \to b}$.  We assume the standard encoding of existential types $\exists_K : (K \to \starsort) \to \starsort$ in \holty{}, with introduction $\pack_K : \PiType{a}{K \to \starsort}{\PiType{b}{K}{a\,b \to \SigmaType{c}{K}{a\,c}}}$ and elimination $\unpack_K : \PiType{a}{K \to \starsort}{(\SigmaType{c}{K}{a\,c}) \to \PiType{b}{\starsort}{(\PiType{c}{K}{a\,c \to b}) \to b}}$.

In intuitionistic logics extended with strong negation, e.g. \cite{nelson1949constructible,gurevich1977intuitionistic}, it is typical to derive a number of useful results by constructing a mapping of provable formulas from the extended logic (with strong negation) into provable formulas of the underlying logic (without strong negation).  The idea is, given a formula of the extended logic, to obtain an equivalent, `reduced', formula in the underlying logic by pushing strong negations inward until they are atomic, and then replacing strongly negative literals $\sneg{a}$ by (intuitionistically) negative ones \revreplace{$\neg\,A$}{$\neg\,a$}.

For example, consider the formulas $\sneg{(p \to q)} \to p \wedge \sneg{q}$ and $\neg(p \to q) \to p \wedge \neg q$, where the latter is obtained by replacing every occurrence of the strong negation by the intuitionistic one.  The first of these formulas is provable in our system \tintcty{}, but the second is not\footnote{The strong negation is definable in \tint{} \cite{oddsson2025strongnegationdefinable2int}, and \tint{} is conservative over intuitionistic propositional logic \cite{wansing2016falsification}, but this formula is not intuitionistically valid.}.  However, we can observe that, by pushing the strong negation inward using the equivalences (a)--(f) of Section~\ref{sec:constructive-falsity}, the first formula is equivalent to $p \wedge \sneg{q} \to p \wedge \sneg{q}$, in which each strong negation is atomic.  Now, when each occurrence of a strong negation is replaced with the corresponding intuitionistic negation, an equivalent formula $p \wedge \neg q \to p \wedge \neg q$ is obtained, and this formula is already provable in intuitionistic logic (without strong negation).

However, this strategy fails in higher-order logic.  Consider the formula:
\begin{equation}\label{eq:hol2-bad-example-1}
  \SigmaType{a}{\starsort}{(a \to p \to q) \wedge (\sneg{a} \to p \wedge \sneg{q})}
\end{equation}
This formula is provable in (the logic of) \holtty{} by supplying the witness $a \coloneqq p \to q$.  However, \camrevadd{this }formula  is already `reduced': all the strong negations are atomic, and after replacing strong by intuitionistic negations, we obtain:
\begin{equation}\label{eq:hol2-bad-example-2}
  \SigmaType{a}{\starsort}{(a \to p \to q) \wedge (\neg a \to p \wedge \neg q)}
\end{equation}
Formulas (\ref{eq:hol2-bad-example-1}) and (\ref{eq:hol2-bad-example-2}) are clearly \emph{not} equivalent, since the latter implies $\neg (p \to q) \to p \wedge \neg q$.  Similar counterexamples can be constructed by exploiting the convertibility of formulas, e.g.
\begin{equation}\label{eq:hol2-bad-example-3}
  (\tyabs{a}{\starsort}{(a \to p \to q) \wedge (\sneg{a} \to p \wedge \sneg{q})})\,(p \to q)
\end{equation}
Hence, for \holtty{} we must construct a more sophisticated mapping.

Our approach is, given a formula of \holtty{}, to replace all subformulas (i.e. subexpressions of kind $\starsort$) by \emph{pairs} of subformulas, \camrevreplace{in which the first component of the pair is (convertible with) the original subformula, and the second component is (convertible with) the equivalent of the strong negation of that subformula according to (a)--(f) of Section~\ref{sec:constructive-falsity}}{in which the first component represents the proof content of the subformula and the second represents its refutation content, following the strong-negation equivalences (a)--(f) of Section~\ref{sec:constructive-falsity}}.
\begin{definition}
  In \holty{} (and \holtty{}), define the following abbreviations, ensuring $a$ in the second clause is chosen fresh with respect to $A$ and $B$, i.e. $a \notin \fv(A,B)$.
  \begin{align*}
    \starsort^2 & \coloneqq (\starsort \to \starsort \to \starsort) \to \starsort
                & \pi_1(A)                                                            & \coloneqq A\,(\tyabs{b}{\starsort}{\tyabs{c}{\starsort}{b}})                 \\
    \abra{A,B}  & \coloneqq \tyabs{a}{\starsort \to \starsort \to \starsort}{a\,A\,B}
                & \pi_2(A)                                                            & \coloneqq A\,(\tyabs{b}{\starsort}{\tyabs{c}{\starsort}{\revreplace{b}{c}}})
  \end{align*}
\end{definition}

When such a pair of subformulas $A : \starsort^2$ is used in a position where a single formula is required, e.g. as an argument to a standard logical operator, like implication, then the first component is projected out.  The action of the strong negation on such pairs is to choose the second component, rather than the first.  Thus, a subformula like $A \to B$ will be mapped to a pair $\abra{\pi_1(\derefT(A)) \to \pi_1(\derefT(B)),\,\pi_1(\derefT(A)) \wedge \pi_2(\derefT(B))}$, where $\derefT(A)$ and $\derefT(B)$ are the \camrevreplace{image}{images} of $A$ and $B$ under the mapping.

So, the idea is that the formula (\ref{eq:hol2-bad-example-3}) above will be translated to (a type convertible with)\footnote{\camrevreplace{However, the fact that conjunction is a defined operator complicates this particular example.}{Here conjunction is treated schematically as a logical connective. In the actual mapping below, its encoding is expanded before translation.}}:
\[
  (\tyabs{a}{\starsort^2}{(\pi_1(a) \to \pi_1(p) \to \pi_1(q)) \wedge (\pi_2(a) \to \pi_1(p) \wedge \pi_2(q))})\,\abra{\pi_1(p) \to \pi_1(q),\,\pi_1(p) \wedge \pi_2(q)}
\]
where now we assume $p,q:\starsort^2$. Thus, an encoding of the strong negation of a particular logical operation is constructed in tandem with the original operation, and this pair of formulas can then be passed along the control flow of the type until the appropriate projections are applied.  Consequently, this formula is provable in both \holtty{} and \holty{}.

\holtty{} is more than just its types, so the complete translation consists of mappings on kinds, types and proof terms.
\begin{definition}[Mapping from \holtty{} to \holty{}]
  In Figure~\ref{fig:hol2-to-hol-mappings}, we define three mappings: $\derefK$ from kinds of \holtty{} to kinds of \holty{}, $\derefT$, from types of \holtty{} to types of \holty{} and, finally, $\derefP$ from terms of \holtty{} to terms of \holty{}.
\end{definition}

\begin{figure}
  \begin{minipage}{.6\textwidth}
    \begin{align*}
      \derefK(k)                & = k                                                                                               \\
      \derefK(\starsort)        & = \starsort^2                                                                                     \\
      \derefK(K_1 \to K_2)      & = \derefK(K_1) \to \derefK(K_2)                                                                   \\
                                &                                                                                                   \\
      \derefT(a)                & = a                                                                                               \\
      \derefT(A\,B)             & = \derefT(A)\,\derefT(B)                                                                          \\
      \derefT(\tyabs{a}{K}{A})  & = \tyabs{a}{\derefK(K)}{\derefT(A)}                                                               \\
      \derefT(A \to B)          & = \abra{\pi_1(\derefT(A)) \to \pi_1(\derefT(B)),\,\pi_1(\derefT(A)) \wedge \pi_2(\derefT(B))}     \\
      \derefT(\PiType{a}{K}{A}) & = \abra{\PiType{a}{\derefK(K)}{\pi_1(\derefT(A))},\,\SigmaType{a}{\derefK(K)}{\pi_2(\derefT(A))}} \\
      \derefT(\sneg{A})         & = \abra{\pi_2(\theta(A)),\,\pi_1(\theta(A))}
    \end{align*}
  \end{minipage}%
  \begin{minipage}{.4\textwidth}
    \begin{align*}
      \derefP(x)               & = x                                                                              \\
      \derefP(\iota_1)         & \camrevreplace{= \pi_1}{= \widehat{\pi}_1}                                       \\
      \derefP(\iota_2)         & \camrevreplace{= \pi_2}{= \widehat{\pi}_2}                                       \\
      \derefP(\cdot{})         & \camrevreplace{= \mkPair{}}{= \widehat{\mkPair}}                                 \\
      \derefP(\pack_K)         & \camrevreplace{= \pack_{\derefK(K)}}{= \widehat{\pack}_{\derefK(K)}}             \\
      \derefP(\unpack_K)       & \camrevreplace{= \mathsf{unpack}_{\derefK(K)}}{= \widehat{\unpack}_{\derefK(K)}} \\
      \derefP(M\,N)            & = \derefP(M)\,\derefP(N)                                                         \\
      \derefP(M\,[A])          & = \derefP(M)\,[\derefT(A)]                                                       \\
      \derefP(\tyabs{x}{A}{M}) & = \tyabs{x}{\derefT_1(A)}{\derefP(M)}                                            \\
      \derefP(\Tyabs{a}{K}{M}) & = \Tyabs{a}{\derefK(K)}{\derefP(M)}
    \end{align*}
  \end{minipage}
  \caption{Mapping from \holtty{} into \holty{}: kinds, types and proofs.}\label{fig:hol2-to-hol-mappings}
  \Description{The left column gives the kind and type translations, including how strong negation swaps the proof and refutation components of a type. The right column gives the term translation for variables, constants, applications, and abstractions.}
\end{figure}
\camrevadd{The hatted constants used in Figure~\ref{fig:hol2-to-hol-mappings} adapt the paired type arguments supplied by $\derefP$ to the ordinary formula arguments expected by the encodings. Their types are given in
  \iftoggle{supplementary}{%
    Appendix~\ref{sec:apx-hol2-consistency}.
  }{%
    the long version of the paper~\cite{li2026logicalfoundationstwosidedtype}
  }%
}
Let us abbreviate $\pi_i(\derefT(A))$ by $\derefT_i(A)$.  This way, the fourth clause of the above definition, for example, can be more simply written as $\derefT(A \to B) = \abra{\derefT_1(A) \to \derefT_1(B),\,\derefT_1(A) \wedge \derefT_2(B)}$.
By an abuse of notation, define the extension of $\theta_1$ and $\theta_2$ pointwise over environments by, for each $i \in \{1,2\}$:
\begin{align*}
  \derefT_i(\varnothing)          & = \varnothing                                \\
  \derefT_i(\Gamma,a\mathord{:}K) & = \derefT_i(\Gamma),a\mathord{:}\derefK(K)   \\
  \derefT_i(\Gamma,x\mathord{:}A) & = \derefT_i(\Gamma),x\mathord{:}\derefT_i(A)
\end{align*}

We show that all three of the mappings interact properly, by demonstrating that proofs of given formulas under assumptions and co-assumptions are preserved.

\begin{lemma}\label{lem:2hol-proof-translations}
  Both of the following:
  \begin{enumerate}
    \item If $\Gamma \types M : A;\,\Delta$ in \holtty{} then $\derefT_1(\Gamma),\derefT_2(\Delta^R) \types \derefP(M) : \derefT_1(A)$ in \holty{}.
    \item If $\Gamma;\,M : A \types \Delta$ in \holtty{} then $\derefT_1(\Gamma),\derefT_2(\Delta^R) \types \derefP(M) : \derefT_2(A)$ in \holty{}.
  \end{enumerate}
\end{lemma}

\noindent
As a corollary, we obtain the consistency of the system.

\begin{theorem}[Consistency]\label{thm:hol2-consistency}
  There is no proof term $M$ such that $\types M : \PiType{a}{\starsort}{a}$ in \holtty{}.
\end{theorem}
\begin{proof}
  Suppose for contradiction that there is such an $M$.  Then by Lemma~\ref{lem:2hol-proof-translations}, there is a \holty{} derivation of $\types \derefP(M) : \derefT_1(\PiType{a}{\starsort}{a})$.  By conversion in \holty{}, there is a derivation of $\types \derefP(M) : \PiType{a}{\starsort^2}{\pi_1\,a}$.  Let $b$ be a fresh type variable, then by weakening, there is a \holty{} derivation of $b\ty:\starsort \types \derefP(M) : \PiType{a}{\starsort^2}{\pi_1\,a}$, and thus one can also derive $\types \Tyabs{b}{\starsort}{\derefP(M)\,[\abra{b,b}]} : \PiType{b}{\starsort}{b}$.  However, this type is not inhabited in \holty{}.
\end{proof}

\section{\holtty{} Proofs and Refutations}\label{sec:hol2-proofs}

In this section, we show a version of the existence property for our system. That is, in any kinding environment, from either a proof of an existential $\SigmaType{a}{K}{A}$ or a proof of a strongly negated universal $\sneg{(\PiType{a}{K}{A})}$ one can obtain directly an appropriate witness of kind $K$.

Unfortunately, our mappings from \holtty{} into \holty{} of the previous section do not shed much light on this aspect of the system, since the structure of a given \holtty{} proof may be changed quite substantially when undergoing the mapping.  For example, a proof $M$ of a formula $\SigmaType{a}{(\starsort \to \starsort) \to \starsort}{A}$ is sent to a proof $\derefP(M)$ of a formula $\SigmaType{a}{(\starsort^2 \to \starsort^2) \to \starsort^2}{\derefT_1(A)}$.  So, although the image of the formula is again an existential, it is not clear how one would obtain the correct \holtty{} witness of kind $(\starsort \to \starsort) \to \starsort$ from the \holty{} witness that has kind $(\starsort^2 \to \starsort^2) \to \starsort^2$, in general.

Therefore, we look to a more \camrevreplace{fine grained}{fine-grained} analysis of the structure of proofs and refutations.  To this end, define reduction on proof\revadd{ and refutation} terms as follows.
\begin{definition}
  The relation of \emph{proof reduction}, written $M \ped N$, is defined inductively as the contextual closure of the following rules:
  \begin{mathpar}
    (\tyabs{x}{A}{M})\,N \ped M[N/x]
    \and
    (\Tyabs{a}{K}{M})\,[A] \ped M[A/a]
    \and
    \iota_i\,[A]\,[B]\,(M_1 \cdot_{C,D} M_2) \ped M_i\\
    \unpack_K\,[\tyabs{a}{K}{A}]\,(\pack_K\,[\tyabs{a}{K}{B}]\,[C]\,N)\,[D]\,(\Tyabs{a}{K}{\tyabs{x}{E}{M}}) \ped M[C/a][N/x]
  \end{mathpar}
\end{definition}

Although our mappings from \holtty{} to \holty{} are not immediately useful to the goal of this section, they do \camrevreplace{have }allow us to obtain one important result very cheaply.  Notice that the mapping $\derefP$ maps redexes to redexes, and is otherwise compatible with the term structure.  Consequently, every proof reduction step in \holtty{} can be mirrored, under the mapping, in \holty{}.

\begin{lemma}\label{lem:2hol-reduction-translation}
  If $M \ped N$ in \holtty{} then $\derefP(M) \pedp \derefP(N)$ in \holty{}.
\end{lemma}

\noindent Here $\!\!\mathord{\pedp}\!\!$ is the \revadd{nonempty }transitive closure of $\!\!\mathord{\ped}\!\!$. As a corollary, we obtain strong normalisation.

\begin{theorem}[Strong Normalisation]\label{thm:hol2-sn}
  Both of the following:
  \begin{itemize}
    \item If $\Gamma;\,M:A \types \Delta$ then $M$ is strongly normalising.
    \item If $\Gamma \types M : A;\,\Delta$ then $M$ is strongly normalising.
  \end{itemize}
\end{theorem}
\begin{proof}
  Suppose $\Gamma \types M : A;\,\Delta$, the other case is similar.  Suppose for contradiction that $M$ starts some infinite reduction sequence.  Then, by Lemma~\ref{lem:2hol-reduction-translation}, $\derefP(M)$ also starts some infinite reduction sequence.  However, it follows from Lemma~\ref{lem:2hol-proof-translations} that $\derefT_1(\Gamma),\derefT_2(\Delta^R) \types \derefP(M) : \derefT_1(A)$, contradicting the strong normalisation of \holty{}.
\end{proof}

To \revreplace{futher}{further} analyse the structure of proofs, we aim to prove an inversion lemma and a canonical forms lemma.  Typically (e.g. in System F$\omega$), these lemmas have to take into account that, due to the implicit nature of the \rlnm{Conv} rule, a given term may be assigned many different types that are connected by convertibility.  In \holtty{}, there is an analogous consideration for the strong negation introduction and elimination rules.  Therefore, it is not enough to consider only convertibility, we must take into account that types may be related by pairs of strong negations.

\newcommand{\nconv}{\cong_\beta}

\begin{definition}
  Define \revadd{the }binary relation $A \nconv B$ as the transitive closure of:
  \[
    \mathord{\conv} \cup \{(\sneg{(\sneg{A})}, A) \mid A \,\in \mathbb{T}\} \cup \{(A,\sneg{(\sneg{A})}) \mid A\,\in \mathbb{T}\}
  \]
  We say that two types $A$ and $B$ are \emph{neg-convertible} just if $A \nconv B$.\revadd{ Here $\mathbb{T}$ is the set of type expressions.}
\end{definition}

Note, however, that this relation is not closed under (one-hole) contexts - the double negation manipulations can only happen at the root of the type.  However, it is clear that \holtty{} derivability is closed under neg-convertibility, in the sense that, if $A \nconv B$ \camrevadd{and $\Gamma \types B : \starsort$ }then both: $\Gamma \types M : A;\,\Delta$ implies $\Gamma \types M : B;\,\Delta$, and $\Gamma;\,M:A \types \Delta$ implies $\Gamma;\,M:B \types \Delta$.

Now we can understand more precisely the necessary conditions under which certain kinds of proof must arise.

\begin{lemma}[Inversion]\label{lem:2hol-inversion}
  All of the following:
  \begin{enumerate}[(1)]
    \item If $\Gamma \types \tyabs{x}{A}{M} : B;\,\Delta$ and $B \nconv B_1 \to B_2$ then $A \conv B_1$ and $\Gamma,\,x\ty:A \types M : B_2;\,\Delta$.
    \item If $\Gamma;\,\tyabs{x}{A}{M} : B \types \Delta$ and $B \nconv \sneg{(B_1 \to B_2)}$ then $A \conv B_1$ and $\Gamma,\,x\ty:A \types M : B_2;\,\Delta$.
    \item If $\Gamma \types \Tyabs{a}{K}{M} : B;\,\Delta$ and $B \nconv \PiType{a}{K'}{C}$ then $K = K'$ and $\Gamma,\,a\ty:K \types M : C;\,\Delta$.
    \item If $\Gamma;\,\Tyabs{a}{K}{M} : B\revreplace{;\,}{\types}\Delta$ and $B \nconv \sneg{(\PiType{a}{K'}{C})}$ then $K = K'$ and $\Gamma,\,a\ty:K \types M : C;\,\Delta$.
    \item If $\Gamma;\,M \cdot_{A_1,A_2} N : B \types \Delta$ and $B \nconv B_1 \to B_2$ then $A_1 \conv B_1$, $A_2 \conv B_2$, $\Gamma \types M : B_1;\,\Delta$ and $\Gamma;\,N:B_2 \types \Delta$.
    \item If $\Gamma \types M \cdot_{A_1,A_2} N : B;\,\Delta$ and $B \nconv \sneg{(B_1 \to B_2)}$ then $A_1 \conv B_1$, $A_2 \conv B_2$, $\Gamma \types M : B_1;\,\Delta$ and $\Gamma;\,N:B_2 \types \Delta$.
    \item If $\Gamma;\,\pack_K\,[\tyabs{a}{K}{A}]\,[C]\,N : B \types \Delta$ and $B \nconv \PiType{a}{K'}{D}$ then $K = K'$ and $A \conv D$ and $\Gamma \types C : K$ and $\Gamma;\,N:A[C/a] \types \Delta$.
    \item If $\Gamma \types \pack_K\,[\tyabs{a}{K}{A}]\,[C]\,N : B;\,\Delta$ and $B \nconv \sneg{(\PiType{a}{K'}{D})}$ then $K = K'$ and $A \conv D$ and $\Gamma \types C : K$ and $\Gamma;\,N:A[C/a] \types \Delta$.
  \end{enumerate}
\end{lemma}

\noindent Moreover, when a proof (or refutation) is in normal form, then its structure reflects that of the formula that it is proving (or refuting).

\begin{lemma}[Canonical Proofs and Refutations]\label{lem:2hol-can-forms}
  Suppose $\Gamma$ is a kinding environment\revadd{, $\Delta=\varnothing$,} and $N$ is in normal form.
  \begin{enumerate}[(i)]
    \item If $\Gamma \types N : A;\,\Delta$ then:
          \begin{enumerate}[(a)]
            \item If $A$ is neg-convertible with an arrow then $N$ is a term abstraction.
            \item If $A$ is neg-convertible with a universal type over $K$ then $N$ is a type abstraction over $K$.
            \item If $A$ is neg-convertible with the strong negation of an arrow, then $N$ is a dot construction.
            \item If $A$ is neg-convertible with the strong negation of a universal type over $K$, then $N$ is a $K$-pack.
          \end{enumerate}
    \item If $\Gamma;\,N:A \types \Delta$ then:
          \begin{enumerate}[(a)]
            \item If $A$ is neg-convertible with an arrow then $N$ is a dot construction.
            \item If $A$ is neg-convertible with a universal type over $K$ then $N$ is a $K$-pack.
            \item If $A$ is neg-convertible with the strong negation of an arrow, then $N$ is a term abstraction.
            \item If $A$ is neg-convertible with the strong negation of a universal type over $K$, then $N$ is a type abstraction over $K$.
          \end{enumerate}
  \end{enumerate}
\end{lemma}

After first obtaining the preservation of proof reduction under type and term substitutions, a straightforward induction establishes that the type of the subject is preserved under proof reduction.

\begin{theorem}[Preservation]\label{thm:hol2-preservation}
  Both of the following:
  \begin{itemize}
    \item If $\Gamma;\,M : A \types \Delta$ and $M \ped N$ then $\Gamma;\,N:A \types \Delta$
    \item If $\Gamma \types M : A;\,\Delta$ and $M \ped N$ then $\Gamma \types N:A;\,\Delta$
  \end{itemize}
\end{theorem}

By combining the above results, we can finally obtain the existence properties.

\begin{theorem}[Existence Property]\label{thm:hol2-existence}
  Suppose $\Gamma$ is a kinding environment.  Then both of:
  \begin{enumerate}[(i)]
    \item If $\Gamma \types M : \SigmaType{a}{K}{A}$ then there \camrevreplace{is}{are} $B$ and $N$ such that $\Gamma \types B : K$ and $\Gamma \types N : A[B/a]$.
    \item If $\Gamma;\, M:\PiType{a}{K}{A} \types $ then there \camrevreplace{is}{are} $B$ and $N$ such that $\Gamma \types B : K$ and $\Gamma;\,N : A[B/a] \types $.
  \end{enumerate}
\end{theorem}
\begin{proof}
  Both are similar, so we only do (i).  Suppose $\Gamma \types M : \SigmaType{a}{K}{A} = \sneg{(\PiType{a}{K}{\sneg{A}})}$.  Then, by Theorem~\ref{thm:hol2-sn}, $M$ is strongly normalising.  Let $Q$ be a normal form of $M$.  Then, by Theorem~\ref{thm:hol2-preservation}, $\Gamma \types Q : \sneg{(\PiType{a}{K}{\sneg{A}})}$.  By Lemma~\ref{lem:2hol-can-forms}, $Q$ is of shape $\pack_K\,[\tyabs{a}{K}{C}]\,[B]\,N$ and, by Lemma~\ref{lem:2hol-inversion}, $C \conv \sneg{A}$ and $\Gamma \types B : K$ and $\Gamma;\,N : C[B/a] \types $.  Therefore, also $\Gamma \types N : A[B/a]$ as required.
\end{proof}
\paragraph{Examples}
The ability to construct refutations is very useful under the propositions-as-types paradigm.  In what follows we will present some examples that have been inspired by the ideas of Jan \revreplace{De}{de} Muijnck-Hughes (and Robert Atkey) for implementing a notion of positive evidence for the falsification of predicates in Idris \cite{deMuijnckHughes2025PositiveNegative,demuinckhughes-github}.
\newcommand{\keyknd}{\mathsf{key}}
\newcommand{\valknd}{\mathsf{val}}

\begin{example}
  Suppose we have kinds $\keyknd$ and $\valknd$ and a type $M : \keyknd \to \valknd \to \starsort$ describing a particular key-value map.
  Then we can define the type:
  \[
    \abbv{All}_M \coloneqq \tyabs{p}{\valknd \to \starsort}{\PiType{x}{\keyknd}{\PiType{y}{\valknd}{M\,x\,y \to p\,y}}}
  \]
  For a given predicate $P\revadd{:\valknd\to\starsort}$, a proof $\camrevreplace{M}{A} : \abbv{All}_{\revreplace{M}{\camrevreplace{A}{M}}}\,P$ is a function that can be used to verify the satisfaction of $P$ for any key-value pair $(x,y)$ in the map.\revadd{ In particular, given type expressions $K:\keyknd$ and $V:\valknd$ and a proof $m:M\,K\,V$, the term $A\,[K]\,[V]\,m$ proves $P\,V$.}  However, it may be that the codomain of $M$ does not satisfy $P$ universally.  In this case we may be able to obtain a proof \revdel{of }$N : \neg\,(\abbv{All}_M\,P)$.  However, this seems much less useful: $N$ is a function which, given an (impossible) proof of $\abbv{All}_M\,\revreplace{p}{P}$, yields a proof of absurdity.  Much more useful is to construct a \emph{refutation} $Q : \sneg{(\abbv{All}_M\,P)}$ which consists of \emph{witnesses} $K:\keyknd$ and $V:\valknd$ together with a proof $Q_1 : M\,K\,V$ and a refutation $Q_2 : \sneg{(P\,V)}$.
\end{example}

\begin{example}
  In higher-order logic it is common to define equality via Aeibniz's principle.  This can be formulated as follows in \holtty{}:
  \[
    \mathord{=}_K \coloneqq \tyabs{xy}{K}{\PiType{z}{K \to \starsort}{z\,x \to z\,y}}
  \]
  If $A$ and $B$ are not equal, then we can attempt to construct a proof $N : \neg(A =_K B)$, but the computational content of this $N$ is not so insightful: given an (impossible) proof of the equality of $A$ and $B$, we can use $N$ to obtain a proof of absurdity.  On the other hand, a refutation $Q : \sneg{(A =_K B)}$ amounts to a witness $P : K \to \starsort$, which is a particular property that is (constructively) satisfied by $A$ but not $B$, i.e. there is a proof of $P\,A$ and a refutation of $P\,B$.
\end{example}

\newcommand{\nodeknd}{\mathsf{n}}
\newcommand{\natknd}{\mathsf{nat}}
\newcommand{\succty}{\mathsf{s}}

\begin{example}\label{ex:hol2-graph}
  Suppose we have \camrevreplace{an}{a} description of a rooted, directed graph as a kind of nodes $\nodeknd$ and types $I : \nodeknd$, representing the root, and $E : \nodeknd \to \nodeknd \to \starsort$, representing the edge relation.  Then we can define the type:
  \[
    \abbv{Reachable}^{\nodeknd}_{I,E} \coloneqq \tyabs{y}{\nodeknd}{\PiType{z}{\nodeknd \to \starsort}{z\,I \wedge (\PiType{uv}{\nodeknd}{z\,u \wedge E\,u\,v \to z\,v}) \to z\,y}}
  \]
  This type encodes a standard second-order definition of reachability, namely that a node $y$ is reachable from the root whenever every subset of nodes containing the root and closed under the edge relation also contains $y$.  In the case that a given $Y:\nodeknd$ is unreachable, one could attempt a proof $N : \neg(\abbv{Reachable}^\nodeknd_{I,E}\,Y)$, but the computational content of such a proof is not so interesting.  By contrast, a \emph{refutation} $Q : \sneg{(\abbv{Reachable}^\nodeknd_{I,E}\,Y)} $ consists of a particular edge-closed subset of nodes that contains $I$ and yet omits $Y$, in other words, an inductive invariant that certifies the unreachability of (bad state) $Y$.

  Let us illustrate this example more concretely.  Suppose we work in a context containing an axiomatisation of a natural number type.  For this simple example, all we will require are:
  \begin{itemize}
    \item A kind $\natknd$.
    \item Two type constructors $0 : \natknd$ and $\succty : \natknd \to \natknd$.
    \item Two proofs witnessing the non-finiteness of $\natknd$:
          \begin{align*}
            h_0 & : \PiType{x}{\natknd}{\sneg{(0 = \succty\,x)}}                              \\
            h_1 & : \PiType{xy}{\natknd}{\sneg{(x = y)} \to \sneg{(\succty\,x = \succty\,y)}}
          \end{align*}
  \end{itemize}
  Define a directed graph by $\abbv{Edge} \coloneqq \tyabs{xy}{\natknd}{y = \succty\,x}$ and let us root the graph at $2$, i.e. define $\abbv{Init} \coloneqq \succty\,(\succty\, 0)$.  Then it follows that $1$ is unreachable, i.e. $\camrevreplace{\sneg{({\abbv{Reachable}^\natknd_{\abbv{Init},\abbv{Edge}}})}\,(\succty\,0)}{\sneg{(\abbv{Reachable}^\natknd_{\abbv{Init},\abbv{Edge}}\,(\succty\,0))}}$.  Since we are proving the strong negation of a universal, we are forced to supply a witness, which we can take as $W \coloneqq \tyabs{n}{\natknd}{\SigmaType{m}{\natknd}{n = \succty\,(\succty\,m)}}$. It is easy to prove, by the properties of Leibniz equality, that this set contains $2$ and is closed under the edge relation.  Moreover, we can establish $\sneg{(W\,(\succty\,0))}$, i.e. $\sneg{(\SigmaType{m}{\natknd}{\succty\,0 = \succty\,(\succty\,m)})}$.  Since this is equivalent to $\PiType{m}{\natknd}{\sneg{(\succty\,0 = \succty\,(\succty\,m))}}$, it follows directly from $h_0$ and $h_1$.
\end{example}

\begin{example}
  In \cite{deMuijnckHughes2025PositiveNegative,demuinckhughes-github}, de Muijnck-Hughes defines a new notion of decidability.  For a unary predicate $p$ to be decidable in this alternative sense, there must be a procedure which, for any instance $x$, either produces a proof of $p\,x$ or produces a proof of $q\,x$.  In many cases of interest, $q\,x$ is an Idris type encoding the strong negation of $p\,x$, but the idea is that one may choose any $q$ subject to the requirement that $p\,x$ and $q\,x$ are inconsistent.  In \holtty{} we have strong negation as a primitive, so we may define our own notion\footnote{Note, our notion and de Muijnck-Hughes' are generally incomparable since we do not assume that $A$ and $\sneg{A}$ are inconsistent.} of strong decidability on kind $K$:
  \[
    \abbv{SDec}_K \coloneqq \tyabs{p}{K \to \starsort}{\PiType{x}{K}{p\,x \vee \sneg{(p\, x)}}}
  \]
  A proof $M : \abbv{SDec}_K\,P$ for a given predicate type $P$ is a function which, for any input $A$ of kind $K$ produces either a proof of $P\,A$ or a refutation of $P\,A$.
\end{example}

\section{Conclusion and Related Work}\label{sec:concl-related}

We have introduced three new two-sided type systems, \tintty{}, \tintcty{} and \holtty{}.  For the former two, we have shown that derivability of a single-subject judgement in the type system coincides with \revreplace{provability of the corresponding logical sequent of}{provability or refutability of the corresponding formula in} \tint{},  and its extension with strong negation, respectively. To the best of our knowledge, there is no existing work on extensions of HOL with strong negation, but we have shown that our system \holtty{} is \revadd{expressively adequate, }consistent, strongly normalising and satisfies the existence and dual existence properties.  Since each of these systems can be naturally understood as an extension of the core rules of two-sided type systems, as identified in \cite{ramsay-walpole-popl24,li-etal-popl26}, \revreplace{this gives a reasonable explanation of the logical foundations of two-sided type systems}{this gives a logical interpretation of that shared core} through the lens of propositions-as-types.\revadd{ The constructive interpretation does not preserve every program-incorrectness judgement of the earlier systems.}

\revdel{
  \paragraph{Two-sided type systems}
  Our systems extend the pure $\lambda$-calculus rules identified in the introduction to \cite{ramsay-walpole-popl24} and which were refined and augmented with complementation in \cite{li-etal-popl26}.  Those works did not consider logical aspects except to illustrate the connection between their co-arrow and logical co-implication.  In our work, the intended meaning of a typing formula $M:A$ on the left-hand side of the judgement, namely that \emph{$M$ is a refutation of $A$}, is rather stronger.  This is necessary for the soundness of our rules.  In the previous works, the intended meaning of such a typing formula is \emph{$M$ is not a term of type $A$}.  For example, the following is an instance of our \rlnm{PiL} of \tintty{}:
  \begin{mathpar}
    \infer{
      M: A \vee B \types
    }
    {
      \pi_1(M) : A \types
    }
  \end{mathpar}
  If we know only that $M$ is not a term of type $A \vee B$, it is not sound to conclude that $\pi_1(M)$ is not a term of type $A$.  The former can be satisfied by $M \coloneqq (P,Q)$ for some $P:A$.  However, if we know that $M$ is a refutation of $A \vee B$, then, indeed, $\pi_1(M)$ must be a refutation of $A$.
}

\paragraph{Bilateral Logics}
Bilateralism can be understood as a part of proof theoretic semantics in which dual concepts in logic, such as assertion and denial, or proof and refutation, are considered symmetrically, and given equal importance \camrevreplace{\cite{ayhan-intro}}{\cite{ayhan-intro-i,ayhan-intro-ii}}.   It is natural, therefore, that we should take a bilateral logic as the basis for the correspondence with two-sided typing.  \revreplace{Our system}{The single-subject fragment of our system} \tintty{} is in correspondence with Wansing's \tint{}, which is a bilateral natural deduction system that was introduced in \cite{wansing2016falsification}.  Wansing illustrated how bilateral systems such as \tint{} can be naturally extended with the strong negation operation in \cite{wansing-2017}, although he did not define a specific, named system.  Oddsson has shown that strong negation is definable in \tint{} using the coimplication \cite{oddsson2025strongnegationdefinable2int}.  We have not gone down that route, and instead defined \tintcty{} \revreplace{explcitly}{explicitly}, because our goal is to understand the correspondence with the two-sided rules of Figure~\ref{fig:prev-typing-rules}.

The coimplication is a key feature of Wansing's logic.  A related line of work studies (non-bilateral) accounts of intuitionistic logic extended with coimplication, starting with \revreplace{Rauzer's}{Rauszer's} Heyting-Brouwer logic \cite{rauszer-hblogic}, and more recently the works of Crolard \cite{crolard01}, \citet*{GoreS20}, \citet*{PintoU18}, \citet*{Tranchini17} and \citet*{cantor2024dual}.  However, excepting the latter, the lack of a symmetrical treatment of refutation causes these logics to \camrevreplace{loose}{lose} some of the good properties of intuitionistic logic, such as the disjunction property.  The work of \cite{cantor2024dual} retains the disjunction property, although the authors do not define a proof system.

Ayhan constructed a $\lambda$-calculus $\lambda^{\mathsf{2Int}}$ in correspondence with \tint{} in \cite{ayhan25}. This calculus is obtained by annotating \tint{} with 2-sorted $\lambda$-terms, so it is not given in a sequent format and there are two distinct judgements corresponding to proof and refutation. Like our two-sided type system (and \tint{}), there is a notion that terms are assigned types under a pair of assumptions and co-assumptions, but unlike \tintty, proof terms and refutation terms are distinct: every term construction, except for truth and falsity, has separate positive and negative forms. \revdel{Consequently, in $\lambda^{\mathsf{2Int}}$ and unlike the systems of this paper, terms have unique principal types.} \revreplace{Otherwise, the relationship between Ayhan's system and our system \tintty{} seems reasonably clear.}{Under corresponding assumptions and co-assumptions, both calculi prove and refute the same formulas of \tint{}; the distinction is in how their evidence is represented.} \revadd{In our system, the same unsigned expression can serve as proof or refutation evidence, with its role determined by the typing judgement. Ayhan's signed syntax instead distinguishes these roles, reflecting her interest in the sense of derivations while allowing their denotations to be identified through conversions and duality. Her subsequent work develops this account for bilateral sequent derivations \cite{ayhan-sequent-sense-denotation}.} \citet*{AbeK23} constructed a $\lambda$-calculus corresponding to Rumfitt's bilateral natural deduction \cite{rumfitt-yes-no}, but it is a symmetric calculus in the style of \citet*{Filinski89}.

\revadd{Ayhan's \emph{Problems and Consequences of Bilateral Notions of (Meta-)Derivability} \cite{ayhan-meta-derivability} distinguishes derivability within sequents from meta-derivability between sequents. At the latter level, one asks what it means to constructively refute an entire sequent. Our correspondence concerns evidence for or against the principal formula of single-subject judgements under assumptions and co-assumptions. It does not refute the whole judgement or establish that a positive judgement is underivable. This is close to the perspective in her Section~4.2. We introduce no additional meta-refutation relation.

  \paragraph{Term Calculi, Duality and Focusing}
  Computational objects on both sides of a sequent also occur in Curien and Herbelin's $\bar{\lambda}\mu\tilde{\mu}$-calculus \cite{curien-herbelin-icfp00}. Its terms, contexts and commands make program--context interaction explicit and expose a duality between call-by-name and call-by-value reduction. Zeilberger relates focused intuitionistic proofs to evaluation order and pattern matching \cite{zeilberger-popl08}, while Munch-Maccagnoni develops a polarised variant of the Curien--Herbelin calculus for focalising cut elimination and classical realisability \cite{munch-maccagnoni-csl09}. These accounts of interaction and polarity provide relevant comparisons for our proof and refutation judgements; an operational translation to the present calculi remains to be investigated.}

\paragraph{Strong Negation}
We identified the concept of strong negation as the logical counterpart to the complementation rules of Figure~\ref{fig:prev-typing-rules}.  Strong negation was introduced by Nelson in \cite{nelson1949constructible} as a replacement for intuitionistic negation.  In particular, the work was motivated essentially by the property that is nowadays often called \emph{constructible falsity}.  This is the first formalization of Nelson-style constructible falsity, which is later known as N3. Nelson also presented a Hilbert-style calculus for this logic, but its metatheory remained tied to the arithmetical framework in which he was working. Later work separated Nelson-style strong negation from this arithmetical background. One of the earliest natural-deduction presentations of Nelson-style strong negation appears in Prawitz’s 1965 work \cite{Prawitz1965-PRANDA}. Gurevich later treated the system as a pure logic: he developed a Hilbert calculus, established a Kripke semantics and completeness theorem, and presented a corresponding Gentzen-style calculus \cite{gurevich1977intuitionistic}. Almukdad and Nelson introduced the paraconsistent Nelson logic N4, which weakens N3 by allowing inconsistency without triviality \cite{almukdad1984constructible}.

In his monograph \cite{wansing93}, Wansing constructed a $\lambda$-calculus corresponding to a version of N4, called $\lambda^c$.  As \camrevadd{in }our systems, the strong negation has no associated introduction or elimination form at the term level.  Consequently, proofs and refutations are identified so that, for example, every term of type $\sneg{(A \wedge B)}$ is also of type $\sneg{A} \vee \sneg{B}$.  Wansing's $\lambda^c$ is presented in Church style, with terms annotated by their types, and his correspondence is with a Gentzen-style system.  The system is not bilateral and only treats the propositional case.

De Muijnck-Hughes, partly inspired by Atkey's incorporation of a version of strong negation in inductive datatypes \cite{data-types-negation-msfp}, has implemented a notion of positive evidence for the falsification of predicates in Idris \cite{deMuijnckHughes2025PositiveNegative,demuinckhughes-github}.  In a talk at TRIPLES in Edinburgh, he set out very clearly the benefits of the approach compared to proving an intuitionistic negation, and this has been the inspiration for our examples of Section~\ref{sec:hol2-proofs}.  Atkey's work also contains a semantics for datatypes with negation based on Chu-spaces, which is closely related to our mapping $\derefT$, since it models types as pairs of their positive and negative evidence, and interprets negation as transposition.

\paragraph{Higher-Order Logic}
Our system \holtty{} is a two-sided version of Geuvers' \holty{} \cite{geuvers-types06,geuvers-phd}, which itself corresponds to the higher-order logic in Church's Simple Theory of Types \cite{church-jsl40}.  Our system is both bilateral and incorporates a strong negation, but we are not aware of any existing work on either higher-order logics or type systems that combine both of these features.

\begin{acks}
  The second author is grateful for the support of the UK Engineering and Physical Sciences Research Council (UKRI2770).  We thank the POPL reviewers for numerous corrections and suggestions that greatly improved the quality of the paper.  The paper also benefitted from useful discussions with colleagues in the Programming Languages Research Group at the University of Bristol.
\end{acks}

\subsection*{Disclosure of the Use of AI in Research}
All research reported in the paper was conducted solely by the authors.

\subsection*{Data Availability Statement}
All proofs are available in the long version of the paper \cite{li2026logicalfoundationstwosidedtype}.

\bibliographystyle{ACM-Reference-Format}
\bibliography{ref}

\iftoggle{supplementary}{%
  \clearpage
  \appendix
  \section{\texorpdfstring{The Logic \tint{} of \cite{wansing2016falsification}}{The Logic 2Int of Wansing (2016)}}\label{sec:apx-2int}

\begin{figure}
  \begingroup
    \renewcommand{\arraystretch}{1.25}%
    \setlength{\tabcolsep}{0pt}%
    \vspace*{-.3em}
    \begin{tabularx}{\linewidth}{
            @{}>{\raggedleft\arraybackslash}p{2em}@{\hspace{0.4em}}
            >{\raggedright\arraybackslash\hsize=.98\hsize\linewidth=\hsize}X@{}
            >{\raggedright\arraybackslash\hsize=1.02\hsize\linewidth=\hsize}X@{}
        }


        \ruletableheading{Axiomatic clauses}

        \axrule{1}
        {\wanp{\overline{A},\,A,\,(\{A\};\,\varnothing)}}

        \axrule{2}
        {\wandp{\dbar{A},\,A,\,(\varnothing;\,\{A\})}}

        \axrule{3}
        {\wanp{\overline{\truefm},\,\truefm,\,(\varnothing;\,\varnothing)}}

        \axrule{4}
        {\wandp{\dbar{\falsefm},\,\falsefm,\,(\varnothing;\,\varnothing)}}
    \end{tabularx}
    \begin{tabularx}{\linewidth}{
            @{}>{\raggedleft\arraybackslash}p{2em}@{\hspace{0.4em}}
            >{\raggedright\arraybackslash}p{.25\linewidth}
            >{\raggedright\arraybackslash}p{\dimexpr\linewidth-2em-0.4em-.25\linewidth\relax}@{}
        }
        \multicolumn{3}{@{}l@{}}{\textbf{Structural rules}} \\

        \tblstrut$(29)$
         &
        $\wanp{\Pi,\,A,\,(\LeftCtx';\,\RightCtx')}$
         &
        $\text{if }\wanp{\Pi,\,A,\,(\LeftCtx;\,\RightCtx)}
            \land (\LeftCtx \subseteq \LeftCtx')
            \land (\RightCtx \subseteq \RightCtx')
            \land \LeftCtx',\RightCtx' \text{ are finite}$
        \\

        \tblstrut$(30)$
         &
        $\wandp{\Pi,\,A,\,(\LeftCtx';\,\RightCtx')}$
         &
        $\text{if }\wandp{\Pi,\,A,\,(\LeftCtx;\,\RightCtx)}
            \land (\LeftCtx \subseteq \LeftCtx')
            \land (\RightCtx \subseteq \RightCtx')
            \land \LeftCtx',\RightCtx' \text{ are finite}$
    \end{tabularx}
\endgroup
  \caption{\tint{}: axioms and structural rules}\label{fig:2int-ax-str}
  \Description{Clauses 1 through 4 give the axioms for assumptions, co-assumptions, truth, and falsity. Rules 29 and 30 allow either side of the bilateral basis to grow while keeping both sides finite.}
\end{figure}

\begin{figure}
  \begingroup
  \renewcommand{\arraystretch}{1.25}%
  \setlength{\tabcolsep}{0pt}%
  \noindent
  \begin{tabularx}{\linewidth}{
      @{}>{\raggedleft\arraybackslash}p{2em}@{\hspace{0.4em}}
      >{\raggedright\arraybackslash\hsize=.98\hsize\linewidth=\hsize}X@{}
      >{\raggedright\arraybackslash\hsize=1.02\hsize\linewidth=\hsize}X@{}
    }

    \ruletableheading{Introductions into proofs}

    \shortrule{5}
    {\wanp{\mprfraction{\Pi_1 \enspace \Pi_2}{A_1 \wedge A_2},\,
        A_1 \wedge A_2,\,
        (\LeftCtx_1 \cup \LeftCtx_2;\,\RightCtx_1 \cup \RightCtx_2)}}
    {\text{if }\wanp{\Pi_1,\,A_1,\,(\LeftCtx_1;\,\RightCtx_1)}
      \land
      \wanp{\Pi_2,\,A_2,\,(\LeftCtx_2;\,\RightCtx_2)}}

    \shortrule{6}
    {\wanp{\mprfraction{\Pi}{A_1 \vee A_2},\,
        A_1 \vee A_2,\,
        (\LeftCtx;\,\RightCtx)}}
    {\text{if }\wanp{\Pi,\,A_1,\,(\LeftCtx;\,\RightCtx)}}

    \shortrule{7}
    {\wanp{\mprfraction{\Pi}{A_1 \vee A_2},\,
        A_1 \vee A_2,\,
        (\LeftCtx;\,\RightCtx)}}
    {\text{if }\wanp{\Pi,\,A_2,\,(\LeftCtx;\,\RightCtx)}}

    \shortrule{8}
    {\wanp{\mprfraction{\Pi}{A_1 \to A_2},\,
        A_1 \to A_2,\,
        (\LeftCtx \setminus \{A_1\};\,\RightCtx)}}
    {\text{if }\wanp{\Pi,\,A_2,\,(\LeftCtx;\,\RightCtx)}}

    \shortrule{9}
    {\wanp{\mprfraction{\Pi_1\enspace\Pi_2}{A_1 \cotol A_2},\,
        A_1 \cotol A_2,\,
        (\LeftCtx_1 \cup \LeftCtx_2;\,\RightCtx_1 \cup \RightCtx_2)}}
    {\text{if }\wandp{\Pi_1,\,A_1,\,(\LeftCtx_1;\,\RightCtx_1)}
      \land
      \wanp{\Pi_2\,A_2,\,(\LeftCtx_2;\,\RightCtx_2)}}

    \ruletableheading{Eliminations from proofs}

    \shortrule{10}
    {\wanp{\mprfraction{\Pi}{A},\,A,\,(\LeftCtx;\,\RightCtx)}}
    {\text{if }\wanp{\Pi,\,\revreplace{\bot}{\falsefm},\,(\LeftCtx;\,\RightCtx)}}

    \shortrule{11}
    {\wanp{\mprfraction{\Pi}{A_1},\,A_1,\,(\LeftCtx;\,\RightCtx)}}
    {\text{if }\wanp{\Pi,\,A_1 \wedge A_2,\,(\LeftCtx;\,\RightCtx)}}

    \shortrule{12}
    {\wanp{\mprfraction{\Pi}{A_2},\,A_2,\,(\LeftCtx;\,\RightCtx)}}
    {\text{if }\wanp{\Pi,\,A_1 \wedge A_2,\,(\LeftCtx;\,\RightCtx)}}

    \longrule{13}
    {\wanp{\mprfraction{\Pi\;\;\Pi_1\;\;\Pi_2}{C},\,C,\,
        (\LeftCtx \cup (\LeftCtx_1 \setminus \{A_1\}) \cup (\LeftCtx_2 \setminus \{A_2\});\,
        \RightCtx \cup \RightCtx_1 \cup \RightCtx_2)}}
    {\text{if }\wanp{\Pi,\,A_1 \vee A_2,\,(\LeftCtx;\,\RightCtx)}
      \land \wanp{\Pi_1,\,C,\,(\LeftCtx_1;\,\RightCtx_1)}
      \land \wanp{\Pi_2,\,C,\,(\LeftCtx_2;\,\RightCtx_2)}}

    \longrule{14}
    {\wanp{\mprfraction{\Pi_1\enspace\Pi_2}{A_2},\,A_2,\,
        (\LeftCtx_1 \cup \LeftCtx_2;\,\RightCtx_1 \cup \RightCtx_2)}}
    {\text{if }\wanp{\Pi_1,\,A_1,\,(\LeftCtx_1;\,\RightCtx_1)}
      \land \wanp{\Pi_2,\,A_1 \to A_2,\,(\LeftCtx_2;\,\RightCtx_2)}}

    \shortrule{15}
    {\wanp{\mprfraction{\Pi}{A_2},\,A_2,\,(\LeftCtx;\,\RightCtx)}}
    {\text{if }\wanp{\Pi,\,A_1 \cotol A_2,\,(\LeftCtx;\,\RightCtx)}}

    \shortrule{16}
    {\wandp{\dblmprfraction{\Pi}{A_1},\,A_1,\,(\LeftCtx;\,\RightCtx)}}
    {\text{if }\wanp{\Pi,\,A_1 \cotol A_2,\,(\LeftCtx;\,\RightCtx)}}
  \end{tabularx}
\endgroup
  \caption{\tint{}: introduction and elimination rules for proof}\label{fig:2int-proof}
  \Description{Clauses 5 through 9 introduce the logical connectives into proofs. Clauses 10 through 16 give the elimination rules. Each row gives the resulting proof and the conditions on its premises.}
\end{figure}

\begin{figure}
  \begingroup
  \renewcommand{\arraystretch}{1.25}%
  \setlength{\tabcolsep}{0pt}%
  \vspace*{-1em}
  \begin{tabularx}{\linewidth}{
      @{}>{\raggedleft\arraybackslash}p{2em}@{\hspace{0.4em}}
      >{\raggedright\arraybackslash\hsize=.98\hsize\linewidth=\hsize}X@{}
      >{\raggedright\arraybackslash\hsize=1.02\hsize\linewidth=\hsize}X@{}
    }
    \ruletableheading{Introductions into dual proofs}

    \shortrule{17}
    {\wandp{\dblmprfraction{\Pi}{A_1 \wedge A_2},\,A_1 \wedge A_2,\,(\LeftCtx;\,\RightCtx)}}
    {\text{if }\wandp{\Pi,\,A_1,\,(\LeftCtx;\,\RightCtx)}}

    \shortrule{18}
    {\wandp{\dblmprfraction{\Pi}{A_1 \wedge A_2},\,A_1 \wedge A_2,\,(\LeftCtx;\,\RightCtx)}}
    {\text{if }\wandp{\Pi,\,A_2,\,(\LeftCtx;\,\RightCtx)}}

    \shortrule{19}
    {\wandp{\dblmprfraction{\Pi_1 \enspace \Pi_2}{A_1 \vee A_2},\,A_1 \vee A_2,\,
        (\LeftCtx_1 \cup \LeftCtx_2;\,\RightCtx_1 \cup \RightCtx_2)}}
    {\text{if }\wandp{\Pi_1,\,A_1,\,(\LeftCtx_1;\,\RightCtx_1)}
      \land \wandp{\Pi_2,\,A_2,\,(\LeftCtx_2;\,\RightCtx_2)}}

    \shortrule{20}
    {\wandp{\dblmprfraction{\Pi_1 \enspace \Pi_2}{A_1 \to A_2},\,A_1 \to A_2,\,
        (\LeftCtx_1 \cup \LeftCtx_2;\,\RightCtx_1 \cup \RightCtx_2)}}
    {\text{if }\wanp{\Pi_1,\,A_1,\,(\LeftCtx_1;\,\RightCtx_1)}
      \land \wandp{\Pi_2,\,A_2,\,(\LeftCtx_2;\,\RightCtx_2)}}

    \shortrule{21}
    {\wandp{\dblmprfraction{\Pi}{A_1 \cotol A_2},\,A_1 \cotol A_2,\,
        (\LeftCtx;\,\RightCtx \setminus \{A_1\})}}
    {\text{if }\wandp{\Pi,\,A_2,\,(\LeftCtx;\,\RightCtx)}}

    \ruletableheading{Eliminations from dual proofs}

    \shortrule{22}
    {\wandp{\dblmprfraction{\Pi}{A},\,A,\,(\LeftCtx;\,\RightCtx)}}
    {\text{if }\wandp{\Pi,\,\top,\,(\LeftCtx;\,\RightCtx)}}

    \longrule{23}
    {\wandp{\dblmprfraction{\Pi\;\;\Pi_1\;\;\Pi_2}{C},\,C,\,
        (\LeftCtx \cup \LeftCtx_1 \cup \LeftCtx_2;\,
        \RightCtx \cup (\RightCtx_1 \setminus \{A_1\}) \cup (\RightCtx_2 \setminus \{A_2\}))}}
    {\text{if }\wandp{\Pi,\,A_1 \wedge A_2,\,(\LeftCtx;\,\RightCtx)}
      \land \wandp{\Pi_1,\,C,\,(\LeftCtx_1;\,\RightCtx_1)}
      \land \wandp{\Pi_2,\,C,\,(\LeftCtx_2;\,\RightCtx_2)}}


    \shortrule{24}
    {\wandp{\dblmprfraction{\Pi}{A_1},\,A_1,\,(\LeftCtx;\,\RightCtx)}}
    {\text{if }\wandp{\Pi,\,A_1 \vee A_2,\,(\LeftCtx;\,\RightCtx)}}

    \shortrule{25}
    {\wandp{\dblmprfraction{\Pi}{A_2},\,A_2,\,(\LeftCtx;\,\RightCtx)}}
    {\text{if }\wandp{\Pi,\,A_1 \vee A_2,\,(\LeftCtx;\,\RightCtx)}}

    \shortrule{26}
    {\wanp{\mprfraction{\Pi}{A_1},\,A_1,\,(\LeftCtx;\,\RightCtx)}}
    {\text{if }\wandp{\Pi,\,A_1 \to A_2,\,(\LeftCtx;\,\RightCtx)}}

    \shortrule{27}
    {\wandp{\dblmprfraction{\Pi}{A_2},\,A_2,\,(\LeftCtx;\,\RightCtx)}}
    {\text{if }\wandp{\Pi,\,A_1 \to A_2,\,(\LeftCtx;\,\RightCtx)}}

    \longrule{28}
    {\wandp{\dblmprfraction{\Pi_1\enspace\Pi_2}{A_2},\,A_2,\,
        (\LeftCtx_1 \cup \LeftCtx_2;\,\RightCtx_1 \cup \RightCtx_2)}}
    {\text{if }\wandp{\Pi_1,\,A_1 \cotol A_2,\,(\LeftCtx_1;\,\RightCtx_1)}
      \land \wandp{\Pi_2,\,A_1,\,(\LeftCtx_2;\,\RightCtx_2)}}
  \end{tabularx}
\endgroup

  \caption{\tint{}: introduction and elimination rules for dual proof}\label{fig:2int-dualproof}
  \Description{Clauses 17 through 21 introduce the logical connectives into dual proofs. Clauses 22 through 28 give the elimination rules. Each row gives the resulting dual proof and the conditions on its premises.}
\end{figure}

\begin{figure}
  \begingroup
    \renewcommand{\arraystretch}{1.25}%
    \setlength{\tabcolsep}{0pt}%
    \noindent
    \begin{tabularx}{\linewidth}{
            @{}>{\raggedleft\arraybackslash}p{2em}@{\hspace{0.4em}}
            >{\raggedright\arraybackslash\hsize=.98\hsize\linewidth=\hsize}X@{}
            >{\raggedright\arraybackslash\hsize=1.02\hsize\linewidth=\hsize}X@{}
        }
        \ruletableheading{Strong negation rules}
        \shortrule{31}
        {\wandp{\dblmprfraction{\Pi}{\sneg{A}},\,\sneg{A},\,(\LeftCtx;\,\RightCtx)}}
        {\text{if }\wanp{\Pi,\,A,\,(\LeftCtx;\,\RightCtx)}}

        \shortrule{32}
        {\wandp{\dblmprfraction{\Pi}{A},\,A,\,(\LeftCtx;\,\RightCtx)}}
        {\text{if }\wanp{\Pi,\,\sneg{A},\,(\LeftCtx;\,\RightCtx)}}

        \shortrule{33}
        {\wanp{\mprfraction{\Pi}{\sneg{A}},\,\revreplace{\sneg{A}}{A}},\,A,\,(\LeftCtx;\,\RightCtx)}
        {\text{if }\wandp{\Pi,\,A,\,(\LeftCtx;\,\RightCtx)}}

        \shortrule{34}
        {\wanp{\mprfraction{\Pi}{A},\,A,\,(\LeftCtx;\,\RightCtx)}}
        {\text{if }\wandp{\Pi,\,\sneg{A},\,(\LeftCtx;\,\RightCtx)}}
    \end{tabularx}
\endgroup

  \caption{\tintc{}: strong negation rules}\label{fig:2int-sneg}
  \Description{Clauses 31 and 32 turn proofs into dual proofs by introducing or eliminating strong negation. Clauses 33 and 34 turn dual proofs into proofs in the same way. The bilateral basis stays fixed.}
\end{figure}

For convenience, we repeat the definition of Wansing's \tint{} from \cite{wansing2016falsification}.  The formulas are defined as follows.

\begin{definition}
  Let \(\Phi=\{p,q,r,\dots\}\) be a denumerable set of atomic formulas.
  Formulas of \tint{} are given by:
  \[
    \begin{array}{rrcl}
      \rlnm{Formulas} & A,\,B & \Coloneqq
                      & p
      \mid \revreplace{\bot}{\falsefm}
      \mid \truefm
      \mid A \to B
      \mid A \land B
      \mid A \lor B
      \mid A \revreplace{\coto}{\cotol} B .
    \end{array}
  \]
\end{definition}

Figures \ref{fig:2int-ax-str}, \ref{fig:2int-proof}, and
\ref{fig:2int-dualproof} present the rules of \tint{}. The first four
clauses are axiomatic. Rules (1) and (2) express the two primitive ways
of appealing to the basis: from an assumption \(A\in L\) one may prove
\(A\), and from a counterassumption \(A\in R\) one may refute \(A\).
Rules (3) and (4) express the corresponding nullary principles:
\(\truefm\) is provable without assumptions, and \camrevreplace{\(\bot\)}{\(\falsefm\)} is refutable
without counterassumptions.

The proof rules in Figure \ref{fig:2int-proof} contain the ordinary
natural-deduction rules for the intuitionistic connectives. Rules
(5)--(8) are the familiar introduction rules for conjunction,
disjunction, and implication, while rules (10)--(14) are the
corresponding elimination rules. Thus, if we restrict attention to:
\[
  \truefm,\quad \revreplace{\bot}{\falsefm},\quad \land,\quad \lor,\quad \to
\]
on the proof side, \tint{} behaves like ordinary intuitionistic natural
deduction, except that derivations are recorded relative to a bilateral
basis \((L;R)\).

The remaining proof rules concern coimplication. Rule (9) introduces
\(A\cotol B\) into proofs from two pieces of evidence: a refutation of \(A\)
and a proof of \(B\). Rules (15) and (16) eliminate a proved
coimplication: from a proof of \(A\cotol B\), one may recover a proof of
\(B\) and a refutation of \(A\). On the proof side, then, \(A\cotol B\)
functions as a constructive form of ``\(B\) but not \(A\)'', where the
``not'' is not classical negation but primitive refutation.

The dual-proof rules in Figure \ref{fig:2int-dualproof} are read in the
opposite direction, from the standpoint of falsification. The pure
dual-proof rules are induced by the same duality that exchanges
\(\truefm\) with \camrevreplace{\(\bot\)}{\(\falsefm\)}, \(\land\) with \(\lor\), and \(\to\) with
\(\cotol\). Thus, to refute a conjunction, one refutes one of its
conjuncts; to refute a disjunction, one refutes both disjuncts; and to
refute a coimplication \(A\cotol B\), one refutes \(B\) under the
temporary counterassumption \(A\). In the last case, the ordinary
discharge of an assumption in implication introduction is mirrored by the
discharge of a counterassumption in coimplication refutation.

The refutation rules for implication are different in kind, because
implication is where falsification has to appeal to both sides of the
calculus. Rule (20) says that a refutation of \(A\to B\) consists of a
proof of \(A\) together with a refutation of \(B\). Rules (26) and (27)
then express the corresponding eliminations: from a refutation of
\(A\to B\), one may extract a proof of \(A\) and a refutation of \(B\).
These rules are mixed, because their premises and conclusions are not all
of the same kind.

The same mixed pattern appears on the proof side for coimplication. A
proof of \(A\cotol B\) consists of a refutation of \(A\) together with a
proof of \(B\), and its eliminations return precisely these two
pieces of evidence. In this sense, the falsity condition for implication
and the truth condition for coimplication coincide:
\[
  \text{refuting } A\to B
  \quad\text{and}\quad
  \text{proving } B\cotol A
\]
both amount to proving \(A\) and refuting \(B\).


Finally, rules (29) and (30) are structural monotonicity rules for the
basis. They say that a proof or a dual proof remains available when the
basis is enlarged by additional assumptions or counterassumptions. Thus
derivability in \tint{} is stable under extensions of the information
available on either side. When rules combine derivations, their
assumption and counterassumption sets are united. When a rule discharges
an assumption or counterassumption, the relevant formula is removed from
the appropriate side of the basis.

\camrevadd{The extension \tintc{} adds the formula former $\sneg{A}$ and the four strong-negation rules in Figure~\ref{fig:2int-sneg} to \tint{}. These rules exchange proof and dual proof while introducing or eliminating strong negation.}

  \section{Supplementary Material for Section~\ref{sec:two-lambda-int}}

In this appendix, we give the details for the proof of the correspondence between \tintc{} and \tintcty{} (also \tint{} and \tintty{}).

\begin{definition}[Erasure]\label{def:proof-term-erasure-appendix}
    We write $\erasure{M:A} = A$ for the erasure of a proof term, (including when $M$ is a variable). We extend this notation to environments by erasing proof terms pointwise and identifying repeated formulas, so that $\erasure{\Gamma}$ and $\erasure{\Delta}$ are finite sets of formulas.
\end{definition}

Given a derivation of a single-subject judgement in \tintcty{}, we can erase the proof and refutation terms to obtain a proof or dual proof in \camrevreplace{\tint{}}{\tintc{}}.

In the opposite direction, as is typical, it is convenient to choose some canonical naming of assumptions (and co-assumptions).

\begin{definition}\label{def:logic-single-subject-appendix}
    Fix once and for all, for each type $A$, two variables $x_A$ and $y_A$, in such a way that the families $\{x_A \mid A \text{ is a type}\}$ and $\{y_A \mid A \text{ is a type}\}$ are jointly pairwise distinct.
    Now let $L = \{A_1,\ldots,A_m\}$ and $R = \{C_1,\ldots,C_n\}$ be finite sets of formulas. We define:
    \begin{itemize}
        \item $L \types' B;\,R$ just if there is a term $M$ and $x_{A_1} \ty: A_1,\ldots,x_{A_m}\ty:A_m \types M : B,\,y_{C_1} \ty: C_1,\ldots,y_{C_n}\ty: C_n$
        \item $L;\,B \types' R$ just if there is a term $M$ and $x_{A_1} \ty: A_1,\ldots,x_{A_m}\ty:A_m,\, M : B \types y_{C_1} \ty: C_1,\ldots,y_{C_n}\ty: C_n$
    \end{itemize}
    If $L=\varnothing$ or $R=\varnothing$, the corresponding family is empty.
\end{definition}

Then, from a proof or a dual proof in \tintc{}, we can obtain a proof term or refutation\camrevadd{ term,} respectively, with assumptions named according to the above scheme.


\begin{definition}
    Two types $A$ and $B$ are said to be \emph{related by} $n$ if
    either $A = \sneg^n B$ or $B = \sneg^n A$, where $\sneg^n$ abbreviates $n$ consecutive occurrences of $\sneg$.
\end{definition}

\begin{lemma}[Judgement form]
    \label{lem:judgement-form}
    Let \(\Gamma \types \Delta\) be derivable. Then the following assertions
    hold.

    \begin{enumerate}
        \item Let \(M:B\in\Gamma\) be a chosen typing formula on the left, so
              that, up to exchange, the judgement decomposes as
              \[
                  \smashunderbrace{\Gamma',\,M:B}{\Gamma}
                  \types
                  \Delta .
              \]
              Suppose that \(\Gamma'\) and \(\Delta\) are disjoint variable
              environments. Equivalently, after removing the chosen typing
              formula \(M:B\), all remaining typing formulas in the judgement
              are variable typings for pairwise distinct variables. Then either
              \(M\) is not a variable, or \(M\) is a variable and one of the
              following alternatives holds:
              \begin{itemize}
                  \item \(M:C\in\Gamma'\), and \(B\) and \(C\) are related by
                        some odd \(n\);
                  \item \(M:C\in\Delta\), and \(B\) and \(C\) are related by
                        some even \(n\).
              \end{itemize}

        \item Let \(M:B\in\Delta\) be a chosen typing formula on the right, so
              that, up to exchange, the judgement decomposes as
              \[
                  \Gamma
                  \types
                  \smashunderbrace{M:B,\,\Delta'}{\Delta} .
              \]
              Suppose that \(\Gamma\) and \(\Delta'\) are disjoint variable
              environments. Equivalently, after removing the chosen typing
              formula \(M:B\), all remaining typing formulas in the judgement
              are variable typings for pairwise distinct variables. Then either
              \(M\) is not a variable, or \(M\) is a variable and one of the
              following alternatives holds:
              \begin{itemize}
                  \item \(M:C\in\Gamma\), and \(B\) and \(C\) are related by
                        some even \(n\);
                  \item \(M:C\in\Delta'\), and \(B\) and \(C\) are related by
                        some odd \(n\).
              \end{itemize}
    \end{enumerate}
\end{lemma}

\begin{proof}
    We prove the two assertions simultaneously by induction on the derivation
    of \(\Gamma\types\Delta\). The only cases requiring genuine analysis are
    \rlnm{Id} and the strong-negation rules. For every other typing rule,
    the principal typing formula in the conclusion has a non-variable term
    former. If the chosen typing formula is this principal typing formula, then the
    non-variable alternative of the corresponding clause holds immediately. If
    the chosen typing formula is not principal, then the principal non-variable
    typing formula remains in one of the side contexts which is required to be a
    variable environment, a contradiction. Hence those cases are vacuous except
    for the immediate non-variable branch.

    It remains to consider \rlnm{Id} and the strong-negation rules. We give
    the details for \rlnm{Id} and \(\sneg_I L\); the cases for
    \(\sneg_E L\), \(\sneg_I R\), and \(\sneg_E R\) are analogous.

    \begin{indproof}
        \indcase{Id}
        The derivation ends with an instance of the identity rule:
        \[
            \infer[Id]{ }{
                \Sigma,\,y:A \types y:A,\,\Pi
            }.
        \]
        Thus, up to exchange,
        \[
            \Gamma=\Sigma,\,y:A
            \qquad\text{and}\qquad
            \Delta=y:A,\,\Pi .
        \]
        We distinguish whether the chosen typing formula is on the left or on
        the right.

        \begin{enumerate}
            \item Suppose first that \(M:B\in\Gamma\) is chosen from the
                  left-hand side, so that, up to exchange,
                  \[
                      \smashunderbrace{\Sigma,\,y:A}{\Gamma',\,M:B}
                      \types
                      y:A,\,\Pi .
                  \]
                  Assume that \(\Gamma'\) and \(\Delta=y:A,\Pi\) are disjoint
                  variable environments. The chosen typing formula must be the
                  left-hand formula \(y:A\) introduced by \rlnm{Id}. Indeed,
                  otherwise \(y:A\) would remain in \(\Gamma'\), while
                  \(y:A\) also appears in \(\Delta\), contradicting
                  disjointness. Hence
                  \[
                      M=y,
                      \qquad
                      B=A,
                      \qquad
                      \Gamma'=\Sigma .
                  \]
                  Since \(y:A\in\Delta\), taking \(C=A\) gives
                  \(M:C\in\Delta\). Moreover \(B=C\), so \(B\) and \(C\) are
                  related by \(0\) negations, and \(0\) is even. Thus the
                  \(\Delta\)-alternative of clause \((1)\) holds.

            \item Suppose now that \(M:B\in\Delta\) is chosen from the
                  right-hand side, so that, up to exchange,
                  \[
                      \Sigma,\,y:A
                      \types
                      \smashunderbrace{y:A,\,\Pi}{M:B,\,\Delta'} .
                  \]
                  Assume that \(\Gamma=\Sigma,y:A\) and \(\Delta'\) are
                  disjoint variable environments. The chosen typing formula
                  must be the right-hand formula \(y:A\) introduced by
                  \rlnm{Id}. Indeed, otherwise \(y:A\) would remain in
                  \(\Delta'\), while \(y:A\) also appears in \(\Gamma\),
                  contradicting disjointness. Hence
                  \[
                      M=y,
                      \qquad
                      B=A,
                      \qquad
                      \Delta'=\Pi .
                  \]
                  Since \(y:A\in\Gamma\), taking \(C=A\) gives
                  \(M:C\in\Gamma\). Moreover \(B=C\), so \(B\) and \(C\) are
                  related by \(0\) negations, and \(0\) is even. Thus the
                  \(\Gamma\)-alternative of clause \((2)\) holds.
        \end{enumerate}

        \indcase{\(\sneg_I L\)}
        The derivation ends as
        \[
            \infer[$\sneg_I L$]{
                \Sigma \types P:A,\,\Pi
            }{
                \Sigma,\,P:\sneg A
                \types
                \Pi
            }.
        \]
        Thus, up to exchange,
        \[
            \Gamma=\Sigma,\,P:\sneg A
            \qquad\text{and}\qquad
            \Delta=\Pi .
        \]
        We distinguish whether the chosen typing formula is on the left or on
        the right.

        \begin{enumerate}
            \item Suppose first that \(M:B\in\Gamma\) is chosen from the
                  left-hand side, so that, up to exchange,
                  \[
                      \smashunderbrace{\Sigma,\,P:\sneg A}{\Gamma',\,M:B}
                      \types
                      \Pi .
                  \]
                  Assume that \(\Gamma'\) and \(\Pi\) are disjoint variable
                  environments. There are two subcases.

                  \begin{itemize}
                      \item Suppose first that \(M:B\) is the principal typing
                            formula \(P:\sneg A\). Then
                            \[
                                M=P,
                                \qquad
                                B=\sneg A,
                                \qquad
                                \Gamma'=\Sigma .
                            \]
                            The premise is
                            \[
                                \Sigma \types P:A,\,\Pi .
                            \]
                            Since \(\Sigma\) and \(\Pi\) are disjoint variable
                            environments, we may apply the induction hypothesis,
                            clause \((2)\), to this premise, choosing \(P:A\)
                            on the right. Thus either \(P\) is not a variable,
                            or \(P\) is a variable and one of the following
                            alternatives holds:
                            \begin{itemize}
                                \item \(P:D\in\Sigma\), and \(A\) and \(D\)
                                      are related by some even \(n\);
                                \item \(P:D\in\Pi\), and \(A\) and \(D\) are
                                      related by some odd \(n\).
                            \end{itemize}
                            If \(P\) is not a variable, then the non-variable
                            alternative of clause \((1)\) holds. Otherwise,
                            suppose \(P\) is a variable. In the first
                            alternative, \(P:D\in\Sigma=\Gamma'\), and adding
                            the leading strong negation changes the parity from
                            even to odd; hence \(B=\sneg A\) and \(D\) are
                            related by some odd number of negations. This gives
                            the \(\Gamma'\)-alternative of clause \((1)\). In
                            the second alternative, \(P:D\in\Pi\), and adding
                            the leading strong negation changes the parity from
                            odd to even; hence \(B=\sneg A\) and \(D\) are
                            related by some even number of negations. This gives
                            the \(\Pi\)-alternative of clause \((1)\).

                      \item Suppose now that \(M:B\) is not the principal typing
                            formula. Then \(M:B\in\Sigma\), so, up to exchange,
                            \[
                                \Sigma=\Sigma',\,M:B
                                \qquad\text{and}\qquad
                                \Gamma'=\Sigma',\,P:\sneg A .
                            \]
                            Since \(\Gamma'\) and \(\Pi\) are disjoint variable
                            environments, \(P\) is a variable and
                            \(\Sigma'\) and \(P:A,\Pi\) are disjoint variable
                            environments. The premise may be written as
                            \[
                                \Sigma',\,M:B \types P:A,\,\Pi .
                            \]
                            Applying the induction hypothesis, clause \((1)\),
                            to this premise, choosing \(M:B\) on the left,
                            gives that either \(M\) is not a variable, or
                            \(M\) is a variable and one of the following
                            alternatives holds:
                            \begin{itemize}
                                \item \(M:D\in\Sigma'\), and \(B\) and \(D\)
                                      are related by some odd \(n\);
                                \item \(M:D\in P:A,\Pi\), and \(B\) and \(D\)
                                      are related by some even \(n\).
                            \end{itemize}
                            If \(M\) is not a variable, then the non-variable
                            alternative of clause \((1)\) holds. Otherwise,
                            suppose \(M\) is a variable. If \(M:D\in\Sigma'\),
                            then \(M:D\in\Gamma'\), so the same odd \(n\)
                            gives the \(\Gamma'\)-alternative of clause \((1)\).
                            If \(M:D\in\Pi\), then the same even \(n\) gives the
                            \(\Pi\)-alternative of clause \((1)\). Finally, if
                            \(M:D\) is the formula \(P:A\), then \(M=P\) and
                            \(D=A\). Since \(B\) and \(A\) are related by some
                            even \(n\), \(B\) and \(\sneg A\) are related by
                            some odd number of negations; and
                            \(P:\sneg A\in\Gamma'\). Taking \(C=\sneg A\) gives
                            the \(\Gamma'\)-alternative of clause \((1)\).
                  \end{itemize}

            \item Suppose now that \(M:B\in\Delta=\Pi\) is chosen from the
                  right-hand side, so that, up to exchange,
                  \[
                      \Sigma,\,P:\sneg A
                      \types
                      \smashunderbrace{M:B,\,\Delta'}{\Pi} .
                  \]
                  Assume that \(\Sigma,P:\sneg A\) and \(\Delta'\) are
                  disjoint variable environments. Then \(P\) is a variable,
                  and \(\Sigma\) and \(P:A,\Delta'\) are disjoint variable
                  environments. The premise may be written as
                  \[
                      \Sigma \types P:A,\,M:B,\,\Delta' .
                  \]
                  Applying the induction hypothesis, clause \((2)\), to this
                  premise, choosing \(M:B\) on the right, gives that either
                  \(M\) is not a variable, or \(M\) is a variable and one of
                  the following alternatives holds:
                  \begin{itemize}
                      \item \(M:D\in\Sigma\), and \(B\) and \(D\) are related
                            by some even \(n\);
                      \item \(M:D\in P:A,\Delta'\), and \(B\) and \(D\) are
                            related by some odd \(n\).
                  \end{itemize}
                  If \(M\) is not a variable, then the non-variable alternative
                  of clause \((2)\) holds. Otherwise, suppose \(M\) is a
                  variable. If \(M:D\in\Sigma\), then \(M:D\in\Gamma\), so the
                  same even \(n\) gives the \(\Gamma\)-alternative of clause
                  \((2)\). If \(M:D\in\Delta'\), then the same odd \(n\) gives
                  the \(\Delta'\)-alternative of clause \((2)\). Finally, if
                  \(M:D\) is the formula \(P:A\), then \(M=P\) and \(D=A\).
                  Since \(B\) and \(A\) are related by some odd \(n\), \(B\) and
                  \(\sneg A\) are related by some even number of negations; and
                  \(P:\sneg A\in\Gamma\). Taking \(C=\sneg A\) gives the
                  \(\Gamma\)-alternative of clause \((2)\).
        \end{enumerate}
    \end{indproof}
\end{proof}

\begin{lemma}[Left Weakening]
    \label{lem:weakening-left}
    Let $\Gamma \types \Delta$ be derivable. Let \(x\) be fresh, i.e.
    \(x\notin\fv(\Gamma,\Delta)\), and let \(A\) be any type.
    For every choice of a displayed formula \(M:B\) in this judgement,
    the following hold.

    \begin{enumerate}
        \item Suppose the chosen formula is on the left, so that the judgement decomposes as
              \[
                  \smashunderbrace{\Gamma',\,M:B}{\Gamma}
                  \types
                  \smashunderbrace{\Delta'}{\Delta}.
              \]
              If \(\Gamma'\) and \(\Delta'\) are disjoint variable environments,
              then
              \[
                  \Gamma',\,x:A,\,M:B \types \Delta'
              \]
              is derivable.

        \item Suppose the chosen formula is on the right, so that the judgement decomposes as
              \[
                  \smashunderbrace{\Gamma'}{\Gamma}
                  \types
                  \smashunderbrace{M:B,\,\Delta'}{\Delta}.
              \]
              If \(\Gamma'\) and \(\Delta'\) are disjoint variable environments,
              then
              \[
                  \Gamma',\,x:A \types M:B,\,\Delta'
              \]
              is derivable.
    \end{enumerate}
\end{lemma}

\begin{proof}
    \label{proof:left-weakening}
    We prove the two assertions simultaneously by induction on the derivation
    of \(\Gamma \types \Delta\).

    \begin{indproof}
        \indcase{Id}
        The derivation ends with an instance of the identity rule:
        \[
            \infer[Id]{ }{
                \smashunderbrace{\Sigma,\,y:C}{\Gamma}
                \types
                \smashunderbrace{y:C,\,\Pi}{\Delta}
            }
            \quad y\notin\fv(\Sigma,\Pi).
        \]
        Since \(x\notin\fv(\Gamma,\Delta)=\fv(\Sigma,y:C,y:C,\Pi)\),
        we have \(x\neq y\).  Hence
        \[
            y\notin\fv(\Sigma,x:A,\Pi).
        \]
        Therefore we may reapply \rlnm{Id} to obtain
        \[
            \infer[Id]{ }{
                \Sigma,\,x:A,\,y:C
                \types
                y:C,\,\Pi
            }
            \quad y\notin\fv(\Sigma,x:A,\Pi).
        \]
        Equivalently, up to exchange,
        \[
            \Gamma,\,x:A \types \Delta .
        \]
        This is exactly the required judgement in both clauses, up to exchange:
        if the chosen formula is on the left, then
        \(\Gamma=\Gamma',M:B\), so
        \[
            \Gamma',\,x:A,\,M:B \types \Delta'
        \]
        follows; if the chosen formula is on the right, then
        \(\Delta=M:B,\Delta'\), so
        \[
            \Gamma',\,x:A \types M:B,\,\Delta'
        \]
        follows.

        \indcase{Non-negation rules}
        We spell out one representative case, namely \(\to R\).  The derivation
        ends as
        \[
            \infer[$\to R$]{
                \Sigma,\,y:C \types P:D,\,\Pi
            }{
                \smashunderbrace{\Sigma}{\Gamma}
                \types
                \smashunderbrace{\abs{y}{P}:C\to D,\,\Pi}{\Delta}
            }
            \quad y\notin\fv(\Sigma,\Pi).
        \]
        By \(\alpha\)-conversion, assume also that \(y\neq x\).  The principal
        formula in the conclusion is
        \[
            \abs{y}{P}:C\to D,
        \]
        whose term is not a variable.

        We claim that the chosen formula \(M:B\) must be the principal
        formula.  \camrevreplace{if}{If} the chosen formula were on the left, then the
        non-variable principal formula \(\abs{y}{P}:C\to D\) would remain in
        the right side context \(\Delta'\), contradicting the assumption that
        \(\Delta'\) is a variable environment.  Similarly, if the chosen
        formula were on the right but passive, then the same principal
        formula would again remain in \(\Delta'\), giving the same
        contradiction.

        Hence
        \[
            M=\abs{y}{P},
            \qquad
            B=C\to D,
            \qquad
            \Gamma'=\Sigma,
            \qquad
            \Delta'=\Pi .
        \]
        Since \(\Gamma'\) and \(\Delta'\) are disjoint variable environments,
        \(\Sigma,y:C\) and \(\Pi\) are also disjoint variable environments,
        because \(y\notin\fv(\Sigma,\Pi)\).

        We now check the freshness condition to apply the IH.  Since
        \[
            x\notin\fv(\Gamma,\Delta)
            =
            \fv(\Sigma,\abs{y}{P}:C\to D,\Pi)
        \]
        and \(x\neq y\), it follows that
        \[
            x\notin\fv(\Sigma,y:C,P:D,\Pi).
        \]
        Thus applying the IH, clause (2) to the premise with chosen formula \(P:D\) gives
        \[
            \Sigma,y:C,x:A \types P:D,\Pi .
        \]
        Up to exchange,
        \[
            \Sigma,x:A,y:C \types P:D,\Pi .
        \]
        Since \(y\notin\fv(\Sigma,\Pi)\) and \(y\neq x\), we have $y\notin\fv(\Sigma,x:A,\Pi)$
        Thus we may reapply \(\to R\) and obtain
        \[
            \Sigma,x:A \types \abs{y}{P}:C\to D,\Pi .
        \]
        Since \(\Gamma'=\Sigma\), \(\Delta'=\Pi\), \(M=\abs{y}{P}\), and
        \(B=C\to D\), this is exactly
        \[
            \Gamma',x:A \types M:B,\Delta' .
        \]

        The remaining non-negation rules are analogous.  If the chosen formula
        is passive, then the non-variable principal formula remains in one of
        the side contexts, contradicting the assumption that the side contexts are
        variable environments.  Hence the chosen formula must be principal.
        For each immediate premise, we check the freshness condition (bound
        variables are first \(\alpha\)-renamed so that they are distinct from
        \(x\) ) and apply the appropriate IH to the corresponding premise formula.

        \indcase{Negation rules}
        We spell out one representative case, namely \(\sneg_I L\).  The
        derivation ends as
        \[
            \infer[$\sneg_I L$]{
                \Sigma \types P:C,\Pi
            }{
                \smashunderbrace{\Sigma,P:\sneg C}{\Gamma}
                \types
                \smashunderbrace{\Pi}{\Delta}
            }.
        \]
        We consider whether the chosen formula \(M:B\) is on the right or on
        the left.

        \begin{enumerate}
            \item Suppose the chosen formula is on the right.  Then, up to
                  exchange,
                  \[
                      \Pi=M:B,\Delta',
                      \qquad
                      \Gamma'=\Sigma,P:\sneg C .
                  \]
                  Thus the premise may be written as
                  \[
                      \Sigma \types M:B,P:C,\Delta' .
                  \]
                  Since \(\Gamma'=\Sigma,P:\sneg C\) and \(\Delta'\) are disjoint
                  variable environments, \(P\) is a variable, and
                  \(\Sigma\) and \(P:C,\Delta'\) are also disjoint variable
                  environments.

                  We now check freshness for the induction hypothesis.  Since
                  \[
                      x\notin\fv(\Gamma,\Delta)
                      =
                      \fv(\Sigma,P:\sneg C,M:B,\Delta'),
                  \]
                  if follows that we also have
                  \[
                      x\notin\fv(\Sigma,M:B,P:C,\Delta').
                  \]
                  Therefore the induction hypothesis of clause (2), applied to
                  the premise with chosen formula \(M:B\), gives
                  \[
                      \Sigma,x:A \types M:B,P:C,\Delta' .
                  \]
                  By exchange,
                  \[
                      \Sigma,x:A \types P:C,M:B,\Delta' .
                  \]
                  Reapplying \(\sneg_I L\) gives
                  \[
                      \Sigma,x:A,P:\sneg C \types M:B,\Delta' .
                  \]
                  Since \(\Gamma'=\Sigma,P:\sneg C\), this is exactly
                  \[
                      \Gamma',x:A \types M:B,\Delta' .
                  \]

            \item Suppose the chosen formula is on the left.  There are two
                  subcases.

                  \begin{itemize}
                      \item Suppose first that the chosen formula is
                            principal.  Then
                            \[
                                M=P,
                                \qquad
                                B=\sneg C,
                                \qquad
                                \Gamma'=\Sigma,
                                \qquad
                                \Delta'=\Pi .
                            \]
                            The premise is
                            \[
                                \Gamma' \types M:C,\Delta' .
                            \]

                            We now check freshness for the induction hypothesis.
                            From
                            \[
                                x\notin\fv(\Gamma,\Delta)
                                =
                                \fv(\Gamma',M:\sneg C,\Delta')
                            \]
                            it follows that we also have
                            \[
                                x\notin\fv(\Gamma',M:C,\Delta'),
                            \]
                            Hence the induction hypothesis of clause (2), gives
                            \[
                                \Gamma',x:A \types M:C,\Delta' .
                            \]
                            Reapplying \(\sneg_I L\) gives
                            \[
                                \Gamma',x:A,M:\sneg C \types \Delta' .
                            \]
                            Since \(B=\sneg C\), this is exactly
                            \[
                                \Gamma',x:A,M:B \types \Delta' .
                            \]

                      \item Suppose now that the chosen formula is passive.
                            Then \(M:B\) is an formula in \(\Sigma\).  Up to
                            exchange, write
                            \[
                                \Sigma=\Sigma',M:B,
                                \qquad
                                \Gamma'=\Sigma',P:\sneg C,
                                \qquad
                                \Delta'=\Pi .
                            \]
                            The premise can be written as
                            \[
                                \Sigma',M:B \types P:C,\Delta' .
                            \]
                            Since \(\Gamma'=\Sigma',P:\sneg C\) and \(\Delta'\)
                            are disjoint variable environments, \(P\) is a
                            variable, and \(\Sigma'\) and \(P:C,\Delta'\) are
                            also disjoint variable environments.

                            We now check freshness for the induction hypothesis.
                            From
                            \[
                                x\notin\fv(\Gamma,\Delta)
                                =
                                \fv(\Sigma',M:B,P:\sneg C,\Delta')
                            \]
                            it follows that
                            \[
                                x\notin\fv(\Sigma',M:B,P:C,\Delta').
                            \]
                            Therefore the induction hypothesis of clause (1),
                            applied to the premise with chosen formula
                            \(M:B\), gives
                            \[
                                \Sigma',x:A,M:B \types P:C,\Delta' .
                            \]
                            Reapplying \(\sneg_I L\) gives
                            \[
                                \Sigma',x:A,M:B,P:\sneg C \types \Delta' .
                            \]
                            By exchange and \(\Gamma'=\Sigma',P:\sneg C\), this
                            is exactly
                            \[
                                \Gamma',x:A,M:B \types \Delta' .
                            \]
                  \end{itemize}
        \end{enumerate}

        The cases for \(\sneg_E L\), \(\sneg_I R\), and \(\sneg_E R\) are analogous.
        In each case, the premise and conclusion contain the same term formula
        \(P\), and only its type changes between \(C\) and \(\sneg C\).  Therefore
        the freshness condition needed for the induction hypothesis follows
        directly from the freshness assumption on the conclusion.  The
        disjoint-variable-environment condition is preserved for the same reason:
        when \(P\) is forced to lie in a side context, that side context being a
        variable environment implies that \(P\) is a variable.
    \end{indproof}

    Therefore left weakening is admissible.
\end{proof}

\begin{lemma}[Right Weakening]
    \label{lem:weakening-right}
    Let $\Gamma \types \Delta$ be derivable. Let \(x\) be fresh, i.e.
    \(x\notin\fv(\Gamma,\Delta)\), and let \(A\) be any type.
    For every choice of a displayed formula \(M:B\) in this judgement,
    the following hold.

    \begin{enumerate}
        \item Suppose the chosen formula is on the left, so that the judgement decomposes as
              \[
                  \smashunderbrace{\Gamma',\,M:B}{\Gamma}
                  \types
                  \smashunderbrace{\Delta'}{\Delta}.
              \]
              If \(\Gamma'\) and \(\Delta'\) are disjoint variable environments,
              then
              \[
                  \Gamma',\,M:B \types x:A,\,\Delta'
              \]
              is derivable.

        \item Suppose the chosen formula is on the right, so that the judgement decomposes as
              \[
                  \smashunderbrace{\Gamma'}{\Gamma}
                  \types
                  \smashunderbrace{M:B,\,\Delta'}{\Delta}.
              \]
              If \(\Gamma'\) and \(\Delta'\) are disjoint variable environments,
              then
              \[
                  \Gamma' \types M:B,\,x:A,\,\Delta'
              \]
              is derivable.
    \end{enumerate}
\end{lemma}
\begin{proof}
    The proof is analogous to that of \cref{lem:weakening-left}.
\end{proof}

\begin{definition}[Assumptions and co-assumptions]
    Let \(\Gamma \types \Delta\) be a judgement. A variable typing
    \(x:A\in\Gamma\), that is, a variable typing appearing on the left-hand
    side of the judgement, is called an \emph{assumption} of the judgement.
    A variable typing \(x:A\in\Delta\), that is, a variable typing appearing
    on the right-hand side of the judgement, is called a \emph{co-assumption}
    of the judgement.
\end{definition}

With this terminology, we can state the substitution lemmas for assumptions
and co-assumptions separately.

\begin{lemma}[Substitution for assumptions]
    \label{lem:sbst-assmpt}
    Let $\Gamma$ and $\Delta$ be disjoint variable environments. The following
    substitution rules are admissible:
    \[
        \infer[Subst-Assm-L]{
            \Gamma, x:A,\,M:B \types \Delta
            \and
            \Gamma \types N:A,\,\Delta
        }{
            \Gamma,\,\Subst{M}{N}{x}:B \types \Delta
        }
        \;x \notin \fv(\Gamma,\Delta)
    \]
    and
    \[
        \infer[Subst-Assm-R]{
            \Gamma, x:A \types M:B,\,\Delta
            \and
            \Gamma \types N:A,\,\Delta
        }{
            \Gamma \types \Subst{M}{N}{x}:B,\,\Delta
        }
        \;x \notin \fv(\Gamma,\Delta).
    \]
\end{lemma}

\begin{lemma}[Substitution for co-assumptions]
    \label{lem:sbst-co-assmpt}
    Let $\Gamma$ and $\Delta$ be disjoint variable environments. The following
    substitution rules are admissible:
    \[
        \infer[Subst-Coassm-L]{
            \Gamma,\,M:B \types x:A,\,\Delta
            \and
            \Gamma,\,N:A \types \Delta
        }{
            \Gamma,\,\Subst{M}{N}{x}:B \types \Delta
        }
        \;x \notin \fv(\Gamma,\Delta)
    \]
    and
    \[
        \infer[Subst-Coassm-R]{
            \Gamma \types M:B,\,x:A,\,\Delta
            \and
            \Gamma,\,N:A\types \Delta
        }{
            \Gamma \types \mf{\Subst{M}{N}{x}:B},\,\Delta
        }
        \;x \notin \fv(\Gamma,\Delta).
    \]
\end{lemma}

We rewrite the two substitution lemmas below in a more proof-friendly, expanded form. It is straightforward to verify that these expanded forms are equivalent to the original lemmas.

\begin{lemma}[Substitution for assumptions\camrevadd{, expanded form}]
    \label{lem:sbst-assmpt-alt}
    Let $\Gamma,\,x:A \types \Delta$ be derivable, where \(x:A\) is the distinguished assumption. Then the following assertions hold.

    \begin{enumerate}
        \item Let \(M:B\in\Gamma\) (hence $M:B$ is distinct from $x:A$) be a chosen typing formula on the left, so
              that, up to exchange, the judgement decomposes as
              \[
                  \smashunderbrace{\Gamma',\,M:B}{\Gamma},\,x:A
                  \types
                  \Delta .
              \]
              Suppose that \(\Gamma',x:A\) and \(\Delta\) are disjoint variable
              environments. Equivalently, after removing the chosen typing
              formula \(M:B\), all remaining typing formulas in the judgement
              are variable typings for pairwise distinct variables. If
              \[
                  \Gamma' \types N:A,\,\Delta
              \]
              is derivable, then
              \[
                  \Gamma',\,\Subst{M}{N}{x}:B \types \Delta
              \]
              is derivable.

        \item Let \(M:B\in\Delta\) be a chosen typing formula on the right, so
              that, up to exchange, the judgement decomposes as
              \[
                  \Gamma,\,x:A
                  \types
                  \smashunderbrace{M:B,\,\Delta'}{\Delta}.
              \]
              Suppose that \(\Gamma,x:A\) and \(\Delta'\) are disjoint variable
              environments. Equivalently, after removing the chosen typing
              formula \(M:B\), all remaining typing formulas in the judgement
              are variable typings for pairwise distinct variables. If
              \[
                  \Gamma \types N:A,\,\Delta'
              \]
              is derivable, then
              \[
                  \Gamma \types \Subst{M}{N}{x}:B,\,\Delta'
              \]
              is derivable.
    \end{enumerate}
\end{lemma}

\begin{lemma}[Substitution for co-assumptions, expanded form]
    \label{lem:sbst-coassmpt-alt}
    Let \(\Gamma \types x:A,\,\Delta\) be derivable, where \(x:A\) is the distinguished co-assumption. Then the
    following assertions hold.

    \begin{enumerate}
        \item Let \(M:B\in\Gamma\) be a chosen typing formula on the left, so
              that, up to exchange,
              \[
                  \smashunderbrace{\Gamma',\,M:B}{\Gamma}
                  \types
                  x:A,\,\Delta .
              \]
              Suppose that \(\Gamma'\) and \(x:A,\Delta\) are disjoint variable
              environments. Equivalently, after removing the chosen typing
              formula \(M:B\), all remaining typing formulas in the judgement
              are variable typings for pairwise distinct variables. If
              \[
                  \Gamma',\,N:A \types \Delta
              \]
              is derivable, then
              \[
                  \Gamma',\,\Subst{M}{N}{x}:B \types \Delta
              \]
              is derivable.

        \item Let \(M:B\in\Delta\) (hence $M:B$ is distinct from $x:A$) be a chosen typing formula on the right, so
              that, up to exchange,
              \[
                  \Gamma
                  \types
                  x:A,\,
                  \smashunderbrace{M:B,\,\Delta'}{\Delta}.
              \]
              Suppose that \(\Gamma\) and \(x:A,\Delta'\) are disjoint variable
              environments. Equivalently, after removing the chosen typing
              formula \(M:B\), all remaining typing formulas in the judgement
              are variable typings for pairwise distinct variables. If
              \[
                  \Gamma,\,N:A \types \Delta'
              \]
              is derivable, then
              \[
                  \Gamma \types \Subst{M}{N}{x}:B,\,\Delta'
              \]
              is derivable.
    \end{enumerate}
\end{lemma}

\begin{proof}
    \camrevreplace{We prove the two assertions simultaneously by induction on the derivation
        of \(\Gamma \types x:A,\,\Delta\).}{We prove all four assertions of \cref{lem:sbst-assmpt-alt,lem:sbst-coassmpt-alt}
        simultaneously by induction on the derivation containing the distinguished
        assumption or co-assumption. We give the co-assumption cases in detail;
        the assumption cases use the same induction with the roles of the two sides
        and their dual typing rules exchanged.} In the binding cases, we first use
    \(\alpha\)-conversion so that the bound variables introduced by the last
    rule are fresh for \(x\) and for the substituting term \(N\).

    \begin{indproof}
        \indcase{Id}
        The derivation ends with
        \[
            \infer[Id]{ }{
                \Sigma,\,y:C \types y:C,\,\Pi
            }
        \]
        Thus, up to exchange, we have $\Gamma=\Sigma,\,y:C$ and $x:A,\,\Delta = y:C,\,\Pi$. We distinguish two cases according to whether \(M:B\) is chosen from the left-hand side or from the
        right-hand side.

        \begin{enumerate}
            \item Suppose first that \(M:B\in\Gamma\) is chosen from the left-hand side.
                  Thus, up to exchange,
                  \[
                      \smashunderbrace{\Sigma,\,y:C}{\Gamma',\,M:B}
                      \types
                      \smashunderbrace{y:C,\,\Pi}{x:A,\,\Delta}.
                  \]
                  We must show
                  \[
                      \Gamma',\,\Subst{M}{N}{x}:B \types \Delta,
                  \]
                  assuming that \(\Gamma'\) and \(x:A,\Delta\) are disjoint variable
                  environments and that
                  \[
                      \Gamma',\,N:A \types \Delta
                  \]
                  is derivable.

                  We first claim that \(M:B\) must be the formula \(y:C\) introduced by
                  \(\rlnm{Id}\) on the left. Otherwise, \(y:C\) would remain in
                  \(\Gamma'\), while \(y:C\) also appears on the right-hand side
                  \(x:A,\Delta\). This contradicts the assumption that \(\Gamma'\) and
                  \(x:A,\Delta\) are disjoint variable environments. Hence
                  \[
                      M=y,
                      \qquad
                      B=C,
                      \qquad
                      \Gamma'=\Sigma.
                  \]

                  We now distinguish whether the distinguished co-assumption \(x:A\) is
                  the formula \(y:C\) introduced by \(\rlnm{Id}\) on the right.

                  \begin{itemize}
                      \item Suppose \(x:A\) is the principal formula on the right.
                            Then we have
                            \[
                                \smashunderbrace{\Sigma}{\Gamma'},\,\smashunderbrace{y:C}{M:B}
                                \types
                                \smashunderbrace{y:C}{x:A},\,\smashunderbrace{\Pi}{\Delta}.
                            \]
                            where
                            \[
                                x=y,\qquad A=C,\qquad \Delta=\Pi.
                            \]
                            Since \(\Subst{M}{N}{x}=\Subst{y}{N}{y}=N\) and $B=C=A$, the required judgement
                            \[
                                \Gamma',\,\Subst{M}{N}{x}:B \types \Delta
                            \]
                            is exactly
                            \[
                                \Gamma',\,N:A \types \Delta,
                            \]
                            which is directly given by the assumption.

                      \item Suppose \(x:A\) is passive. Then $x:A \in \Pi$ and \(y:C\in\Delta\).
                            Write, up to exchange, $\Pi = x:A,\,\Pi'$, the judgement decomposes as
                            \[
                                \smashunderbrace{\Sigma}{\Gamma'},\,\smashunderbrace{y:C}{M:B}
                                \types x:A,\, \smashunderbrace{y:C,\,\Pi'}{\Delta}.
                            \]
                            where
                            \[
                                \Delta = y:C,\,\Pi'.
                            \]
                            Since \(\Gamma'\) and \(x:A,\Delta\) are disjoint variable
                            environments and $y:C \in \Delta$, there must be $y \neq x$. Since $M=y$ and \(y\neq x\), $\Subst{M}{N}{x}=\Subst{y}{N}{x}=y$.
                            Moreover since $B=C$ and $\Gamma'=\Sigma$, the required judgement
                            \[
                                \Gamma',\,\Subst{M}{N}{x}:B \types \Delta
                            \]
                            is exactly
                            \[
                                \Sigma,\,y:C \types y:C,\,\Pi'.
                            \]
                            This is obtained as an instance of \(\rlnm{Id}\).
                  \end{itemize}

            \item Suppose now that \(M:B\in\Delta\) is chosen from the right-hand side.
                  Thus, up to exchange,
                  \[
                      \smashunderbrace{\Sigma,\,y:C}{\Gamma}
                      \types \smashunderbrace{y:C,\,\Pi}{x:A,\,M:B,\,\Delta'}
                  \]
                  We must show
                  \[
                      \Gamma \types \Subst{M}{N}{x}:B,\,\Delta',
                  \]
                  assuming that \(\Gamma\) and \(x:A,\Delta'\) are disjoint variable
                  environments and that
                  \[
                      \Gamma,\,N:A \types \Delta'
                  \]
                  is derivable.

                  First, the distinguished co-assumption \(x:A\) cannot be the formula
                  \(y:C\) introduced by \(\rlnm{Id}\) on the right. Indeed, if
                  \(x:A=y:C\), then \(y\) occurs in both \(\Gamma=\Sigma,y:C\) and
                  \(x:A,\Delta'\), contradicting the assumption that \(\Gamma\) and
                  \(x:A,\Delta'\) are disjoint variable environments. Therefore, \(x:A\) is
                  passive and in $\Pi$.

                  We now claim that \(M:B\) must be the formula \(y:C\) introduced by
                  \(\rlnm{Id}\) on the right. Otherwise, \(y:C\) would remain in
                  \(\Delta'\), while \(y:C\in\Gamma\), contradicting the assumption
                  that \(\Gamma\) and \(x:A,\Delta'\) are disjoint variable
                  environments. Hence
                  \[
                      \smashunderbrace{\Sigma,\,y:C}{\Gamma}
                      \types
                      \smashunderbrace{y:C}{M:B},\,\smashunderbrace{\Pi}{x:A,\,\Delta'}.
                  \]
                  where
                  \[
                      M=y,\qquad B=C, \qquad \Pi=x:A,\,\Delta'.
                  \]

                  Since \(\Gamma\) and \(x:A,\Delta'\) are disjoint variable environments and
                  $y:C \in \Gamma$, there must be $y \neq x$.

                  Since $M=y$ and \(y\neq x\), $\Subst{M}{N}{x}=\Subst{y}{N}{x}=y$.
                  Moreover since $B=C$ and \(\Gamma=\Sigma,y:C\), the required judgement
                  \[
                      \Gamma \types \Subst{M}{N}{x}:B,\,\Delta',
                  \]
                  is exactly
                  \[
                      \Sigma,\,y:C \types y:C,\,\Delta'.
                  \]
                  This is obtained as \camrevreplace{a}{an} instance of \(\rlnm{Id}\).
        \end{enumerate}

        \indcase{Non-negation rules}
        We spell out one representative case, namely \(\to R\). The derivation ends as
        \[
            \infer[$\to R$]{
                \Sigma,\,y:C \types P:D,\,\Pi
            }{
                \Sigma
                \types
                \abs{y}{P}:C\to D,\,\Pi
            }
            \;y\notin\fv(\Sigma,\Pi).
        \]
        By \(\alpha\)-conversion, we may assume that \(y\notin\fv(N)\) and $y \neq x$.
        Since the distinguished co-assumption \(x:A\) is a variable typing, it cannot be the principal typing formula $\abs{y}{P}:C\to D$. Hence \(x:A \in \Pi\). Write, up to exchange, $\Pi = x:A,\,\Pi'$. Then
        \[
            \smashunderbrace{\Sigma}{\Gamma}
            \types x:A,\,
            \smashunderbrace{\abs{y}{P}:C\to D,\,\Pi'}{\Delta}
        \]
        \[
            \Gamma=\Sigma
            \qquad\text{and}\qquad
            \Delta=\abs{y}{P}:C\to D,\,\Pi'.
        \]
        Moreover, from \(y\notin\fv(\Sigma,\Pi)\) and $y\neq x$, we have
        \[
            y\notin\fv(\Sigma,x:A,\Pi').
        \]

        We distinguish two cases according to whether the typing formula
        \(M:B\) is chosen from the left-hand side or from the right-hand side.

        \begin{enumerate}
            \item Suppose first that \(M:B\in\Gamma\) is chosen from the left-hand side.
                  Thus, up to exchange,
                  \[
                      \smashunderbrace{\Sigma}{\Gamma',\,M:B}
                      \types
                      x:A,\,\smashunderbrace{\abs{y}{P}:C\to D,\,\Pi'}{\Delta}
                  \]
                  We must show
                  \[
                      \Gamma',\,\Subst{M}{N}{x}:B \types \Delta,
                  \]
                  assuming that \(\Gamma'\) and \(x:A,\Delta\) are disjoint variable
                  environments and that
                  \[
                      \Gamma',\,N:A \types \Delta
                  \]
                  is derivable.

                  However, this case is impossible. Indeed, the context
                  \(x:A,\Delta\) contains the typing formula
                  \[
                      \abs{y}{P}:C\to D,
                  \]
                  whose term is not a variable. Therefore \(x:A,\Delta\) is not a
                  variable environment, contradicting the assumption that
                  \(\Gamma'\) and \(x:A,\Delta\) are disjoint variable environments.

            \item Suppose now that \(M:B\in\Delta\) is chosen from the right-hand side.
                  Thus, up to exchange,
                  \[
                      \smashunderbrace{\Sigma}{\Gamma}
                      \types
                      x:A,\,
                      \smashunderbrace{\abs{y}{P}:C\to D,\,\Pi'}{M:B,\,\Delta'}.
                  \]
                  We must show
                  \[
                      \Gamma \types \Subst{M}{N}{x}:B,\,\Delta',
                  \]
                  assuming that \(\Gamma\) and \(x:A,\Delta'\) are disjoint variable
                  environments and that
                  \[
                      \Gamma,\,N:A \types \Delta'
                  \]
                  is derivable.

                  We claim that \(M:B\) must be the principal typing formula
                  \[
                      \abs{y}{P}:C\to D.
                  \]
                  Otherwise, this principal typing formula would remain in
                  \(\Delta'\). Then \(x:A,\Delta'\) would contain a typing formula whose
                  term is not a variable, contradicting the assumption that
                  \(\Gamma\) and \(x:A,\Delta'\) are disjoint variable environments.

                  Hence
                  \[
                      M=\abs{y}{P},
                      \qquad
                      B=C\to D,
                      \qquad
                      \Gamma=\Sigma,
                      \qquad
                      \Delta'=\Pi'.
                  \]
                  Therefore the assumed derivation $\Gamma,\,N:A \types \Delta'$
                  is
                  \[
                      \Sigma,\,N:A \types \Pi'.
                  \]

                  We now derive the auxiliary judgement needed to apply the induction
                  hypothesis to the premise of \(\to R\). Since
                  \(\Gamma=\Sigma\) and \(x:A,\Delta' (= x:A,\Pi')\) are disjoint
                  variable environments, \(\Sigma\) and \(\Pi'\) must also be disjoint
                  variable environments. Moreover, we have justified that
                  \[
                      y\notin\fv(\Sigma,N:A,\Pi').
                  \]
                  Hence, by left weakening applied to
                  \[
                      \Sigma,\,N:A \types \Pi',
                  \]
                  we obtain
                  \[
                      \Sigma,\,y:C,\,N:A \types \Pi'.
                  \]

                  Now consider the premise
                  \[
                      \Sigma,\,y:C \types P:D,\,x:A,\,\Pi'.
                  \]
                  The distinguished co-assumption is still \(x:A\), and the chosen
                  typing formula \(P:D\) is on the right. We may apply the induction
                  hypothesis, clause \((2)\), because:
                  \[
                      \Sigma,\,y:C
                      \quad\text{and}\quad
                      x:A,\,\Pi'
                  \]
                  are disjoint variable environments, using
                  \(y\notin\fv(\Sigma,x:A,\Pi')\), and because we have just derived
                  \[
                      \Sigma,\,y:C,\,N:A \types \Pi'.
                  \]
                  Therefore the induction hypothesis gives
                  \[
                      \Sigma,\,y:C
                      \types
                      \Subst{P}{N}{x}:D,\,\Pi'.
                  \]

                  Since \(y\notin\fv(\Sigma,\Pi')\), we may reapply \(\to R\) and
                  obtain
                  \[
                      \Sigma
                      \types
                      \abs{y}{\Subst{P}{N}{x}}:C\to D,\,\Pi'.
                  \]
                  Because \(y\neq x\) and \(y\notin\fv(N)\), capture-avoiding
                  substitution gives
                  \[
                      \Subst*{\abs{y}{P}}{N}{x}
                      =
                      \abs{y}{(\Subst{P}{N}{x})}.
                  \]
                  Hence the above judgement is exactly
                  \[
                      \Gamma
                      \types
                      \Subst{M}{N}{x}:B,\,\Delta',
                  \]
                  as required.
        \end{enumerate}

        The remaining non-negation rules are analogous. In each case, the distinguished
        co-assumption \(x:A\) is a variable typing and therefore cannot be the
        non-variable principal typing formula introduced by the last rule. Hence
        \(x:A\) must be passive. The disjoint-variable-environment condition then
        forces the chosen typing formula \(M:B\) to be principal in every non-vacuous
        case; otherwise the non-variable principal typing formula would remain in one
        of the side contexts that is required to be a variable environment. For each
        immediate premise, we first \(\alpha\)-rename bound variables so that they are
        fresh for \(N\), apply weakening where necessary to the derivation typing
        \(N:A\), then apply the induction hypothesis to the corresponding premise
        typing formula, and finally reapply the same non-negation rule.

        \indcase{Negation rules}
        We spell out one representative case, namely \(\sneg_I L\). The derivation ends
        as
        \[
            \infer[$\sneg_I L$]{
                \Sigma \types P:C,\,\Pi
            }{
                \Sigma,\,P:\sneg C
                \types
                \Pi
            }.
        \]
        Thus, up to exchange,
        \[
            \smashunderbrace{\Sigma,\,P:\sneg C}{\Gamma}
            \types
            \smashunderbrace{\Pi}{x:A,\,\Delta}
        \]
        where
        \[
            \Gamma=\Sigma,\,P:\sneg C
            \qquad\text{and}\qquad
            \Pi=x:A,\,\Delta.
        \]
        We distinguish two cases according to whether the chosen typing formula
        \(M:B\) is from the left-hand side or from the right-hand side.

        \begin{enumerate}
            \item Suppose first that \(M:B\in\Gamma\) is chosen from the left-hand side.
                  Thus, up to exchange,
                  \[
                      \smashunderbrace{\Sigma,\,P:\sneg C}{\Gamma',\,M:B}
                      \types
                      \smashunderbrace{\Pi}{x:A,\,\Delta}.
                  \]
                  We must show
                  \[
                      \Gamma',\,\Subst{M}{N}{x}:B \types \Delta,
                  \]
                  assuming that \(\Gamma'\) and \(x:A,\Delta\) are disjoint variable
                  environments and that
                  \[
                      \Gamma',\,N:A \types \Delta
                  \]
                  is derivable.

                  We now distinguish whether the chosen typing formula \(M:B\) is the
                  principal formula \(P:\sneg C\) introduced by \(\sneg_I L\).

                  \begin{itemize}
                      \item Suppose first that \(M:B\) is principal. Then
                            \[
                                M=P,
                                \qquad
                                B=\sneg C,
                                \qquad
                                \Gamma'=\Sigma.
                            \]
                            The premise of the last rule is
                            \begin{equation}
                                \Sigma \types P:C,\,x:A,\,\Delta
                                \label{proof:sbst-coassmpt-alt-case-simIL-left-principal-premise}
                            \end{equation}
                            We claim that $P:C$ is distinct from $x:A$. Indeed, since
                            \(\Gamma'=\Sigma\) and \(\Gamma'\) and \(x:A,\Delta\)
                            are disjoint variable environments, \(\Sigma\) and
                            \(x:A,\Delta\) are disjoint variable environments. Apply
                            \cref{lem:judgement-form}, clause \((2)\), to
                            \cref{proof:sbst-coassmpt-alt-case-simIL-left-principal-premise},
                            choosing the typing formula \(P:C\) on the right and
                            taking the remaining right context to be \(x:A,\Delta\).
                            If \(P\) is not a variable, then certainly
                            \(P:C\neq x:A\). Otherwise, the judgement-form lemma gives
                            either \(P:D\in\Sigma\), with \(C\) and \(D\) related by
                            some even number of negations, or \(P:D\in x:A,\Delta\),
                            with \(C\) and \(D\) related by some odd number of
                            negations. If \(P:C=x:A\), the first alternative
                            contradicts the disjointness of \(\Sigma\) and
                            \(x:A,\Delta\). In the second alternative, pairwise
                            distinctness of the variable environment \(x:A,\Delta\)
                            forces \(P:D=x:A\), hence \(C=D=A\), but then \(C\) and
                            \(D\) are related by \(0\) negations, which is even, not
                            odd. This contradiction proves \(P:C\neq x:A\).

                            Therefore we may apply the induction
                            hypothesis, clause \((2)\), to this premise, choosing the
                            typing formula \(P:C\) on the right. The required
                            disjointness condition is exactly that \(\Sigma\) and
                            \(x:A,\Delta\) are disjoint variable environments, which
                            follows from the present assumption
                            \[
                                \Gamma'=\Sigma
                                \quad\text{and}\quad
                                \Gamma' \text{ and } x:A,\Delta
                                \text{ are disjoint variable environments}.
                            \]
                            The second derivation needed for the induction hypothesis is
                            precisely
                            \[
                                \Sigma,\,N:A \types \Delta,
                            \]
                            which is derivable by assumption.

                            Therefore the induction hypothesis gives
                            \[
                                \Sigma \types \Subst{P}{N}{x}:C,\,\Delta.
                            \]
                            Reapplying \(\sneg_I L\) gives
                            \[
                                \Sigma,\,\Subst{P}{N}{x}:\sneg C \types \Delta.
                            \]
                            Since \(M=P\), \(B=\sneg C\), and \(\Gamma'=\Sigma\), this is
                            exactly
                            \[
                                \Gamma',\,\Subst{M}{N}{x}:B \types \Delta.
                            \]

                      \item Suppose now that \(M:B\) is not principal. Then
                            \(M:B\in\Sigma\). Write, up to exchange,
                            \[
                                \Sigma=\Sigma',\,M:B.
                            \]
                            Hence
                            \[
                                \Gamma'=\Sigma',\,P:\sneg C,
                            \]
                            and the premise of the last rule may be written as
                            \[
                                \Sigma',\,M:B \types P:C,\,x:A,\,\Delta.
                            \]

                            Since \(\Gamma'=\Sigma',P:\sneg C\) and \(x:A,\Delta\) are
                            disjoint variable environments, \(P\) is a variable, and
                            \(\Sigma'\) and \(x:A,P:C,\Delta\) are disjoint variable
                            environments.

                            By assumption, we have a derivation of
                            \[
                                \Gamma',\,N:A \types \Delta,
                            \]
                            that is,
                            \[
                                \Sigma',\,P:\sneg C,\,N:A \types \Delta.
                            \]
                            By exchange, this is
                            \[
                                \Sigma',\,N:A,\,P:\sneg C \types \Delta.
                            \]
                            Applying the negation rule inverse to \(\sneg_I L\), namely
                            \(\sneg_E R\), gives
                            \[
                                \Sigma',\,N:A \types P:C,\,\Delta.
                            \]

                            We may now apply the induction hypothesis, clause \((1)\),
                            to the premise
                            \[
                                \Sigma',\,M:B \types P:C,\,x:A,\,\Delta,
                            \]
                            choosing \(M:B\) from the left-hand side. The side-context
                            condition is exactly that \(\Sigma'\) and
                            \(x:A,P:C,\Delta\) are disjoint variable environments, as
                            checked above. Hence the induction hypothesis gives
                            \[
                                \Sigma',\,\Subst{M}{N}{x}:B \types P:C,\,\Delta.
                            \]
                            Reapplying \(\sneg_I L\) gives
                            \[
                                \Sigma',\,\Subst{M}{N}{x}:B,\,P:\sneg C \types \Delta.
                            \]
                            By exchange, and since
                            \(\Gamma'=\Sigma',P:\sneg C\), this is exactly
                            \[
                                \Gamma',\,\Subst{M}{N}{x}:B \types \Delta.
                            \]
                  \end{itemize}

            \item Suppose now that \(M:B\in\Delta\) is chosen from the right-hand side.
                  Thus, up to exchange,
                  \[
                      \smashunderbrace{\Sigma,\,P:\sneg C}{\Gamma}
                      \types
                      x:A,\,
                      \smashunderbrace{M:B,\,\Delta'}{\Delta}.
                  \]
                  We must show
                  \[
                      \Gamma \types \Subst{M}{N}{x}:B,\,\Delta',
                  \]
                  assuming that \(\Gamma\) and \(x:A,\Delta'\) are disjoint variable
                  environments and that
                  \[
                      \Gamma,\,N:A \types \Delta'
                  \]
                  is derivable.

                  Since
                  \[
                      \Gamma=\Sigma,\,P:\sneg C
                  \]
                  and \(\Gamma\) is a variable environment, \(P\) is a variable. Since
                  \(\Gamma\) and \(x:A,\Delta'\) are disjoint variable environments,
                  \(\Sigma\) and \(x:A,P:C,\Delta'\) are also disjoint variable
                  environments.

                  The premise of the last rule is
                  \[
                      \Sigma \types P:C,\,x:A,\,M:B,\,\Delta'.
                  \]

                  By assumption, we have a derivation of
                  \[
                      \Gamma,\,N:A \types \Delta',
                  \]
                  that is,
                  \[
                      \Sigma,\,P:\sneg C,\,N:A \types \Delta'.
                  \]
                  By exchange, this is
                  \[
                      \Sigma,\,N:A,\,P:\sneg C \types \Delta'.
                  \]
                  Applying \(\sneg_E R\) gives
                  \[
                      \Sigma,\,N:A \types P:C,\,\Delta'.
                  \]

                  We now apply the induction hypothesis, clause \((2)\), to the premise
                  \[
                      \Sigma \types P:C,\,x:A,\,M:B,\,\Delta',
                  \]
                  choosing \(M:B\) from the right-hand side. The required second
                  derivation is
                  \[
                      \Sigma,\,N:A \types P:C,\,\Delta',
                  \]
                  obtained above. Hence the induction hypothesis gives
                  \[
                      \Sigma \types P:C,\,\Subst{M}{N}{x}:B,\,\Delta'.
                  \]
                  Reapplying \(\sneg_I L\) gives
                  \[
                      \Sigma,\,P:\sneg C
                      \types
                      \Subst{M}{N}{x}:B,\,\Delta'.
                  \]
                  Since \(\Gamma=\Sigma,P:\sneg C\), this is exactly
                  \[
                      \Gamma \types \Subst{M}{N}{x}:B,\,\Delta'.
                  \]
        \end{enumerate}

        \camrevreplace{The cases for \(\sneg_E L\), \(\sneg_I R\), and \(\sneg_E R\) are analogous.}{The remaining negation cases in which the distinguished co-assumption is passive are analogous.}

        \camrevadd{If the distinguished co-assumption is principal, it changes sides in the premise, so we use the assumption-substitution part of the simultaneous induction hypothesis. For example, suppose the last rule is \(\sneg_I R\), the distinguished co-assumption is \(x:\sneg C\), and the chosen subject is on the right. The last rule has the form
            \[
                \infer[$\sneg_I R$]{
                    \Gamma,x:C \types M:B,\Delta'
                }{
                    \Gamma \types x:\sneg C,M:B,\Delta'
                }.
            \]
            The given substituting derivation \(\Gamma,N:\sneg C\types\Delta'\) yields \(\Gamma\types N:C,\Delta'\) by \(\sneg_E R\). Applying clause (2) of the assumption-substitution induction hypothesis to the strictly smaller premise therefore gives \(\Gamma\types\Subst{M}{N}{x}:B,\Delta'\). If the chosen subject is on the left, the same argument uses clause (1). For a principal distinguished co-assumption in \(\sneg_E R\), use \(\sneg_I R\) to turn the given refutation of \(C\) into a proof of \(\sneg C\), and again apply assumption substitution to the premise.}

        \camrevadd{Conversely, when a distinguished assumption is principal in \(\sneg_I L\) or \(\sneg_E L\), apply \(\sneg_E L\) or \(\sneg_I L\), respectively, to the substituting derivation and use the co-assumption-substitution induction hypothesis on the premise. In each case the remaining variable environments are unchanged, so the required disjointness conditions are preserved. A formula moved into a side context cannot coincide with the chosen subject on that side of the premise: otherwise that premise would consist only of disjoint variable environments, contrary to \cref{lem:judgement-form}.}
    \end{indproof}

    Therefore substitution for co-assumptions is admissible.
\end{proof}

We can obtain the following renaming lemma as a corollary of the substitution lemmas.
\begin{corollary}[Renaming]
    \label{lem:renaming}
    Let \(x\) and \(y\) be distinct variables of the same type \(A\).
    For each of the rules below, suppose that
    \[
        y \notin \fv(\Gamma,\Delta,M).
    \]

    \begin{enumerate}[(i)]
        \item \emph{Renaming an assumption.}
              Suppose that \(\Gamma,x:A\) and \(\Delta\) are disjoint variable
              environments. Then the following rules are admissible:
              \[
                  \infer[Ren-Assm-L]{
                      \Gamma,\,x:A,\,M:B \types \Delta
                  }{
                      \Gamma,\,y:A,\,\Subst{M}{y}{x}:B \types \Delta
                  }
              \]
              and
              \[
                  \infer[Ren-Assm-R]{
                      \Gamma,\,x:A \types M:B,\,\Delta
                  }{
                      \Gamma,\,y:A
                      \types
                      \Subst{M}{y}{x}:B,\,\Delta
                  }.
              \]

        \item \emph{Renaming a co-assumption.}
              Suppose that \(\Gamma\) and \(x:A,\Delta\) are disjoint variable
              environments. Then the following rules are admissible:
              \[
                  \infer[Ren-Coassm-L]{
                      \Gamma,\,M:B \types x:A,\,\Delta
                  }{
                      \Gamma,\,\Subst{M}{y}{x}:B
                      \types
                      y:A,\,\Delta
                  }
              \]
              and
              \[
                  \infer[Ren-Coassm-R]{
                      \Gamma
                      \types
                      M:B,\,x:A,\,\Delta
                  }{
                      \Gamma
                      \types
                      \Subst{M}{y}{x}:B,\,y:A,\,\Delta
                  }.
              \]
    \end{enumerate}
\end{corollary}

\begin{proof}
    Both assertions follow from weakening, identity, and the corresponding
    substitution lemmas.

    For assumption renaming, first add \(y:A\) to the left-hand side of the
    given premise by admissible weakening. The identity rule gives
    \[
        \Gamma,\,y:A \types y:A,\,\Delta.
    \]
    We may therefore apply respectively
    \(\rlnm{Subst-Assm-L}\) or \(\rlnm{Subst-Assm-R}\), with left
    context \(\Gamma,y:A\) and substituting term \(N=y\). This gives
    \[
        \Gamma,\,y:A,\,\Subst{M}{y}{x}:B \types \Delta
    \]
    or
    \[
        \Gamma,\,y:A
        \types
        \Subst{M}{y}{x}:B,\,\Delta,
    \]
    respectively.

    Co-assumption renaming is analogous.
\end{proof}

\begin{lemma}[From $\tintc{}$ to $\tintcty{}$]
    The following hold simultaneously:
    \begin{enumerate}[(i)]
        \item If $\wanp{\Pi,\,A,\,(L,R)}$ in \tintc{}, then $L \vdash' A;R$ in \tintcty{}.
        \item If $\wandp{\Pi,\,A\,(L,R)}$ in \tintc{}, then $L;A \vdash' R$ in \tintcty{}.
    \end{enumerate}
\end{lemma}

\begin{proof}
    We prove (i) and (ii) simultaneously by mutual induction on the given derivation \camrevreplace{in Wansing's $\tint{}$ system}{in $\tintc{}$, the strong-negation extension of Wansing's $\tint{}$}.

    \camrevadd{We strengthen the simultaneous induction so that each constructed subject has free variables only in its displayed assumption and co-assumption environments. This invariant holds for identity and unit cases and is preserved by the term constructors, capture-avoiding substitutions, and strong-negation steps below. Thus the extra variables used when enlarging the contexts are fresh for the subject, as required by \cref{lem:weakening-left,lem:weakening-right}.} For (i), assume $\mathrm{P}(\Pi,A,(L\camrevreplace{,}{;}\,R))$. We proceed by cases on the last rule. Whenever a premise is of the form $\mathrm{DP}(\Pi,B(L'\camrevreplace{,}{;}\,R'))$, we appeal to the mutual induction hypothesis for claim (ii). The cases are as follows.
    \begin{indproof}
        \indcase{1} Here $\wanp{\overline{A},\,A,\,({A};\,\varnothing)}$ with $\LeftCtx={A}$ and $\RightCtx=\varnothing$, so it remains to show that $A \vdash' A$. By Definition~\ref{def:logic-single-subject-appendix}, this means showing that there exists a term $M$ such that $x_A:A \types M:A$. But $x_A:A \types x_A:A$ by \rlnm{ID}, so the result follows by taking $M=x_A$.
        \indcase{3} Here $\wanp{\overline{\truefm},\,\truefm,\,(\varnothing;\,\varnothing)}$ with $\LeftCtx=\varnothing$ and $\RightCtx=\varnothing$, so it remains to show that $\vdash' \truefm$. By Definition~\ref{def:logic-single-subject-appendix}, this means showing that there exists a term $M$ such that $\types M:\truefm$. But $\types \unittm:\truefm$ by the $\truefm$R rule, so the result follows by taking $M=\unittm$.

        \indcase{5}
        Here we have $\wanp{\mprfraction{\Pi_1 \enspace \Pi_2}{A_1 \wedge A_2},\,
                A_1 \wedge A_2,\,
                (\LeftCtx_1 \cup \LeftCtx_2;\,\RightCtx_1 \cup \RightCtx_2)}$
        with premises:
        \[
            \wanp{\Pi_1,\,A_1,\,(\LeftCtx_1;\,\RightCtx_1)}
            \quad\text{and}\quad
            \wanp{\Pi_2,\,A_2,\,(\LeftCtx_2;\,\RightCtx_2)}
        \]
        By the IH on \camrevreplace{$\wanp{\Pi_1,\,A_1,\,(\LeftCtx_1;\,\RightCtx_1)}$}{both premises}, there exist terms $M_1$ and $M_2$ such that
        \[
            \smashunderbrace{(x_B : B)_{B \in L_1}}{\Gamma_1} \types M_1 : A_1,\,\smashunderbrace{(y_C : C)_{C \in R_1}}{\Delta_1}
        \]
        \[
            \smashunderbrace{(x_B : B)_{B \in L_2}}{\Gamma_2} \types M_2 : A_2,\,\smashunderbrace{(y_C : C)_{C \in R_2}}{\Delta_2}
        \]

        By admissible weakening, we may enlarge both derivations to the common context
        $\Gamma = \Gamma_1 \cup \Gamma_2$ on the left and $\Delta = \Delta_1 \cup \Delta_2$ on the right, obtaining
        \[
            \smashunderbrace{(x_B : B)_{B \in L_1 \cup L_2}}{\Gamma} \types M_i : A_i,\,\smashunderbrace{(y_C : C)_{C \in R_1 \cup R_2}}{\Delta}\;i=1,\,2
        \]
        Hence, by the $\wedge$R rule,
        \[
            (x_B : B)_{B \in L_1 \cup L_2} \types (M_1,\,M_2):A_1 \wedge A_2,\,(y_C : C)_{C \in R_1 \cup R_2}
        \]
        Therefore, by Definition~\ref{def:logic-single-subject-appendix}, this shows
        \[
            \LeftCtx_1 \cup \LeftCtx_2 \vdash' A_1 \wedge A_2;\,\RightCtx_1 \cup \RightCtx_2,
        \]
        as required.

        \indcase{6}
        Here we have $\wanp{\mprfraction{\Pi}{A_1 \vee A_2},\, A_1 \vee A_2,\,(\LeftCtx;\,\RightCtx)}$
        with premise
        \[
            \wanp{\Pi,\,A_1,\,(\LeftCtx;\,\RightCtx)}
        \]
        Applying the IH then gives $\LeftCtx \vdash' A_1;\,\RightCtx$.
        By Definition~\ref{def:logic-single-subject-appendix}, this means $\exists M$ such that
        \[
            (x_B : B)_{B \in \LeftCtx} \types M : A_1,\,(y_C : C)_{C \in \RightCtx}
        \]
        Applying the $\vee$R$_1$ rule, we obtain
        \[
            (x_B : B)_{B \in \LeftCtx} \types \inj_1(M) : A_1 \vee A_2,\,(y_C : C)_{C \in \RightCtx}
        \]
        Therefore, by Definition~\ref{def:logic-single-subject-appendix},
        \[
            \LeftCtx \vdash' A_1 \vee A_2;\,\RightCtx,
        \]
        as required.

        \indcase{7}
        This is analogous \camrevreplace{with}{to} case (6).

        \indcase{8}
        Here we have $\wanp{\mprfraction{\Pi}{A_1 \to A_2},\,
                A_1 \to A_2,\,
                (\LeftCtx \setminus \{A_1\};\,\RightCtx)}$
        with premise
        \[
            \wanp{\Pi,\,A_2,\,(\LeftCtx;\,\RightCtx)}
        \]
        By the IH, $\LeftCtx \vdash' A_2;\,\RightCtx$. By Definition~\ref{def:logic-single-subject-appendix}, there exists a term $M$ such that
        \begin{equation}
            \label{proof:2int-to-typing-case8-eq1}
            (x_B : B)_{B \in \LeftCtx} \types M : A_2,\,(y_C : C)_{C \in \RightCtx}
        \end{equation}
        Let $\Gamma=(x_B:B)_{B\in \LeftCtx\setminus\{A_1\}}$ and $\Delta=(y_C:C)_{C\in \RightCtx}$. We distinguish two cases.
        \begin{itemize}
            \item If $A_1 \in \LeftCtx$, then $(x_B:B)_{B\in \LeftCtx}$ is exactly $\Gamma,\,x_{A_1}:A_1$. Hence \eqref{proof:2int-to-typing-case8-eq1} becomes $\Gamma,\,x_{A_1}:A_1 \types M:A_2,\,\Delta$.

            \item If $A_1 \notin \LeftCtx$, then by \camrevreplace{\rlnm{WeakL2}}{\cref{lem:weakening-left}, clause (2),} applied to \eqref{proof:2int-to-typing-case8-eq1}, we obtain $\Gamma,\,x_{A_1}:A_1 \types M:A_2,\,\Delta$.
        \end{itemize}
        Thus, in either case,
        \begin{equation}
            \label{proof:2int-to-typing-case8-eq2}
            \Gamma,\,x_{A_1}:A_1 \types M:A_2,\,\Delta
        \end{equation}
        Moreover, by the choice of the distinguished variables in Definition~\ref{def:logic-single-subject-appendix}, $x_{A_1}\notin \fv(\Gamma,\,\Delta)$. We may therefore apply \rlnm{$\to$R} to \eqref{proof:2int-to-typing-case8-eq2} and obtain
        \[
            \Gamma \types \mf{\abs{x_{A_1}}{M}:A_1\to A_2},\,\Delta.
        \]
        That is,
        \[
            (x_B:B)_{B\in \LeftCtx\setminus\{A_1\}}
            \types
            \mf{\abs{x_{A_1}}{M}:A_1\to A_2},\,
            (y_C:C)_{C\in \RightCtx}
        \]
        Therefore, by Definition~\ref{def:logic-single-subject-appendix}, $\LeftCtx\setminus\{A_1\}\vdash' A_1\to A_2;\,\RightCtx$ as required.

        \indcase{9}
        Here we have
        \[
            \wanp{\mprfraction{\Pi_1\quad\Pi_2}{A_1 \cotol A_2},\,
                A_1 \cotol A_2,\,
                (\LeftCtx_1 \cup \LeftCtx_2;\,\RightCtx_1 \cup \RightCtx_2)}
        \]
        with premises $\wandp{\Pi_1,\,A_1,\,(\LeftCtx_1;\,\RightCtx_1)}$ and $\wanp{\Pi_2,\,A_2,\,(\LeftCtx_2;\,\RightCtx_2)}$.

        By the IH applied to $\wandp{\Pi_1,\,A_1,\,(\LeftCtx_1;\,\RightCtx_1)}$, we have $\LeftCtx_1;\,A_1 \vdash' \RightCtx_1$. By Definition~\ref{def:logic-single-subject-appendix}, there exists a term $M$ such that
        \begin{equation}
            \label{proof:2int-to-typing-case9-eq2}
            (x_B:B)_{B\in \LeftCtx_1},\,M:A_1 \types (y_C:C)_{C\in \RightCtx_1}
        \end{equation}
        By the IH applied to $\wanp{\Pi_2,\,A_2,\,(\LeftCtx_2;\,\RightCtx_2)}$, we have $\LeftCtx_2 \vdash' A_2;\,\RightCtx_2$. By Definition~\ref{def:logic-single-subject-appendix}, there exists a term $N$ such that
        \begin{equation}
            \label{proof:2int-to-typing-case9-eq1}
            (x_B:B)_{B\in \LeftCtx_2} \types N:A_2,\,(y_C:C)_{C\in \RightCtx_2}
        \end{equation}
        By admissible weakening, we may enlarge both \eqref{proof:2int-to-typing-case9-eq1} and \eqref{proof:2int-to-typing-case9-eq2} to the common contexts $\Gamma=(x_B:B)_{B\in \LeftCtx_1\cup \LeftCtx_2}$ and $\Delta=(y_C:C)_{C\in \RightCtx_1\cup \RightCtx_2}$, obtaining
        \begin{equation}
            \label{proof:2int-to-typing-case9-eq4}
            \Gamma,\,M:A_1 \types \Delta
        \end{equation}
        and
        \begin{equation}
            \label{proof:2int-to-typing-case9-eq3}
            \Gamma \types N:A_2,\,\Delta
        \end{equation}
        Applying \rlnm{$\cotot$R} to \eqref{proof:2int-to-typing-case9-eq3} and \eqref{proof:2int-to-typing-case9-eq4}, we obtain
        \[
            \Gamma \types M\cdot N : A_1 \cotot A_2,\,\Delta.
        \]
        That is,
        \[
            (x_B:B)_{B\in \LeftCtx_1\cup \LeftCtx_2}
            \types
            M\cdot N : A_1 \cotot A_2,\,
            (y_C:C)_{C\in \RightCtx_1\cup \RightCtx_2}
        \]
        Therefore, by Definition~\ref{def:logic-single-subject-appendix}, $\LeftCtx_1 \cup \LeftCtx_2 \vdash' A_1 \cotol A_2;\,\RightCtx_1 \cup \RightCtx_2$.

        \indcase{10}
        Here we have $\wanp{\mprfraction{\Pi}{A},\,A,\,(\LeftCtx;\,\RightCtx)}$ with premise
        \[
            \wanp{\Pi,\,\falsefm,\,(\LeftCtx;\,\RightCtx)}
        \]
        By the IH, $\LeftCtx \vdash' \falsefm;\,\RightCtx$. By Definition~\ref{def:logic-single-subject-appendix}, there exists a term $M$ such that
        \[
            (x_B:B)_{B\in \LeftCtx} \types M:\falsefm,\,(y_C:C)_{C\in \RightCtx}
        \]
        Applying \rlnm{{R-elim}}, we obtain
        \[
            (x_B:B)_{B\in \LeftCtx} \types \abort(M):A,\,(y_C:C)_{C\in \RightCtx}
        \]
        Therefore, by Definition~\ref{def:logic-single-subject-appendix}, $\LeftCtx \vdash' A;\,\RightCtx$, as required.

        \indcase{11}
        Here we have $\wanp{\mprfraction{\Pi}{A_1},\,A_1,\,(\LeftCtx;\,\RightCtx)}$ with premise
        \[
            \wanp{\Pi,\,(A_1 \wedge A_2),\,(\LeftCtx;\,\RightCtx)}
        \]
        By the IH, $\LeftCtx \vdash' A_1 \wedge A_2;\,\RightCtx$. By Definition~\ref{def:logic-single-subject-appendix}, there exists a term $M$ such that
        \begin{equation}
            \label{proof:2int-to-typing-case11-eq1}
            (x_B:B)_{B\in \LeftCtx} \types M:A_1 \wedge A_2,\,(y_C:C)_{C\in \RightCtx}
        \end{equation}
        Applying the projection rule on the right to \eqref{proof:2int-to-typing-case11-eq1}, we obtain
        \[
            (x_B:B)_{B\in \LeftCtx} \types \pi_1(M):A_1,\,(y_C:C)_{C\in \RightCtx}
        \]
        Therefore, by Definition~\ref{def:logic-single-subject-appendix}, $\LeftCtx \vdash' A_1;\,\RightCtx$, as required.

        \indcase{12}
        This is analogous \camrevreplace{with}{to} case (11).

        \indcase{13}
        Here we have
        \[
            \wanp{\mprfraction{\Pi\;\;\Pi_1\;\;\Pi_2}{C},\,C,\,(\LeftCtx \cup (\LeftCtx_1 \setminus \{A_1\}) \cup (\LeftCtx_2 \setminus \{A_2\});\,\RightCtx \cup \RightCtx_1 \cup \RightCtx_2)}
        \]
        with premises
        \begin{subequations}
            \begin{equation}
                \label{proof:2int-to-typing-case13-eq1}
                \wanp{\Pi,\,(A_1 \vee A_2),\,(\LeftCtx;\,\RightCtx)}
            \end{equation}
            \begin{equation}
                \label{proof:2int-to-typing-case13-eq2}
                \wanp{\Pi_1,\,C,\,(\LeftCtx_1;\,\RightCtx_1)}
            \end{equation}
            \begin{equation}
                \label{proof:2int-to-typing-case13-eq3}
                \wanp{\Pi_2,\,C,\,(\LeftCtx_2;\,\RightCtx_2)}.
            \end{equation}
        \end{subequations}
        Applying the IH to the three premises yields $\LeftCtx \vdash' A_1 \vee A_2;\,\RightCtx$, $\LeftCtx_1 \vdash' C;\,\RightCtx_1$, and $\LeftCtx_2 \vdash' C;\,\RightCtx_2$, respectively. By Definition~\ref{def:logic-single-subject-appendix}, there exist terms $M$, $N$, and $P$ such that
        \begin{subequations}
            \begin{equation}
                \label{proof:2int-to-typing-case13-eq4}
                (x_B:B)_{B\in \LeftCtx} \types M:A_1 \vee A_2,\,(y_D:D)_{D\in \RightCtx}
            \end{equation}
            \begin{equation}
                \label{proof:2int-to-typing-case13-eq5}
                (x_B:B)_{B\in \LeftCtx_1} \types N:C,\,(y_D:D)_{D\in \RightCtx_1}
            \end{equation}
            \begin{equation}
                \label{proof:2int-to-typing-case13-eq6}
                (x_B:B)_{B\in \LeftCtx_2} \types P:C,\,(y_D:D)_{D\in \RightCtx_2}.
            \end{equation}
        \end{subequations}
        Let $\Gamma=(x_B:B)_{B\in \LeftCtx \cup (\LeftCtx_1\setminus\{A_1\}) \cup (\LeftCtx_2\setminus\{A_2\})}$ and $\Delta=(y_D:D)_{D\in \RightCtx \cup \RightCtx_1 \cup \RightCtx_2}$. Moreover, choose fresh variables $u$ and $v$ such that $u,v \notin \fv(\Gamma,\Delta,M,N,P)\camrevadd{\cup\{x_{A_1},x_{A_2}\}}$ and $u \neq v$.

        By admissible weakening on \eqref{proof:2int-to-typing-case13-eq4}, we obtain
        \begin{equation}
            \label{proof:2int-to-typing-case13-eq7}
            \Gamma \types M:A_1 \vee A_2,\,\Delta
        \end{equation}

        For the second branch, we distinguish two cases.
        \begin{itemize}
            \item If $A_1 \in \LeftCtx_1$, then $(x_B:B)_{B\in \LeftCtx_1}$ is of the form $\Gamma_1,\,x_{A_1}:A_1$ where $\Gamma_1=(x_B:B)_{B\in \LeftCtx_1\setminus\{A_1\}}$.
                  By the renaming lemma, clause \textup{(i)}, with
                  \(x=x_{A_1}\) and \(y=u\), we obtain
                  \[
                      \Gamma_1,\,u:A_1
                      \types
                      \Subst{N}{u}{x_{A_1}}:C,\,
                      (y_D:D)_{D\in\RightCtx_1}.
                  \]
                  \camrevadd{By admissible weakening, this gives}
                  \[
                      \Gamma,\,u:A_1 \types \camrevreplace{N'}{\Subst{N}{u}{x_{A_1}}}:C,\,\Delta.
                  \]

            \item If $A_1 \notin \LeftCtx_1$, then $(x_B:B)_{B\in \LeftCtx_1} \subseteq \Gamma$. Hence, by admissible weakening applied to \eqref{proof:2int-to-typing-case13-eq5}, we obtain
                  \[
                      \Gamma,\,u:A_1 \types N:C,\,\Delta.
                  \]
        \end{itemize}
        Thus, in either case, there exists a term $N_0$ such that
        \begin{equation}
            \label{proof:2int-to-typing-case13-eq8}
            \Gamma,\,u:A_1 \types N_0:C,\,\Delta
        \end{equation}

        Similarly, for the third branch, we distinguish two cases.
        \begin{itemize}
            \item \revadd{If $A_2\in\LeftCtx_2$, then }\revreplace{By}{by} the renaming lemma, clause \textup{(i)}, with
                  \(x=x_{A_2}\) and \(y=v\), we obtain
                  \[
                      \Gamma_2,\,v:A_2
                      \types
                      \Subst{P}{v}{x_{A_2}}:C,\,
                      (y_D:D)_{D\in\RightCtx_2},
                  \]
                  where
                  \[
                      \Gamma_2=(x_B:B)_{B\in\LeftCtx_2\setminus\{A_2\}}.
                  \]
                  \revadd{Weakening extends this derivation to the common contexts $\Gamma,v:A_2$ and $\Delta$.}

            \item If $A_2 \notin \LeftCtx_2$, then by admissible weakening applied to \eqref{proof:2int-to-typing-case13-eq6}, we obtain
                  \[
                      \Gamma,\,v:A_2 \types P:C,\,\Delta.
                  \]
        \end{itemize}
        Thus, in either case, there exists a term $P_0$ such that
        \begin{equation}
            \label{proof:2int-to-typing-case13-eq9}
            \Gamma,\,v:A_2 \types P_0:C,\,\Delta
        \end{equation}
        Since $u,v \notin \fv(\Gamma,\Delta)$, we may apply \rlnm{CaseR} to \eqref{proof:2int-to-typing-case13-eq7}, \eqref{proof:2int-to-typing-case13-eq8}, and \eqref{proof:2int-to-typing-case13-eq9}, obtaining
        \[
            \camrevreplacegathered{
                \Gamma \types \mf{\casetm{M}{N_0}{P_0}:C},\,\Delta
            }{
                \Gamma \types \mf{\mathsf{case}\,M\,\mathsf{of}\,\{\,\inj_1(u)\mapsto N_0;\,\inj_2(v)\mapsto P_0\,\}:C},\,\Delta
            }
        \]
        That is,
        \[
            \camrevreplacegathered{
                (x_B:B)_{B\in \LeftCtx \cup (\LeftCtx_1\setminus\{A_1\}) \cup (\LeftCtx_2\setminus\{A_2\})}
                \types
                \mf{\casetm{M}{N_0}{P_0}:C},
                (y_D:D)_{D\in \RightCtx \cup \RightCtx_1 \cup \RightCtx_2}
            }{
                \begin{aligned}
                     & (x_B:B)_{B\in \LeftCtx \cup (\LeftCtx_1\setminus\{A_1\}) \cup (\LeftCtx_2\setminus\{A_2\})}
                    \\
                     & \quad\types \mf{\mathsf{case}\,M\,\mathsf{of}\,\{\,\inj_1(u)\mapsto N_0;\,\inj_2(v)\mapsto P_0\,\}:C},
                    \\
                     & \quad (y_D:D)_{D\in \RightCtx \cup \RightCtx_1 \cup \RightCtx_2}
                \end{aligned}
            }
        \]
        Therefore, by Definition~\ref{def:logic-single-subject-appendix}, $\LeftCtx \cup (\LeftCtx_1\setminus\{A_1\}) \cup (\LeftCtx_2\setminus\{A_2\}) \vdash' C;\,\RightCtx \cup \RightCtx_1 \cup \RightCtx_2$, as required.

        \indcase{14}
        Here we have
        \[
            \wanp{\mprfraction{\Pi_1\quad\Pi_2}{A_2},\,A_2,\,(\LeftCtx_1 \cup \LeftCtx_2;\,\RightCtx_1 \cup \RightCtx_2)}
        \]
        with premises
        \begin{subequations}
            \begin{equation}
                \label{proof:2int-to-typing-case14-eq1}
                \wanp{\Pi_1,\,A_1,\,(\LeftCtx_1;\,\RightCtx_1)}
            \end{equation}
            \begin{equation}
                \label{proof:2int-to-typing-case14-eq2}
                \wanp{\Pi_2,\,(A_1 \to A_2),\,(\LeftCtx_2;\,\RightCtx_2)}.
            \end{equation}
        \end{subequations}
        Applying the IH to the two premises yields $\LeftCtx_1 \vdash' A_1;\,\RightCtx_1$ and $\LeftCtx_2 \vdash' A_1 \to A_2;\,\RightCtx_2$, respectively. By Definition~\ref{def:logic-single-subject-appendix}, there exist terms $M$ and $N$ such that
        \begin{subequations}
            \begin{equation}
                \label{proof:2int-to-typing-case14-eq3}
                (x_B:B)_{B\in \LeftCtx_1} \types M:A_1,\,(y_C:C)_{C\in \RightCtx_1}
            \end{equation}
            \begin{equation}
                \label{proof:2int-to-typing-case14-eq4}
                (x_B:B)_{B\in \LeftCtx_2} \types N:A_1 \to A_2,\,(y_C:C)_{C\in \RightCtx_2}.
            \end{equation}
        \end{subequations}
        Let $\Gamma=(x_B:B)_{B\in \LeftCtx_1 \cup \LeftCtx_2}$ and $\Delta=(y_C:C)_{C\in \RightCtx_1 \cup \RightCtx_2}$. Since $(x_B:B)_{B\in \LeftCtx_1} \subseteq \Gamma$ and $(x_B:B)_{B\in \LeftCtx_2} \subseteq \Gamma$, admissible weakening applied to \eqref{proof:2int-to-typing-case14-eq3} and \eqref{proof:2int-to-typing-case14-eq4} yields
        \begin{subequations}
            \begin{equation}
                \label{proof:2int-to-typing-case14-eq5}
                \Gamma \types M:A_1,\,\Delta
            \end{equation}
            \begin{equation}
                \label{proof:2int-to-typing-case14-eq6}
                \Gamma \types N:A_1 \to A_2,\,\Delta.
            \end{equation}
        \end{subequations}
        Applying \rlnm{AppR} to \eqref{proof:2int-to-typing-case14-eq6} and \eqref{proof:2int-to-typing-case14-eq5}, we obtain
        \[
            \Gamma \types N\,M:A_2,\,\Delta.
        \]
        That is, $(x_B:B)_{B\in \LeftCtx_1 \cup \LeftCtx_2} \types N\,M:A_2,\,(y_C:C)_{C\in \RightCtx_1 \cup \RightCtx_2}$. By Definition~\ref{def:logic-single-subject-appendix}, this means $\LeftCtx_1 \cup \LeftCtx_2 \vdash' A_2;\,\RightCtx_1 \cup \RightCtx_2$, as required.

        \indcase{15}
        Here we have $\wanp{\mprfraction{\Pi}{\revreplace{A_1}{A_2}},\,\revreplace{A_1}{A_2},\,(\LeftCtx;\,\RightCtx)}$ with premise
        \begin{equation}
            \label{proof:2int-to-typing-case15-eq1}
            \wanp{\Pi,\,(A_1 \revreplace{\coto}{\cotol} A_2),\,(\LeftCtx;\,\RightCtx)}
        \end{equation}
        Applying the IH to \eqref{proof:2int-to-typing-case15-eq1} yields $\LeftCtx \vdash' A_1 \revreplace{\coto}{\cotol} A_2;\,\RightCtx$. By Definition~\ref{def:logic-single-subject-appendix}, there exists a term $M$ such that
        \begin{equation}
            \label{proof:2int-to-typing-case15-eq2}
            (x_B:B)_{B\in \LeftCtx} \types M:A_1 \revreplace{\coto}{\cotot} A_2,\,(y_C:C)_{C\in \RightCtx}
        \end{equation}
        Applying \rlnm{Iota2R} to \eqref{proof:2int-to-typing-case15-eq2}, we obtain
        \[
            (x_B:B)_{B\in \LeftCtx} \types \iota_2(M):A_2,\,(y_C:C)_{C\in \RightCtx}
        \]
        Therefore, by Definition~\ref{def:logic-single-subject-appendix}, $\LeftCtx \vdash' A_2;\,\RightCtx$, as required.

        \indcase{33}
        Here we have
        \[
            \wanp{\mprfraction{\Pi}{\sneg{A}},\,\sneg{A},\,(\LeftCtx;\,\RightCtx)}
        \]
        with premise
        \begin{equation}
            \label{proof:2int-to-typing-case33-eq1}
            \wandp{\Pi,\,A,\,(\LeftCtx;\,\RightCtx)}.
        \end{equation}
        Applying the mutual induction hypothesis to
        \eqref{proof:2int-to-typing-case33-eq1} yields
        \[
            \LeftCtx;\,A \vdash' \RightCtx.
        \]
        By Definition~\ref{def:logic-single-subject-appendix}, there exists a term $M$
        such that
        \begin{equation}
            \label{proof:2int-to-typing-case33-eq2}
            (x_B:B)_{B\in \LeftCtx},\,M:A \types (y_C:C)_{C\in \RightCtx}.
        \end{equation}
        Applying \rlnm{$\sneg_{I}$ R} to
        \eqref{proof:2int-to-typing-case33-eq2}, we obtain
        \[
            (x_B:B)_{B\in \LeftCtx} \types M:\sneg{A},\,(y_C:C)_{C\in \RightCtx}.
        \]
        Therefore, by Definition~\ref{def:logic-single-subject-appendix},
        $\LeftCtx \vdash' \sneg{A};\,\RightCtx$, as required.

        \indcase{34}
        Here we have
        \[
            \wanp{\mprfraction{\Pi}{A},\,A,\,(\LeftCtx;\,\RightCtx)}
        \]
        with premise
        \begin{equation}
            \label{proof:2int-to-typing-case34-eq1}
            \wandp{\Pi,\,\sneg{A},\,(\LeftCtx;\,\RightCtx)}.
        \end{equation}
        Applying the mutual induction hypothesis to
        \eqref{proof:2int-to-typing-case34-eq1} yields
        \[
            \LeftCtx;\,\sneg{A} \vdash' \RightCtx.
        \]
        By Definition~\ref{def:logic-single-subject-appendix}, there exists a term $M$
        such that
        \begin{equation}
            \label{proof:2int-to-typing-case34-eq2}
            (x_B:B)_{B\in \LeftCtx},\,M:\sneg{A} \types (y_C:C)_{C\in \RightCtx}.
        \end{equation}
        Applying \rlnm{$\sneg_{E}$ R} to
        \eqref{proof:2int-to-typing-case34-eq2}, we obtain
        \[
            (x_B:B)_{B\in \LeftCtx} \types M:A,\,(y_C:C)_{C\in \RightCtx}.
        \]
        Therefore, by Definition~\ref{def:logic-single-subject-appendix},
        $\LeftCtx \vdash' A;\,\RightCtx$, as required.
    \end{indproof}
    \smallskip
    \noindent
    For claim (ii), assume $\mathrm{DP}(\Pi,A(L\camrevreplace{,}{;}\,R))$. Again we proceed by cases on the last rule.
    Whenever a premise is of the form $\mathrm{P}(\Pi,B,(L'\camrevreplace{,}{;}\,R'))$, we appeal to the
    mutual induction hypothesis for claim (i).
    \begin{indproof}
        \indcase{2} Here $\wandp{\dbar{A},\,A,\,(\varnothing;\,\{A\})}$ with $\LeftCtx=\varnothing$ and $\RightCtx=\{A\}$, so it remains to show that $\varnothing;\,A \vdash' A$. By Definition~\ref{def:logic-single-subject-appendix}, this means showing that there \revreplace{exist}{exists} a term $M$ such that $M:A \types y_A:A$. But $y_A:A \types y_A:A$ by \rlnm{ID}, so the result follows by taking $M=y_A$.
        \indcase{4} Here $\wandp{\dbar{\falsefm},\,\falsefm,\,(\varnothing;\,\varnothing)}$ with $\LeftCtx=\varnothing$ and $\RightCtx=\varnothing$, so it remains to show that $\varnothing;\,\falsefm \vdash' \varnothing$. By Definition~\ref{def:logic-single-subject-appendix}, this means showing that there exists a term $M$ such that $M:\falsefm \types$. But $\unittm:\falsefm \types$ by the $\falsefm$L rule, so the result follows by taking $M=\unittm$.
        \indcase{16}
        Here we have $\wandp{\dblmprfraction{\Pi}{A_1},\,\revreplace{A_2}{A_1},\,(\LeftCtx;\,\RightCtx)}$ with premise
        \begin{equation}
            \label{proof:2int-to-typing-case16-eq1}
            \wanp{\Pi,\,(A_1 \cotol A_2),\,(\LeftCtx;\,\RightCtx)}
        \end{equation}
        Applying the mutual induction hypothesis to \eqref{proof:2int-to-typing-case16-eq1} yields $\LeftCtx \vdash' A_1 \cotol A_2;\,\RightCtx$. By Definition~\ref{def:logic-single-subject-appendix}, there exists a term $M$ such that
        \begin{equation}
            \label{proof:2int-to-typing-case16-eq2}
            (x_B:B)_{B\in \LeftCtx} \types M:A_1 \cotot A_2,\,(y_C:C)_{C\in \RightCtx}
        \end{equation}
        Applying \rlnm{Iota1L} to \eqref{proof:2int-to-typing-case16-eq2}, we obtain
        \[
            (x_B:B)_{B\in \LeftCtx},\,\iota_1(M):A_1 \types (y_C:C)_{C\in \RightCtx}
        \]
        Therefore, by Definition~\ref{def:logic-single-subject-appendix}, $\LeftCtx;\,A_1 \vdash' \RightCtx$, as required.

        \indcase{17}
        Here we have $\wandp{\dblmprfraction{\Pi}{A_1 \wedge A_2},\,(A_1 \wedge A_2),\,(\LeftCtx;\,\RightCtx)}$ with premise
        \begin{equation}
            \label{proof:2int-to-typing-case17-eq1}
            \wandp{\Pi,\,A_1,\,(\LeftCtx;\,\RightCtx)}
        \end{equation}
        Applying the IH to \eqref{proof:2int-to-typing-case17-eq1} yields $\LeftCtx;\,A_1 \vdash' \RightCtx$. By Definition~\ref{def:logic-single-subject-appendix}, there exists a term $M$ such that
        \begin{equation}
            \label{proof:2int-to-typing-case17-eq2}
            (x_B:B)_{B\in \LeftCtx},\,M:A_1 \types (y_C:C)_{C\in \RightCtx}
        \end{equation}
        Applying $\wedge$L$_1$ to \eqref{proof:2int-to-typing-case17-eq2}, we obtain
        \[
            (x_B:B)_{B\in \LeftCtx},\,\inj_1(M):A_1 \wedge A_2 \types (y_C:C)_{C\in \RightCtx}
        \]
        Therefore, by Definition~\ref{def:logic-single-subject-appendix}, $\LeftCtx;\,A_1 \wedge A_2 \vdash' \RightCtx$, as required.

        \revadd{\indcase{18}
            This is analogous to case (17), using $\wedge$L$_2$ and $\inj_2$ in place of $\wedge$L$_1$ and $\inj_1$.}

        \indcase{19}
        Here we have $\wandp{\dblmprfraction{\Pi_1 \quad \Pi_2}{A_1 \vee A_2},\,(A_1 \vee A_2),\,(\LeftCtx_1 \cup \LeftCtx_2;\,\RightCtx_1 \cup \RightCtx_2)}$
        with premises
        \begin{subequations}
            \begin{equation}
                \label{proof:2int-to-typing-case19-eq1}
                \wandp{\Pi_1,\,A_1,\,(\LeftCtx_1;\,\RightCtx_1)}
            \end{equation}
            \begin{equation}
                \label{proof:2int-to-typing-case19-eq2}
                \wandp{\Pi_2,\,A_2,\,(\LeftCtx_2;\,\RightCtx_2)}.
            \end{equation}
        \end{subequations}
        Applying the IH to the two premises yields $\LeftCtx_1;\,A_1 \vdash' \RightCtx_1$ and $\LeftCtx_2;\,A_2 \vdash' \RightCtx_2$, respectively. By Definition~\ref{def:logic-single-subject-appendix}, there exist terms $M$ and $N$ such that
        \begin{subequations}
            \begin{equation}
                \label{proof:2int-to-typing-case19-eq3}
                (x_B:B)_{B\in \LeftCtx_1},\,M:A_1 \types (y_C:C)_{C\in \RightCtx_1}
            \end{equation}
            \begin{equation}
                \label{proof:2int-to-typing-case19-eq4}
                (x_B:B)_{B\in \LeftCtx_2},\,N:A_2 \types (y_C:C)_{C\in \RightCtx_2}.
            \end{equation}
        \end{subequations}
        Let $\Gamma=(x_B:B)_{B\in \LeftCtx_1 \cup \LeftCtx_2}$ and $\Delta=(y_C:C)_{C\in \RightCtx_1 \cup \RightCtx_2}$. Since $(x_B:B)_{B\in \LeftCtx_1} \subseteq \Gamma$ and $(x_B:B)_{B\in \LeftCtx_2} \subseteq \Gamma$, admissible weakening applied to \eqref{proof:2int-to-typing-case19-eq3} and \eqref{proof:2int-to-typing-case19-eq4} yields
        \begin{subequations}
            \begin{equation}
                \label{proof:2int-to-typing-case19-eq5}
                \Gamma,\,M:A_1 \types \Delta
            \end{equation}
            \begin{equation}
                \label{proof:2int-to-typing-case19-eq6}
                \Gamma,\,N:A_2 \types \Delta.
            \end{equation}
        \end{subequations}
        Applying \rlnm{$\vee$L} to \eqref{proof:2int-to-typing-case19-eq5} and \eqref{proof:2int-to-typing-case19-eq6}, we obtain
        \[
            \Gamma,\,(M,\,N):A_1 \vee A_2 \types \Delta.
        \]
        That is,
        \[
            (x_B:B)_{B\in \LeftCtx_1 \cup \LeftCtx_2},\,(M,\,N):A_1 \vee A_2 \types (y_C:C)_{C\in \RightCtx_1 \cup \RightCtx_2}
        \]
        Therefore, by Definition~\ref{def:logic-single-subject-appendix}, $\LeftCtx_1 \cup \LeftCtx_2;\,A_1 \vee A_2 \vdash' \RightCtx_1 \cup \RightCtx_2$, as required.

        \indcase{20}
        Here we have $\wandp{\dblmprfraction{\Pi_1 \quad \Pi_2}{A_1 \to A_2},\,(A_1 \to A_2),\,(\LeftCtx_1 \cup \LeftCtx_2;\,\RightCtx_1 \cup \RightCtx_2)}$
        with premises
        \begin{subequations}
            \begin{equation}
                \label{proof:2int-to-typing-case20-eq1}
                \wanp{\Pi_1,\,A_1,\,(\LeftCtx_1;\,\RightCtx_1)}
            \end{equation}
            \begin{equation}
                \label{proof:2int-to-typing-case20-eq2}
                \wandp{\Pi_2,\,A_2,\,(\LeftCtx_2;\,\RightCtx_2)}.
            \end{equation}
        \end{subequations}
        Applying the IH to the two premises yields $\LeftCtx_1 \vdash' A_1;\,\RightCtx_1$ and $\LeftCtx_2;\,A_2 \vdash' \RightCtx_2$, respectively. By Definition~\ref{def:logic-single-subject-appendix}, there exist terms $M$ and $N$ such that
        \begin{subequations}
            \begin{equation}
                \label{proof:2int-to-typing-case20-eq3}
                (x_B:B)_{B\in \LeftCtx_1} \types M:A_1,\,(y_C:C)_{C\in \RightCtx_1}
            \end{equation}
            \begin{equation}
                \label{proof:2int-to-typing-case20-eq4}
                (x_B:B)_{B\in \LeftCtx_2},\,N:A_2 \types (y_C:C)_{C\in \RightCtx_2}.
            \end{equation}
        \end{subequations}
        Let $\Gamma=(x_B:B)_{B\in \LeftCtx_1 \cup \LeftCtx_2}$ and $\Delta=(y_C:C)_{C\in \RightCtx_1 \cup \RightCtx_2}$. Since $(x_B:B)_{B\in \LeftCtx_1} \subseteq \Gamma$ and $(x_B:B)_{B\in \LeftCtx_2} \subseteq \Gamma$, admissible weakening applied to \eqref{proof:2int-to-typing-case20-eq3} and \eqref{proof:2int-to-typing-case20-eq4} yields
        \begin{subequations}
            \begin{equation}
                \label{proof:2int-to-typing-case20-eq5}
                \Gamma \types M:A_1,\,\Delta
            \end{equation}
            \begin{equation}
                \label{proof:2int-to-typing-case20-eq6}
                \Gamma,\,N:A_2 \types \Delta.
            \end{equation}
        \end{subequations}
        Applying \rlnm{$\to$L} to \eqref{proof:2int-to-typing-case20-eq5} and \eqref{proof:2int-to-typing-case20-eq6}, we obtain
        \[
            \Gamma,\,M\cdot N:A_1 \to A_2 \types \Delta.
        \]
        That is,
        \[
            (x_B:B)_{B\in \LeftCtx_1 \cup \LeftCtx_2},\,M\cdot N:A_1 \to A_2 \types (y_C:C)_{C\in \RightCtx_1 \cup \RightCtx_2}
        \]
        Therefore, by Definition~\ref{def:logic-single-subject-appendix}, $\LeftCtx_1 \cup \LeftCtx_2;\,A_1 \to A_2 \vdash' \RightCtx_1 \cup \RightCtx_2$, as required.

        \[
            \infer[$\cotot$L]{ \Gamma,\,M:A \types y:B,\,\Delta }{ \Gamma,\,\mf{\abs{y}{M}:B \cotot A} \types \Delta }\;y \notin \fv(\Gamma,\,\Delta)
        \]

        \indcase{21}
        Here we have $\wandp{\dblmprfraction{\Pi}{A_1 \cotol A_2},\,(A_1 \cotol A_2),\,(\LeftCtx;\,\RightCtx \setminus \{A_1\})}$
        with premise
        \begin{equation}
            \label{proof:2int-to-typing-case21-eq1}
            \wandp{\Pi,\,A_2,\,(\LeftCtx;\,\RightCtx)}
        \end{equation}
        Applying the IH to \eqref{proof:2int-to-typing-case21-eq1} yields $\LeftCtx;\,A_2 \vdash' \RightCtx$. By Definition~\ref{def:logic-single-subject-appendix}, there exists a term $M$ such that
        \begin{equation}
            \label{proof:2int-to-typing-case21-eq2}
            (x_B:B)_{B\in \LeftCtx},\,M:A_2 \types (y_C:C)_{C\in \RightCtx}
        \end{equation}
        Let $\Gamma=(x_B:B)_{B\in \LeftCtx}$ and $\Delta=(y_C:C)_{C\in \RightCtx\setminus\{A_1\}}$. Since $(y_C:C)_{C\in \RightCtx} \subseteq y_{A_1}:A_1,\,\Delta$, admissible weakening applied to \eqref{proof:2int-to-typing-case21-eq2} yields
        \begin{equation}
            \label{proof:2int-to-typing-case21-eq3}
            \Gamma,\,M:A_2 \types y_{A_1}:A_1,\,\Delta
        \end{equation}
        Moreover, by the choice of the distinguished variables in Definition~\ref{def:logic-single-subject-appendix}, $\revreplace{y_{A_2}}{y_{A_1}}\notin \fv(\Gamma,\Delta)$. We may therefore apply \rlnm{$\revreplace{\coto}{\cotot}$L} to \eqref{proof:2int-to-typing-case21-eq3} and obtain
        \[
            \Gamma,\,\mf{\abs{y_{A_1}}{M}:A_1 \cotot A_2} \types \Delta.
        \]
        That is,
        \[
            (x_B:B)_{B\in \LeftCtx},\,\mf{\abs{y_{A_1}}{M}:A_1 \cotot A_2} \types (y_C:C)_{C\in \RightCtx\setminus\{A_1\}}
        \]
        Therefore, by Definition~\ref{def:logic-single-subject-appendix}, $\LeftCtx;\,A_1 \cotol A_2 \vdash' \RightCtx\setminus\{A_1\}$, as required.

        \indcase{22}
        Here we have $\wandp{\dblmprfraction{\Pi}{A},\,A,\,(\LeftCtx;\,\RightCtx)}$ with premise
        \[
            \wandp{\Pi,\,\truefm,\,(\LeftCtx;\,\RightCtx)}
        \]
        By the IH, $\LeftCtx;\,\truefm \vdash' \RightCtx$. By Definition~\ref{def:logic-single-subject-appendix}, there exists a term $M$ such that
        \begin{equation}
            \label{proof:2int-to-typing-case22-eq1}
            (x_B:B)_{B\in \LeftCtx},\,M:\truefm \types (y_C:C)_{C\in \RightCtx}
        \end{equation}
        Applying $\rlnm{$\truefm$L\text{-elim}}$ to \eqref{proof:2int-to-typing-case22-eq1}, we obtain
        \[
            (x_B:B)_{B\in \LeftCtx},\,\revreplace{\trueelim(M)}{\abort(M)}:A \types (y_C:C)_{C\in \RightCtx}
        \]
        Therefore, by Definition~\ref{def:logic-single-subject-appendix}, $\LeftCtx;\,A \vdash' \RightCtx$, as required.

        \indcase{23}
        Here we have
        \[
            \wandp{\dblmprfraction{\Pi\;\;\Pi_1\;\;\Pi_2}{C},\,C,\,(\LeftCtx \cup \LeftCtx_1 \cup \LeftCtx_2;\,\RightCtx \cup (\RightCtx_1 \setminus \{A_1\}) \cup (\RightCtx_2 \setminus \{A_2\}))}
        \]
        with premises
        \begin{subequations}
            \begin{equation}
                \label{proof:2int-to-typing-case23-eq1}
                \wandp{\Pi,\,(A_1 \wedge A_2),\,(\LeftCtx;\,\RightCtx)}
            \end{equation}
            \begin{equation}
                \label{proof:2int-to-typing-case23-eq2}
                \wandp{\Pi_1,\,C,\,(\LeftCtx_1;\,\RightCtx_1)}
            \end{equation}
            \begin{equation}
                \label{proof:2int-to-typing-case23-eq3}
                \wandp{\Pi_2,\,C,\,(\LeftCtx_2;\,\RightCtx_2)}.
            \end{equation}
        \end{subequations}
        Applying the IH to the three premises yields $\LeftCtx;\,A_1 \wedge A_2 \vdash' \RightCtx$, $\LeftCtx_1;\,C \vdash' \RightCtx_1$, and $\LeftCtx_2;\,C \vdash' \RightCtx_2$, respectively. By Definition~\ref{def:logic-single-subject-appendix}, there exist terms $M$, $N$, and $P$ such that
        \begin{subequations}
            \begin{equation}
                \label{proof:2int-to-typing-case23-eq4}
                (x_B:B)_{B\in \LeftCtx},\,M:A_1 \wedge A_2 \types (y_D:D)_{D\in \RightCtx}
            \end{equation}
            \begin{equation}
                \label{proof:2int-to-typing-case23-eq5}
                (x_B:B)_{B\in \LeftCtx_1},\,N:C \types (y_D:D)_{D\in \RightCtx_1}
            \end{equation}
            \begin{equation}
                \label{proof:2int-to-typing-case23-eq6}
                (x_B:B)_{B\in \LeftCtx_2},\,P:C \types (y_D:D)_{D\in \RightCtx_2}.
            \end{equation}
        \end{subequations}
        Let $\Gamma=(x_B:B)_{B\in \LeftCtx \cup \LeftCtx_1 \cup \LeftCtx_2}$ and $\Delta=(y_D:D)_{D\in \RightCtx \cup (\RightCtx_1\setminus\{A_1\}) \cup (\RightCtx_2\setminus\{A_2\})}$. Choose fresh variables $u$ and $v$ such that $u,v \notin \fv(\Gamma,\Delta,M,N,P)\camrevadd{\cup\{y_{A_1},y_{A_2}\}}$ and $u\neq v$.

        By admissible weakening on \eqref{proof:2int-to-typing-case23-eq4}, we obtain
        \begin{equation}
            \label{proof:2int-to-typing-case23-eq7}
            \Gamma,\,M:A_1 \wedge A_2 \types \Delta
        \end{equation}

        We now show that the second premise can be brought to the form required by \rlnm{CaseL}. We distinguish two cases.
        \begin{itemize}
            \item If $A_1 \in \RightCtx_1$, then writing $\Delta_1=(y_D:D)_{D\in \RightCtx_1\setminus\{A_1\}}$, we may rewrite \eqref{proof:2int-to-typing-case23-eq5} as $(x_B:B)_{B\in \LeftCtx_1},\,N:C \types y_{A_1}:A_1,\,\Delta_1$. Since $u \notin \fv((x_B:B)_{B\in \LeftCtx_1},\Delta_1,N)$, admissible renaming yields
                  \[
                      (x_B:B)_{B\in \LeftCtx_1},\,\mf{\Subst{N}{u}{y_{A_1}}:C} \types u:A_1,\,\Delta_1.
                  \]
                  Since $(x_B:B)_{B\in \LeftCtx_1} \subseteq \Gamma$ and $\Delta_1 \subseteq \Delta$, admissible weakening gives
                  \[
                      \Gamma,\,\mf{\Subst{N}{u}{y_{A_1}}:C} \types u:A_1,\,\Delta.
                  \]

            \item If $A_1 \notin \RightCtx_1$, then $(y_D:D)_{D\in \RightCtx_1} \subseteq \Delta$, so by admissible weakening applied to \eqref{proof:2int-to-typing-case23-eq5}, we obtain
                  \[
                      \Gamma,\,N:C \types u:A_1,\,\Delta.
                  \]
        \end{itemize}
        Thus, in either case, there exists a term $N_0$ such that
        \begin{equation}
            \label{proof:2int-to-typing-case23-eq8}
            \Gamma,\,N_0:C \types u:A_1,\,\Delta
        \end{equation}

        Similarly, for the third premise, we distinguish two cases.
        \begin{itemize}
            \item If $A_2 \in \RightCtx_2$, then writing $\Delta_2=(y_D:D)_{D\in \RightCtx_2\setminus\{A_2\}}$, we may rewrite \eqref{proof:2int-to-typing-case23-eq6} as $(x_B:B)_{B\in \LeftCtx_2},\,P:C \types y_{A_2}:A_2,\,\Delta_2$. Since $v \notin \fv((x_B:B)_{B\in \LeftCtx_2},\Delta_2,P)$, admissible renaming yields
                  \[
                      (x_B:B)_{B\in \LeftCtx_2},\,\mf{\Subst{P}{v}{y_{A_2}}:C} \types v:A_2,\,\Delta_2.
                  \]
                  Since $(x_B:B)_{B\in \LeftCtx_2} \subseteq \Gamma$ and $\Delta_2 \subseteq \Delta$, admissible weakening gives
                  \[
                      \Gamma,\,\mf{\Subst{P}{v}{y_{A_2}}:C} \types v:A_2,\,\Delta.
                  \]

            \item If $A_2 \notin \RightCtx_2$, then $(y_D:D)_{D\in \RightCtx_2} \subseteq \Delta$, so by admissible weakening applied to \eqref{proof:2int-to-typing-case23-eq6}, we obtain
                  \[
                      \Gamma,\,P:C \types v:A_2,\,\Delta.
                  \]
        \end{itemize}
        Thus, in either case, there exists a term $P_0$ such that
        \begin{equation}
            \label{proof:2int-to-typing-case23-eq9}
            \Gamma,\,P_0:C \types v:A_2,\,\Delta
        \end{equation}

        Since $u,v \notin \fv(\Gamma,\Delta)$, we may apply \rlnm{CaseL} to \eqref{proof:2int-to-typing-case23-eq7}, \eqref{proof:2int-to-typing-case23-eq8}, and \eqref{proof:2int-to-typing-case23-eq9}, obtaining
        \[
            \Gamma,\,\mf{\revreplace{\casetm{M}{N_0}{P_0}}{\mathsf{case}\,M\,\mathsf{of}\,\{\,\inj_1(u)\mapsto N_0;\,\inj_2(v)\mapsto P_0\,\}}:C} \types \Delta.
        \]
        That is,
        \[
            \begin{aligned}
                 & (x_B:B)_{B\in \LeftCtx \cup \LeftCtx_1 \cup \LeftCtx_2},                                                                           \\
                 & \quad \mf{\revreplace{\casetm{M}{N_0}{P_0}}{\mathsf{case}\,M\,\mathsf{of}\,\{\,\inj_1(u)\mapsto N_0;\,\inj_2(v)\mapsto P_0\,\}}:C}
                \types                                                                                                                                \\
                 & \quad (y_D:D)_{D\in \RightCtx \cup (\RightCtx_1\setminus\{A_1\}) \cup (\RightCtx_2\setminus\{A_2\})}
            \end{aligned}
        \]
        Therefore, by Definition~\ref{def:logic-single-subject-appendix}, $\LeftCtx \cup \LeftCtx_1 \cup \LeftCtx_2;\,C \vdash' \RightCtx \cup (\RightCtx_1\setminus\{A_1\}) \cup (\RightCtx_2\setminus\{A_2\})$, as required.

        \indcase{24}
        Here we have $\wandp{\dblmprfraction{\Pi}{A_1},\,A_1,\,(\LeftCtx;\,\RightCtx)}$ with premise
        \begin{equation}
            \label{proof:2int-to-typing-case24-eq1}
            \wandp{\Pi,\,(A_1 \vee A_2),\,(\LeftCtx;\,\RightCtx)}
        \end{equation}
        Applying the IH to \eqref{proof:2int-to-typing-case24-eq1} yields $\LeftCtx;\,A_1 \vee A_2 \vdash' \RightCtx$. By Definition~\ref{def:logic-single-subject-appendix}, there exists a term $M$ such that
        \begin{equation}
            \label{proof:2int-to-typing-case24-eq2}
            (x_B:B)_{B\in \LeftCtx},\,M:A_1 \vee A_2 \types (y_C:C)_{C\in \RightCtx}
        \end{equation}
        Applying \rlnm{PiL} to \eqref{proof:2int-to-typing-case24-eq2}, we obtain
        \[
            (x_B:B)_{B\in \LeftCtx},\,\pi_1(M):A_1 \types (y_C:C)_{C\in \RightCtx}
        \]
        Therefore, by Definition~\ref{def:logic-single-subject-appendix}, $\LeftCtx;\,A_1 \vdash' \RightCtx$, as required.

        \indcase{25}
        This is analogous \camrevreplace{with}{to} case (24).

        \indcase{26}
        Here we have $\wanp{\mprfraction{\Pi}{A_1},\,A_1,\,(\LeftCtx;\,\RightCtx)}$ with premise
        \begin{equation}
            \label{proof:2int-to-typing-case26-eq1}
            \wandp{\Pi,\,A_1 \to A_2,\,(\LeftCtx;\,\RightCtx)}
        \end{equation}
        Applying the IH to \eqref{proof:2int-to-typing-case26-eq1} yields $\LeftCtx;\,A_1 \to A_2 \vdash' \RightCtx$. By Definition~\ref{def:logic-single-subject-appendix}, there exists a term $M$ such that
        \begin{equation}
            \label{proof:2int-to-typing-case26-eq2}
            (x_B:B)_{B\in \LeftCtx},\,M:A_1 \to A_2 \types (y_C:C)_{C\in \RightCtx}
        \end{equation}
        Applying \rlnm{Iota1R} to \eqref{proof:2int-to-typing-case26-eq2}, we obtain
        \[
            (x_B:B)_{B\in \LeftCtx} \types \iota_1(M):A_1,\,(y_C:C)_{C\in \RightCtx}
        \]
        Therefore, by Definition~\ref{def:logic-single-subject-appendix}, $\camrevreplace{\LeftCtx;\,A_1 \vdash' \RightCtx}{\LeftCtx \vdash' A_1;\,\RightCtx}$, as required.

        \indcase{27}
        Here we have $\wandp{\dblmprfraction{\Pi}{A_2},\,A_2,\,(\LeftCtx;\,\RightCtx)}$ with premise
        \begin{equation}
            \label{proof:2int-to-typing-case27-eq1}
            \wandp{\Pi,\,(A_1 \to A_2),\,(\LeftCtx;\,\RightCtx)}.
        \end{equation}
        Applying the IH to \eqref{proof:2int-to-typing-case27-eq1} yields $\LeftCtx;\,A_1 \to A_2 \vdash' \RightCtx$. By Definition~\ref{def:logic-single-subject-appendix}, there exists a term $M$ such that
        \begin{equation}
            \label{proof:2int-to-typing-case27-eq2}
            (x_B:B)_{B\in \LeftCtx},\,M:A_1 \to A_2 \types (y_C:C)_{C\in \RightCtx}.
        \end{equation}
        Applying \rlnm{Iota2L} to \eqref{proof:2int-to-typing-case27-eq2}, we obtain
        \[
            (x_B:B)_{B\in \LeftCtx},\,\iota_2(M):A_2 \types (y_C:C)_{C\in \RightCtx}.
        \]
        Therefore, by Definition~\ref{def:logic-single-subject-appendix}, $\LeftCtx;\,A_2 \vdash' \RightCtx$, as required.

        \indcase{28}
        Here we have $\wandp{\dblmprfraction{\Pi_1\quad\Pi_2}{A_2},\,A_2,\,(\LeftCtx_1 \cup \LeftCtx_2;\,\RightCtx_1 \cup \RightCtx_2)}$
        with premises
        \begin{subequations}
            \begin{equation}
                \label{proof:2int-to-typing-case28-eq1}
                \wandp{\Pi_1,\,(A_1 \cotol A_2),\,(\LeftCtx_1;\,\RightCtx_1)}
            \end{equation}
            \begin{equation}
                \label{proof:2int-to-typing-case28-eq2}
                \wandp{\Pi_2,\,A_1,\,(\LeftCtx_2;\,\RightCtx_2)}.
            \end{equation}
        \end{subequations}
        Applying the IH to the two premises yields $\LeftCtx_1;\,A_1 \cotol A_2 \vdash' \RightCtx_1$ and $\LeftCtx_2;\,A_1 \vdash' \RightCtx_2$, respectively. By Definition~\ref{def:logic-single-subject-appendix}, there exist terms $M$ and $N$ such that
        \begin{subequations}
            \begin{equation}
                \label{proof:2int-to-typing-case28-eq3}
                (x_B:B)_{B\in \LeftCtx_1},\,M:A_1 \cotot A_2 \types (y_C:C)_{C\in \RightCtx_1}
            \end{equation}
            \begin{equation}
                \label{proof:2int-to-typing-case28-eq4}
                (x_B:B)_{B\in \LeftCtx_2},\,N:A_1 \types (y_C:C)_{C\in \RightCtx_2}.
            \end{equation}
        \end{subequations}
        Let $\Gamma=(x_B:B)_{B\in \LeftCtx_1 \cup \LeftCtx_2}$ and $\Delta=(y_C:C)_{C\in \RightCtx_1 \cup \RightCtx_2}$. Since $(x_B:B)_{B\in \LeftCtx_1} \subseteq \Gamma$ and $(x_B:B)_{B\in \LeftCtx_2} \subseteq \Gamma$, admissible weakening applied to \eqref{proof:2int-to-typing-case28-eq3} and \eqref{proof:2int-to-typing-case28-eq4} yields
        \begin{subequations}
            \begin{equation}
                \label{proof:2int-to-typing-case28-eq5}
                \Gamma,\,M:A_1 \cotot A_2 \types \Delta
            \end{equation}
            \begin{equation}
                \label{proof:2int-to-typing-case28-eq6}
                \Gamma,\,N:A_1 \types \Delta.
            \end{equation}
        \end{subequations}
        Applying the \rlnm{AppL} rule to \eqref{proof:2int-to-typing-case28-eq5} and \eqref{proof:2int-to-typing-case28-eq6}, we obtain
        \[
            \Gamma,\,M\,N:A_2 \types \Delta.
        \]
        That is,
        \[
            (x_B:B)_{B\in \LeftCtx_1 \cup \LeftCtx_2},\,M\,N:A_2 \types (y_C:C)_{C\in \RightCtx_1 \cup \RightCtx_2}
        \]
        Therefore, by Definition~\ref{def:logic-single-subject-appendix}, $\LeftCtx_1 \cup \LeftCtx_2;\,A_2 \vdash' \RightCtx_1 \cup \RightCtx_2$, as required.

        \indcase{29}
        Here we have $\wanp{\Pi,\,A,\,(\LeftCtx';\,\RightCtx')}$ with premise
        \begin{subequations}
            \begin{equation}
                \label{proof:2int-to-typing-case29-eq1}
                \wanp{\Pi,\,A,\,(\LeftCtx;\,\RightCtx)}
            \end{equation}
            \begin{equation}
                \label{proof:2int-to-typing-case29-eq2}
                \LeftCtx \subseteq \LeftCtx'
            \end{equation}
            \begin{equation}
                \label{proof:2int-to-typing-case29-eq3}
                \RightCtx \subseteq \RightCtx'
            \end{equation}
        \end{subequations}
        where $\LeftCtx'$ and $\RightCtx'$ are finite. Applying the IH to \eqref{proof:2int-to-typing-case29-eq1} yields $\LeftCtx \vdash' A;\,\RightCtx$. By Definition~\ref{def:logic-single-subject-appendix}, this means there exists a term $M$ such that
        \begin{equation}
            \label{proof:2int-to-typing-case29-eq4}
            (x_B:B)_{B\in \LeftCtx} \types M:A,\,(y_C:C)_{C\in \RightCtx}
        \end{equation}
        Let $\Gamma=(x_B:B)_{B\in \LeftCtx'}$ and $\Delta=(y_C:C)_{C\in \RightCtx'}$. By \eqref{proof:2int-to-typing-case29-eq2} and \eqref{proof:2int-to-typing-case29-eq3}, we have $(x_B:B)_{B\in \LeftCtx} \subseteq \Gamma$ and $(y_C:C)_{C\in \RightCtx} \subseteq \Delta$. Hence, by admissible weakening applied to \eqref{proof:2int-to-typing-case29-eq4}, we obtain
        \[
            \Gamma \types M:A,\,\Delta.
        \]
        That is,
        \[
            (x_B:B)_{B\in \LeftCtx'} \types M:A,\,(y_C:C)_{C\in \RightCtx'}
        \]
        Therefore, by Definition~\ref{def:logic-single-subject-appendix}, $\LeftCtx' \vdash' A;\,\RightCtx'$, as required.
        \indcase{30}
        This is analogous to case (29).

        \indcase{31}
        Here we have
        \[
            \wandp{\dblmprfraction{\Pi}{\sneg{A}},\,\sneg{A},\,(\LeftCtx;\,\RightCtx)}
        \]
        with premise
        \begin{equation}
            \label{proof:2int-to-typing-case31-eq1}
            \wanp{\Pi,\,A,\,(\LeftCtx;\,\RightCtx)}.
        \end{equation}
        Applying the mutual induction hypothesis to
        \eqref{proof:2int-to-typing-case31-eq1} yields
        \[
            \LeftCtx \vdash' A;\,\RightCtx.
        \]
        By Definition~\ref{def:logic-single-subject-appendix}, there exists a term $M$
        such that
        \begin{equation}
            \label{proof:2int-to-typing-case31-eq2}
            (x_B:B)_{B\in \LeftCtx} \types M:A,\,(y_C:C)_{C\in \RightCtx}.
        \end{equation}
        Applying \rlnm{$\sneg_{I}$ L} to
        \eqref{proof:2int-to-typing-case31-eq2}, we obtain
        \[
            (x_B:B)_{B\in \LeftCtx},\,M:\sneg{A} \types (y_C:C)_{C\in \RightCtx}.
        \]
        Therefore, by Definition~\ref{def:logic-single-subject-appendix},
        $\LeftCtx;\,\sneg{A} \vdash' \RightCtx$, as required.

        \indcase{32}
        Here we have
        \[
            \wandp{\dblmprfraction{\Pi}{A},\,A,\,(\LeftCtx;\,\RightCtx)}
        \]
        with premise
        \begin{equation}
            \label{proof:2int-to-typing-case32-eq1}
            \wanp{\Pi,\,\sneg{A},\,(\LeftCtx;\,\RightCtx)}.
        \end{equation}
        Applying the mutual induction hypothesis to
        \eqref{proof:2int-to-typing-case32-eq1} yields
        \[
            \LeftCtx \vdash' \sneg{A};\,\RightCtx.
        \]
        By Definition~\ref{def:logic-single-subject-appendix}, there exists a term $M$
        such that
        \begin{equation}
            \label{proof:2int-to-typing-case32-eq2}
            (x_B:B)_{B\in \LeftCtx} \types M:\sneg{A},\,(y_C:C)_{C\in \RightCtx}.
        \end{equation}
        Applying \rlnm{$\sneg_{E}$ L} to
        \eqref{proof:2int-to-typing-case32-eq2}, we obtain
        \[
            (x_B:B)_{B\in \LeftCtx},\,M:A \types (y_C:C)_{C\in \RightCtx}.
        \]
        Therefore, by Definition~\ref{def:logic-single-subject-appendix},
        $\LeftCtx;\,A \vdash' \RightCtx$, as required.
    \end{indproof}
\end{proof}

\begin{lemma}[Strong-negation shifts in Wansing contexts]
    \label{lem:wansing-context-negation-shift}
    Let $L$ and $R$ be finite sets of formulas. The following transformations
    are admissible for proofs:
    \begin{enumerate}[(a)]
        \item If $\wanp{\Pi,\,C,\,(L\cup\{B\};\,R)}$, then, for some $\Pi'$,
              \[
                  \wanp{\Pi',\,C,\,(L;\,R\cup\{\sneg{B}\})}.
              \]
        \item If $\wanp{\Pi,\,C,\,(L\cup\{\sneg{B}\};\,R)}$, then, for some
              $\Pi'$,
              \[
                  \wanp{\Pi',\,C,\,(L;\,R\cup\{B\})}.
              \]
        \item If $\wanp{\Pi,\,C,\,(L;\,R\cup\{B\})}$, then, for some $\Pi'$,
              \[
                  \wanp{\Pi',\,C,\,(L\cup\{\sneg{B}\};\,R)}.
              \]
        \item If $\wanp{\Pi,\,C,\,(L;\,R\cup\{\sneg{B}\})}$, then, for some
              $\Pi'$,
              \[
                  \wanp{\Pi',\,C,\,(L\cup\{B\};\,R)}.
              \]
    \end{enumerate}
    The same four transformations are admissible for dual proofs, with every
    judgement $\wanp{-}$ replaced by the corresponding judgement
    $\wandp{-}$.
\end{lemma}

\camrevadd{
    \begin{proof}
        First consider clauses (a) and (b) for proofs. In (a), apply implication introduction (8) to the given proof of \(C\), discharging \(B\), to obtain a proof of \(B\to C\) under \((L\setminus\{B\};R)\). The co-assumption axiom (2), followed by rule (34), gives a proof of \(B\) under \((\varnothing;\{\sneg B\})\). Implication elimination (14) and weakening (29) now give the required proof of \(C\). For (b), discharge \(\sneg B\) instead, and obtain its proof from the co-assumption \(B\) by rules (2) and (33).

        Dually, consider clauses (c) and (d) for dual proofs. In (c), coimplication introduction (21) turns the given dual proof of \(C\) into a dual proof of \(B\cotol C\) under \((L;R\setminus\{B\})\). Rules (1) and (32) give a dual proof of \(B\) from the assumption \(\sneg B\). Coimplication elimination (28) and weakening (30) give the required dual proof of \(C\). For (d), discharge \(\sneg B\) instead and derive its dual proof from the assumption \(B\) by rules (1) and (31).

        Finally, clauses (a) and (b) for dual proofs follow by first turning a dual proof of \(C\) into a proof of \(\sneg C\) using (33), applying the corresponding proof transformation, and returning to a dual proof of \(C\) using (32). Clauses (c) and (d) for proofs follow in the same way using (31) and (34) around the dual-proof transformations.
    \end{proof}
}

\begin{lemma}[From $\tintcty{}$ to $\tintc{}$]
    The following hold simultaneously:
    \begin{enumerate}[(i)]
        \item If $\Gamma \types M:A,\,\Delta$ in \tintcty{} and $\Gamma$ and $\Delta$ are disjoint variable \revreplace{contexts}{environments}, then there exists some $\Pi$, such that
              $\wanp{\Pi,\,A,\,(\erasure{\Gamma},\,\erasure{\Delta})}$ in \tintc{}.
        \item If $\Gamma,\,M:A \types \Delta$ in \tintcty{} and $\Gamma$ and $\Delta$ are disjoint variable \revreplace{contexts}{environments}, then there exists some $\Pi$, such that
              $\wandp{\Pi,\,A,\,(\erasure{\Gamma},\,\erasure{\Delta})}$ in \tintc{}.
    \end{enumerate}
\end{lemma}

\begin{proof}
    We prove (i) and (ii) simultaneously by mutual induction on the given typing derivation. Throughout, the displayed formula $M:A$ is the distinguished subject formula; $\Gamma$ and $\Delta$ are the remaining formulas in the judgement.

    We first record where the disjoint-variable-\revreplace{contexts}{environments} hypothesis is used. Apart from \rlnm{Id} and the four strong-negation rules, the principal typing formula in the conclusion of every typing rule has a non-variable term.
    If such \camrevreplace{the}{a} principal formula is not $M:A$, it would remain in $\Gamma$ or in $\Delta$. That is impossible because $\Gamma$ and $\Delta$ are variable \revreplace{contexts}{environments}. Consequently, in every non-negation case treated below, the principal formula of the last rule is exactly the distinguished subject formula. In particular, for claim (i) every non-negation left rule is impossible and for claim (ii) every non-negation right rule is impossible.
    \smallskip
    \noindent
    For (i), assume $\Gamma \types M:A,\,\Delta$. We proceed by cases on the last typing rule.
    Whenever a premise is of the form $\Gamma',\,N:B \types \Delta'$, we appeal to the mutual induction hypothesis for claim (ii). The cases of (i) are as follows.
    \begin{indproof}
        \indcase{\rlnm{Id}}
        The last rule is an instance of identity,
        \[
            \infer*[left=\rlnm{Id}]{ }{
                \Sigma,\,x:B \types x:B,\,\Pi
            }.
        \]
        The distinguished formula $M:A$ must be the principal formula
        $x:B$ on the right. Indeed, otherwise that right-hand formula
        $x:B$ would belong to $\Delta$, while the principal left-hand
        formula $x:B$ would belong to $\Gamma$, contradicting the
        disjointness of $\Gamma$ and $\Delta$. Hence, up to exchange,
        \[
            M=x,\qquad A=B,\qquad
            \Gamma=\Sigma,x:B,\qquad \Delta=\Pi.
        \]
        By Wansing rule (1),
        \begin{equation}
            \label{proof:typing-to-2int-caseidr-eq1}
            \wanp{\overline{B},\,B,\,(\{B\};\,\varnothing)}.
        \end{equation}
        Since $\{B\}\subseteq\erasure{\Gamma}$ and
        $\varnothing\subseteq\erasure{\Delta}$, structural rule (29)
        yields
        \[
            \wanp{\Pi,\,A,\,(\erasure{\Gamma};\,\erasure{\Delta})}
        \]
        for a suitable $\Pi$, as required.

        \indcase{\rlnm{$\truefm$R}}
        Here the given derivation ends with
        \[
            \infer*[left=\rlnm{$\truefm$R}]{ }{
                \Gamma \types \unittm:\truefm,\,\Delta
            }
        \]
        By Wansing rule (3), we have
        \begin{equation}
            \label{proof:typing-to-2int-casetopr-eq1}
            \wanp{\overline{\truefm},\,\truefm,\,(\varnothing;\,\varnothing)}
        \end{equation}
        Let $\Pi=\overline{\truefm}$. Since $\varnothing \subseteq \erasure{\Gamma}$ and $\varnothing \subseteq \erasure{\Delta}$, applying Wansing's structural rule (29) to \ref{proof:typing-to-2int-casetopr-eq1} gives
        \[
            \wanp{\Pi,\,\truefm,\,(\erasure{\Gamma};\,\erasure{\Delta})}
        \]
        as required.

        \indcase{\rlnm{$\wedge$R}}
        Here the given derivation ends with
        \[
            \infer*[left=\rlnm{$\wedge$R}]{
                \Gamma \types M:A_1,\,\Delta
                \and
                \Gamma \types N:A_2,\,\Delta
            }{
                \Gamma \types (M,\,N):A_1 \wedge A_2,\,\Delta
            }
        \]
        By the induction hypothesis on the two premises, there exist some $\Pi_1$, $\Pi_2$ such that
        \begin{subequations}
            \begin{equation}
                \label{proof:typing-to-2int-caseandr-eq1}
                \wanp{\Pi_1,\,A_1,\,(\erasure{\Gamma};\,\erasure{\Delta})}
            \end{equation}
            \begin{equation}
                \label{proof:typing-to-2int-caseandr-eq2}
                \wanp{\Pi_2,\,A_2,\,(\erasure{\Gamma};\,\erasure{\Delta})}
            \end{equation}
        \end{subequations}
        Applying Wansing rule (5) to \ref{proof:typing-to-2int-caseandr-eq1} and \ref{proof:typing-to-2int-caseandr-eq2} gives
        \[
            \wanp{\mprfraction{\Pi_1 \quad \Pi_2}{(A_1 \wedge A_2)},\,(A_1 \wedge A_2),\,((\erasure{\Gamma} \cup \erasure{\Gamma});\,(\erasure{\Delta} \cup \erasure{\Delta}))}
        \]
        Since $\erasure{\Gamma} \cup \erasure{\Gamma}=\erasure{\Gamma}$ and $\erasure{\Delta} \cup \erasure{\Delta}=\erasure{\Delta}$, this yields
        \[
            \wanp{\Pi,\,A_1 \wedge A_2,\,(\erasure{\Gamma};\,\erasure{\Delta})}
        \]
        as required, where $\Pi=\mprfraction{\Pi_1 \quad \Pi_2}{(A_1 \wedge A_2)}$.

        \indcase{\rlnm{$\vee$R$_1$}}
        Here the given derivation ends with
        \[
            \infer*[left=\rlnm{$\vee$R$_1$}]{
                \Gamma \types M:A_1,\,\Delta
            }{
                \Gamma \types \inj_1(M):A_1 \vee A_2,\,\Delta
            }
        \]
        By the induction hypothesis on the premise, there exists some $\Pi$ such that
        \begin{equation}
            \label{proof:typing-to-2int-caseorr1-eq1}
            \wanp{\Pi,\,A_1,\,(\erasure{\Gamma};\,\erasure{\Delta})}
        \end{equation}
        Applying Wansing rule (6) to \ref{proof:typing-to-2int-caseorr1-eq1} gives
        \[
            \wanp{\mprfraction{\Pi}{(A_1 \vee A_2)},\,(A_1 \vee A_2),\,(\erasure{\Gamma};\,\erasure{\Delta})}
        \]
        as required, where $\Pi'=\mprfraction{\Pi}{(A_1 \vee A_2)}$

        \indcase{\rlnm{$\vee$R$_2$}}
        This is analogous to the case where the last typing rule is \rlnm{$\vee$R$_1$}, using Wansing rule (7) in place of rule (6).

        \indcase{\rlnm{$\to$R}}
        Here the given derivation ends with
        \[
            \infer*[
            left=\rlnm{$\to$R},
            right={$x \notin \fv(\Gamma,\,\Delta)$}
            ]{
                \Gamma,\,x:A_1 \types M:A_2,\,\Delta
            }{
                \Gamma \types \abs{x}{M}:A_1 \to A_2,\,\Delta
            }
        \]
        By the induction hypothesis on the premise, there exists some $\Pi$ such that
        \begin{equation}
            \label{proof:typing-to-2int-caseimpr-eq1}
            \wanp{\Pi,\,A_2,\,(\erasure{\Gamma,\,x:A_1};\,\erasure{\Delta})}
        \end{equation}
        Applying Wansing rule (8) to \ref{proof:typing-to-2int-caseimpr-eq1} gives
        \[
            \wanp{\mprfraction{\Pi}{(A_1 \to A_2)},\,(A_1 \to A_2),\,(\erasure{\Gamma,\,x:A_1}\setminus\{A_1\};\,\erasure{\Delta})}
        \]
        Since $\erasure{\Gamma,\,x:A_1}=\erasure{\Gamma}\cup\{A_1\}$, we have $\erasure{\Gamma,\,x:A_1}\setminus\{A_1\}=\erasure{\Gamma}\setminus\{A_1\}$. Hence
        \[
            \wanp{\mprfraction{\Pi}{(A_1 \to A_2)},\,(A_1 \to A_2),\,(\erasure{\Gamma}\setminus\{A_1\};\,\erasure{\Delta})}
        \]
        Finally, since $\erasure{\Gamma}\setminus\{A_1\} \subseteq \erasure{\Gamma}$, applying Wansing's structural rule (29) yields
        \[
            \wanp{\Pi',\,A_1 \to A_2,\,(\erasure{\Gamma};\,\erasure{\Delta})}
        \]
        as required, where $\Pi'=\mprfraction{\Pi}{(A_1 \to A_2)}$.

        \indcase{\rlnm{$\cotot$R}}
        Here the given derivation ends with
        \[
            \infer*[left=\rlnm{$\cotot$R}]{
                \Gamma,\,M:A \types \Delta
                \and
                \Gamma \types N:B,\,\Delta
            }{
                \Gamma \types M\cdot N:A \cotot B,\,\Delta
            }
        \]
        By the mutual IH on the two premises, there exist some $\Pi_1$, $\Pi_2$ such that
        \begin{subequations}
            \begin{equation}
                \label{proof:typing-to-2int-casecotor-eq2}
                \wandp{\Pi_1,\,A,\,(\erasure{\Gamma};\,\erasure{\Delta})}.
            \end{equation}
        \end{subequations}
        \begin{equation}
            \label{proof:typing-to-2int-casecotor-eq1}
            \wanp{\Pi_2,\,B,\,(\erasure{\Gamma};\,\erasure{\Delta})}
        \end{equation}
        Applying Wansing rule (9) to \ref{proof:typing-to-2int-casecotor-eq1} and \ref{proof:typing-to-2int-casecotor-eq2} gives
        \[
            \wanp{\mprfraction{\Pi_1 \quad \Pi_2}{(A \cotol B)},\,(A \cotol B),\,((\erasure{\Gamma}\cup\erasure{\Gamma});\,(\erasure{\Delta}\cup\erasure{\Delta}))}
        \]
        Since $\erasure{\Gamma}\cup\erasure{\Gamma}=\erasure{\Gamma}$ and $\erasure{\Delta}\cup\erasure{\Delta}=\erasure{\Delta}$, this yields
        \[
            \wanp{\Pi',\,A \cotol B,\,(\erasure{\Gamma};\,\erasure{\Delta})}
        \]
        as required, where $\Pi'=\mprfraction{\Pi_1 \quad \Pi_2}{(A \cotol B)}$.

        \indcase{\rlnm{$\falsefm$R-elim}}
        Here the given derivation ends with
        \[
            \infer*[left=\rlnm{$\falsefm$R-elim}]{
                \Gamma \types M:\falsefm,\,\Delta
            }{
                \Gamma \types \abort(M):A,\,\Delta
            }
        \]
        By the induction hypothesis on the premise, there exists some $\Pi$ such that
        \begin{equation}
            \label{proof:typing-to-2int-casebotrelim-eq1}
            \wanp{\Pi,\,\falsefm,\,(\erasure{\Gamma};\,\erasure{\Delta})}
        \end{equation}
        Applying Wansing rule (10) to \ref{proof:typing-to-2int-casebotrelim-eq1} gives
        \[
            \wanp{\Pi',\,A,\,(\erasure{\Gamma};\,\erasure{\Delta})}
        \]
        as required, where $\Pi'=\mprfraction{\Pi}{A}$.

        \indcase{\rlnm{AppR}}
        Here the given derivation ends with
        \[
            \infer*[left=\rlnm{AppR}]{
                \Gamma \types M:A_1 \to A_2,\,\Delta
                \and
                \Gamma \types N:A_1,\,\Delta
            }{
                \Gamma \types M\,N:A_2,\,\Delta
            }
        \]
        By the induction hypothesis on the two premises, there exist some $\Pi_1$, $\Pi_2$ such that
        \begin{subequations}
            \begin{equation}
                \label{proof:typing-to-2int-caseappr-eq1}
                \wanp{\Pi_1,\,A_1 \to A_2,\,(\erasure{\Gamma};\,\erasure{\Delta})}
            \end{equation}
            \begin{equation}
                \label{proof:typing-to-2int-caseappr-eq2}
                \wanp{\Pi_2,\,A_1,\,(\erasure{\Gamma};\,\erasure{\Delta})}
            \end{equation}
        \end{subequations}
        Applying Wansing rule (14) to \ref{proof:typing-to-2int-caseappr-eq2} and \ref{proof:typing-to-2int-caseappr-eq1} gives
        \[
            \wanp{\mprfraction{\Pi_2\quad\Pi_1}{A_2},\,A_2,\,((\erasure{\Gamma}\cup\erasure{\Gamma});\,(\erasure{\Delta}\cup\erasure{\Delta}))}
        \]
        Since $\erasure{\Gamma}\cup\erasure{\Gamma}=\erasure{\Gamma}$ and $\erasure{\Delta}\cup\erasure{\Delta}=\erasure{\Delta}$, this yields
        \[
            \wanp{\Pi,\,A_2,\,(\erasure{\Gamma};\,\erasure{\Delta})}
        \]
        as required, where $\Pi=\mprfraction{\Pi_2\quad\Pi_1}{A_2}$.

        \indcase{\rlnm{PiR}}
        Here the given derivation ends with
        \[
            \infer*[left=\rlnm{PiR}]{
                \Gamma \types M:A_1 \wedge A_2,\,\Delta
            }{
                \Gamma \types \pi_i(M):A_i,\,\Delta
            }
        \]
        By the induction hypothesis on the premise, there exists some $\Pi$ such that
        \begin{equation}
            \label{proof:typing-to-2int-casepir-eq1}
            \wanp{\Pi,\,A_1 \wedge A_2,\,(\erasure{\Gamma};\,\erasure{\Delta})}
        \end{equation}
        If $i=1$, then \camrevreplace{applying}{apply} Wansing rule (11) to \ref{proof:typing-to-2int-casepir-eq1}. If $i=2$, then \camrevreplace{applying}{apply} Wansing rule (12) to \ref{proof:typing-to-2int-casepir-eq1}. In either case the application gives:
        \[
            \wanp{\Pi',\,A_i,\,(\erasure{\Gamma};\,\erasure{\Delta})}
        \]
        as required, where $\Pi'=\mprfraction{\Pi}{A_i}$.

        \indcase{\rlnm{CaseR}}
        Here the given derivation ends with
        \[
            \infer*[left=\rlnm{CaseR},
            right={$x,\,y \notin \fv(\Gamma,\,\Delta)$}]{
                \Gamma \types M:A_1 \vee A_2,\,\Delta
                \and
                \Gamma,\,x:A_1 \types N:C,\,\Delta
                \and
                \Gamma,\,y:A_2 \types P:C,\,\Delta
            }{
                \Gamma \types \casetm{M}{N}{P}:C,\,\Delta
            }
        \]
        By the induction hypothesis on the three premises, there exist some $\Pi$, $\Pi_1$, $\Pi_2$ such that
        \begin{subequations}
            \begin{equation}
                \label{proof:typing-to-2int-casecaser-eq1}
                \wanp{\Pi,\,A_1 \vee A_2,\,(\erasure{\Gamma};\,\erasure{\Delta})}
            \end{equation}
            \begin{equation}
                \label{proof:typing-to-2int-casecaser-eq2}
                \wanp{\Pi_1,\,C,\,(\erasure{\Gamma,\,x:A_1};\,\erasure{\Delta})}
            \end{equation}
            \begin{equation}
                \label{proof:typing-to-2int-casecaser-eq3}
                \wanp{\Pi_2,\,C,\,(\erasure{\Gamma,\,y:A_2};\,\erasure{\Delta})}
            \end{equation}
        \end{subequations}
        Since $\erasure{\Gamma,\,x:A_1}=\erasure{\Gamma}\cup\{A_1\}$ and $\erasure{\Gamma,\,y:A_2}=\erasure{\Gamma}\cup\{A_2\}$, applying Wansing rule (13) to \ref{proof:typing-to-2int-casecaser-eq1}, \ref{proof:typing-to-2int-casecaser-eq2}, and \ref{proof:typing-to-2int-casecaser-eq3} gives
        \[
            \wanp{\mprfraction{\Pi\;\;\Pi_1\;\;\Pi_2}{C},\,C,\,(\erasure{\Gamma} \cup ((\erasure{\Gamma}\cup\{A_1\}) \setminus \{A_1\}) \cup ((\erasure{\Gamma}\cup\{A_2\}) \setminus \{A_2\});\,\erasure{\Delta} \cup \erasure{\Delta} \cup \erasure{\Delta})}.
        \]
        Since $(\erasure{\Gamma}\cup\{A_1\}) \setminus \{A_1\} \subseteq \erasure{\Gamma}$, $(\erasure{\Gamma}\cup\{A_2\}) \setminus \{A_2\} \subseteq \erasure{\Gamma}$, and $\erasure{\Delta} \cup \erasure{\Delta} \cup \erasure{\Delta}=\erasure{\Delta}$, this yields
        \[
            \wanp{\Pi',\,C,\,(\erasure{\Gamma};\,\erasure{\Delta})}
        \]
        as required, where $\Pi'=\mprfraction{\Pi\;\;\Pi_1\;\;\Pi_2}{C}$.

        \indcase{\rlnm{Iota1R}}
        Here the given derivation ends with
        \[
            \infer*[left=\rlnm{Iota1R}]{
                \Gamma,\,M:A_1 \to A_2 \types \Delta
            }{
                \Gamma \types \iota_1(M):A_1,\, \Delta
            }
        \]
        By the induction hypothesis on the premise, there exists some $\Pi$ such that
        \begin{equation}
            \label{proof:typing-to-2int-caseiota1r-eq1}
            \wandp{\Pi,\,A_1 \to A_2,\,(\erasure{\Gamma};\,\erasure{\Delta})}
        \end{equation}
        Applying Wansing rule (26) to \ref{proof:typing-to-2int-caseiota1r-eq1} gives
        \[
            \wanp{\Pi',\,A_1,\,(\erasure{\Gamma};\,\erasure{\Delta})}
        \]
        as required, where $\Pi'=\mprfraction{\Pi}{A_1}$.

        \indcase{\rlnm{Iota2R}}
        Here the given derivation ends with
        \[
            \infer*[left=\rlnm{Iota2R}]{
                \Gamma \types M:A_1 \cotot A_2,\,\Delta
            }{
                \Gamma \types \iota_2(M):A_2,\,\Delta
            }
        \]
        By the mutual induction hypothesis on the premise, there exists some $\Pi$ such that
        \begin{equation}
            \label{proof:typing-to-2int-caseiota2r-eq1}
            \wanp{\Pi,\,A_1 \cotol A_2,\,(\erasure{\Gamma};\,\erasure{\Delta})}
        \end{equation}
        Applying Wansing rule (15) to \ref{proof:typing-to-2int-caseiota2r-eq1} gives
        \[
            \wanp{\Pi',\,A_2,\,(\erasure{\Gamma};\,\erasure{\Delta})}
        \]
        as required, where $\Pi'=\mprfraction{\Pi}{A_2}$.

        \indcase{\rlnm{$\sneg_I$ R}}
        The last rule has the form
        \[
            \infer*[left=\rlnm{$\sneg_I$ R}]{
                \Sigma,\,P:B \types \Pi
            }{
                \Sigma \types P:\sneg{B},\,\Pi
            }.
        \]
        There are two possibilities.
        \begin{enumerate}
            \item Suppose that the principal formula $P:\sneg{B}$ is the
                  distinguished subject. Then
                  \[
                      M=P,\qquad A=\sneg{B},\qquad
                      \Gamma=\Sigma,\qquad \Delta=\Pi.
                  \]
                  The mutual induction hypothesis for the premise gives, for
                  some $\Theta$,
                  \[
                      \wandp{\Theta,\,B,\,
                          (\erasure{\Gamma};\,\erasure{\Delta})}.
                  \]
                  Wansing rule (\revreplace{34}{33}) therefore gives
                  \[
                      \wanp{\Theta',\,\sneg{B},\,
                          (\erasure{\Gamma};\,\erasure{\Delta})},
                  \]
                  which is the required conclusion.

            \item Suppose that the principal formula is passive. Since it
                  lies in the variable \camrevreplace{context}{environment} $\Delta$, its term is a
                  variable, say $P=z$. Up to exchange, write
                  \[
                      \Delta=z:\sneg{B},\,\Delta_0,
                      \qquad
                      \Pi=M:A,\,\Delta_0,
                      \qquad
                      \Gamma=\Sigma.
                  \]
                  Disjointness of $\Gamma$ and $\Delta$ implies that
                  $\Gamma,z:B$ and $\Delta_0$ are disjoint variable \revreplace{contexts}{environments}.
                  Applying the induction hypothesis for claim (i) to the
                  premise, with $M:A$ still distinguished, gives
                  \[
                      \wanp{\Theta,\,A,\,
                          (\erasure{\Gamma,z:B};\,
                          \erasure{\Delta_0})}.
                  \]
                  Since
                  \[
                      \erasure{\Gamma,z:B}=\erasure{\Gamma}\cup\{B\},
                      \qquad
                      \erasure{\Delta}
                      =\erasure{\Delta_0}\cup\{\sneg{B}\},
                  \]
                  clause (a) of
                  \cref{lem:wansing-context-negation-shift} yields
                  \[
                      \wanp{\Theta',\,A,\,
                          (\erasure{\Gamma};\,\erasure{\Delta})}.
                  \]
        \end{enumerate}

        \indcase{\rlnm{$\sneg_E$ R}}
        The last rule has the form
        \[
            \infer*[left=\rlnm{$\sneg_E$ R}]{
                \Sigma,\,P:\sneg{B} \types \Pi
            }{
                \Sigma \types P:B,\,\Pi
            }.
        \]
        If the principal formula $P:B$ is the distinguished subject, the
        mutual induction hypothesis gives a dual proof of $\sneg{B}$ with
        context $(\erasure{\Gamma};\erasure{\Delta})$, and Wansing rule (\camrevreplace{35}{34})
        gives the required proof of $B$.

        Otherwise the principal formula is passive in $\Delta$, so
        $P=z$ is a variable and, up to exchange,
        \[
            \Delta=z:B,\,\Delta_0,
            \qquad
            \Pi=M:A,\,\Delta_0,
            \qquad
            \Gamma=\Sigma.
        \]
        The induction hypothesis applied to the premise with $M:A$
        distinguished gives
        \[
            \wanp{\Theta,\,A,\,
                (\erasure{\Gamma,z:\sneg{B}};\,
                \erasure{\Delta_0})}.
        \]
        Clause (b) of \cref{lem:wansing-context-negation-shift} gives
        \[
            \wanp{\Theta',\,A,\,
                (\erasure{\Gamma};\,\erasure{\Delta})},
        \]
        as required.

        \indcase{\rlnm{$\sneg_I$ L}}
        The last rule has the form
        \[
            \infer*[left=\rlnm{$\sneg_I$ L}]{
                \Sigma \types P:B,\,\Pi
            }{
                \Sigma,\,P:\sneg{B} \types \Pi
            }.
        \]
        Its principal formula is on the left and hence is passive relative
        to claim (i). Because it belongs to the variable \camrevreplace{context}{environment} $\Gamma$,
        $P=z$ is a variable. Up to exchange,
        \[
            \Gamma=\Sigma,z:\sneg{B},
            \qquad
            \Pi=M:A,\,\Delta.
        \]
        The premise, with $M:A$ distinguished, is
        \[
            \Sigma \types M:A,\,z:B,\,\Delta.
        \]
        Its side contexts $\Sigma$ and $z:B,\Delta$ are disjoint variable
        \camrevreplace{contexts}{environments}, so the induction hypothesis gives
        \[
            \wanp{\Theta,\,A,\,
                (\erasure{\Sigma};\,
                \erasure{z:B,\Delta})}.
        \]
        Clause (c) of \cref{lem:wansing-context-negation-shift} moves $B$
        from the right context to $\sneg{B}$ in the left context and yields
        \[
            \wanp{\Theta',\,A,\,
                (\erasure{\Gamma};\,\erasure{\Delta})}.
        \]

        \indcase{\rlnm{$\sneg_E$ L}}
        The last rule has the form
        \[
            \infer*[left=\rlnm{$\sneg_E$ L}]{
                \Sigma \types P:\sneg{B},\,\Pi
            }{
                \Sigma,\,P:B \types \Pi
            }.
        \]
        Again the principal formula is passive, hence $P=z$ is a variable.
        Up to exchange,
        \[
            \Gamma=\Sigma,z:B,
            \qquad
            \Pi=M:A,\,\Delta.
        \]
        Applying the induction hypothesis to
        \[
            \Sigma \types M:A,\,z:\sneg{B},\,\Delta
        \]
        gives a proof of $A$ with context
        $(\erasure{\Sigma};\erasure{z:\sneg{B},\Delta})$.
        Clause (d) of \cref{lem:wansing-context-negation-shift} therefore
        yields the required proof with context
        $(\erasure{\Gamma};\erasure{\Delta})$.
    \end{indproof}

    \smallskip
    \noindent
    For (ii), assume $\Gamma,\,M:A \types \Delta$. Again we proceed by cases on the last typing rule.
    Whenever a premise is of the form $\Gamma' \types N:B,\,\Delta'$, we appeal to the mutual induction hypothesis for claim (i). The cases are as follows.
    \begin{indproof}
        \indcase{\rlnm{Id}}
        The last rule is an instance of identity,
        \[
            \infer*[left=\rlnm{Id}]{ }{
                \Sigma,\,x:B \types x:B,\,\Pi
            }.
        \]
        The distinguished formula $M:A$ must be the principal formula
        $x:B$ on the left. Otherwise the principal left-hand formula
        $x:B$ would belong to $\Gamma$, while the principal right-hand
        formula $x:B$ would belong to $\Delta$, contradicting
        disjointness. Hence, up to exchange,
        \[
            M=x,\qquad A=B,\qquad
            \Gamma=\Sigma,\qquad \Delta=x:B,\Pi.
        \]
        By Wansing rule (2),
        \begin{equation}
            \label{proof:typing-to-2int-caseidl-eq1}
            \wandp{\dbar{B},\,B,\,(\varnothing;\,\{B\})}.
        \end{equation}
        Since $\varnothing\subseteq\erasure{\Gamma}$ and
        $\{B\}\subseteq\erasure{\Delta}$, structural rule (30)
        yields
        \[
            \wandp{\Pi,\,A,\,(\erasure{\Gamma};\,\erasure{\Delta})}
        \]
        for a suitable $\Pi$, as required.

        \indcase{\rlnm{$\falsefm$L}}
        Here the given derivation ends with
        \[
            \infer*[left=\rlnm{$\falsefm$L}]{ }{
                \Gamma,\,\unittm:\falsefm \types \Delta
            }
        \]
        By Wansing rule (4), we have
        \begin{equation}
            \label{proof:typing-to-2int-casebotl-eq1}
            \wandp{\dbar{\falsefm},\,\falsefm,\,(\varnothing;\,\varnothing)}
        \end{equation}
        Let $\Pi=\dbar{\falsefm}$. Since $\varnothing \subseteq \erasure{\Gamma}$ and $\varnothing \subseteq \erasure{\Delta}$, applying Wansing's structural rule (30) to \ref{proof:typing-to-2int-casebotl-eq1} gives
        \[
            \wandp{\Pi,\,\falsefm,\,(\erasure{\Gamma};\,\erasure{\Delta})}
        \]
        as required.

        \indcase{\rlnm{$\wedge$L$_1$}}
        Here the given derivation ends with
        \[
            \infer*[left=\rlnm{$\wedge$L$_1$}]{
                \Gamma,\,M:A_1 \types \Delta
            }{
                \Gamma,\,\inj_1(M):A_1 \wedge A_2 \types \Delta
            }
        \]
        By the induction hypothesis on the premise, there exists some $\Pi$ such that
        \begin{equation}
            \label{proof:typing-to-2int-caseandl1-eq1}
            \wandp{\Pi,\,A_1,\,(\erasure{\Gamma};\,\erasure{\Delta})}
        \end{equation}
        Applying Wansing rule (17) to \ref{proof:typing-to-2int-caseandl1-eq1} gives
        \[
            \wandp{\dblmprfraction{\Pi}{A_1 \wedge A_2},\,(A_1 \wedge A_2),\,(\erasure{\Gamma};\,\erasure{\Delta})}
        \]
        as required.

        \indcase{\rlnm{$\wedge$L$_2$}}
        This is analogous to the case where the last typing rule is \rlnm{$\wedge$L$_1$}, using Wansing rule (18) in place of rule (17).

        \indcase{\rlnm{$\vee$L}}
        Here the given derivation ends with
        \[
            \infer*[left=\rlnm{$\vee$L}]{
                \Gamma,\,M:A_1 \types \Delta
                \and
                \Gamma,\,N:A_2 \types \Delta
            }{
                \Gamma,\,(M,\,N):A_1 \vee A_2 \types \Delta
            }
        \]
        By the induction hypothesis on the two premises, there exist some $\Pi_1$, $\Pi_2$ such that
        \begin{subequations}
            \begin{equation}
                \label{proof:typing-to-2int-caseorl-eq1}
                \wandp{\Pi_1,\,A_1,\,(\erasure{\Gamma};\,\erasure{\Delta})}
            \end{equation}
            \begin{equation}
                \label{proof:typing-to-2int-caseorl-eq2}
                \wandp{\Pi_2,\,A_2,\,(\erasure{\Gamma};\,\erasure{\Delta})}.
            \end{equation}
        \end{subequations}
        Applying Wansing rule (19) to \ref{proof:typing-to-2int-caseorl-eq1} and \ref{proof:typing-to-2int-caseorl-eq2} gives
        \[
            \wandp{\dblmprfraction{\Pi_1 \quad \Pi_2}{(A_1 \vee A_2)},\,(A_1 \vee A_2),\,((\erasure{\Gamma} \cup \erasure{\Gamma});\,(\erasure{\Delta} \cup \erasure{\Delta}))}
        \]
        Since $\erasure{\Gamma} \cup \erasure{\Gamma}=\erasure{\Gamma}$ and $\erasure{\Delta} \cup \erasure{\Delta}=\erasure{\Delta}$, this yields
        \[
            \wandp{\Pi,\,A_1 \vee A_2,\,(\erasure{\Gamma};\,\erasure{\Delta})}
        \]
        as required, where $\Pi=\dblmprfraction{\Pi_1 \quad \Pi_2}{(A_1 \vee A_2)}$.

        \indcase{\rlnm{$\to$L}}
        Here the given derivation ends with
        \[
            \infer*[left=\rlnm{$\to$L}]{
                \Gamma \types M:A_1,\,\Delta
                \and
                \Gamma,\,N:A_2 \types \Delta
            }{
                \Gamma,\,M\cdot N:A_1 \to A_2 \types \Delta
            }
        \]
        By the IH on the two premises, there exist some $\Pi_1$, $\Pi_2$ such that
        \begin{subequations}
            \begin{equation}
                \label{proof:typing-to-2int-caseimpl-eq1}
                \wanp{\Pi_1,\,A_1,\,(\erasure{\Gamma};\,\erasure{\Delta})}
            \end{equation}
            \begin{equation}
                \label{proof:typing-to-2int-caseimpl-eq2}
                \wandp{\Pi_2,\,A_2,\,(\erasure{\Gamma};\,\erasure{\Delta})}.
            \end{equation}
        \end{subequations}
        Applying Wansing rule (20) to \ref{proof:typing-to-2int-caseimpl-eq1} and \ref{proof:typing-to-2int-caseimpl-eq2} gives
        \[
            \wandp{\dblmprfraction{\Pi_1 \quad \Pi_2}{(A_1 \to A_2)},\,(A_1 \to A_2),\,((\erasure{\Gamma} \cup \erasure{\Gamma});\,(\erasure{\Delta} \cup \erasure{\Delta}))}
        \]
        Since $\erasure{\Gamma} \cup \erasure{\Gamma}=\erasure{\Gamma}$ and $\erasure{\Delta} \cup \erasure{\Delta}=\erasure{\Delta}$, this yields
        \[
            \wandp{\Pi,\,A_1 \to A_2,\,(\erasure{\Gamma};\,\erasure{\Delta})}
        \]
        as required, where $\Pi=\dblmprfraction{\Pi_1 \quad \Pi_2}{(A_1 \to A_2)}$.

        \indcase{\rlnm{$\cotot$L}}
        Here the given derivation ends with
        \[
            \infer*[
            left=\rlnm{$\cotot$L},
            right={$x \notin \fv(\Gamma,\,\Delta)$}]{
                \Gamma,\,M:B \types x:A,\,\Delta
            }{
                \Gamma,\,\abs{x}{M}:A \cotot B \types \Delta
            }
        \]
        By the induction hypothesis on the premise, there exists some $\Pi$ such that
        \begin{equation}
            \label{proof:typing-to-2int-casecotol-eq1}
            \wandp{\Pi,\,B,\,(\erasure{\Gamma};\,\erasure{x:A,\,\Delta})}
        \end{equation}
        Applying Wansing rule (21) to \ref{proof:typing-to-2int-casecotol-eq1} gives
        \[
            \wandp{\dblmprfraction{\Pi}{(A \cotol B)},\,(A \cotol B),\,(\erasure{\Gamma};\,\erasure{x:A,\,\Delta}\setminus\{A\})}
        \]
        Let $\Pi'=\dblmprfraction{\Pi}{(A \cotol B)}$. Since $\erasure{x:A,\,\Delta}\setminus\{A\} \subseteq \erasure{\Delta}$, applying Wansing's structural rule (30) yields
        \[
            \wandp{\Pi',\,A \cotol B,\,(\erasure{\Gamma};\,\erasure{\Delta})}
        \]
        as required.

        \indcase{\rlnm{$\truefm$L-elim}}
        Here the given derivation ends with
        \[
            \infer*[left=\rlnm{$\truefm$L-elim}]{
                \Gamma,\,M:\truefm \types \Delta
            }{
                \Gamma,\,\abort(M):A \types \Delta
            }
        \]
        By the induction hypothesis on the premise, there exists some $\Pi$ such that
        \begin{equation}
            \label{proof:typing-to-2int-casetoprelim-eq1}
            \wandp{\Pi,\,\truefm,\,(\erasure{\Gamma};\,\erasure{\Delta})}
        \end{equation}
        Applying Wansing rule (22) to \ref{proof:typing-to-2int-casetoprelim-eq1} gives
        \[
            \wandp{\Pi',\,A,\,(\erasure{\Gamma};\,\erasure{\Delta})}
        \]
        as required, where $\Pi'=\dblmprfraction{\Pi}{A}$.

        \indcase{\rlnm{AppL}}
        Here the given derivation ends with
        \[
            \infer*[left=\rlnm{AppL}]{
                \Gamma,\,M:A \cotot B \types \Delta
                \and
                \Gamma,\,N:A \types \Delta
            }{
                \Gamma,\,M\,N:B \types \Delta
            }.
        \]
        By the induction hypothesis on the two premises, there exist some $\Pi_1$, $\Pi_2$ such that
        \begin{subequations}
            \begin{equation}
                \label{proof:typing-to-2int-caseappl-eq1}
                \wandp{\Pi_1,\,A \cotol B,\,(\erasure{\Gamma};\,\erasure{\Delta})}
            \end{equation}
            \begin{equation}
                \label{proof:typing-to-2int-caseappl-eq2}
                \wandp{\Pi_2,\,A,\,(\erasure{\Gamma};\,\erasure{\Delta})}.
            \end{equation}
        \end{subequations}
        Applying Wansing rule (28) to \ref{proof:typing-to-2int-caseappl-eq1} and \ref{proof:typing-to-2int-caseappl-eq2} gives
        \[
            \wandp{\dblmprfraction{\Pi_1\enspace\Pi_2}{B},\,B,\,((\erasure{\Gamma}\cup\erasure{\Gamma});\,(\erasure{\Delta}\cup\erasure{\Delta}))}.
        \]
        Since $\erasure{\Gamma}\cup\erasure{\Gamma}=\erasure{\Gamma}$ and $\erasure{\Delta}\cup\erasure{\Delta}=\erasure{\Delta}$, this yields
        \[
            \wandp{\Pi,\,B,\,(\erasure{\Gamma};\,\erasure{\Delta})}
        \]
        as required, where $\Pi=\dblmprfraction{\Pi_1\enspace\Pi_2}{B}$.

        \indcase{\rlnm{PiL}}
        Here the given derivation ends with
        \[
            \infer*[left=\rlnm{PiL}]{
                \Gamma,\,M:A_1 \vee A_2 \types \Delta
            }{
                \Gamma,\,\pi_i(M):A_i \types \Delta
            }
        \]
        By the induction hypothesis on the premise, there exists some $\Pi$ such that
        \begin{equation}
            \label{proof:typing-to-2int-casepil-eq1}
            \wandp{\Pi,\,A_1 \vee A_2,\,(\erasure{\Gamma};\,\erasure{\Delta})}
        \end{equation}
        If $i=1$, then \camrevreplace{applying}{apply} Wansing rule (24) to \ref{proof:typing-to-2int-casepil-eq1}. If $i=2$, then \camrevreplace{applying}{apply} Wansing rule (25) to \ref{proof:typing-to-2int-casepil-eq1}. In either case, the application gives
        \[
            \wandp{\Pi',\,A_i,\,(\erasure{\Gamma};\,\erasure{\Delta})}
        \]
        as required, where $\Pi'=\dblmprfraction{\Pi}{A_i}$.

        \indcase{\rlnm{CaseL}}
        Here the given derivation ends with
        \[
            \infer*[
            left=\rlnm{CaseL},
            right={$x,\,y \notin \fv(\Gamma,\,\Delta)$}]{
                \Gamma,\,M:A_1 \wedge A_2 \types \Delta
                \and
                \Gamma,\,N:C \types x:A_1,\,\Delta
                \and
                \Gamma,\,P:C \types y:A_2,\,\Delta
            }{
                \Gamma,\,\casetm{M}{N}{P}:C \types \Delta
            }
        \]
        By the induction hypothesis on the three premises, there \camrevreplace{exists}{exist} some $\Pi$, $\Pi_1$, $\Pi_2$ such that
        \begin{subequations}
            \begin{equation}
                \label{proof:typing-to-2int-casecasel-eq1}
                \wandp{\Pi,\,A_1 \wedge A_2,\,(\erasure{\Gamma};\,\erasure{\Delta})}
            \end{equation}
            \begin{equation}
                \label{proof:typing-to-2int-casecasel-eq2}
                \wandp{\Pi_1,\,C,\,(\erasure{\Gamma};\,\erasure{x:A_1,\,\Delta})}
            \end{equation}
            \begin{equation}
                \label{proof:typing-to-2int-casecasel-eq3}
                \wandp{\Pi_2,\,C,\,(\erasure{\Gamma};\,\erasure{y:A_2,\,\Delta})}.
            \end{equation}
        \end{subequations}
        Since $\erasure{x:A_1,\,\Delta}=\{A_1\}\cup\erasure{\Delta}$ and $\erasure{y:A_2,\,\Delta}=\{A_2\}\cup\erasure{\Delta}$, applying Wansing rule (23) to \ref{proof:typing-to-2int-casecasel-eq1}, \ref{proof:typing-to-2int-casecasel-eq2}, and \ref{proof:typing-to-2int-casecasel-eq3} gives
        \[
            \wandp{\dblmprfraction{\Pi\;\;\Pi_1\;\;\Pi_2}{C},\,C,\,(\erasure{\Gamma} \cup \erasure{\Gamma} \cup \erasure{\Gamma};\,\erasure{\Delta} \cup ((\{A_1\}\cup\erasure{\Delta}) \setminus \{A_1\}) \cup ((\{A_2\}\cup\erasure{\Delta}) \setminus \{A_2\}))}.
        \]
        Since $\erasure{\Gamma} \cup \erasure{\Gamma} \cup \erasure{\Gamma}=\erasure{\Gamma}$, $(\{A_1\}\cup\erasure{\Delta}) \setminus \{A_1\} \subseteq \erasure{\Delta}$, and $(\{A_2\}\cup\erasure{\Delta}) \setminus \{A_2\} \subseteq \erasure{\Delta}$, this yields
        \[
            \wandp{\Pi',\,C,\,(\erasure{\Gamma};\,\erasure{\Delta})}
        \]
        as required, where $\Pi'=\dblmprfraction{\Pi\;\;\Pi_1\;\;\Pi_2}{C}$.
        \indcase{\rlnm{Iota1L}}
        Here the given derivation ends with
        \[
            \infer*[left=\rlnm{Iota1L}]{
                \Gamma \types M:A_1 \cotot A_2,\,\Delta
            }{
                \Gamma,\,\iota_1(M):A_1 \types \Delta
            }.
        \]
        By the induction hypothesis on the premise, there exists some $\Pi$ such that
        \begin{equation}
            \label{proof:typing-to-2int-caseiota1l-eq1}
            \wanp{\Pi,\,A_1 \cotol A_2,\,(\erasure{\Gamma};\,\erasure{\Delta})}.
        \end{equation}
        Applying Wansing rule (16) to \ref{proof:typing-to-2int-caseiota1l-eq1} gives
        \[
            \wandp{\Pi',\,A_1,\,(\erasure{\Gamma};\,\erasure{\Delta})}
        \]
        as required, where $\Pi'=\dblmprfraction{\Pi}{A_1}$.

        \indcase{\rlnm{Iota2L}}
        Here the given derivation ends with
        \[
            \infer*[left=\rlnm{Iota2L}]{
                \Gamma,\,M:A_1 \to A_2 \types \Delta
            }{
                \Gamma,\,\iota_2(M):A_2 \types \Delta
            }
        \]
        By the induction hypothesis on the premise, there exists some $\Pi$ such that
        \begin{equation}
            \label{proof:typing-to-2int-caseiota2l-eq1}
            \wandp{\Pi,\,A_1 \to A_2,\,(\erasure{\Gamma};\,\erasure{\Delta})}
        \end{equation}
        Applying Wansing rule (27) to \ref{proof:typing-to-2int-caseiota2l-eq1} gives
        \[
            \wandp{\Pi',\,A_2,\,(\erasure{\Gamma};\,\erasure{\Delta})}
        \]
        as required, where $\Pi'=\dblmprfraction{\Pi}{A_2}$.

        \indcase{\rlnm{$\sneg_I$ L}}
        The last rule has the form
        \[
            \infer*[left=\rlnm{$\sneg_I$ L}]{
                \Sigma \types P:B,\,\Pi
            }{
                \Sigma,\,P:\sneg{B} \types \Pi
            }.
        \]
        We distinguish whether the principal formula is the distinguished
        subject.
        \begin{enumerate}
            \item If $P:\sneg{B}$ is the distinguished subject, then
                  \[
                      M=P,\qquad A=\sneg{B},\qquad
                      \Gamma=\Sigma,\qquad \Delta=\Pi.
                  \]
                  The induction hypothesis for claim (i) gives a proof of
                  $B$ with context
                  $(\erasure{\Gamma};\erasure{\Delta})$. Wansing rule (\revreplace{32}{31})
                  gives the required dual proof of $\sneg{B}$.

            \item If the principal formula is passive, it belongs to the
                  variable \camrevreplace{context}{environment} $\Gamma$, so $P=z$ is a variable. Up to
                  exchange, there is a variable \camrevreplace{context}{environment} $\Sigma'$ such that
                  \[
                      \Sigma=\Sigma',M:A,
                      \qquad
                      \Gamma=\Sigma',z:\sneg{B},
                      \qquad
                      \Delta=\Pi.
                  \]
                  The premise, with $M:A$ distinguished, is
                  \[
                      \Sigma',M:A \types z:B,\,\Delta.
                  \]
                  Its side contexts are disjoint variable \revreplace{contexts}{environments}, so the
                  induction hypothesis for claim (ii) gives
                  \[
                      \wandp{\Theta,\,A,\,
                          (\erasure{\Sigma'};\,
                          \erasure{z:B,\Delta})}.
                  \]
                  Clause (c) of
                  \cref{lem:wansing-context-negation-shift}, in its dual-proof
                  form, gives
                  \[
                      \wandp{\Theta',\,A,\,
                          (\erasure{\Gamma};\,\erasure{\Delta})}.
                  \]
        \end{enumerate}

        \indcase{\rlnm{$\sneg_E$ L}}
        The last rule has the form
        \[
            \infer*[left=\rlnm{$\sneg_E$ L}]{
                \Sigma \types P:\sneg{B},\,\Pi
            }{
                \Sigma,\,P:B \types \Pi
            }.
        \]
        If the principal formula $P:B$ is the distinguished subject, the
        induction hypothesis gives a proof of $\sneg{B}$ with context
        $(\erasure{\Gamma};\erasure{\Delta})$, and Wansing rule (\camrevreplace{33}{32}) gives the
        required dual proof of $B$.

        Otherwise $P=z$ is a variable and, up to exchange,
        \[
            \Sigma=\Sigma',M:A,
            \qquad
            \Gamma=\Sigma',z:B,
            \qquad
            \Delta=\Pi.
        \]
        Applying the induction hypothesis for claim (ii) to
        \[
            \Sigma',M:A \types z:\sneg{B},\,\Delta
        \]
        gives a dual proof of $A$ with context
        $(\erasure{\Sigma'};\erasure{z:\sneg{B},\Delta})$. Clause (d) of
        \cref{lem:wansing-context-negation-shift}, in its dual-proof form,
        yields the required dual proof with context
        $(\erasure{\Gamma};\erasure{\Delta})$.

        \indcase{\rlnm{$\sneg_I$ R}}
        The last rule has the form
        \[
            \infer*[left=\rlnm{$\sneg_I$ R}]{
                \Sigma,\,P:B \types \Pi
            }{
                \Sigma \types P:\sneg{B},\,\Pi
            }.
        \]
        Its principal formula is on the right and hence is passive relative
        to claim (ii). Thus $P=z$ is a variable. Up to exchange,
        \[
            \Sigma=\Gamma,M:A,
            \qquad
            \Delta=z:\sneg{B},\,\Pi.
        \]
        Applying the induction hypothesis for claim (ii) to the premise
        \[
            \Gamma,z:B,\,M:A \types \Pi
        \]
        gives a dual proof of $A$ with context
        $(\erasure{\Gamma,z:B};\erasure{\Pi})$. Clause (a) of
        \cref{lem:wansing-context-negation-shift}, in its dual-proof form,
        moves $B$ to $\sneg{B}$ on the right and gives the required context
        $(\erasure{\Gamma};\erasure{\Delta})$.

        \indcase{\rlnm{$\sneg_E$ R}}
        The last rule has the form
        \[
            \infer*[left=\rlnm{$\sneg_E$ R}]{
                \Sigma,\,P:\sneg{B} \types \Pi
            }{
                \Sigma \types P:B,\,\Pi
            }.
        \]
        Again the principal formula is passive, so $P=z$ is a variable and,
        up to exchange,
        \[
            \Sigma=\Gamma,M:A,
            \qquad
            \Delta=z:B,\,\Pi.
        \]
        The induction hypothesis applied to
        \[
            \Gamma,z:\sneg{B},\,M:A \types \Pi
        \]
        gives a dual proof of $A$ with context
        $(\erasure{\Gamma,z:\sneg{B}};\erasure{\Pi})$. Clause (b) of
        \cref{lem:wansing-context-negation-shift}, in its dual-proof form,
        yields the required dual proof with context
        $(\erasure{\Gamma};\erasure{\Delta})$.
    \end{indproof}
\end{proof}



  \section{Supplementary Material for Section~\ref{sec:2hol}}

This appendix details the admissibility of the context switching rules.  For this, we first note that the system satisfies weakening and exchange among assumptions and co-assumptions.

\begin{lemma}[Exchange]\label{lem:hol2-exchange}
  In the following, let $\alpha$ and $\beta$ be, each, either a type variable or a term variable and let $S$ and $T$ be, each, either a kind or a type.
  \begin{enumerate}[(i)]
    \item If $\Gamma_1,\alpha\ty:S,\beta\ty:T,\Gamma_2;\,M:A \types \Delta$ and $\alpha \notin \fv(T)$, then $\Gamma_1,\beta\ty:T,\alpha\ty:S,\Gamma_2;\,M:A \types \Delta$.
    \item If $\Gamma_1,\alpha\ty:S,\beta\ty:T,\Gamma_2 \types M:A;\, \Delta$ and $\alpha \notin \fv(T)$, then $\camrevreplace{\Gamma_1,\beta\ty:T,\alpha\ty:S,\Gamma_2;\,M:A \types \Delta}{\Gamma_1,\beta\ty:T,\alpha\ty:S,\Gamma_2 \types M:A;\,\Delta}$.
    \item If $\Gamma;\,M:A \types \Delta_1,x\ty:A,y\ty:B,\Delta_2$ \camrevreplace{then $\Gamma;\,M:A \types \Delta_1,x\ty:A,y\ty:B,\Delta_2$}{then $\Gamma;\,M:A \types \Delta_1,y\ty:B,x\ty:A,\Delta_2$.}
    \item If $\Gamma \types M:A;\, \Delta_1,x\ty:A,y\ty:B,\Delta_2$ \camrevreplace{then $\Gamma;\,M:A \types \Delta_1,x\ty:A,y\ty:B,\Delta_2$}{then $\Gamma \types M:A;\,\Delta_1,y\ty:B,x\ty:A,\Delta_2$.}
  \end{enumerate}
\end{lemma}
\begin{proof}
  All by induction on the derivation.  The only interesting cases are when the proof is concluded by \rlnm{Id-L} or \rlnm{Id-R}.  In these cases, note that the side conditions on these rules do not depend on the order of assumptions or co-assumptions, and that the condition $\alpha \notin \fv(T)$ ensures that the premise concerning the well-formedness of the context remains valid after exchange.  All the remaining cases follow immediately from the induction hypotheses.
\end{proof}

Note that the height of the proof tree is preserved through \camrevdel{the }the above operations of exchange.

\begin{lemma}[Weakening]\label{lem:hol2-weakening}
  Both of the following:
  \begin{enumerate}
    \item If $\Gamma \types M:A;\,\Delta$ and $\types \Gamma,x\ty:B,\Delta^R$, then $\Gamma,x\ty:B \types M : A;\,\Delta$ and $\Gamma \types M:A;\,x\ty:B,\Delta$.
    \item If $\Gamma;\,M:A \types \Delta$ and $\types \Gamma,x\ty:B,\Delta^R$, then $\Gamma,x\ty:B;\,M:A \types \Delta$ and $\Gamma;\,M:A \types x\ty:B,\Delta$.
  \end{enumerate}
\end{lemma}
\begin{proof}
  The proof is a straightforward induction on derivations.  We show only a few of the more interesting cases:
  \begin{description}
    \item[\rlnm{Id-L},\rlnm{Id-R}] These cases are symmetrical, so we do the former only.  In this case, $M$ is some variable $x'$, we have $\types \Gamma,\Delta^R$ and $x'\ty:A \in \Delta$.  Suppose $\types \Gamma,x\ty:B,\Delta^R$, then the result follows from \rlnm{Id-L}.
    \item[\rlnm{$\forall$-Elim-L}] In this case, $M$ has shape $\unpack_K\,[\tyabs{a}{K}{B}]\,M'\,[A]\,(\Tyabs{a}{K}{\tyabs{x'}{\sneg{B'}}{N}})$.  Spp $\types \Gamma,x\ty:B,\Delta^R$.  Then, by the freshness of $a$ and $x'$, it follows that $\Gamma,a\ty:K,x\ty:\camrevreplace{B}{B'},x'\ty:B',\Delta^R$.  It follows from the induction hypothesis that all of the following:
          \begin{enumerate}[(i)]
            \item $\Gamma,x\ty:B;\,M':\PiType{a}{K}{B'} \types \Delta$
            \item $\Gamma;\,M':\PiType{a}{K}{B'} \types x\ty:B,\Delta$
            \item $\Gamma,a\ty:K,x\ty:B;\,N:A \types x'\ty:B',\Delta$
            \item $\Gamma,a\ty:K;\,N:A \types x\ty:B,x'\ty:B',\Delta$
          \end{enumerate}
          By exchange, we can obtain $\Gamma,x:B,a:K;\,N:A \types x':B',\Delta$ from (iii), and together with (i), using \revreplace{\rlnm{$\forall$-Elim-R}}{\rlnm{$\forall$-Elim-L}}we obtain the first part of the conclusion.  Similarly, we can obtain $\Gamma,a:K;\,N:A \types x':B',x:B,\Delta$, and then, together with (ii), we obtain the second part of the conclusion.
    \item[\rlnm{$\sneg{}$-Intro-L}]  In this case, $A$ is of shape $\sneg{C}$ and we aim to show (2).  Suppose, $\types \Gamma,x\ty:B,\Delta^R$.  Then it follows from the induction hypothesis, case (1), that $\Gamma,x\ty:B \types M:C;\,\Delta$ and $\Gamma \types M:C;\,x\ty:B,\Delta$.  Therefore, the conclusion follows from \rlnm{$\sneg{}$-Intro-L}.
    \item[\rlnm{$\sneg{}$-Intro-R}]  In this case, $A$ is of shape $\sneg{C}$ and we aim to show (1).  Suppose, $\types \Gamma,x\ty:B,\Delta^R$.  Then it follows from the induction hypothesis, case (2), that $\Gamma,x\ty:B;\,M:C \types \Delta$ and $\Gamma;\,M:C \types x\ty:B,\Delta$.  Therefore, the conclusion follows from \rlnm{$\sneg{}$-Intro-R}.
  \end{description}
\end{proof}

\begin{lemma}[Context Switching]
  All of the following:
  \begin{enumerate}[(i)]
    \item $\Gamma;\,M:B \types x\ty:A,\,\Delta$ implies $\Gamma,\,x\ty:\sneg{A};\,M:B \types \Delta$
    \item $\Gamma \types M:B;\,x\ty:A,\,\Delta$ implies $\Gamma,\,x\ty:\sneg{A} \types M:B;\,\Delta$
    \item $\Gamma,\,x\ty:A \types M:B;\,\Delta$ implies $\Gamma \types M:B;\,x\ty:\sneg{A},\, \Delta$
    \item $\Gamma,\,x\ty:A;\,M:B \types \Delta$ implies $\Gamma;\,M:B \types x\ty:\sneg{A},\,\Delta$
  \end{enumerate}
\end{lemma}
\begin{proof}
  We prove cases (i) and (ii) simultaneously by induction on the height of the derivation.  Cases (iii) and (iv) are completely symmetrical.
  \begin{description}
    \item[\rlnm{Id-L}] Here we only consider (i).  There are two subcases.  If $M=x$ then, by the well-formedness of the context, necessarily $A=B$ and \camrevreplace{we can obtain a proof of $\Gamma,\,x\ty:\sneg{A};\,x:A \types \Delta$, therefore, we can obtain a proof of the conclusion by \rlnm{$\sneg{}$-Elim-R} and \rlnm{Id-R}}{\rlnm{Id-R} gives $\Gamma,\,x\ty:\sneg{A} \types x:\sneg{A};\,\Delta$, and \rlnm{$\sneg{}$-Elim-R} then gives the conclusion}.  If $M \neq x$, then the conclusion follows immediately by \rlnm{Id-L}.
    \item[\rlnm{Id-R}] Here we only consider (ii) and the proof is symmetrical to the previous case.
    \item[\rlnm{$\sneg{}$-Intro-L}]  Here $B$ has shape $\sneg{C}$ \camrevadd{and }we only consider (i).  It follows from the induction hypothesis, case (ii), that $\Gamma,\,x\ty:\sneg{A} \types M:C;\,\Delta$ and so the conclusion follows from \rlnm{$\sneg{}$-Intro-L}.
    \item[\rlnm{$\sneg{}$-Intro-R}] Here $B$ has shape $\sneg{C}$ and we only consider (ii).  It follows from the induction hypothesis, case (i), that $\Gamma,x\ty:\sneg{A};\,M:C \types \Delta$ and so the conclusion follows from \rlnm{$\sneg{}$-Intro-R}.
    \item[\rlnm{$\sneg{}$-Elim-L}] Here we only consider (ii).  It follows from the induction hypothesis, case (i), that $\Gamma,\,x\ty:\sneg{A};\,M:\sneg{B} \types \Delta$ and so the conclusion follows from \rlnm{$\sneg{}$-Elim-L}.
    \item[\rlnm{$\sneg{}$-Elim-R}] Here we only consider (i).  It follows from the induction hypothesis, case (ii), that $\Gamma,\,x\ty:\sneg{A} \types M : \sneg{B};\,\Delta$ and so the conclusion follows from \rlnm{$\sneg{}$-Elim-R}.
    \item[\rlnm{$\forall$-Elim-L}] In this case, the term $M$ is of shape $\unpack_K\,[\tyabs{a}{K}{B'}]\,M'\,[C]\,(\Tyabs{a}{K}{\tyabs{y}{\camrevreplace{B'}{\sneg{B'}}}{N}})$.  We only consider (i).  It follows from the induction hypothesis that $\Gamma,\,\camrevreplace{y}{x}\ty:\sneg{A};\,M':\PiType{a}{K}{B'} \types \Delta$.  By applying exchange to the right side of the second premise we obtain a proof of $\Gamma,\,a\ty:K;\,N:C \types x\ty:A,\,y\ty:B',\,\Delta$ of the same height.  Therefore, we can apply the induction hypothesis to obtain $\Gamma,\,a\ty:K,\,x\ty:\sneg{A};\,N:C \types y\ty:B',\,\Delta$.  By applying exchange to the left-hand side, we obtain $\Gamma,\,x\ty:\sneg{A},\,a\ty:K;\,N:C \types y\ty:B',\,\Delta$ and then the result follows from \rlnm{$\forall$-Elim-L}.
    \item[\rlnm{$\forall$-Intro-R}, \rlnm{$\to$-Intro-R}] These cases can be obtained in a similar way, as a combination of exchange and the induction hypotheses.
    \item All other cases follow immediately from the induction hypotheses.
  \end{description}
\end{proof}
  \section{Supplementary Material for Section~\ref{sec:hol2-power}}

This appendix contains some details concerning the development of the expressive power of \holtty{} that were omitted from the main text.

\subsection{\texorpdfstring{The Type $A \cotot B$}{The Type A coimplies B}}
We defined coimplication by:
\begin{align*}
  A \cotot B & \coloneqq \sneg{(\sneg{A} \to \sneg{B})}
\end{align*}

Due to the admissibility of the context switching rules (Figure~\ref{fig:ctx-switching}), the introduction rules for coimplication are admissible:
\begin{mathpar}
  \infer*{
    \infer*{
      \infer*{
        \infer*{
          \Gamma;\,M:B \types x:A,\,\Delta
        }
        {
          \Gamma,\,x:\sneg{A};\,M:B \types \Delta
        }
      }
      {
        \Gamma,\,x:\sneg{A} \types M:\sneg{B};\, \Delta
      }
    }
    {
      \Gamma \types \tyabs{x}{\sneg{A}}{M} : \sneg{A} \to \sneg{B};\,\Delta
    }
  }
  {
    \Gamma;\,\tyabs{x}{\sneg{A}}{M} : A \cotot B \types \Delta
  }
  \and
  \infer*{
    \infer*{
      \infer*{
        \Gamma;\,M:A \types \Delta
      }
      {
        \Gamma \types M : \sneg{A};\,\Delta
      }
      \and
      \infer*{
        \Gamma \types N : B;\,\Delta
      }
      {
        \Gamma;\,N:\sneg{B} \types \Delta
      }
    }
    {
      \Gamma;\, M \cdot_{\sneg{A},\sneg{B}} N : \sneg{A} \to \sneg{B} \types \Delta
    }
  }
  {
    \Gamma \types M \cdot_{\sneg{A},\sneg{B}} N : A \cotot B;\,\Delta
  }
\end{mathpar}
\camrevreplace{Similarly, the elimination rules can be derived, $i \in \{1,2\}$:}{Similarly, the elimination rules can be derived:}
\begin{mathpar}
  \infer*{
    \infer*{
      \infer*{
        \Gamma;\,M:A \cotot B \types \Delta
      }
      {
        \Gamma \types M : \sneg{A} \to \sneg{B};\,\Delta
      }
      \and
      \infer*{
        \Gamma;\,N:A \types \Delta
      }
      {
        \Gamma \types N:\sneg{A};\,\Delta
      }
    }
    {
      \Gamma \types M\,N : \sneg{B};\,\Delta
    }
  }
  {
    \Gamma;\,M\,N:B \types \Delta
  }
  \and
  \infer*{
    \infer*{
      \infer*{
        \Gamma \types M : A_1 \cotot A_2;\,\Delta
      }
      {
        \Gamma;\,M : \sneg{A_1} \to \sneg{A_2} \types \Delta
      }
    }
    {
      \Gamma \types \camrevreplace{\iota_i}{\iota_1}\,[\sneg{A_1}]\,[\sneg{A_2}]\,M : \sneg{\camrevreplace{A_i}{A_1}};\,\Delta
    }
  }
  {
    \Gamma;\,\camrevreplace{\iota_i}{\iota_1}\,[\sneg{A_1}]\,[\sneg{A_2}]\,M : \camrevreplace{A_i}{A_1} \types \Delta
  }
  \and
  \camrevadd{\infer*{
      \infer*{
        \infer*{\Gamma \types M : A_1 \cotot A_2;\,\Delta}
        {\Gamma;\,M : \sneg{A_1} \to \sneg{A_2} \types \Delta}
      }
      {\Gamma;\,\iota_2\,[\sneg{A_1}]\,[\sneg{A_2}]\,M : \sneg{A_2} \types \Delta}
    }
    {\Gamma \types \iota_2\,[\sneg{A_1}]\,[\sneg{A_2}]\,M : A_2;\,\Delta}}
\end{mathpar}

\subsection{\texorpdfstring{The Type $A \refwedge B$}{The Type A | B}}

In Section~\ref{sec:hol2-power}, we defined:
\begin{align*}
  A \refwedge B & \coloneqq \SigmaType{b}{\starsort}{(A \cotot b) \cotot (B \cotot b) \cotot b}                                                                                                                                                \\
  \inj_1        & \coloneqq \Tyabs{a}{\starsort}{\Tyabs{b}{\starsort}{\tyabs{z}{\camrevreplace{a}{\sneg{a}}}{\Tyabs{c}{\starsort}{\tyabs{x}{\sneg{(a \cotot c)}}{\tyabs{y}{\sneg{(b \cotot c)}}{x\,z}}}}}}                                     \\
  \inj_2        & \coloneqq \Tyabs{a}{\starsort}{\Tyabs{b}{\starsort}{\tyabs{z}{\camrevreplace{b}{\sneg{b}}}{\Tyabs{c}{\starsort}{\tyabs{x}{\sneg{(a \cotot c)}}{\tyabs{y}{\sneg{(b \cotot c)}}{y\,z}}}}}}                                     \\
  \refcase      & \coloneqq \Tyabs{a}{\starsort}{\Tyabs{b}{\starsort}{\tyabs{x}{\camrevreplace{(a \refwedge b)}{\sneg{(a \refwedge b)}}}{\Tyabs{c}{\starsort}{\tyabs{f}{\sneg{(a \cotot c)}}{\tyabs{g}{\sneg{(b \cotot c)}}{x\,[c]\,f\,g}}}}}}
\end{align*}

We now establish the types of these constructions in detail.
\camrevreplace{Below, the first of the injections, the other is symmetrical.}{Below, we derive the type of the first injection; the other is symmetrical.}  Let $\Gamma \coloneqq [a:\starsort,b:\starsort,c:\starsort]$ in the following:
\begin{mathpar}
  \infer*{
    \infer*{
      \infer*{
        \infer*{
          \infer*{
            \infer*{
              \infer*{
                \infer*{ }
                {
                  \Gamma;\,x:a \cotot c \types \camrevreplace{x}{y}\ty:b \cotot c,\,\camrevreplace{y}{x}\ty:a \cotot c,\,z\ty:a
                }
                \and
                \infer*{ }
                {
                  \Gamma;\,z:\camrevreplace{A}{a} \types y\ty:b \cotot c,\,x\ty:a \cotot c,\,z\ty:a
                }
              }
              {
                \Gamma;\,x\,z:c \types y\ty:b \cotot c,\,x\ty:a \cotot c,\,z\ty:a
              }
            }
            {
              \Gamma;\,\tyabs{y}{\sneg{(b \cotot c)}}{x\,z}:(b \cotot c) \cotot c \types x\ty:a \cotot c,\,z\ty:a
            }
          }
          {
            \Gamma;\,\tyabs{x}{\sneg{(a \cotot c)}}{\tyabs{y}{\sneg{(b \cotot c)}}{x\,z}}:(a \cotot c) \cotot (b \cotot c) \cotot c \types z\ty:a
          }
        }
        {
          \camrevreplace{\Gamma}{a\ty:\starsort,b\ty:\starsort};\,\Tyabs{c}{\starsort}{\tyabs{x}{\sneg{(a \cotot c)}}{\tyabs{y}{\sneg{(b \cotot c)}}{x\,z}}}:a \refwedge b \types z\ty:a
        }
      }
      {
        \camrevdel{\Gamma;\,}a\ty:\starsort,\,b\ty:\starsort\camrevreplace{,}{;}\,\tyabs{z}{\camrevreplace{a}{\sneg{a}}}{\Tyabs{c}{\starsort}{\tyabs{x}{\sneg{(a \cotot c)}}{\tyabs{y}{\sneg{(b \cotot c)}}{x\,z}}}}:a \cotot (a \refwedge b) \types
      }
    }
    {
      \camrevdel{\Gamma;\,}a\ty:\starsort\camrevreplace{,}{;}\,\Tyabs{b}{\starsort}{\tyabs{z}{\camrevreplace{a}{\sneg{a}}}{\camrevadd{\Tyabs{c}{\starsort}{}}\tyabs{x}{\sneg{(a \cotot c)}}{\tyabs{y}{\sneg{(b \cotot c)}}{x\,z}}}}:\SigmaType{b}{\starsort}{a \cotot (a \refwedge b)} \types
    }
  }
  {
    \camrevdel{\Gamma,\,}\inj_1 : \SigmaType{a}{\starsort}{\SigmaType{b}{\starsort}{a \cotot (a \refwedge b)}} \types
  }
\end{mathpar}
To save on the page width, let $\Gamma \coloneqq [a:\starsort,b:\starsort,c:\starsort]$ and $\camrevreplace{\Delta \coloneqq [x\ty:(a \refwedge b),\,f\ty:a \cotot c,\,g\ty:b \cotot c]}{\Delta \coloneqq [g\ty:b \cotot c,\,f\ty:a \cotot c,\,x\ty:(a \refwedge b)]}$ in the following.
\begin{mathpar}
  \adjustbox{max width=\linewidth}{$\displaystyle
      \infer*{
        \infer*{
          \infer*{
            \infer*{
              \infer*{
                \infer*{
                  \infer*{
                    \infer*{
                      \infer*{
                        \infer*{ }{\Gamma;\,x:\revreplace{A}{a} \refwedge \revreplace{B}{b} \types \Delta}
                        \and
                        \infer*{ }{\Gamma \types c : \starsort}
                      }
                      {
                        \Gamma;\,x\,[c] : (a \cotot c) \cotot (b \cotot c) \cotot c \types \Delta
                      }
                      \and
                      \infer*{ }
                      {
                        \Gamma;\,f: a \cotot c \types \Delta
                      }
                    }
                    {
                      \Gamma;\,x\,[c]\,f:\revdel{\sneg{}}(b \cotot c) \cotot c \types \Delta
                    }
                    \and
                    \infer*{ }
                    {
                      \Gamma;\,g:b \cotot c \types \Delta
                    }
                  }
                  {
                    \Gamma;\,x\,[c]\,f\,g : c \types \Delta
                  }
                }
                {
                  \Gamma;\,\tyabs{g}{\sneg{(b \cotot c)}}{x\,[c]\,f\,g} : (b \cotot c) \cotot c\revreplace{,\,x\ty:(a  \refwedge b),\,f:a \cotot c}{\types f\ty:a\cotot c,\,x\ty:(a\refwedge b)}
                }
              }
              {
                \Gamma;\,\tyabs{f}{\sneg{(a \cotot c)}}{\tyabs{g}{\sneg{(b \cotot c)}}{x\,[c]\,f\,g}} : (a \cotot c) \cotot (b \cotot c) \cotot c\revreplace{,}{\types}\,x\ty:(a  \refwedge b)
              }
            }
            {
              a,b\ty:\starsort;\,\Tyabs{c}{\starsort}{\tyabs{f}{\sneg{(a \cotot c)}}{\tyabs{g}{\sneg{(b \cotot c)}}{x\,[c]\,f\,g}}} : (a \refwedge b)\revreplace{,}{\types}\,x\ty:(a  \refwedge b)
            }
          }
          {
            a,b\ty:\starsort;\,\tyabs{x}{\sneg{(a \refwedge b)}}{\Tyabs{c}{\starsort}{\tyabs{f}{\sneg{(a \cotot c)}}{\tyabs{g}{\sneg{(b \cotot c)}}{x\,[c]\,f\,g}}}} : (a \refwedge b) \cotot (a \refwedge b)\revadd{\types}
          }
        }
        {
          a\ty:\starsort;\,\Tyabs{b}{\starsort}{\tyabs{x}{\sneg{(a \refwedge b)}}{\Tyabs{c}{\starsort}{\tyabs{f}{\sneg{(a \cotot c)}}{\tyabs{g}{\sneg{(b \cotot c)}}{x\,[c]\,f\,g}}}}} : \SigmaType{b}{\starsort}{(a \refwedge b) \cotot (a \refwedge b)}\revadd{\types}
        }
      }
      {
        \refcase\,: \SigmaType{a}{\starsort}{\SigmaType{b}{\starsort}{(a \refwedge b) \cotot (a \refwedge b)}} \types
      }
    $}
\end{mathpar}

  \section{Supplementary Material in Connection with Section~\ref{sec:hol2-consistency}}\label{sec:apx-hol2-consistency}

In this appendix, we give an explicit \revreplace{definitition}{definition} of Geuvers' \holty{} and then give the details of the consistency proof for our \holtty{} that were omitted from the main text.

\subsection{Geuvers' \holty{}}

\begin{definition}[\holty{} Syntax]\label{def:hol-syntax}
  Let us assume countably infinite supplies of kind variables $k$; type variables $a$, $b$, $c$; and term variables $x$, $y$, $z$.  The \emph{kinds}, typically $K$; \emph{types}, typically $A, B$; \emph{terms}, typically $M$, $N$; and environments, typically $\Gamma,\Delta$; of \holty{} are described by the following grammars:
  \[
    \begin{array}{rrcl}
      \rlnm{Kinds} & K             & \Coloneqq & k \mid \starsort \mid K_1 \to K_2                                   \\
      \rlnm{Types} & A,B           & \Coloneqq & a \mid A\,B \mid \tyabs{a}{K}{A} \mid A \to B \mid \PiType{a}{K}{A} \\
      \rlnm{Terms} & M,N           & \Coloneqq & x \mid M\,N \mid M\,[A] \mid \tyabs{x}{A}{M} \mid \Tyabs{a}{K}{M}   \\
      \rlnm{Envs}  & \Gamma,\Delta & \Coloneqq & \varnothing \mid \Gamma,a\mathord{:}K \mid \Gamma,x\mathord{:}A
    \end{array}
  \]
  We consider the types $\tyabs{a}{K}{A}$, $\PiType{a}{K}{A}$; and the terms $\tyabs{x}{A}{M}$ and $\Tyabs{a}{K}{M}$; as variable binders, and we treat them accordingly.
\end{definition}

\begin{definition}[\holty{} Type Conversion]
  Conversion $A \conv B$ between two types $A$ and $B$, is defined as the reflexive, symmetric, transitive closure of the following inductively defined relation:
  \begin{mathpar}
    \infer[CRedex]{ }{
      (\tyabs{a}{K}{A})\,B \conv A[B/a]
    }
    \\
    \infer[CAbs]{
      A \conv B
    }
    {
      \tyabs{a}{K}{A} \conv \tyabs{a}{K}{B}
    }
    \and
    \infer[CApp]{
      A \conv B
      \and
      C \conv D
    }
    {
      A\,C \conv B\,D
    }
    \and
    \infer[CForall]{
      A \conv B
    }
    {
      \PiType{a}{K}{A} \conv \PiType{a}{K}{B}
    }
    \and
    \infer[CArrow]{
      A \conv B
      \and
      C \conv D
    }
    {
      A \to C \conv B \to D
    }
  \end{mathpar}
\end{definition}

\begin{definition}[\holty{} Kinding]
  Well-formedness of an environment $\Gamma$, written $\types \Gamma$, and the assignment of a kind $K$ to a type $A$ in the environment $\Gamma$, written $\Gamma \types A : K$, are given simultaneously by the following inductive definition.
  \begin{mathpar}
    \infer[Empty]{ }{
      \types \varnothing
    }
    \and
    \infer[Kinding]{
      \types \Gamma
    }
    {
      \types \Gamma,a\ty:K
    }\;a \notin \subj(\Gamma)
    \and
    \infer[Typing]{
      \Gamma \types A : \starsort
    }
    {
      \types \Gamma,x\ty:A
    }\;x \notin \subj(\Gamma)
    \\
    \infer[Var]{
      \types \Gamma
    }
    {
      \Gamma \types a : K
    }\;a\mathord{:}K \in \Gamma
    \and
    \infer[Abs]{
      \Gamma,\,a\ty:K_1 \types B : K_2
    }
    {
      \Gamma \types \tyabs{a}{K_1}{B} : K_1 \to K_2
    }
    \and
    \infer[App]{
      \Gamma \types A : K_1 \to K_2
      \and
      \Gamma \types B : K_1
    }
    {
      \Gamma \types A\,B : K_2
    }
    \and
    \infer[Arr]{
      \Gamma \types A : \starsort
      \and
      \Gamma \types B : \starsort
    }
    {
      \Gamma \types A \to B : \starsort
    }
    \and
    \infer[All]{
      \Gamma,\,a\ty:K \types A : \starsort
    }
    {
      \Gamma \types \PiType{a}{K}{A} : \starsort
    }
  \end{mathpar}
\end{definition}

\begin{definition}[Proof Terms]
  The assignment of a type $A$ to a (proof) term $M$ under assumptions $\Gamma$ is written $\Gamma \types M : A$.  The theory is defined by induction.
  \begin{mathpar}
    %
    %
    %
    %
    \infer[Id]
    {
      \types \Gamma
    }
    {
      \Gamma \types x:A
    }\;x\mathord{:}A \in \Gamma
    \and
    %
    %
    \infer[Conv]{
      \Gamma \types M : B
      \and
      \Gamma \types A : \starsort
    }
    {
      \Gamma \types M : A
    }\;A =_\beta B
    \\
    %
    %
    %
    \infer[$\forall$-Intro-R]{
      \Gamma,a : K \vdash M : B
    }
    {
      \Gamma \vdash \Tyabs{a}{K}{M} : \PiType{a}{K}{B}
    }
    \and
    %
    %
    \infer[$\forall$-Elim-R]{
      \Gamma \types M : \PiType{a}{K}{B}
      \and
      \Gamma \types A : K
    }
    {
      \Gamma \types M\,[A] : B[A/a]
    }
    \\
    \infer[$\to$-Intro-R]{
      \Gamma,\,x\ty:A \types M : B
    }
    {
      \Gamma \types \tyabs{x}{A}{M} : A \to B
    }
    \and
    \infer[$\to$-Elim-R]{
      \Gamma \types M : A \to B
      \and
      \Gamma \types N : A
    }
    {
      \Gamma \types M\,N : B
    }
  \end{mathpar}
  In \rlnm{$\forall$-Intro-R} and \rlnm{$\to$-Intro-R} we require that the variables $a$ and $x$ are fresh for their context, in the sense that neither \camrevreplace{occur}{occurs} free in $\Gamma$.
\end{definition}

\subsection{Consistency of \holtty{}}

We assume the following abbreviations in \holty{}:
\begin{align*}
  A \wedge B                       & \coloneqq \PiType{a}{\starsort}{(A \to B \to a) \to a}                                                                                                                     \\
  \mkPair                          & \coloneqq \Tyabs{a}{\starsort}{\Tyabs{b}{\starsort}{\tyabs{x}{a}{\tyabs{y}{b}{\Tyabs{c}{\starsort}{\tyabs{z}{a \to b \to c}{z\,x\,y}}}}}}                                  \\
  \pi_1                            & \coloneqq \Tyabs{a}{\starsort}{\Tyabs{b}{\starsort}{\tyabs{x}{(a \wedge b)}{\revreplace{x\,a\,}{x\,[a]\,}(\tyabs{y}{a}{\tyabs{z}{b}{y}})}}}                                \\
  \pi_2                            & \coloneqq \Tyabs{a}{\starsort}{\Tyabs{b}{\starsort}{\tyabs{x}{(a \wedge b)}{\revreplace{x\,a\,(\tyabs{y}{a}{\tyabs{z}{b}{y}})}{x\,[b]\,(\tyabs{y}{a}{\tyabs{z}{b}{z}})}}}} \\
  \revreplace{\Sigma_K}{\exists_K} & \coloneqq \tyabs{a}{K \to \starsort}{\PiType{b}{\starsort}{(\PiType{c}{K}{a\,c \to b}) \to b}}                                                                             \\
  \pack_K                          & \coloneqq \Tyabs{a}{K \to \starsort}{\Tyabs{b}{K}{\tyabs{x}{a\,b}{\Tyabs{c}{\starsort}{\tyabs{y}{(\PiType{b}{K}{a\,b \to c})}{\revreplace{y\,b\,x}{y\,[b]\,x}}}}}}         \\
  \mathsf{unpack}_K                & \coloneqq \Tyabs{a}{K \to \starsort}{\tyabs{x}{\revreplace{\Sigma_K}{\exists_K}\,a}{x}}
\end{align*}

We require the following results concerning environment permutation, type weakening and environment consistency, respectively, in \holty{}:
\begin{lemma}\label{lem:holty-prereqs}
  In \holty{}:
  \begin{enumerate}[(i)]
    \item If $\Gamma_1,\,\alpha_1:T_1,\alpha_2:T_2,\,\Gamma_2 \types M : A$ and $\alpha_1,\alpha_2$ are term or type variables, and $T_1, T_2$ are kinds or types, and $\alpha_1$ does not occur in $\fv(T_2)$, then $\Gamma_1,\,\alpha_2:T_2,\alpha_1:T_1,\,\Gamma_2 \types M : A$.
    \item If $\Gamma_1 \types M : B$ and $ \types \Gamma_1,\Gamma_2$ then $\Gamma_1,\Gamma_2 \types M : B$.
    \item If $\Gamma \types M : A$ then $\types \Gamma$.
  \end{enumerate}
\end{lemma}




We aim to map the assumption environment $\Gamma$ to $\derefT_1(\Gamma)$ and the co-assumptions environment $\Delta$ to $\derefT_2(\Delta^R)$ and then concatenate them $\derefT_1(\Gamma),\derefT_2(\Delta^R)$ to obtain a single \holty{} environment containing all the (co-)assumptions.  However, this mixed environment consists of types obtained by both $\derefT_1$ and $\derefT_2$, so it is neither $\theta_1(\Gamma)$ nor $\derefT_2(\Gamma)$, for any $\Gamma$.  Hence, it is convenient to define a notion of \holty{} environments that are related to a given \holtty{} environment (e.g. $\Gamma,\Delta^R$) but agnostic about which of $\derefT_1$ and $\derefT_2$ are used at each component.
\begin{definition}
  Define the set $\derefT[\Gamma]$ of $\derefT$-acceptable \holty{} environments, indexed by \holtty{} environments $\Gamma$, by induction:
  \begin{itemize}
    \item $\varnothing \in \derefT[\camrevreplace{\Gamma}{\varnothing}]$.
    \item If $\Gamma' \in \derefT[\Gamma]$ then $\Gamma',a\mathord{:}\derefK(K) \in \derefT[\Gamma,\,a\mathord{:}K]$.
    \item If $\Gamma' \in \derefT[\Gamma]$ and $i \in \{1,2\}$ then $\Gamma',x\mathord{:}\derefT_i(A) \in \derefT[\camrevreplace{\Gamma'}{\Gamma},x\mathord{:}A]$.
  \end{itemize}
  Clearly, for each $i \in \{1,2\}$, $\derefT_i(\Gamma) \in \derefT[\Gamma]$.
\end{definition}

Using this definition, we can state and prove the preservation of kinding.
\begin{lemma}\label{lem:hol2-kind-mapping-preservation}
  Both of the following:
  \begin{enumerate}[(i)]
    \item If $\ \types \Gamma$ in \holtty{} and $\Gamma' \in \derefT[\Gamma]$ then $\ \types \Gamma'$ in $\lambda$HOL.
    \item If $\Gamma \types A : K$ in \holtty{} and $\Gamma' \in \derefT[\Gamma]$, then $\Gamma' \types \derefT(A) : \derefK(K)$ in $\lambda$HOL.
  \end{enumerate}
\end{lemma}
\begin{proof}
  The proof is by induction.  In each case, we suppose $\Gamma' \in \derefT[\Gamma]$.
  \begin{description}
    \item[\rlnm{Var}] In this case, $A$ is a type variable $a$ and $a\mathord{:}K \in \Gamma$ and so $a\mathord{:}\derefK(K) \in \Gamma'$.  It follows from the induction hypothesis that $\types \Gamma'$, whence the result follows by \rlnm{Var} in $\lambda$HOL.
    \item[\rlnm{Abs}] In this case, $A$ is an abstraction $\tyabs{a}{K_1}{B}$ and $K$ is of shape $K_1 \to K_2$.  By definition $\Gamma',a\mathord{:}\derefK(K_1) \in \derefT[\Gamma,a\mathord{:}K_1]$ and so it follows from the induction hypothesis that $\Gamma',a\mathord{:}\derefK(K_1) \types \derefT(B) : \derefK(K_2)$ holds in $\lambda$HOL.  Therefore, it follows by \rlnm{Abs} that $\Gamma' \types \tyabs{a}{\derefK(K_1)}{\theta(B)} : \derefK(K_1) \to \derefK(K_2)$ is derivable in $\lambda$HOL.
    \item[\rlnm{App}] In this case, $A$ is an application $A'B$ and it follows from the induction hypotheses that there is some kind $K_1$ such that $\Gamma' \types \derefT(A') : \derefK(K_1) \to \derefK(K)$ and $\Gamma' \types \derefT(B) : \derefK(K_1)$ in $\lambda$HOL.  Therefore, by \rlnm{App}, $\Gamma' \types \derefT(A'\,B) : \derefK(K)$ in $\lambda$HOL.
    \item[\rlnm{Arr}] In this case, $A$ is of shape $A' \to B$ and $K$ is just $\starsort$.  Then it follows from the induction hypotheses that $\Gamma' \types \derefT(A') : \starsort^2$ and $\Gamma' \types \derefT(B) : \starsort^2$ in $\lambda$HOL.  Therefore, for $i \in \{1,2\}$, $\Gamma' \types \pi_i(\derefT(A')) : \starsort$ and $\Gamma' \types \pi_i(\derefT(B))\revadd{: \starsort}$ in $\lambda$HOL.  Hence, it follows that $\Gamma' \types \abra{\pi_1(\derefT(\revreplace{A}{A'})) \to \pi_1(\derefT(B)),\,\pi_1(\derefT(\camrevreplace{A}{A'})) \wedge \pi_2(\derefT(B))} : \starsort^2$ in $\lambda$HOL, as required.
    \item[\rlnm{All}] In this case, $A$ is of shape $\PiType{a}{K'}{A'}$ and $K$ is just $\starsort$.  Then, by definition, $\Gamma',a\mathord{:}\derefK(K') \in \camrevreplace{\derefT(\Gamma,a\mathord{:}K')}{\derefT[\Gamma,a\mathord{:}K']}$ and it follows from the induction hypothesis that $\Gamma',a\mathord{:}\derefK(K') \types \derefT(A') : \starsort^2$ in $\lambda$HOL.  Therefore, for $i \in \{1,2\}$, $\Gamma',a\mathord{:}\derefK(K') \types \pi_i(\derefT(A')) : \starsort$ in $\lambda$HOL.  Hence, by \rlnm{All}, $\Gamma' \types \PiType{a}{\derefK(K')}{\pi_1(\derefT(A'))} : \starsort$ and by \rlnm{Ex}, $\Gamma' \types \SigmaType{a}{\derefK{\revadd{(}K'\revadd{)}}}{\pi_2(\derefT(A'))} : \starsort$  in $\lambda$HOL.  Therefore, also $\Gamma' \types \abra{\PiType{a}{\derefK(\camrevreplace{K}{K'})}{\pi_1(\derefT(\camrevreplace{A}{A'}))},\,\SigmaType{a}{\derefK(\camrevreplace{K}{K'})}{\pi_2(\derefT(\camrevreplace{A}{A'}))}} : \camrevreplace{\starsort}{\starsort^2}$ in $\lambda$HOL.
    \item[\rlnm{Sneg}] In this case, $A$ is of shape $\sneg{A'}$ and $K$ is just $\starsort$.  Then it follows from the induction hypothesis that $\Gamma' \types \derefT(A') : \starsort^2$ in $\lambda$HOL. Therefore, $\Gamma' \types \abra{\pi_2(\derefT(A')),\pi_1(\derefT(A'))} : \starsort^2$ in $\lambda$HOL.
    \item[\rlnm{Empty}] Clear.
    \item[\rlnm{Kinding}] Then $\Gamma$ has shape $\Gamma_0,a\mathord{:}K$ and $a \notin \subj(\Gamma_0)$.  Therefore, $\Gamma'$ has shape $\Gamma'_0,a\mathord{:}\derefK(K)$ with $\Gamma'_0 \in \derefT[\Gamma_0]$.  It follows from the induction hypothesis that $\types \Gamma'_0$ in $\lambda$HOL.  Therefore, the result follows from \rlnm{Kinding} in $\lambda$HOL.
    \item[\rlnm{Typing}] Then $\Gamma$ has shape $\Gamma_0,x\mathord{:}A$ and $x \notin \subj(\camrevreplace{\Gamma'}{\Gamma_0})$.  Therefore, $\Gamma'$ has shape $\Gamma'_0,\camrevreplace{a}{x}\mathord{:}\derefT_i(A)$ for some $i \in \{1,2\}$ with $\Gamma'_0 \in \derefT[\Gamma_0]$.  It follows from the induction hypothesis that $\Gamma'_0 \types \derefT(A) : \starsort^2$ in $\lambda$HOL, and therefore also $\Gamma'_0 \types \derefT_i(A) : \starsort$.  Hence, $\types \Gamma'_0,x\mathord{:}\derefT_i(A)$ by \rlnm{Typing} in $\lambda$HOL.
  \end{description}
\end{proof}

When considering the mapping of proofs, we need to take into account that a proof step \rlnm{Conv} may involve conversion from one formula to another.

\begin{lemma}[Translation Type Substitution]\label{lem:translation-subst}
  $\derefT(A[B/a]) = \derefT(A)[\derefT(B)/a]$
\end{lemma}
\begin{proof}
  By induction on the structure of $A$:
  \begin{itemize}
    \item If $A$ is a variable there are two cases.  If $A$ is the distinguished variable $a$ of the statement, then $\derefT(A[B/a]) = \derefT(B) = \derefT(A)[\derefT(B)/a]$.  Otherwise, $\derefT(A[B/a]) = \derefT(A) = \derefT(A)[\derefT(B)/a]$.
    \item If $A$ is an application $A_1\,A_2$ then $\derefT(A[B/a]) = \derefT(A_1[B/a]\,A_2[B/a])$.  It follows from the induction hypotheses that $\derefT(A_i[B/a]) = \derefT(A_i)[\derefT(B)/a]$ from which the result follows by definition.
    \item \camrevreplace{If $A$ is an abstraction $\tyabs{a}{K}{A'}$ then $\derefT(A[B/a]) = \tyabs{a}{\derefK(K)}{\derefT(A'[B/a])}$.  Then the result follows from the induction hypothesis and by definition.}{If $A$ is an abstraction $\tyabs{c}{K}{A'}$, choose the bound variable $c$ fresh for $a$ and $B$. Then $\derefT(A[B/a]) = \tyabs{c}{\derefK(K)}{\derefT(A'[B/a])}$. The result follows from the induction hypothesis and by definition.}
    \item If $A$ is an arrow $A_1 \to A_2$ then $\derefT(A[B/a])$ is:
          \[
            \abra{\pi_1(\derefT(A_1[B/a])) \to \pi_1(\derefT(A_2[B/a])),\pi_1(\derefT(A_1[B/a])) \wedge \pi_2(\derefT(A_2[B/a]))}
          \]
          Then it follows from the induction hypotheses that this is:
          \[
            \abra{\pi_1(\derefT(A_1)[\derefT(B)/a]) \to \pi_1(\derefT(A_2)[\derefT(B)/a]),\pi_1(\derefT(A_1)[\derefT(B)/a]) \wedge \pi_2(\derefT(A_2)[\derefT(B)/a])}
          \]
          which, since each $\pi_i$ is closed, is exactly $\derefT(A)[\derefT(B)/a]$.
    \item If $A$ is a universal type $\PiType{c}{K}{C}$ then $\derefT(A[B/a])$ is:
          \[
            \abra{\PiType{c}{\derefK(K)}{\pi_1(\derefT(\camrevreplace{A}{C}[B/a]))}, \SigmaType{c}{\derefK(K)}{\pi_2(\derefT(\camrevreplace{A}{C}[B/a]))}}
          \]
          It follows from the induction hypotheses that this is just:
          \[
            \abra{\PiType{c}{\derefK(K)}{\pi_1(\derefT(\camrevreplace{A}{C})[\derefT(B)/a])}, \SigmaType{c}{\derefK(K)}{\pi_2(\derefT(\camrevreplace{A}{C})[\derefT(B)/a])}}
          \]
          which, since each $\pi_i$ is closed, is exactly $\derefT(A)[\derefT(B)/a]$.
    \item \camrevreplace{If $A$ is a strong negation $\sneg{B}$ then $\derefT(A[B/a]) = \sneg{\derefT(A[B/a])}$ and the result follows from the induction hypothesis.}{If $A$ is a strong negation $\sneg{C}$ then $\derefT(A[B/a]) = \abra{\pi_2(\derefT(C[B/a])),\pi_1(\derefT(C[B/a]))}$. Applying the induction hypothesis to $C$ gives $\derefT(A)[\derefT(B)/a]$, as required.}
  \end{itemize}
\end{proof}

\begin{lemma}[Conversion Translation]\label{lem:2hol-conv-trans}
  If $A \conv B$ in \holtty{} then $\derefT(A) \conv \derefT(B)$ in \holty{}.
\end{lemma}
\begin{proof}
  The proof is by induction on convertibility.
  \begin{description}
    \item[\rlnm{CRefl}] In this case, the result is immediate.
    \item[\rlnm{CSymm}] In this case, it follows from the induction hypothesis that $\revreplace{\derefT(N) \conv \derefT(M)}{\derefT(B) \conv \derefT(A)}$ in \holty, whence the result follows from the symmetry of $\conv$ in \holty{}.
    \item[\rlnm{CTrans}] In this case, it follows from the induction hypotheses that, for \revreplace{some $P$}{some $C$}, $\revreplace{\derefT(M) \conv \derefT(P)}{\derefT(A) \conv \derefT(C)}$ and $\revreplace{\derefT(P) \conv \derefT(N)}{\derefT(C) \conv \derefT(B)}$ in \holty{}, and the result follows from the transitivity of that relation.
    \item[\rlnm{CArrow}] In this case, $A$ is of shape $A' \to C$ and $B$ of shape $B' \to D$ and it follows from the induction hypotheses that $\derefT(A') \conv \derefT(B')$ and $\derefT(C) \conv \derefT(D)$ in \holty{}.  Then, by the compatibility rules of conversion in \holty{}, as required:
          \[
            \abra{\pi_1(\derefT(A')) \to \pi_1(\derefT(C)), \pi_1(\derefT(A')) \wedge \pi_2(\derefT(C))} \conv \abra{\pi_1(\derefT(B')) \to \pi_1(\derefT(D)), \pi_1(\derefT(B')) \wedge \pi_2(\derefT(D))}
          \]
    \item[\rlnm{CAbs}] In this case, $A$ is of shape $\tyabs{a}{K}{A'}$ and $B$ of shape $\tyabs{a}{K}{B'}$ and it follows from the induction hypothesis that \camrevreplace{$A' \conv B'$}{$\derefT(A') \conv \derefT(B')$} in \holty{}.  Therefore, the result follows by \rlnm{CAbs} in \holty{}.
    \item[\rlnm{CForall}] In this case, $A$ is of shape $\PiType{a}{K}{A'}$ and $B$ of shape $\PiType{a}{K}{B'}$ and it follows from the induction hypothesis that \camrevreplace{$A' \conv B'$}{$\derefT(A') \conv \derefT(B')$} in \holty{}.  Then it follows from the compatibility rules of conversion in \holty{} that, as required:
          \[
            \abra{\PiType{a}{\derefK(K)}{\derefT_1(A')},\SigmaType{a}{\derefK(K)}{\derefT_2(A')}} \conv \abra{\PiType{a}{\derefK(K)}{\derefT_1(B')},\SigmaType{a}{\derefK(K)}{\derefT_2(B')}}
          \]
    \item[\rlnm{CApp}] In this case, $A$ is of shape $A'\,C$ and $B$ is of shape $B'\,D$ and it follows from the induction hypotheses that $\derefT(A') \conv \derefT(B')$ and $\derefT(C) \conv \derefT(D)$ in \holty{}.  Therefore $\derefT(A'\,C) \conv \derefT(B'\,D)$ follows by \rlnm{CApp} in \holty{}.
    \item[\rlnm{CSneg}] In this case, $A$ is of shape $\sneg{A'}$ and $B$ is of shape $\sneg{B'}$ and it follows from the induction hypothesis that $\derefT(A') \conv \derefT(B')$ in \holty{}.  Therefore, it follows from the compatibility rules of conversion in \holty{} that $\abra{\derefT_2(A'),\derefT_1(A')} \conv \abra{\derefT_2(B'),\derefT_1(B')}$.
    \item[\rlnm{CRedex}] In this case, $A$ is of shape $(\tyabs{a}{K}{A'})\,B'$ and $B$ is of shape $A'\,[B'/a]$.  It follows from \rlnm{CRedex} in \holty{} that $(\tyabs{a}{\derefK(K)}{\derefT(A')})\,\derefT(B') \conv \derefT(A')[\derefT(B')/a]$, and then the result follows by definition and by Lemma~\ref{lem:translation-subst}.
  \end{description}
\end{proof}

\camrevadd{The lifted constants used in Figure~\ref{fig:hol2-to-hol-mappings} have the following types in \holty{}. Here $L$ is any target kind; the bound variables are chosen fresh.}
\begin{align*}
  \camrevadd{\widehat{\pi}_1}     & \camrevadd{: \PiType{ab}{\starsort^2}{(\pi_1(a)\wedge\pi_2(b))\to\pi_1(a)}}                           \\
  \camrevadd{\widehat{\pi}_2}     & \camrevadd{: \PiType{ab}{\starsort^2}{(\pi_1(a)\wedge\pi_2(b))\to\pi_2(b)}}                           \\
  \camrevadd{\widehat{\mkPair}}   & \camrevadd{: \PiType{ab}{\starsort^2}{\pi_1(a)\to\pi_2(b)\to(\pi_1(a)\wedge\pi_2(b))}}                \\
  \camrevadd{\widehat{\pack}_L}   & \camrevadd{: \PiType{f}{L\to\starsort^2}{\PiType{b}{L}{\pi_2(f\,b)\to\SigmaType{a}{L}{\pi_2(f\,a)}}}} \\
  \camrevadd{\widehat{\unpack}_L} & \camrevadd{: \PiType{f}{L\to\starsort^2}{(\SigmaType{a}{L}{\pi_2(f\,a)})\to}}                         \\
                                  & \camrevadd{\qquad\PiType{c}{\starsort^2}{(\PiType{a}{L}{\pi_2(f\,a)\to\pi_2(c)})\to\pi_2(c)}}
\end{align*}
\camrevadd{Each typing follows by instantiating the ordinary encoding at the displayed projections and then introducing the type or term parameters. Thus translated proof terms are typed directly with these constants. Reducing an adapter merely exposes the corresponding ordinary encoding; no subject expansion principle for arbitrary terms is needed.}

\begin{ulemma}[Restatement of Lemma~\ref{lem:2hol-proof-translations}]
  All of the following:
  \begin{enumerate}
    \item If $\Gamma \types M : A;\,\Delta$ in \holtty{} then $\derefT_1(\Gamma),\derefT_2(\Delta^R) \types \derefP(M) : \derefT_1(A)$ in \holty{}.
    \item If $\Gamma;\,M : A \types \Delta$ in \holtty{} then $\derefT_1(\Gamma),\derefT_2(\Delta^R) \types \derefP(M) : \derefT_2(A)$ in \holty{}.
  \end{enumerate}
\end{ulemma}
\begin{proof}
  The proof is by mutual induction on the derivations.
  \begin{description}
    \item[\rlnm{Id-L}] In this case, $M$ is a variable $x$ and we may assume \camrevreplace{$\types \Gamma,\Delta$}{$\types \Gamma,\Delta^R$} and $x:A \in \revreplace{\Delta}{\Delta^R}$.  By Lemma~\ref{lem:hol2-kind-mapping-preservation}, since $\derefT_1(\Gamma),\derefT_2(\Delta^R) \in \derefT[\Gamma,\Delta^R]$, it follows that $\types \derefT_1(\Gamma),\derefT_2(\Delta^R)$.  Moreover, $x:\derefT_2(A) \in \derefT_2(\Delta^R)$.  Hence, the result follows by \camrevreplace{\rlnm{Var} in \holty{}}{\rlnm{Id} in \holty{}}.
    \item[\rlnm{Id-R}] Symmetrical to the previous case.
    \item[\rlnm{Conv}] In this case, we may assume $A \conv B$, $\Gamma \types A : \starsort$, and therefore that $\derefT(A) \conv \derefT(B)$ and $\derefT_1(\Gamma) \types \derefT(A) : \starsort^2$ by Lemmas~\ref{lem:2hol-conv-trans},~\ref{lem:hol2-kind-mapping-preservation}.  It follows from the induction hypothesis that $\derefT_1(\Gamma),\derefT_2(\Delta^R) \types \derefP(M) : \derefT_1(B)$ in \holty{} and so, by Lemma~\ref{lem:holty-prereqs}, $\types \derefT_1(\Gamma),\derefT_2(\Delta^R)$ in \holty{}.  Therefore, by Lemma~\ref{lem:holty-prereqs}, $\derefT_1(\Gamma),\derefT_2(\Delta^R) \types \derefT(A) : \starsort^2$ in \holty{}, and therefore that $\derefT_1(\Gamma),\derefT_2(\Delta^R) \types \derefT_1(A) : \starsort$.  Then, the result follows by \rlnm{Conv} in \holty{}.
    \item[\rlnm{$\forall$-Intro-L}] In this case, $M$ is of shape $\pack_K\,[\tyabs{a}{K}{B}]\,[A']\,N$ and $A$ is $\PiType{a}{K}{B}$.  We may assume $\Gamma \types A' : K$ and so it follows from Lemma~\ref{lem:hol2-kind-mapping-preservation} that $\derefT_1(\Gamma) \types \derefT(A') : \derefK(K)$.  It follows from the induction hypothesis that $\derefT_1(\Gamma),\derefT_2(\Delta^R) \types \derefP(N) : \camrevreplace{\derefT_1}{\derefT_2}(B[A'/a])$.  Therefore, it follows from \holty{} Environment Consistency that $\types \derefT_1(\Gamma),\derefT_2(\Delta^R)$ and so by Type Weakening (Lemma~\ref{lem:holty-prereqs}), $\derefT_1(\Gamma),\derefT_2(\Delta^R) \types \derefT(A') : \derefK(K)$.  By Lemma~\ref{lem:translation-subst}, $\camrevreplace{\derefT_1}{\derefT_2}(B[A'/a]) = \camrevreplace{\pi_1}{\pi_2}\,(\derefT(B)[\derefT(A')/a]) = \camrevreplace{\derefT_1}{\derefT_2}(B)[\derefT(A')/a]$.  Hence, in \holty{}, we can use \rlnm{$\exists$-Intro} to obtain:
          \[
            \derefT_1(\Gamma),\derefT_2(\Delta^R) \types \pack_{\derefK(K)}\,[\tyabs{a}{\derefK(K)}{\camrevreplace{\derefT_1}{\derefT_2}(B)}]\,[\derefT(A')]\,\derefP(N) : \SigmaType{a}{\derefK(K)}{\camrevreplace{\derefT_1}{\derefT_2}(B)}
          \]
          \camrevadd{Up to conversion of its type arguments, the displayed term agrees with the reduct of the directly typable application
            \[
              \derefP(\pack_K)\,[\derefT(\tyabs{a}{K}{B})]\,[\derefT(A')]\,\derefP(N).
            \]
            Its family reduces to $\tyabs{a}{\derefK(K)}{\derefT_2(B)}$.} This final type is, \camrevreplace{by definition}{by conversion}, $\derefT_2(\PiType{a}{K}{B})$, as required.
    \item[\rlnm{$\forall$-Elim-L}] In this case, $M$ is of shape $\unpack_K\,[\tyabs{a}{K}{B}]\,M'\,[A]\,(\Tyabs{a}{K}{\tyabs{x}{\camrevreplace{B}{\sneg{B}}}{N}})$.  We may assume that \camrevreplace{$a \notin \subj(\Gamma,\Delta)$}{$a \notin \fv(\Gamma,\Delta,A)$} and $x \notin \subj(\Gamma,\Delta)$.  It follows from the induction hypotheses that both of the following are derivable in \holty{}:
          \begin{itemize}
            \item $\derefT_1(\Gamma),\derefT_2(\Delta^R) \types \derefP(M') : \derefT_2(\PiType{a}{K}{B})$
            \item $\derefT_1(\Gamma),a:\derefK(K),\derefT_2(\Delta^R),x:\derefT_2(B) \types \derefP(N) : \derefT_2(A)$
          \end{itemize}
          It follows from the \rlnm{Conv} rule in \holty{} that, therefore also
          \[
            \derefT_1(\Gamma),\derefT_2(\Delta^R) \types \derefP(M') : \SigmaType{a}{\derefK(K)}{\derefT_2(B)}.
          \]
          Since $a$ does not occur freely in $\Delta$, by an environment permutation in \holty{} we may also obtain:
          \[
            \derefT_1(\Gamma),\derefT_2(\Delta^R),a:\derefK(K),x:\derefT_2(B) \types \derefP(N) : \derefT_2(A)
          \]
          From this, it follows in \holty{} that:
          \[
            \derefT_1(\Gamma),\derefT_2(\Delta^R) \types \Tyabs{a}{\derefK(K)}{\tyabs{x}{\derefT_2(B)}{\derefP(N)}} : \PiType{a}{\derefK(K)}{\derefT_2(B) \to \derefT_2(A)}
          \]
          and, therefore, that, in \holty{}:
          \[
            \derefT_1(\Gamma),\derefT_2(\Delta^R) \types \unpack_{\derefK(K)}\,[\tyabs{a}{\derefK(K)}{\derefT_2(B)}]\,\derefP(M')\,[\derefT_2(A)]\,(\Tyabs{a}{\derefK(K)}{\tyabs{x}{\derefT_2(B)}{\derefP(N)}}) : \derefT_2(A)
          \]
          \camrevadd{The application of $\widehat{\unpack}_{\derefK(K)}$ in $\derefP(M)$ has this result type directly. The translated continuation has annotation $\derefT_1(\sneg{B})$, which is convertible to $\derefT_2(B)$, so target context conversion (Lemma~\ref{lem:holty-prereqs}(iv)) types its body under this annotation; after the administrative adapter reductions, its type arguments and annotations are therefore convertible to those displayed above.}
    \item[\rlnm{$\forall$-Intro-R}] In this case, $M$ is of shape $\Tyabs{a}{K}{M'}$ and $A$ of shape $\PiType{a}{K}{B}$ and we may assume that $a \notin \fv(\Gamma,\Delta)$.  It follows from the induction hypothesis that $\derefT_1(\Gamma),a:\derefK(K),\derefT_2(\Delta^R) \camrevreplace{\types M':B}{\types \derefP(M'):\derefT_1(B)}$ in \holty{}.  After a permutation of environments, we obtain $\derefT_1(\Gamma),\derefT_2(\Delta^R) \types \derefP(\Tyabs{a}{K}{M'}) : \camrevreplace{\PiType{a}{\derefK(K)}{\derefT(B)}}{\PiType{a}{\derefK(K)}{\derefT_1(B)}}$ by \camrevreplace{\rlnm{$\forall$-Intro}}{\rlnm{$\forall$-Intro-R}} in \holty{}.  The result then follows from \rlnm{Conv} in \holty{}.
    \item[\rlnm{$\forall$-Elim-R}] In this case, $M$ is of shape $M'\,[A']$ and $A$ is of shape $B[A'/a]$ and we may assume $\Gamma \types A' : K$.  It follows from the induction hypothesis that $\derefT_1(\Gamma),\derefT_2(\Delta^R) \types \derefP(M') : \derefT_1(\PiType{a}{K}{B})$.  By \rlnm{Conv}, hence $\derefT_1(\Gamma),\derefT_2(\Delta^R) \types \derefP(M') \camrevreplace{\types}{:} \PiType{a}{\derefK(K)}{\derefT_1(B)}$.  Therefore, $\types \derefT_1(\Gamma),\derefT_2(\Delta^R)$ in \holty{}.  It follows from Lemma~\ref{lem:hol2-kind-mapping-preservation} that, in \holty{}, $\derefT_1(\Gamma) \types \camrevreplace{\derefT_1}{\derefT}(A') : \derefK(K)$, and, by Type Weakening (Lemma~\ref{lem:holty-prereqs}), $\derefT_1(\Gamma),\derefT_2(\Delta^R) \types \camrevreplace{\derefT_1}{\derefT}(A') : \derefK(K)$.  Therefore, we obtain $\derefT_1(\Gamma),\derefT_2(\Delta^R) \types \derefP(M')\,[\camrevreplace{\derefT_1}{\derefT}(A')] : \derefT_1(B)[\camrevreplace{\derefT_1}{\derefT}(A')/a]$ in \holty{} by \rlnm{$\forall$-Elim\camrevadd{-R}} and the result follows from Translation Substitution (Lemma~\ref{lem:translation-subst}).
    \item[\rlnm{$\to$-Intro-L}] In this case, $M$ is of shape $M' \cdot{} N$ and $A$ is of shape $A' \to B$.  It follows from the induction hypotheses that $\derefT_1(\Gamma),\derefT_2(\Delta^R) \types \derefP(M') : \derefT_1(A')$ and $\derefT_1(\Gamma),\derefT_2(\Delta^R) \types \derefP(N) : \derefT_2(B)$ in \holty{}.  Therefore, by the derived rule \rlnm{$\wedge$-Intro} we obtain
          \[
            \begin{aligned}
              \derefT_1(\Gamma),\derefT_2(\Delta^R) & \types \camrevreplace{(\derefP(M'), \derefP(N))}{\mkPair\,[\derefT_1(A')]\,[\derefT_2(B)]\,\derefP(M')\,\derefP(N)} \\
                                                    & \qquad : \derefT_1(A') \wedge \derefT_2(B)
            \end{aligned}
          \]
          in \holty{}, and \camrevreplace{the required result follows from \rlnm{Conv} in \holty{}.}{the required result follows from the type of $\widehat{\mkPair}$ and \rlnm{Conv} in \holty{}.}
    \item[\rlnm{$\to$-Elim-L$_1$}] In this case, $M$ is of shape $\iota_1\,M'$.  It follows from the induction hypothesis that $\derefT_1(\Gamma),\derefT_2(\Delta^R) \types \derefP(M') : \derefT_2(A \to B)$ in \holty{}.  By \rlnm{Conv}, then $\derefT_1(\Gamma),\derefT_2(\Delta^R) \types \derefP(M') : \derefT_1(A) \wedge \derefT_2(B)$ in \holty{}. \camrevreplace{The result then follows by the derived rule \rlnm{$\wedge$-Elim$_1$}.}{The result then follows by the type of $\widehat{\pi}_1$, which instantiates the derived rule \rlnm{$\wedge$-Elim$_1$} at the projected arguments.}
    \item[\rlnm{$\to$-Elim-L$_2$}] In this case, $M$ is of shape $\iota_2\,M'$.  It follows from the induction hypothesis that $\derefT_1(\Gamma),\derefT_2(\Delta^R) \types \derefP(M') : \derefT_2(A' \to A)$ in \holty{}.  By \rlnm{Conv}, then $\derefT_1(\Gamma),\derefT_2(\Delta^R) \types \derefP(M') : \derefT_1(A') \wedge \derefT_2(A)$ in \holty{}. \camrevreplace{The result then follows by the derived rule \rlnm{$\wedge$-Elim$_2$}.}{The result then follows by the type of $\widehat{\pi}_2$, which instantiates the derived rule \rlnm{$\wedge$-Elim$_2$} at the projected arguments.}
    \item[\rlnm{$\to$-Intro-R}] In this case, $M$ is of shape $\camrevreplace{\tyabs{x}{A'}{M}}{\tyabs{x}{A'}{M'}}$ and $A$ is of shape $A' \to B$.  We may assume that $x$ does not occur in $\Gamma$ or $\Delta$.  It follows from the induction hypothesis that $\derefT_1(\Gamma),x:\derefT_1(\revreplace{A}{A'}),\derefT_2(\Delta^R) \types \derefP(M') : \derefT_1(B)$ in \holty{}.  By an environment permutation, we obtain $\derefT_1(\Gamma),\derefT_2(\Delta^R),x:\derefT_1(\revreplace{A}{A'}) \types \derefP(M') : \derefT_1(B)$ in \holty{} and $\derefT_1(\Gamma),\derefT_2(\Delta^R) \types \tyabs{x}{\derefT_1(A')}{\derefP(M')} : \derefT_1(A') \to \derefT_1(B)$.  The result then follows from \rlnm{Conv} in \holty{}.
    \item[\rlnm{$\to$-Elim-R}] In this case $M$ is of shape $M'\,N$.  It follows from the induction hypotheses that there is some type $A'$ such that the following are derivable in \holty{}:
          \begin{itemize}
            \item $\derefT_1(\Gamma),\derefT_2(\Delta^R) \types \derefP(M') : \derefT_1(A' \to A)$
            \item $\derefT_1(\Gamma),\derefT_2(\Delta^R) \types \derefP(N) : \derefT_1(A')$
          \end{itemize}
          Therefore, by \rlnm{Conv}, we also obtain $\derefT_1(\Gamma),\derefT_2(\Delta^R) \types \derefP(M') : \derefT_1(A') \to \derefT_1(A)$ in \holty{}.  The result then follows by \camrevreplace{\rlnm{$\to$-Elim}}{\rlnm{$\to$-Elim-R}} in \holty{}.
    \item[\rlnm{$\sneg{}$-Intro-L}] In this case, $A$ is of shape $\sneg{A'}$ and it follows from the induction hypothesis that $\derefT_1(\Gamma),\derefT_2(\Delta^R) \types \derefP(M) : \derefT_1(A')$.  Therefore, by \rlnm{Conv}, $\derefT_1(\Gamma),\derefT_2(\Delta^R) \types \derefP(M) : \pi_2(\abra{\derefT_2(A'),\derefT_1(A')})$, in \holty{}, as required.
    \item[\rlnm{$\sneg{}$-Intro-R}] In this case, $A$ is of shape $\sneg{A'}$ and it follows from the induction hypothesis that $\derefT_1(\Gamma),\derefT_2(\Delta^R) \types \derefP(M) : \derefT_2(A')$.  Therefore, by \rlnm{Conv}, we can obtain $\derefT_1(\Gamma),\derefT_2(\Delta^R) \types \derefP(M) : \pi_1(\abra{\derefT_2(A'),\derefT_1(A')})$, in \holty{}, as required.
    \item[\rlnm{$\sneg{}$-Elim-L}] In this case, it follows from the induction hypothesis that $\derefT_1(\Gamma),\derefT_2(\Delta^R) \types \derefP(M) : \derefT_2(\sneg{A})$ in \holty{}.  Therefore, we can also obtain $\derefT_1(\Gamma),\derefT_2(\Delta^R) \types \derefP(M) : \derefT_1(A)$, in \holty{}, as required.
    \item[\rlnm{$\sneg{}$-Elim-R}] In this case, it follows from the induction hypothesis that $\derefT_1(\Gamma),\derefT_2(\Delta^R) \types \derefP(M) : \derefT_1(\sneg{A})$ in \holty{}.  Therefore, we can also obtain $\derefT_1(\Gamma),\derefT_2(\Delta^R) \types \derefP(M) : \derefT_2(A)$, in \holty{}, as required.
  \end{description}
\end{proof}

  \section{Supplementary Material for Section~\ref{sec:hol2-proofs}}

In this appendix we include the details of the development of the existence property for \holtty{}.

\begin{lemma}[Translation Term Substitution]\label{lem:2hol-derefP-subst}
  Both of the following:
  \begin{enumerate}[(i)]
    \item $\derefP(M)[\derefP(N)/x] = \derefP(M[N/x])$
    \item $\derefP(M)[\derefT(A)/a] = \derefP(M[A/a])$
  \end{enumerate}
\end{lemma}
\begin{proof}
  Part (i) is a straightforward induction on $M$.
  Part (ii) is an induction on $M$:
  \begin{itemize}
    \item If $M$ is a variable, $\cdot$, $\iota_1$, $\iota_2$, $\pack_K$ or $\unpack_K$ then the result is clear since $a$ cannot occur free.
    \item If $M$ is a term application, then the result follows from the induction hypotheses.
    \item If $M$ is a type application $M'\,[A']$, then
          \[
            (\derefP(M')\,[\derefT(A')])[\derefT(A)/a] = \derefP(M')[\derefT(A)/a]\camrevreplace{\,[A'[\derefT(A)/a]]}{\,[\derefT(A')[\derefT(A)/a]]}
          \]
          and the result follows from the induction hypothesis and Lemma~\ref{lem:translation-subst}.
    \item If $M$ is a term abstraction $\tyabs{x}{A'}{M'}$ then:
          \[
            (\tyabs{x}{\derefT_1(A')}{\derefP(M')})[\derefT(A)\revadd{/a}] = \tyabs{x}{\derefT_1(A')[\derefT(A)/a]}{\derefP(M')[\derefT(A)/a]}
          \]
          Note: $\derefT_1(A')[\derefT(A)/a] = \pi_1(\derefT(A'))[\derefT(A)/a] = \pi_1(\derefT(A')[\derefT(A)/a])$ and this is $\derefT_1(A'[A/a])$ by Lemma~\ref{lem:translation-subst}.  Then the result follows from the induction hypothesis.
    \item If $M$ is a type abstraction $\Tyabs{b}{K}{M'}$ then $(\Tyabs{b}{\camrevreplace{K}{\delta(K)}}{\derefP(M')})[\derefT(A)/a] = \Tyabs{b}{\camrevreplace{K}{\delta(K)}}{\derefP(M')[\derefT(A)/a]}$ and the result follows immediately from the induction hypothesis.
  \end{itemize}
\end{proof}

\begin{ulemma}[Restatement of Lemma ~\ref{lem:2hol-reduction-translation}]
  If $M \ped N$ in \holtty{} then $\derefP(M) \pedp \derefP(N)$ in \holty{}.
\end{ulemma}
\begin{proof}
  The proof is by induction on the reduction.
  \begin{itemize}
    \item If the reduction is the contraction of a term $\beta$-redex then $M$ is of shape $(\tyabs{x}{A}{M'})\,N'$ and $N$ is $M'[N'/x]$.  Then $\derefP(M) = (\tyabs{x}{\derefT_1(\revreplace{A'}{A})}{\derefP(M')})\,\derefP(N')$.  This contracts to $\derefP(M')[\derefP(N')/x]$ in \holty{} which is exactly $\derefP(N)$ by Lemma~\ref{lem:2hol-derefP-subst}.
    \item If the reduction is the contraction of a type $\beta$-redex then $M$ is of shape $(\Tyabs{a}{K}{M'})\,[A]$ and $N$ is $M'[A/a]$.  Then $\derefP(M) = (\Tyabs{a}{\camrevreplace{K}{\delta(K)}}{\derefP(M')})\,[\derefT(A)/a]$ and it contracts to $\derefP(M')[\derefT(A)/a]$ in \holty{}.  This is exactly $\derefP(N)$ by Lemma~\ref{lem:2hol-derefP-subst}.
    \item If the reduction is the contraction of an unpack-pack redex, then $M$ is of shape:
          \[
            \unpack_K\,[\tyabs{a}{K}{A}]\,(\pack_K\,[\tyabs{a}{K}{B}]\,[C]\,\camrevreplace{N}{N_0})\,[D]\,(\Tyabs{a}{K}{\tyabs{x}{E}{M'}})
          \]
          \camrevreplace{Then it follows that $\derefP(M)$ is exactly:}{By reducing the lifted adapters and the ordinary unpack-pack encoding, $\derefP(M)$ reduces in a nonempty sequence to:}
          \camrevreplace{\[
              \unpack_K\,[\tyabs{a}{K}{\derefT(A)}]\,(\pack_K\,[\tyabs{a}{K}{\derefT(B)}]\,[\derefT(C)]\,\derefP(N))\,[\derefT(D)]\,(\Tyabs{a}{K}{\tyabs{x}{\derefT_1(E)}{\derefP(M')}})
            \]}{\[
              (\Tyabs{a}{\derefK(K)}{\tyabs{x}{\derefT_1(E)}{\derefP(M')}})\,[\derefT(C)]\,\derefP(N_0)
            \]}
          \camrevreplace{This reduces, in \holty{}, to $\derefP(M')[\derefT(C)/a][\derefP(N)/x]$ in 9 steps, and this is exactly $\derefP(N)$.}{Two further $\beta$-contractions yield $\derefP(M')[\derefT(C)/a][\derefP(N_0)/x]$, which equals $\derefP(M'[C/a][N_0/x])=\derefP(N)$ by Lemma~\ref{lem:2hol-derefP-subst}.}
    \item If the reduction is the contraction of a dot-redex, then $M$ is of shape $\iota_i\,[A]\,[B]\,(M_1 \cdot_{C,D} M_2)$ and $N$ is $M_i$.  Then $\derefP(M) = \pi_i\,[\derefT(A)]\,[\derefT(B)]\,(\mathsf{mkPair}\,[\derefT(C)]\,[\derefT(D)]\,\derefP(M_1)\,\derefP(M_2))$, which reduces to $\derefP(M_i)$ in 11 steps in \holty{}, as required.
    \item The remaining cases follow directly from the induction hypothesis.
  \end{itemize}
\end{proof}

\begin{lemma}[Basic Properties of Neg-Conversion]
  All of the following:
  \begin{enumerate}[(i)]
    \item $\nconv$ is an equivalence.
    \item If $A_1 \to A_2 \nconv B_1 \to B_2$ then $A_1 \conv B_1$ and $\revreplace{A_1}{A_2} \conv B_2$.
    \item If $\PiType{a}{K}{A} \nconv \PiType{a}{K'}{A'}$ then $K=K'$ and $A \conv A'$.
    \item $\PiType{a}{K}{A} \not\nconv B \to C$.
    \item If $\PiType{a}{K}{A} \nconv \sneg{B}$ then $B \nconv \sneg{(\PiType{a}{K}{A})}$.
    \item If $A_1 \to A_2 \nconv \sneg{B}$ then $B \nconv \sneg{(A_1 \to A_2)}$.
  \end{enumerate}
\end{lemma}
\begin{proof}
  Parts (i) -- (iv) \camrevreplace{follows}{follow} immediately by the definition and the properties of $\beta$-conversion.  For (v) and (vi) we can observe the following invariant: if $\sneg{}^k(A) \nconv B$ and $A$ is not a strong negation, then $B = \sneg{}^m(C)$ with $k \equiv m\,(\text{mod}\,2)$ and $C \conv A$.
\end{proof}

\begin{remark}
  The introduction of the concept of neg-conversion is natural given our system's commitment to the idea that the refutations of some type $A$ are identified with proofs of $\sneg{A}$\footnote{The same approach is taken in Wansing's $\lambda^C$, a term assignment system for a propositional logic with strong negation \cite{wansing93}.}.  So, the strong negation rules of the type system regard the movement between refutations of $A$ and proofs of $\sneg{A}$ implicitly (and similarly refutations of $\sneg{A}$ and proofs of $A$), similarly to the treatment of conversion between types.  The development could be done with introduction and elimination forms, say $\sigma$ and $\sigma^{-1}$, with the contraction $\sigma^{-1}\,[A]\,(\sigma\,[B]\,M) \ped M$ and the rules:
  \begin{mathpar}
    \infer{
      \Gamma \types M:A;\, \Delta
    }
    {
      \Gamma;\, \sigma\,[A]\,M : \sneg{A} \types \Delta
    }
    \and
    \infer{
      \Gamma;\,M:A \types \Delta
    }
    {
      \Gamma \types \sigma\,[A]\,M : \sneg{A};\,\Delta
    }
    \and
    \infer{
      \Gamma;\,M:\sneg{A} \types \Delta
    }
    {
      \Gamma \types \sigma^{-1}\,[A]\,M : A;\,\Delta
    }
    \and
    \infer{
      \Gamma \types M:\sneg{A};\, \Delta
    }
    {
      \Gamma;\, \sigma^{-1}\,[A]\,M : A \types \Delta
    }
  \end{mathpar}
  So that proofs of $A$ and refutations of $\sneg{A}$ are merely in correspondence.  Then the translation into \holty{} could interpret both $\sigma$ and $\sigma^{-1}$ as the identity, $\derefP(\sigma) = \derefP(\sigma^{-1}) = \Tyabs{a}{\starsort^2}{\tyabs{x}{a}{x}}$.
\end{remark}

\begin{ulemma}[Restatement of Inversion, Lemma~\ref{lem:2hol-inversion}]
  All of the following:
  \begin{enumerate}[(1)]
    \item If $\Gamma \types \tyabs{x}{A}{M} : B;\,\Delta$ and $B \nconv B_1 \to B_2$ then $A \conv B_1$ and $\Gamma,\,x\ty:A \types M : B_2;\,\Delta$.
    \item If $\Gamma;\,\tyabs{x}{A}{M} : B \types \Delta$ and $B \nconv \sneg{(B_1 \to B_2)}$ then $A \conv B_1$ and $\Gamma,\,x\ty:A \types M : B_2;\,\Delta$.
    \item If $\Gamma \types \Tyabs{a}{K}{M} : B;\,\Delta$ and $B \nconv \PiType{a}{K'}{C}$ then $K = K'$ and $\Gamma,\,a\ty:K \types M : C;\,\Delta$.
    \item If $\Gamma;\,\Tyabs{a}{K}{M} : B\camrevreplace{;}{\types}\Delta$ and $B \nconv \sneg{(\PiType{a}{K'}{C})}$ then $K = K'$ and $\Gamma,\,a\ty:K \types M : C;\,\Delta$.
    \item If $\Gamma;\,M \cdot_{A_1,A_2} N : B \types \Delta$ and $B \nconv B_1 \to B_2$ then $A_1 \conv B_1$, $A_2 \conv B_2$, $\Gamma \types M : B_1;\,\Delta$ and $\Gamma;\,N:B_2 \types \Delta$.
    \item If $\Gamma \types M \cdot_{A_1,A_2} N : B;\,\Delta$ and $B \nconv \sneg{(B_1 \to B_2)}$ then $A_1 \conv B_1$, $A_2 \conv B_2$, $\Gamma \types M : B_1;\,\Delta$ and $\Gamma;\,N:B_2 \types \Delta$.
    \item If $\Gamma;\,\pack_K\,[\tyabs{a}{K}{A}]\,[C]\,N : B \types \Delta$ and $B \nconv \PiType{a}{K'}{D}$ then $K = K'$ and $A \conv D$ and $\Gamma \types C : K$ and $\Gamma;\,N:A[C/a] \types \Delta$.
    \item If $\Gamma \types \pack_K\,[\tyabs{a}{K}{A}]\,[C]\,N : B;\,\Delta$ and $B \nconv \sneg{(\PiType{a}{K'}{D})}$ then $K = K'$ and $A \conv D$ and $\Gamma \types C : K$ and $\Gamma;\,N:A[C/a] \types \Delta$.
  \end{enumerate}
\end{ulemma}
\begin{proof}
  We prove all parts simultaneously by induction on derivability.
  \begin{description}
    \item[\rlnm{Id-L},\rlnm{Id-R}] Do not apply.
    \item[\rlnm{Conv}] Each part then follows immediately from the induction hypothesis.
    \item[\rlnm{$\forall$-Intro-L}] In this case, $B$ is of shape $\PiType{a}{K}{A}$ and we aim to verify the conclusion of case (7).  By inspection, $\Gamma \types C : K$ and $\Gamma;\,N:A[C/a] \types \Delta$.  Assume $\PiType{a}{K}{A} \nconv \PiType{a}{K'}{D}$ and it follows from the properties of neg-conversion that $K=K'$ and $A \conv D$, as required.
    \item[\rlnm{$\forall$-Intro-R}] In this case, $B$ is of shape $\PiType{a}{K}{B'}$ and we aim to show (3).  By inspection, $\Gamma,\,a:K \types M : B';\,\Delta$.  Suppose $B \nconv \PiType{a}{K'}{C}$.  Then it follows from the properties of neg-conversion that $K=K'$ and $B' \conv C$.  Therefore, also $\Gamma,\,a:K \types M : C;\,\Delta$, as required.
    \item[\rlnm{$\forall$-Elim-L},\rlnm{$\forall$-Elim-R}] The result holds vacuously.
    \item[\rlnm{$\to$-Intro-L}] In this case, $B$ is of shape $A_1 \to A_2$ and the premises give $\Gamma \types M :A_1;\,\Delta$ and $\Gamma;\, N:A_2 \types \Delta$.  By the properties of neg-conversion, $B_1 \conv A_1$ and $B_2 \conv A_2$.  Therefore, the result follows by \rlnm{Conv} (\rlnm{$\sneg{}$-Intro-R} and \rlnm{$\sneg{}$-Elim-R}).
    \item[\rlnm{$\to$-Elim-L$_i$},\rlnm{$\to$-Elim-R}] The result holds vacuously.
    \item[\rlnm{$\to$-Intro-R}] In this case, $B$ is of shape $A \to A'$ and the premise gives $\Gamma,x:A \types M:A';\,\Delta$.  By the properties of neg-conversion, $B_1 \conv A$ and $B_2 \conv A'$.  Therefore, the result follows by \rlnm{Conv}.
    \item[\rlnm{$\sneg{}$-Intro-L}] In this case, each of the clauses (2), (4), (5) and (7) follow\camrevadd{s} from the induction hypothesis clauses (1), (3), (6) and (8) respectively.
    \item[\rlnm{$\sneg{}$-Intro-R}] In this case, each of the clauses (1), (3), (6) and (8) follow\camrevadd{s} from the induction hypothesis clauses (2), (4), (5) and (7) respectively.
    \item[\rlnm{$\sneg{}$-Elim-L}]  In this case, each of the clauses (1), (3), (6) and (8) follow\camrevadd{s} from the induction hypothesis clauses (2), (4), (5) and (7) respectively.
    \item[\rlnm{$\sneg{}$-Elim-R}]  In this case, each of the clauses (2), (4), (5) and (7) follow\camrevadd{s} from the induction hypothesis clauses (1), (3), (6) and (8) respectively.
  \end{description}
\end{proof}

\begin{ulemma}[Restatement of Canonical Proofs and Refutations, Lemma~\ref{lem:2hol-can-forms}]
  Suppose $\Gamma$ is a kinding environment \camrevadd{$\Delta=\varnothing$, } and $N$ is in normal form.
  \begin{enumerate}[(i)]
    \item If $\Gamma \types N : A;\,\Delta$ then:
          \begin{enumerate}[(a)]
            \item If $A$ is neg-convertible with an arrow then $N$ is a term abstraction.
            \item If $A$ is neg-convertible with a universal type over $K$ then $N$ is a type abstraction over $K$.
            \item If $A$ is neg-convertible with the strong negation of an arrow, then $N$ is a dot construction.
            \item If $A$ is neg-convertible with the strong negation of a universal type over $K$, then $N$ is a $K$-pack.
          \end{enumerate}
    \item If $\Gamma;\,N:A \types \Delta$ then:
          \begin{enumerate}[(a)]
            \item If $A$ is neg-convertible with an arrow then $N$ is a dot construction.
            \item If $A$ is neg-convertible with a universal type over $K$ then $N$ is a $K$-pack.
            \item If $A$ is neg-convertible with the strong negation of an arrow, then $N$ is a term abstraction.
            \item If $A$ is neg-convertible with the strong negation of a universal type over $K$, then $N$ is a type abstraction over $K$.
          \end{enumerate}
  \end{enumerate}
\end{ulemma}
\begin{proof}
  The proof is by induction on the derivation relation.
  \begin{description}
    \item[\rlnm{Id-L},\rlnm{Id-R}] These cases cannot occur since $\Gamma$ is a kinding environment \camrevadd{and $\Delta=\varnothing$}.
    \item[\rlnm{Conv}] In this case, each case of (i) follows immediately from the induction hypothesis and (ii) cannot occur.
    \item[\rlnm{$\forall$-Intro-L}] In this case, (ii)(b) follows immediately.  By the properties of neg-conversion, no other subcase can occur.
    \item[\rlnm{$\forall$-Intro-R}] In this case, (i)(b) follows immediately.   By the properties of neg-conversion, no other subcase can occur.
    \item[\rlnm{$\forall$-Elim-L}] In this case, $N$ has shape $\unpack_K\,[\tyabs{a}{K}{\camrevreplace{B}{\sneg{B}}}]\,M\,[C]\,(\Tyabs{a}{K}{\tyabs{x}{B}{N'}})$ and it follows from the induction hypothesis that $M$ is a $K$-pack, contradicting the assumption that $N$ is a normal form.
    \item[\rlnm{$\forall$-Elim-R}] In this case, $N$ has shape $M\,[A']$ and it follows from the induction hypothesis that $M$ is a type abstraction, contradicting the assumption that $N$ is in normal form.
    \item[\rlnm{$\to$-Intro-L}] In this case, (ii)(a) follows immediately.  By the properties of neg-conversion, no other subcase can occur.
    \item[\rlnm{$\to$-Intro-R}] In this case, (i)(a) follows immediately.  By the properties of neg-conversion, no other subcase can occur.
    \item[\rlnm{$\to$-Elim-L$_1$}] In this case, $N$ is of shape $\iota_1\,M$ and it follows from the induction hypothesis that $M$ is a dot construction, contradicting the assumption that $N$ is in normal form.
    \item[\rlnm{$\to$-Elim-L$_2$}] In this case, $N$ is of shape $\iota_2\,M$ and it follows from the induction hypothesis that $M$ is a dot construction, contradicting the assumption that $N$ is in normal form.
    \item[\rlnm{$\to$-Elim-R}] In this case, $N$ is of shape $M\,N'$ and it follows from the induction hypothesis that $M$ is an abstraction, contradicting the assumption that $N$ is in normal form.
    \item[\rlnm{$\sneg{}$-Intro-L}] In this case, $A$ is of shape $\sneg{A'}$.  If $A$ is neg-convertible with an arrow, then it must be that $A'$ is neg-convertible with the strong negation of an arrow, and then we obtain (ii)(a) by induction hypothesis (i)(c).  If $A$ is neg-convertible with a universal type over $K$, then $A'$ must be neg-convertible with the strong negation of a universal type over $K$, and then we obtain (ii)(b) by induction hypothesis (i)(d).  If $A$ is neg-convertible with the strong negation of an arrow, it must be that $A'$ is neg-convertible with an arrow, then (ii)(c) follows from induction hypothesis (i)(a).  If $A$ is neg-convertible with the strong negation of a universal over $K$, then it must be that $A'$ is neg-convertible with a universal over $K$ and (ii)(d) follows from induction hypothesis (i)(b).  By the properties of neg-conversion, no other subcase can occur.
    \item[\rlnm{$\sneg{}$-Intro-R}] Symmetrical to the above.
    \item[\rlnm{$\sneg{}$-Elim-L}] There are four subcases but all are similar so we only treat the first.  If $A$ is neg-convertible with an arrow, then $\sneg{A}$ is neg-convertible with the strong negation of an arrow and so (i)(a) follows from induction hypothesis (ii)(c).
    \item[\rlnm{$\sneg{}$-Elim-R}] Symmetrical to the above.
  \end{description}
\end{proof}

\begin{lemma}[Type Substitution]\label{lem:2hol-ty-subst}
  Both of the following:
  \begin{enumerate}[(i)]
    \item If $\Gamma_1,\,a\ty:K,\,\Gamma_2;\,M:A \types \Delta$ and $\Gamma_1 \types B : K$ then $\Gamma_1,\,\Gamma_2[B/a];\,M[B/a]:A[B/a] \types \Delta[B/a]$.
    \item If $\Gamma_1,\,a\ty:K,\,\Gamma_2 \types M:A;\, \Delta$ and $\Gamma_1 \types B : K$ then $\Gamma_1,\,\Gamma_2[B/a] \types M[B/a]:A[B/a];\, \Delta[B/a]$.
  \end{enumerate}
\end{lemma}
\begin{proof}
  First note that, by a standard argument on simply typed $\lambda$-calculus, if $\Gamma,\,a:K_1 \types A:K_2$ and $\Gamma \types B:K_1$, then $\Gamma \types A[B/a] : K_2$.  The proof proper is by induction on the derivations, but the majority of cases follow immediately from the induction hypothesis and the definition of substitution, so we just give the more interesting cases below.
  \begin{description}
    \item[\rlnm{Id-L},\rlnm{Id-R}] Both cases are symmetrical, so we consider the first.  In this case, $M$ is some variable $x$ and $x:A \in \Delta$. We have $\types \Gamma_1,a:K,\Gamma_2,\Delta^R$.   Hence, by the noted property of kinding above, since we assume $\Gamma_1 \types B:K$, we obtain by induction $\types \Gamma_1,\Gamma_2[B/a],(\Delta[B/a])^R$.   Therefore, $\Gamma_1, \Gamma_2[B/a]; x:A[B/a] \types \Delta[B/a]$ follows again from \rlnm{Id-L}.
    \item[\rlnm{Conv}] In this case, we have $A \conv B'$ for some $B'$.  It follows from the properties of convertibility that, therefore, $A[B/a] \conv B'[B/a]$.  Therefore, the result follows from the induction hypothesis.
    \item[\rlnm{$\forall$-Elim-L}] In this case, $M$ is of shape $\unpack_{\revreplace{K}{K'}}\,[\tyabs{a'}{K'}{B'}]\,M'\,[A]\,(\Tyabs{a'}{K'}{\tyabs{x}{\sneg{B'}}{N}})$.  We may assume, by the variable convention, that $a \neq a'$\revadd{ and $a'$ is fresh for $B$, $\Gamma_1$, $\Gamma_2$ and $\Delta$}.  It follows from the induction hypothesis that both of the following are derivable:
          \begin{enumerate}[(i)]
            \item $\Gamma_1,\Gamma_2[B/a];\,M'\revadd{[B/a]} : \PiType{a'}{K'}{B'[B/a]} \types \Delta[B/a]$
            \item $\Gamma_1,\Gamma_2[B/a],a'\ty:K';\,N\revadd{[B/a]} : A[B/a] \types x\ty:B'[B/a],\,\Delta[B/a]$
          \end{enumerate}
          Therefore, we can immediately conclude, using \rlnm{$\forall$-Elim-L}:
          \[
            \begin{aligned}
              \Gamma_1,\Gamma_2[B/a];\; & \unpack_{\revreplace{K}{K'}}\,[\tyabs{a'}{K'}{B'[B/a]}]\,M'\revadd{[B/a]}\,[A[B/a]]           \\
                                        & \quad(\Tyabs{a'}{K'}{\tyabs{x}{\sneg{B'[B/a]}}{N\revadd{[B/a]}}}) : A[B/a] \types \Delta[B/a]
            \end{aligned}
          \]
  \end{description}

\end{proof}

\begin{lemma}[Term Substitution]\label{lem:2hol-term-subst}
  All of the following:
  \begin{enumerate}[(i)]
    \item If $\Gamma_1,\,x\ty:A,\,\Gamma_2;\,M:B \types \Delta$ and $\Gamma_1 \types N:A;\,\Delta$ then $\Gamma_1,\Gamma_2;\,M[N/x] : B \types \Delta$.
    \item If $\Gamma_1,\,x\ty:A,\,\Gamma_2 \types M:B;\, \Delta$ and $\Gamma_1 \types N:A;\,\Delta$ then $\Gamma_1,\Gamma_2 \types M[N/x] : B;\, \Delta$.
    \item If $\Gamma;\,M:B \types \Delta_2,x\ty:A,\Delta_1$ and $\Gamma;\,N:A \types \Delta_1$ then $\Gamma;\,M[N/x] : B \types \Delta_2,\Delta_1$.
    \item If $\Gamma \types M:B;\,\Delta_2,x\ty:A,\Delta_1$ and $\Gamma;\,N:A \types \Delta_1$ then $\Gamma \types M[N/x] : B;\, \Delta_2,\Delta_1$.
  \end{enumerate}
\end{lemma}
\begin{proof}
  Clauses (i),(ii) are symmetrical with (iii),(iv) so we only do the former.  The proof is by induction on the derivation, but most of the cases follow immediately from the induction hypothesis, so we only show the more interesting cases below.
  \begin{description}
    \item[\rlnm{Id-L}] In this case, we aim to show (i).  We have that $M$ is a variable $x'$ and $x'\ty:B \in \Delta$.  It follows, therefore, from the premise $\types \Gamma,\Delta^R$ that $x \neq x'$.  Hence, $x'[N/x] = \camrevreplace{x}{x'}$ and so the result follows again from \rlnm{Id-L}.
    \item[\rlnm{Id-R}] In this case, we aim to show (ii).  We have that $M$ is a variable $x'$ and $x'\ty:B \in \Gamma$ and $\Gamma = \Gamma_1,x:A,\Gamma_2$.  If $x=x'$ then, by the premise $\types \Gamma,\Delta^R$ it follows that $A=B$.  By definition, $x'[N/x'] = N$ and so we have $\Gamma_1,\Gamma_2 \types N:A;\,\Delta$ by weakening the assumption.  If $x \neq x'$ then, $x'[N/x] = x'$ and so the result follows immediately by \rlnm{Id-R}.
    \item[\rlnm{$\sneg{}$-Intro-L}] In this case, we aim to show (i).  We have that $B$ is of shape $\sneg{C}$.  It follows from the induction hypothesis, case (ii), that $\Gamma_1,\Gamma_2 \types M[N/x] : C;\,\Delta$ and so the result follows from \rlnm{$\sneg{}$-Intro-L}.
    \item[\rlnm{$\sneg{}$-Intro-R}] In this case, we aim to show (ii).  We have that $B$ is of shape $\sneg{C}$.  It follows from the induction hypothesis, case (i), that $\Gamma_1,\Gamma_2;\,M[N/x] : C \types \Delta$ and so the result follows from \rlnm{$\sneg{}$-Intro-R}.
  \end{description}
\end{proof}

\begin{utheorem}[Restatement of Preservation, Theorem~\ref{thm:hol2-preservation}]
  Both of the following:
  \begin{itemize}
    \item If $\Gamma;\,M : A \types \Delta$ and $M \ped N$ then $\Gamma;\,N:A \types \Delta$
    \item If $\Gamma \types M : A;\,\Delta$ and $M \ped N$ then $\Gamma \types N:A;\,\Delta$
  \end{itemize}
\end{utheorem}
\begin{proof}
  We prove both parts simultaneously by induction on the derivations.
  \begin{description}
    \item[\rlnm{Id-L},\rlnm{Id-R}] Subject cannot reduce.
    \item[\rlnm{Conv}] The result follows immediately from the induction hypothesis.
    \item[\rlnm{$\forall$-Intro-L}] In this case, $M$ is of shape $\pack_K\,[\tyabs{a}{K}{B}]\,[A']\,N$ and so, the only reduction may be via a reduction step of $N$.  In this case, the result follows immediately from the induction hypothesis.
    \item[\rlnm{$\forall$-Intro-R}] In this case, $M$ is of shape $\Tyabs{a}{K}{M'}$ and the only reduction may be via a reduction step of $M'$.  Therefore, the result follows from the induction hypothesis.
    \item[\rlnm{$\forall$-Elim-L}] In this case, $M$ is of shape $\unpack_K\,[\tyabs{a}{K}{B}]\,M'\,[A]\,(\Tyabs{a}{K}{\tyabs{x}{\revreplace{B}{\sneg{B}}}{N}})$.  If the reduction step is via $M'$ or $N$, in both cases the result follows from the induction hypotheses.  Otherwise, the whole term may be a redex when $M'$ is of shape $\pack_K\,[\tyabs{a}{K}{D}]\,[E]\,N'$.  By inspection of the rule, $\Gamma,\,a\ty:K;\,N:A \types x\ty:B,\,\Delta$. By Inversion (Lemma~\ref{lem:2hol-inversion}(7)), $D \conv B$ and $\Gamma \types E : K$ and $\Gamma;\;N' : D[E/a] \types \Delta$.  By the compatibility of conversion, therefore $\Gamma;\;N':B[E/a] \types \Delta$.  Hence, it follows from the type substitution lemma (Lemma~\ref{lem:2hol-ty-subst}) and the freshness of $a$ that $\Gamma;\,N[E/a]:A \types x\ty:B[E/a]\camrevadd{,}\,\Delta$.  Then it follows from the term substitution lemma (Lemma~\ref{lem:2hol-term-subst}) and the freshness of $x$ that $\Gamma;\,N[E/a][N'/x] : A \types \Delta$, as required.
    \item[\rlnm{$\forall$-Elim-R}] In this case, $M$ is of shape $M'\,[A']$ and $A$ is of shape $B[A'/a]$.  We may assume, by the variable convention, that $a$ does not occur in $\Gamma$ or $\Delta$.  In case the reduction is a step inside $\revreplace{M}{M'}$, the result follows from the induction hypothesis.  Otherwise, it must be that $M'$ is a type abstraction $\Tyabs{a}{K}{Q}$ and $N$ is $Q[A'/a]$.  It follows from Inversion (Lemma~\ref{lem:2hol-inversion}) and the first premise that $\Gamma,a\ty:K \types Q : B$.  By the type substitution lemma (Lemma~\ref{lem:2hol-ty-subst}(ii)) and the second premise we obtain $\Gamma \types Q[A'/a] : B[A'/a];\,\Delta[A'/a]$.  However, $a$ was assumed not to occur in $\Delta$ and so the result follows.
    \item[\rlnm{$\to$-Intro-L}] In this case, $M$ is of shape $M' \cdot_{A,B} N'$ and the only reduction step may have occurred in either $M'$ or $N'$ separately.  In both cases, the result follows from the induction hypotheses.
    \item[\rlnm{$\to$-Intro-R}] In this case, $M$ is of shape $\tyabs{x}{A'}{M'}$ and $A$ of shape $A' \to B$.  The reduction step must have occurred inside $M'$ and so the result follows from the induction hypothesis.
    \item[\rlnm{$\to$-Elim-L$_i$}] Both subcases are similar, so we treat $i=1$.  Here, $M$ is of shape $\iota_1\,[A_1]\,[A_2]\,M'$ and $A$ is $A_1$.  If the reduction step was internal to $M'$ then the result follows from the induction hypothesis.  Otherwise, $M'$ must itself be of shape $M_1 \cdot_{B_1,B_2} M_2$ and $N$ therefore $M_1$.  By inversion (Lemma~\ref{lem:2hol-inversion}), and the first premise, $A_1 \conv B_1$, $A_2 \conv B_2$ and $\Gamma \types M_1 : A_1;\,\Delta$, as required.
    \item[\rlnm{$\to$-Elim-R}] In this case, $M$ is of shape $M'\,N'$ and the first premise of shape $\Gamma \types M' : A' \to A\camrevadd{;\,\Delta}$.  If the reduction was internal to $M'$ or $N'$, then the result follows from the induction hypothesis.  Otherwise, $M'$ must be of shape $\tyabs{x}{A''}{P}$ and $N$ is therefore $P[N'/x]$.  It follows from the first premise and inversion (Lemma~\ref{lem:2hol-inversion}) that $A'' \conv A'$ and $\Gamma,x\ty:A'' \types P : A;\,\Delta$.  By the second premise and \rlnm{Conv}, $\Gamma \types N' : A'';\,\Delta$.  Then the result follows by substitution (Lemma~\ref{lem:2hol-term-subst}(ii)).
    \item[\rlnm{$\sneg{}$-Intro-L},\rlnm{$\sneg{}$-Intro-R},\rlnm{$\sneg{}$-Elim-L},\rlnm{$\sneg{}$-Elim-R}] In these cases, the result follows from the induction hypothesis.
  \end{description}
\end{proof}
}{}

\end{document}
\endinput